\documentclass[a4paper,10pt]{article}
\usepackage{wrapfig}
\usepackage{amsmath,amsthm,amssymb}
\usepackage{enumerate}
\usepackage{hyperref}

\DeclareMathAlphabet{\mathcall}{T1}{calligra}{m}{n}

\usepackage{graphicx}
\usepackage{epsfig}
\usepackage{todonotes}
\usepackage{color,tikz}
\usetikzlibrary{patterns, decorations.markings,arrows, calc}
\usepackage{lipsum}
\usepackage[toc, page]{appendix}

\usepackage{amsfonts}
\definecolor{darkblue}{rgb}{0.1,0.1,0.45}
\usepackage{ulem}

\usepackage{epsfig}

\usepackage{epstopdf}

\usetikzlibrary{arrows, decorations.markings}

\usepackage{color}
\usepackage{framed}
\newtheorem*{prop*}{Proposition}

\newtheorem{RHP}{Riemann-Hilbert problem}
\newtheorem{col*}{Corollary}%[section]

\newcommand{\e}{\mathrm{e}}
\newcommand{\ii}{\mathrm{i}}
\renewcommand{\i}{\mathrm{i}}
\newcommand{\Ai}{\mathrm{Ai}}

\renewcommand{\d}{\mathrm{d}}
\renewcommand{\Im}{\operatorname{Im}}
\renewcommand{\(}{\left(}
\renewcommand{\)}{\right)}

\newcommand{\ol}{\overline}

\newcommand{\C}{\mathbb{C}}
\newcommand{\R}{\mathbb{R}}

\providecommand{\C}{\mathcal} \providecommand{\Db}{\mathbb}

\renewcommand{\Re}{\operatorname{Re}}

\renewcommand{\l}{\mathfrak{l}}
\renewcommand{\r}{\mathfrak{r}}

\newcommand{\widebar}{\overline}

\renewcommand{\l}{-}
\renewcommand{\r}{+}

 \newtheorem{theorem}{Theorem}[section]
    \newtheorem{corollary}[theorem]{Corollary}
    \newtheorem{proposition}[theorem]{Proposition}
    
    \newtheorem{Definition}[theorem]{Definition}
    
    \newtheorem{Remark}[theorem]{Remark}
    \newenvironment{remark}{\begin{Remark}\rm}{\end{Remark}}
    \newtheorem{Example}[theorem]{Example}
    
    \newtheorem{Assumption}[theorem]{Assumption}

\newtheorem{lemma}[theorem]{Lemma}

\numberwithin{equation}{section}
\begin{document}

%\title{On the form of a modulated (hyper)-elliptic wave for MKdV (compression wave)}
\title{On the long time asymptotic behaviour  of the modified Korteweg de Vries equation  with step-like initial data}
\author{{\normalsize Tamara \textsc{Grava}$^\dagger$ and \normalsize  Alexander  \textsc{Minakov}$^\ddagger$}\\[1mm]
{\scriptsize{$^{\dagger}$SISSA, via Bonomea 265, 34136 Trieste, Italy and School of Mathematics, University of Bristol, UK}}\\
{\scriptsize $^{\ddagger}$Institut de Recherche en Math\'{e}matique et Physique (IRMP), Universit\'{e} Catholique de Louvain
(UCL), Louvain-la-Neuve, Belgium}\\
{\scriptsize{grava@sissa.it \qquad oleksandr.minakov@uclouvain.be}}\\
}
\date{}
\maketitle

\begin{abstract}
We study the long time asymptotic behaviour of  the  solution  $q(x,t) $, $x\in\mathbb{R}$, $t\in\mathbb{R}^+$,  of the modified Korteweg de Vries equation  (MKdV)  $q_t+6q^2q_x+q_{xxx}=0$
  with step-like initial datum  
\begin{equation*}
q(x,0)\to
%\left\{\begin{array}{ll}
\begin{cases}
c_-\quad \mbox{ for $x\to-\infty$}\\
c_+\quad \mbox{ for $x\to +\infty$},
\end{cases}
%\end{array}\right.
\end{equation*}
with $c_->c_+\geq 0$.
For the  step %shock
 initial data
\begin{equation*}
q(x,0)=
%\left\{\begin{array}{ll}
\begin{cases}
c_-\quad \mbox{ for $x\leq0$},\\
c_+\quad \mbox{ for $x>0$},
\end{cases}
%\end{array}\right.
\end{equation*}
the solution develops an oscillatory region  called dispersive shock wave region that   connects the two constant regions $c_+$ and $c_-$.
We show that the dispersive shock wave  is described by a modulated periodic travelling wave solution of the  MKdV  equation where the modulation parameters evolve
according to \color{black}{a} Whitham modulation equation.
%The periodic travelling wave solution $q_{periodic}(x,t; \beta_1,\beta_2,\beta_3,x_0)$ is parametrized by three constants $\beta_1>\beta_2>\beta_3$ and a phase shift $x_0$.
%The dispersive shock wave is asymptotically described by $\beta_1=c_-$, $\beta_3=c_+$, $\beta_2=\beta_2(x,t)$ and   $x_0=x_0(x,t)$,  where  the dependence on time  of $\beta_2(x,t)$
% satisfies the Whitham modulation equations. 
 The oscillatory region is expanding within a cone in the $(x,t)$  plane defined as 
$ -6c_{-}^2+12c_{+}^2<\frac{x}{t}<4c_{-}^2+2c_{+}^2,$ with $t\gg 1$.
  
For step-like initial data we show that the solution decomposes for long times in three main regions:
\begin{itemize}
\item a   region where solitons  and  breathers   travel with positive velocities on a constant background $c_+$;
\item an expanding oscillatory region   {\color{black} (that generically  contains breathers)};
\item a region  of breathers travelling with negative velocities on the constant background $c_-$. 
\end{itemize}
When the oscillatory region does not contain  breathers, the form of the asymptotic  solution coincides up to a phase shift with the dispersive shock wave solution obtained for the step initial data.
The phase shift depends on  the solitons, breathers and the radiation of the initial data.
This shows that   the dispersive shock wave is a coherent structure that interacts in an elastic way with solitons,  breathers and radiation.
\end{abstract}

%\begin{keywords}
%Integrable system, Riemann-Hilbert problem, dispersive shock waves, long time asymptotic analysis.
%\end{keywords}
%
%% REQUIRED
%\begin{AMS}
%35Q15, 35Q51, 35Q53
%\end{AMS}

\section{Introduction}
We consider the   Cauchy problem for the focusing modified Korteweg--de Vries  (MKdV) equation 
\begin{equation}\label{MKdV}
q_t(x,t)+6q^2(x,t)q_x(x,t)+q_{xxx}(x,t)=0,\quad x\in\mathbb{R},\;t\in\mathbb{R}^+,
\end{equation}
with a step--like initial datum $q_0(x)$  of the form
\begin{equation}\label{ic0}
q(x,0)=
q_0(x)\to c_{\pm}\qquad {\rm as }\quad x\to\pm\infty,
\end{equation}
where $c_{\pm}$ are some  real  constants.
We are interested in the long-time behavior of the solution.

Due to symmetries $q\mapsto-q$ and $x\mapsto-x, t\to-t,$ it is enough to consider the case
\begin{equation}\label{ic_assump}
 c_{-}\geq |c_{+}|,
\end{equation}
\color{black} since the other cases of mutual location of the constants $c_-, c_+$ can be reduced to it; however, the cases $c_->c_+\geq0$ and $c_->0>c_+\geq-c_-$  though qualitatively similar, leads to a quite different asymptotic analysis.
In the present manuscript   we restrict ourselves to the case $c_->c_+\geq0$  and  we discuss  briefly the differences with respect to the case 
$c_+<0,$ $|c_+|\leq c_-$  in the Appendix~\ref{WhithamApp}.

The   focusing MKdV equation  is a canonical model  for  the  description  of  nonlinear  long  waves  when there is a polarity symmetry, and it has many
physical applications; \color{black}{ in particular, } this includes waves in a quantized film \cite{PS}  internal ocean waves \cite{Grimshaw}, ion acoustic waves in a two component plasma \cite{RTP}.
\color{black}
The MKdV equation  is an integrable equation  \cite{Wadati} with an infinite number of conserved quantities.
For the  class of initial data  considered, the classical mass and momentum   have to be replaced by the conserved quantities
  \begin{equation*}
  \begin{split}&
  H_0 = \int\limits_{-\infty} ^ {x} (q(\widetilde x, t)-c_-)\d\widetilde x + \int\limits_{x}^{+\infty} (q(\widetilde x, t) - c_+) \d \widetilde x +(c_--c_+)x-2(c_-^3-c_+^3)t,
  \\&
    H_1 = \int\limits_{-\infty} ^ {x} (q^2(\widetilde x, t)-c^2_-)\d\widetilde x + \int\limits_{x}^{+\infty} (q^2(\widetilde x, t) - c^2_+) \d \widetilde x +(c^2_--c^2_+)x-4(c_-^4-c_+^4)t.
    \end{split}
  \end{equation*}

% \todo{Maybe add some other applications?}
  
%Our goal is to show that for generic initial data  satisfying the conditions
%\begin{equation*}
%\end{equation*}
%the initial value problem evolves for long times into a set of  solitons and breathers on a constant background and a dispersive  shock wave.  The dispersive shock wave takes the form
%of a modulated travelling wave 
%\begin{equation*}
%q(x,t)=
%\end{equation*}
%
%
%Such a problem was considered in El, Marchant, Leach, KM.
%

%In particular we consider the case $c_{+}>0$
%Furthermore do to the scaling invariance 
%\begin{equation*}
%q(x,t)\longrightarrow \lambda q(\lambda x,\lambda^3t),
%\end{equation*}
%which is generated by the vector field $1+x\partial_x+3t\partial t$,
%it is sufficient to consider step-like initial data of the form
%\begin{equation}\label{ic}
%q_0(x)\to\begin{cases}c\qquad {\rm as}\quad
%x\to+\infty,\;\;0\leq c<1\\1\qquad {\rm as }\quad x\to-\infty.\end{cases}
%\end{equation}
The study of the long-time asymptotic behaviour  of integrable  dispersive equations with   initial datum vanishing at infinity 
 was initiated in the mid-seventies  using the inverse scattering \color{black}{method} in the works of  Ablowitz and Segur \cite{AS} and Manakov and Zakharov \cite{MZ}. In the seminal paper \cite{DZ93} Deift and Zhou  
  introduced   the  steepest descent  method  for oscillatory  Riemann-Hilbert (RH) problems  to study the long-time asymptotic behaviour of the defocusing MKdV   equation with initial data vanishing at infinity. Such technique was extensively implemented in the asymptotic analysis of a wide variety of integrable problems  ( see e.g. \cite{DIZ}) \color{black}{(which in turn can be applied to some near-integrable cases, like the long-time behaviour of the perturbed defocusing nonlinear Schr\"odinger equation \cite{DZnonint}).}
  An extension of the  steepest descent  method  for oscillatory RH problems, called $\bar{\partial}$ method, was introduced in \cite{MM1}  and applied  to study the long-time behaviour of integrable dispersive equations with   initial data with low regularity  \cite{BJM, CL, CL2, DMM} in the strongly nonlinear regime.  In particular, for the MKdV equation  there is a vast body of literature studying   existence of solution
  for initial data with low regularity (see e.g. \cite{Linares}). 
  Regarding the weakly nonlinear regime, the long-time asymptotic behaviour with small initial data   is quite similar for the focusing and defocusing MKdV equation and it was also  obtained without using  the integrability property 
   in \cite{Germain}, \cite{Harrop} and  \cite{Hayashi}.
  In the last years a vast literature of results concerns the long-time dynamics of initial boundary value problems of nonlinear dispersive equations. For a review see \cite{BoutetKotlyarovShepelskyZheng}.

  %In particular for the MKdV equation there is a vast body of literature studying local and global  well-posedness of the Cauchy problem with initial data in certain Sobolev space and for a summuray of known result
  %we refer the reader to   \cite{Linares}.  Furthemore, for   
%The main feature of the long-time asymptotic description of the solution  of integrable dispersive equations with initial data that decay at infinity is that the solution  decomposes for a large value of time in a soliton region (when solitons exists) and  a decaying radiation region consisting  of several sub-regions with different  decaying rates.
%
The first results on the long time asymptotic analysis of Cauchy problems  with step like initial data were obtained for the Korteweg-de Vries (KdV) equation. Physicists have begun to understand the qualitative behaviour of the solution with the pioneering work of Gurevich and Pitaevsky \cite{GP}, who, working in the framework of Whitham theory \cite{W}, predicted  the  appearance of high oscillations  called "dispersive shock waves".  These oscillations were described by modulated travelling waves.
This phenomenon was justified rigorously in the pioneering work of Khruslov \cite{Kh2} who, working in the framework of the inverse scattering theory, obtained formulas for the first finite number of peaks  of the oscillations (which were called {\it asymptotic} solitons to distinguish them from  the usual solitons). This approach was extended to many other integrable models \cite{KK}.
Note that for a long time it was believed that Khruslov's solitons can be obtained from the Gurevich-Pitaevsky dispersive shock wave, until it was shown in \cite{KM2015, BM} that dispersive shock waves (expressed in terms of elliptic functions) do not describe the asymptotic behaviour  of the wave in Khruslov's region, and the full matching of the two regions was obtained in \cite{BM}  for MKdV and in \cite{CG} for KdV.
Using ansatz for solutions of RH problems, Bikbaev \cite{Bikb1} obtained interesting results for step-like quasi-periodic initial data.
 The KdV dispersive shock wave solution   emerges also in the small dispersion limit \cite{DVZ}. In particular, for the exact step initial data, the long time asymptotic and the small dispersion asymptotic description are equivalent, while for step like initial data the two asymptotic descriptions are   quite different. Indeed the dispersive shock wave obtained from  the long time asymptotic limit is always described  by the self-similar solution of the Whitham modulation equations \cite{W}, while in the small dispersion limit  this is not generically the case  (see e.g. \cite{Grava_LN}, \cite{GT}).
 
 Implementation of the rigorous asymptotic analysis to step-like Cauchy problems for integrable equations started in the papers \cite{BV, BIK}.
 Since then,  the   long-time asymptotic behaviour of dispersive equations with {\it step-like}  initial conditions has been studied for KdV in  \cite{EGKT, EGT16},  for the nonlinear Schr\"odinger equation in \cite{BiM, BKS11, BV, RS20, MLS20,Jenkins, CJ}, for the Camassa-Holm equation in \cite{M16}.
For the MKdV equation the analysis was initiated in the work \cite{KK} and  later  in \cite{KM, KM2, KM2015, BM} via the asymptotic analysis of the RH problem,  in \cite{Leach, Marchant} via matching ansatz method and in \cite{EHS} via the Whitham method.  
%The long time asymptotic of the Cauchy problem for the MKdV equation with steplike initial condition was initiated
%by Kotlyarov and Krhuslov \cite{KK} where they showed via inverse scattering, the existence of a train of asymptotic solitons emerging from steplike initial conditions. The rigorous asymptotic analysis was developed through steepest descent for oscillatory Riemann-Hilbert  problems
%by Kotlyarov and the second author  of the present manuscript  \cite{} where different step-like initial conditions were considered. 
The main feature in the long-time behaviour that distinguish  step-like initial conditions from decaying initial conditions is the formation of an oscillatory region that connects the  different behaviour at $\pm\infty$ of the solution.   These oscillatory regions are typically described by elliptic or hyperelliptic   modulated waves. 

%The study of  the long-time behaviour of  nonlinear dispersive equations with shock initial conditions,  was initiated in the papers of Gurevich, Pitaevsky   \cite{GP}  and Khruslov \cite{Kh2}, where they considered the long-time behaviour of  the Korteweg-de Vries equation (KdV) with a step initial condition. 
%Dispersive shock waves have  been widely used in applications, such as the resonant flow of a stratified fluid over topography, see \cite{GS, Smyth} for a review see \cite{}.
%

%startChanges
%\color{black}
%Implementation of the rigorous RH problem scheme to step-like problems started from the papers \cite{BV}, \cite{BIK} \cite{, where the authors succeeded in constructing of appropriate $g$-functions. The purpose of that $g$-function is to replace in asymptotic analysis the phase of a linearized equation solved via Fourier transform.
%
%
%\todo{}

%With this respect it is remarkable that an elliptic asymptotics was found even in the case of decaying initial conditions, namely in the
%(collisionless) dispersive shock wave region that appears in the long-time asymptotics for KdV \cite{AS, DVZ}, but in this case the amplitude of the oscillations is decaying in time.

\color{black}

The scattering problem  for MKdV  with  non vanishing initial condition was developed in \cite{KM}, \cite{KM2}, \cite{AK91}.
The linear spectral problem  is a non self-adjoint problem  and for  a step-like  initial data  $q_0(x)$ as in \eqref{ic0} and satisfying certain assumption (see below) the Zakharov-Shabat or AKNS operator for \eqref{MKdV}
has a continuous spectrum $r(k):\Sigma\to \C$ where $\Sigma=\mathbb{R}\cup[-\i c_-,\i c_-]$ and  {\color{black}  generically  it might have   discrete spectrum anywhere in  $\mathbb{C}\backslash \{\mathbb{R}\cup [-\i c_-,\i c_-]\}$, }
that corresponds to the zeros of $a(k)$, the inverse of the transmission coefficient.
  Pure imaginary  couples of conjugated eigenvalues correspond to solitons, while quadruplets of complex conjugated eigenvalues correspond to breathers  \cite{Wadati}. A breather is a solution that is 
periodic in the time variable and decay exponentially in the space variable. 
 Unlike the KdV equation the MKdV equation can have higher order solitons and breathers.  In this manuscript we consider the  generic case when only first order solitons and breathers appear. Nongeneric cases are considered in the Appendix~\ref{sect_nongen}.
  First order solitons and breathers are    the  fundamental localised non radiating solutions   of the MKdV equation.
Since the MKdV equation is not Galilean invariant,    solitary wave solutions  and breathers on a constant background $c>0$  cannot be  mapped to solutions on zero background. 

Our main result   contained in Theorem~\ref{thrm:asymp:rl} below,  is to show that the long-time asymptotic solution of the MKdV equation  with step-like initial data  of the form \eqref{ic0} with $c_->c_+\geq 0$ decomposes into three main regions:
\begin{itemize}
\item a region of  solitons  and  breathers  on a constant background $c_{+}$  travelling in the positive direction;
\item a  dispersive shock wave region, which connects the two different asymptotic behaviours of the initial data and interact elastically with breathers  and solitons.
 This  region is  described by a modulated travelling wave solution of  MKdV, or  by  a modulated travelling wave solution and breathers on an elliptic background.
 \item a  region  of breathers  on a constant background $c_{-}$  {\color{black} travelling with a slower speed with respect to the dispersive shock wave}. This region contains also radiation decaying in time.

\end{itemize}

The localised travelling wave   solution on a constant background  $c>0$ is parametrised by two  real  constants  $\nu$ and  $\kappa_0>c$, where \color{black} the points $\pm\i \kappa_0$  constitute the \color{black} discrete spectrum of the Zakharov- Shabat
 linear operator, (namely they are the simple zeros of $a(k)$) \color{black}{ and $\nu$ is the corresponding norming constant}, and  takes the form
%The MKdV equation admits a family of solitary wave solution of the form
%\begin{equation*}
%q(x,t)=Q_{v}(x-vt),\quad Q_v(\xi)=\pm\dfrac{\sqrt{v}}{\cosh(\sqrt{v}\xi)}.
%\end{equation*}
%The MKdV equation is not Galilean invariant, and it admits  solitary wave solutions on a constant background $c$ of the form
\begin{equation}\label{soliton_intro}
q_{soliton}(x,t;c,\kappa_0,x_0)=c- \frac{2\, \mathrm{sign} (\nu)(\kappa_0^2-c^2)}{\kappa_0\cosh\left[2\sqrt{\kappa_0^2-c^2}\(x-(2c^2+4\kappa_0^2)t\)+x_0\right]-\mbox{sign}(\nu) c},
\end{equation}
  where  the phase shift $x_0$ depends on the spectral data via the relation $$x_0=\log\dfrac{2(\kappa_0^2-c^2)}{|\nu|\kappa_0}\in\mathbb{R}.$$ The solution with  $\nu<0$ is called soliton  and corresponds to  a positive hump, while the solution with $\nu>0$  is called antisoliton and corresponds to a negative hump. In both cases  the speed  is $4\kappa_0^2 +2c^2$, namely the speed of the soliton increases with the size of the step.
   The  maximal   amplitude of   the soliton is $2\kappa_0-c$ while the  minimal  amplitude  of the antisoliton is  $-2\kappa_0-c$.
\color{black}{This means that the values of the soliton span over the interval $[0,2\kappa_0-c],$ with a pronounced peak at the maximal value $2\kappa_0-c,$ and the antisoliton ranges from $0$ to $-2\kappa_0-c,$ with a pronounced peak at the minimal value $-2\kappa_0-c.$}

%\todo{What is amplitude? Is it $q$? Or is it $q_{max}-q_{min}?$}
%For the antisoliton in the limit when $\kappa_0\to c$  we obtain   the rational solution \cite{Ono}
%\begin{equation*}
%q_{rational}=c-\dfrac{4c}{4c^2(x-6c^2t)^2+1}.
%\end{equation*}

%\begin{equation*}
%q_{soliton}(x,t;c)=c\pm\dfrac{v}{\sqrt{v+4c^2}\cosh\left(\sqrt{v}(x-(v+6c^2)t)\right)\pm2c},
%\end{equation*}
%%where the sign $+$ holds for $2c>\sqrt{v}$ and the sign $-$ holds for $2c<\sqrt{v}$
%  where $v>0$ and the  point  spectrum of the associated  Lax operator is $ik=i\sqrt{(\frac{v}{4}+c^2)}$ (\red{Check!}).
%From the above formula, the soliton speed is $v+6c^2$ and the maximum  of the  soliton amplitude is $\sqrt{v+4c^2}-c$ for the solution with the $+$ sign,
%and the minimum amplitude is $-\sqrt{v+4c^2}-c$  for the solution with the $-$ sign. This latter solution is also called antisoliton.
%
%For the antisoliton in the limit when $v\to 0$  we obtain   the rational solution \cite{Ono}
%\begin{equation*}
%q_{rational}=c-\dfrac{4c}{4c^2(x-6c^2t)^2+1}.
%\end{equation*}
The  breather solution  on a constant background $c$ have been obtained using the bilinear method  and Darboux transformations  in \cite{Grim1},  and inverse scattering in \cite{Alejo}.
In this manuscript we obtain the breather on a constant background as a solution of a RH  problem that is parametrised by the  complex number   $\kappa$,  with $\Re\kappa> 0$, $\Im\kappa>0$,  \color{black} 
 and by the complex parameter $\nu$. 
%\todo{Probably it is better to change $\nu$ to another letter, so not to mix with soliton, where the analogous parameter is $-\i\nu.$
%Or discuss this in the text.
% It is resolved now. We have factor $\i$ both for solitons and for breathers now.
%}
Introducing the   complex number $\chi$  defined as $\chi=\chi_1+\i \chi_2=\sqrt{\kappa^2+c^2}$, with $\chi_1>0$, $\chi_2>0,$  the breather solution  on a constant background takes the form
\begin{equation}
  \label{q_breath_intro}
  \begin{split}
& q_{breather}(x,t;c,\kappa,\nu)=c+\\
 &+2\partial_x\arctan\left[\dfrac{{|\chi|}\cos\varphi +\frac{c|\nu| \chi_1^2}{2|\chi|^2 \chi_2}\e^{- 2 Z \chi_2}}{\frac{|\chi|^2}{|\nu|}\e^{2 Z \chi_2}+\frac{\chi_1^2(|\chi|^2-c^2)}{4|\chi|^2\chi_2^2}|\nu|\e^{ -2 Z \chi_2}+c \sin\left(\varphi-\theta_2\right)}\right]\,,
\end{split}
\end{equation}
\noindent  where
 \begin{equation*}
 Z=x+4t(3\chi_1^2-\chi_2^2-\frac32c^2),
 \qquad
 \varphi=2(Z-8t|\chi|^2)\chi_1+\theta_1-\theta_2,
 \end{equation*}
 and  phases \color{black}{$\theta_1=\mbox{arccos}\dfrac{-\Im \nu}{|\nu|}$} and $\theta_2=\mbox{arccos}\dfrac{\chi_1}{|\chi|}$.
 %and $\theta_3=\log\left[\dfrac{2|\chi|^2\chi_2}{|\nu|\ \chi_1\ \cdot|\sqrt{|\chi|^2-c^2}| }\right]$.
 \color{black}{Note that the denominator in \eqref{q_breath_intro} can be written as $\cosh$ when $|\chi|>c$ and as $\sinh$ when $|\chi|<c$. Despite having a $\sinh$ in the denominator, the expression remains regular  (see Remark~\ref{Remark_b}  and cfr.  \cite{Alejo}).
 \color{black}
 On the line $Z=0$    the breather oscillates with period $\dfrac{\pi}{8|\chi|^2\chi_1}$
 and the envelope of the oscillations moves with  a speed 
 \begin{equation}
 \label{breather_speed}
 V=4\chi_2^2+6c^2-12\chi_1^2,\;\;\;\;\; \chi_1=\Re\sqrt{\kappa^2+c^2},\;\;\chi_2=\Im \sqrt{\kappa^2+c^2}.
 \end{equation}
 We observe that for fixed $\kappa$ and large values of $c$ the velocity of the breathers is always negative.
 The level set  of the spectrum   in the complex  $\kappa$-plane corresponding to breathers with   equal speed is shown in Figure~\ref{fig_breather}.
 \begin{figure}[htb]
 \begin{center}
\hskip-1cm\epsfig{width=4.1cm,figure=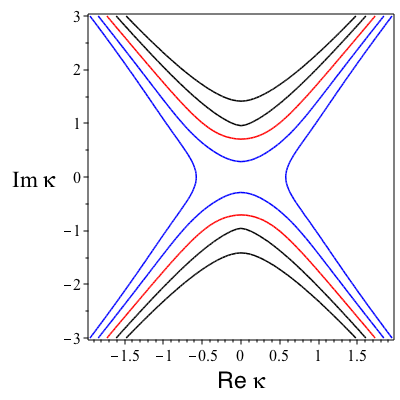}
\epsfig{width=4.5cm,figure=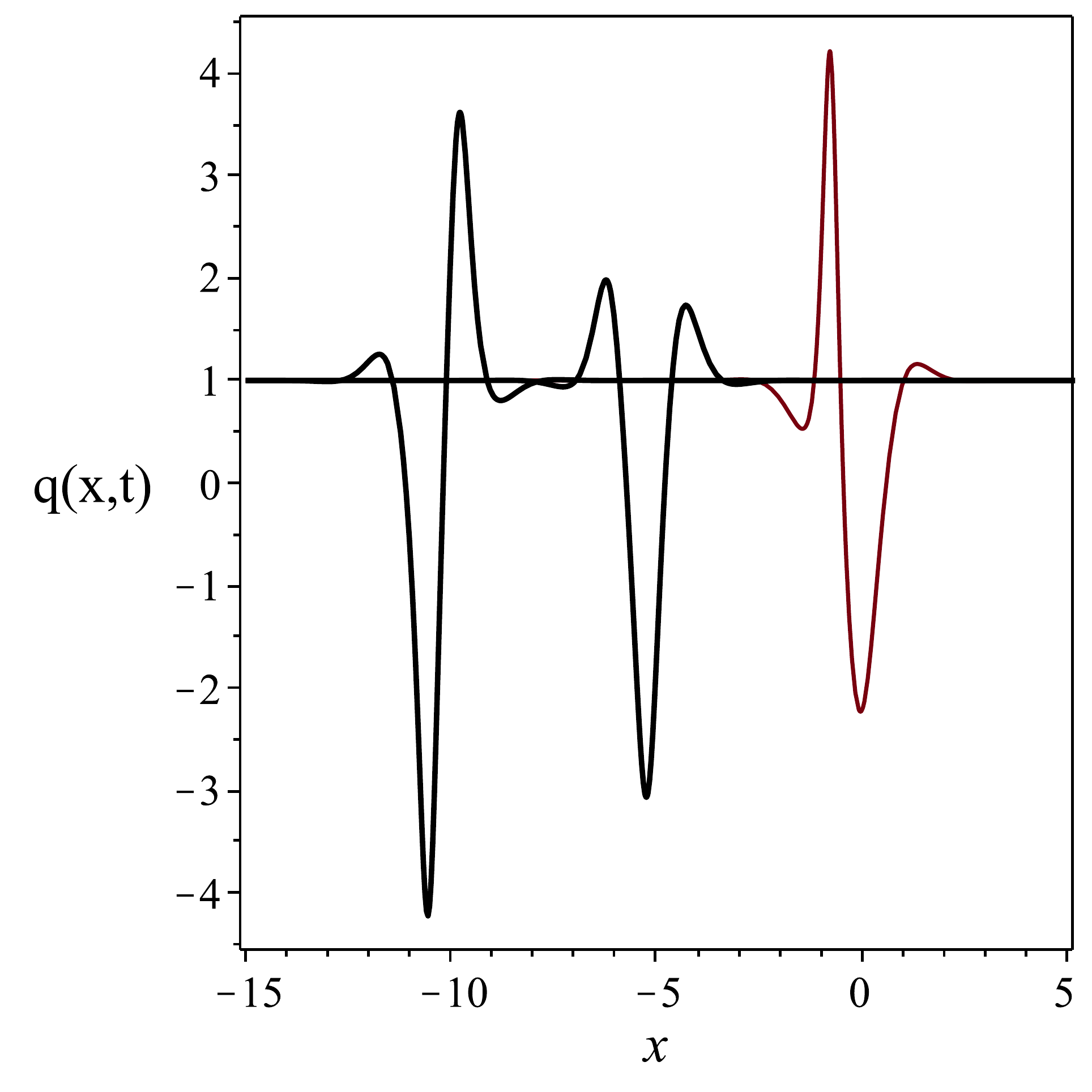}
\epsfig{width=4.5cm,figure=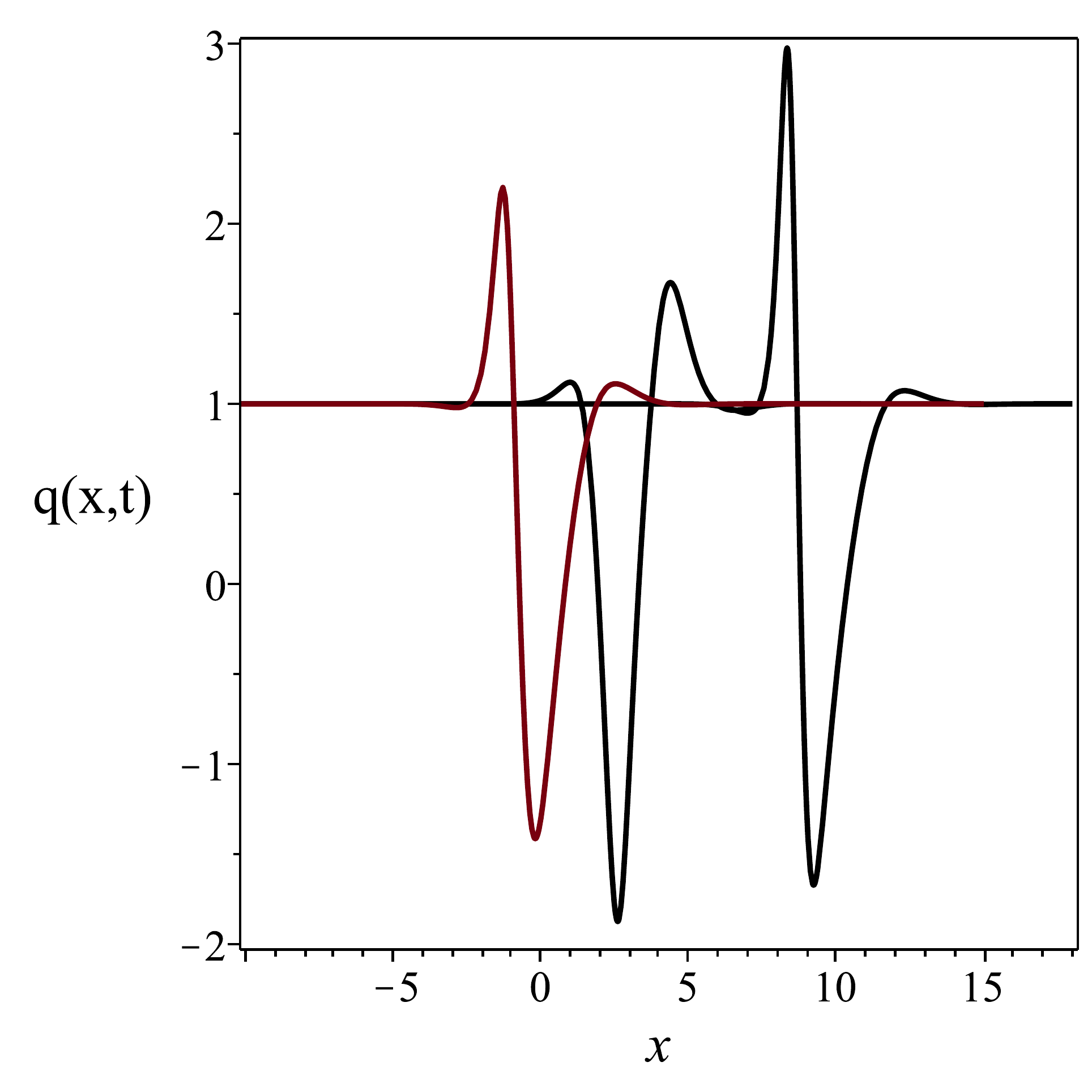}
\end{center}
\caption{The  level  set of curves  of the breather speed  $V$ in \eqref{breather_speed} in the plane $(\Re\kappa,\Im\kappa)$ for  $c=1$. In black the level curves with positive velocity, in blue with negative velocity and in red the line with zero velocity. The snapshot  of a breather at  $t=0$ (red) and at later times in black. In the first figure $c=1,\Re\kappa=1,\Im \kappa=1.5$ and $V<0$ and in  the second figure $c=1,\Re\kappa=0.5,\Im\kappa=1.5$ and $V>0$.}
 \label{fig_breather}
 \end{figure}
% We observe that a soliton  on a constant background  has a speed $4\kappa_0^2+2c^2$  for some $\kappa_0>c$.
% Therefore a soliton and a breather    can have the same speed provided we chose $\chi$ such that 
% $\chi_2^2-3\chi_1^2=\kappa_0^2-c^2$.

 It has been shown in \cite{Grim1} that solitons and breathers on a constant background  $c>0$  interact in an elastic way as in the case $c=0$.
 The breather \eqref{q_breath_intro} turns into a pair of soliton and antisoliton  \eqref{soliton_intro} on a constant background when we let $\chi_1=0$  and $\chi_2>0$.
% \todo{notations for breather and soliton are misleading.
 %Probably it is better to use different letters $\nu, \hat\nu$ for breather and soliton.
% Now, we have $\i$ in all the cases, both for solitons and for breathers.
 %}
 
%\todo[inline]{It rather turns into a double-pole soliton, i.e. qualitatively a pair of soliton and antisoliton.}

For $c=0$  we have $ \chi_1=\Re\kappa$ and $ \chi_2=\Im \kappa$  and the  solution \eqref{q_breath_intro} reduces to the standard  breather~\cite{WO82}
\begin{align*}
q_{breather}(x,t;c=0,\kappa,\nu)&=2\partial_x\arctan\left[\dfrac{(\Im \kappa) \cos \varphi}{(\Re \kappa) \cosh\Theta}\right]\\
&=
-4(\Im \kappa)(\Re \kappa) \dfrac{(\Im \kappa) \sinh\Theta\cos \varphi+(\Re \kappa) \cosh \Theta\sin \varphi}{(\Im \kappa)^2\cos^2( \varphi)+(\Re\kappa)^2\cosh^2\Theta}\,,
\end{align*}
where 
\begin{equation*}
\Theta=2\Im \kappa\left(x+4\left(3(\Re \kappa)^2-(\Im \kappa)^2\right)t\right)+\log\dfrac{2\Im \kappa |\kappa|}{\Re\kappa|\nu|},
\end{equation*}
and\color{black}{\begin{equation*}
 \varphi(x,t)=2\Re\kappa\left(x+4\left((\Re \kappa)^2-3(\Im \kappa)^2\right)t\right)+\mbox{arccos}\dfrac{-\Im \nu}{|\nu|}-\mbox{arccos}\dfrac{\Re \kappa}{|\kappa|}.
\end{equation*}}
It has been shown in \cite{CGM} that   the formation of breathers is generic  for certain 
\color{black}
compactly supported
\color{black} initial conditions.

\noindent 
The periodic travelling  wave solution of the MKdV equation  takes the form  (see Appendix\ref{Appendix_travelling})
\begin{equation}
\label{periodic_intro}
\begin{split}
q_{per}(x,t;\beta_1,\beta_2,\beta_3,x_0&)=-\beta_1-\beta_2-\beta_3+\\
&+\frac{2(\beta_2+\beta_3)(\beta_1+\beta_3)}{\beta_2+\beta_3-(\beta_2-\beta_1)\mathrm{cn}^2\(\sqrt{\beta_3^2-\beta_1^2} (x-{\mathcal V}t)+x_0|m\)},
\end{split}\end{equation}
%\begin{equation}\label{q_trav_intro}q_{periodic}(x,t)=e_4+\frac{(e_1-e_4)(e_2-e_4)}{e_1-e_4-(e_1-e_2)\mathrm{cn}^2\(\frac{\sqrt{(e_1-e_3)(e_2-e_4)}}{2}\ (x-vt)+x_0|m\)},\end{equation}
where $\beta_3>\beta_2>\beta_1$,  the speed ${\mathcal V}=2(\beta_1^2+\beta_2^2+\beta_3^2)$ and $x_0$ is an arbitrary phase. 
The function   $\mathrm{cn}(z|m)$ is the Jacobi elliptic function of modulus 
% $ m=\frac{(e_1-e_2)(e_3-e_4)}{(e_1-e_3)(e_2-e_4)}$ with $e_1>e_2>e_3>e_4$, $\sum_{j=1}^4 e_j=0$
\begin{equation}
\label{mm}
m^2=\dfrac{\beta_2^2-\beta_1^2}{\beta_3^2-\beta_1^2},
\end{equation}
and $\mathrm{cn}(z+2K(m)|m)=-\mathrm{cn}(z|m)$ where $K(m)=\int_0^{\frac{\pi}{2}}\dfrac{ds}{\sqrt{1-m^2\sin^2s}}$ is the  complete  elliptic integral of the first kind.
The periodic solution \eqref{periodic_intro} has wave number  $\mathrm{k}$, frequency $\omega$   and amplitude $a$  given by 
\begin{equation*}
\mathrm{k}=\dfrac{\pi\sqrt{\beta_3^2-\beta_1^2}}{K(m)},  \quad \omega=\mathcal{V}{\mathrm k},\quad a=2(\beta_2-\beta_1),
\end{equation*}
respectively.  When $m\to 1,$  the travelling wave solution (\ref{periodic_intro})  converges to the soliton solution  (\ref{soliton_intro})   with $\beta_2=\beta_3=\kappa_0$ and $\beta_1=c$.
 \color{black}

 Our main result concerns the asymptotic description  for large times of the MKdV initial value problem for step-like initial data of the form \eqref{ic0}.
 Before stating our result, we  remark that the description   in \cite{KM2}  of the long-time behaviour  of the solution of MKdV with  the  shock  initial data  
 \begin{equation}\label{ic0exact}
q_0(x)=\begin{cases}
c_{+}\qquad {\rm for}\quad
x>0,\\c_{-}\qquad {\rm for }\quad x\leq0, \end{cases}
\end{equation}
with $c_->c_+>0$  is as follows: there are two constant regions $\frac{x}{t}<-6c_{-}^2+12c_{+}^2-\delta$ and $\frac{x}{t}>4c_{-}^2+2c_{+}^2+\delta$ for any sufficiently small $\delta>0$,   where the solution $q(x,t)=c_\mp+ \it{o}(1)$   as $t\to\infty$  respectively. The   solution that connects the   two constant regions is oscillatory and it is  described   in terms of a  genus $2$ quasi-periodic solution. 
In our  asymptotic analysis  we  show that such a genus 2 solution \color{black}{ is in fact  a genus 1 solution,} and can be reduced to the modulated travelling wave solution  (\ref{periodic_intro}) of MKdV, namely 
 \begin{equation}
 \label{qp1}
   q(x,t)=q_{per}(x,t,c_{-},d,c_{+},x_0)+O(t^{-1}), \quad  -6c_{-}^2+12c_{+}^2+\delta<\frac{x}{t}<4c_{-}^2+2c_{+}^2-\delta,
   \end{equation}
    where   $d=d(x,t)$ depends on   space and time  according to 
 \begin{equation}
 \label{Whitham}
 \dfrac{x}{t}=W_2(c_{+},d,c_{-}).
 \end{equation}
 Here 
 \begin{equation}
 \label{Whitham_d0}
 W_2(\beta_1,\beta_2,\beta_3)=2(\beta_1^2+\beta^2_2+\beta_3^2)+4\dfrac{(\beta_2^2-\beta_1^2)(\beta_2^2-\beta_3^2)}{\beta_2^2-\beta_3^2+(\beta_3^2-\beta_1^2)\frac{E(m)}{K(m)}},
 \end{equation}
 with $\beta_1\leq \beta_2\leq \beta_3$ and   $E(m)=\int_0^{\frac{\pi}{2}}\sqrt{1-m^2\sin^2\theta}d\theta$  the complete elliptic integral of the second kind. We have  that
  $$W_2(c_{+},c_{+},c_{-}) =-6c_{-}^2+12c_{+}^2<\frac{x}{t}<  4c_{-}^2+2c_{+}^2=W_2(c_{+},c_{-},c_{-}).$$  
  \color{black}
  \begin{figure}[htb]
 \begin{center}
 \includegraphics[scale=0.3]{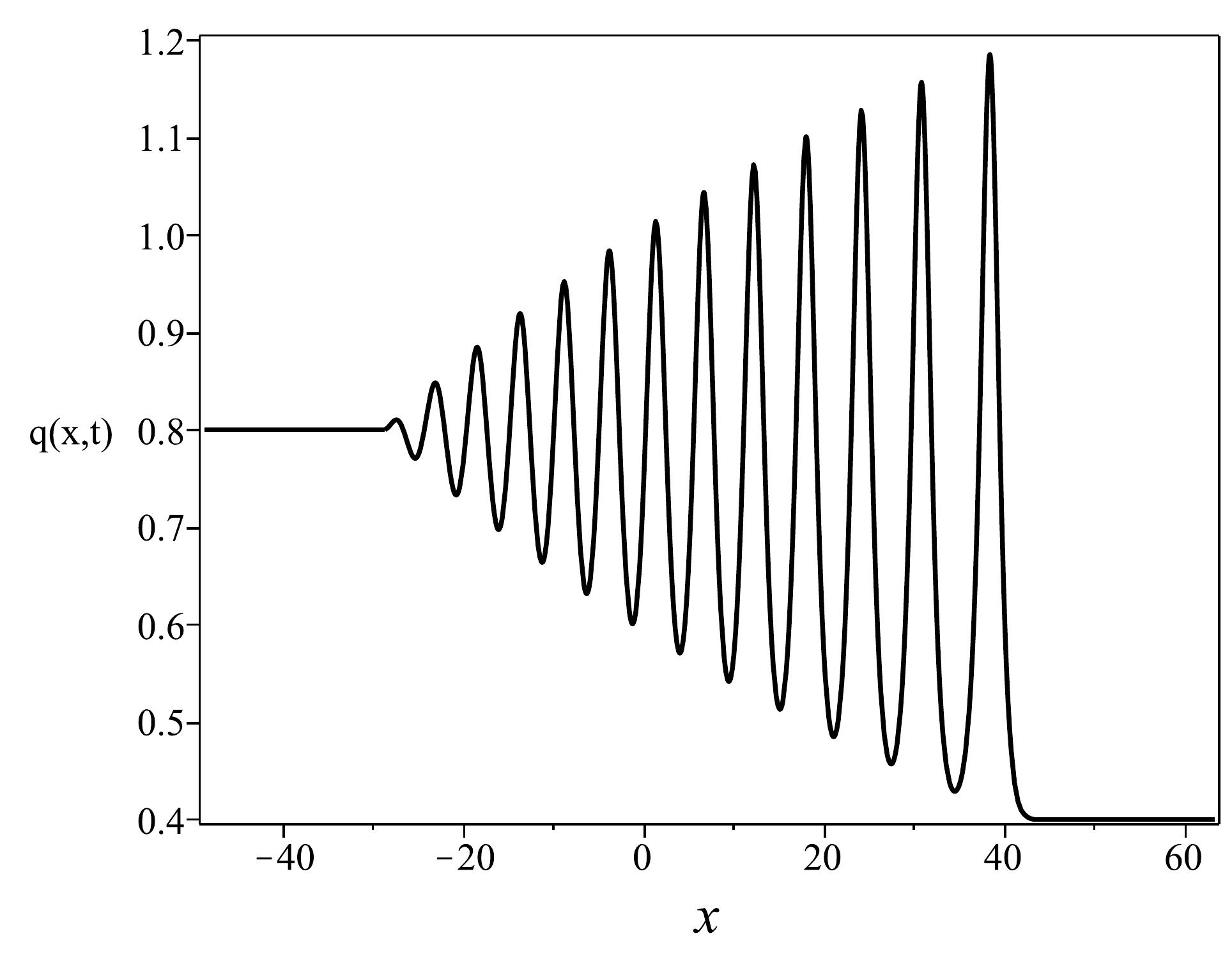}
  \includegraphics[scale=0.3]{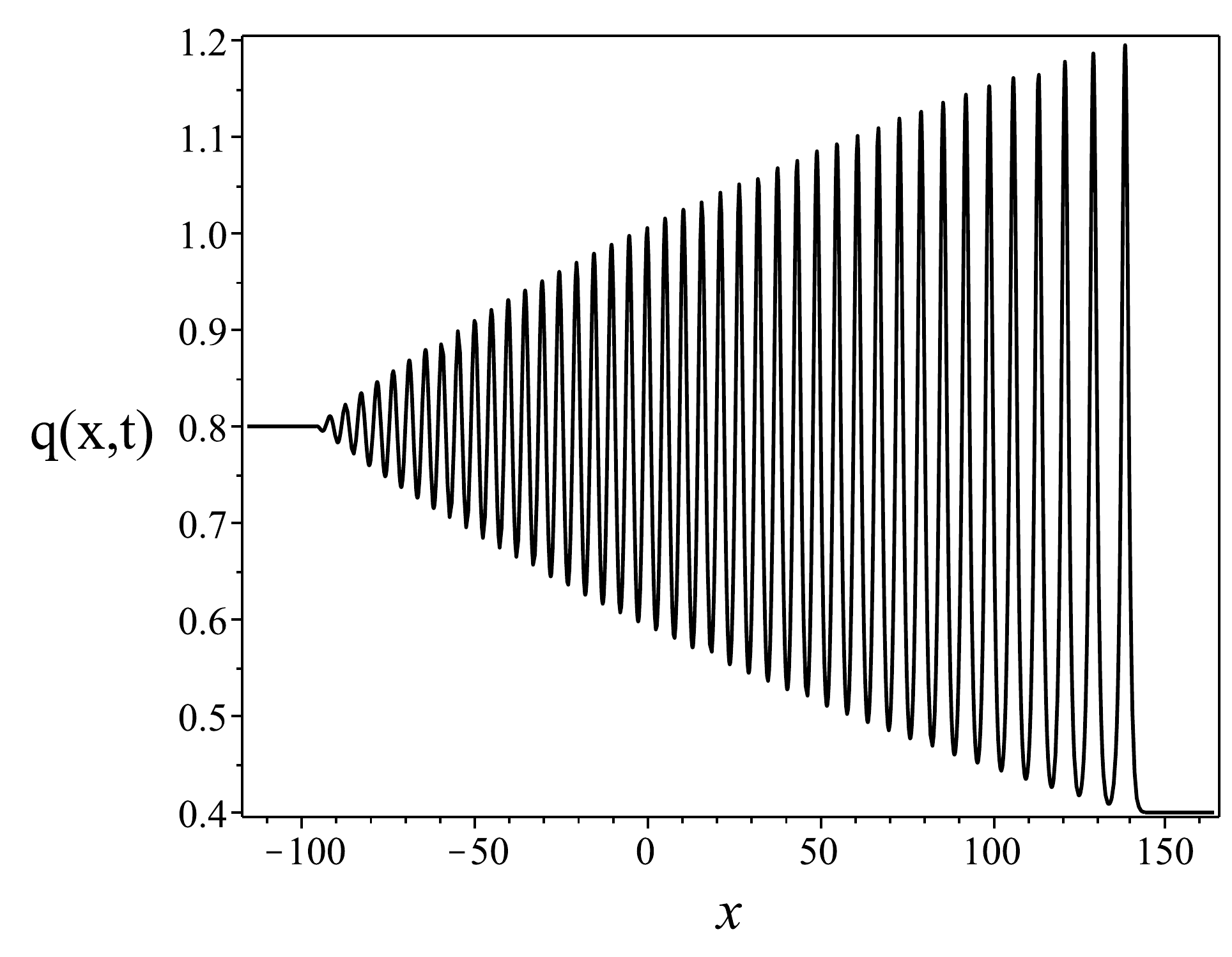}
\end{center}
\caption{The solution of
the modulated travelling wave  \eqref{qp1} with $x_0=K(m)$, $c_+=0.4$, $c_-=0.8$  and $t=15$ (left), $t=50$ (right). We observe that the leading oscillation is approximately a soliton  of maximal height   $2c_--c_+$. }
 \label{fig02}
 \end{figure}  
  The quantity  $W_2(\beta_1,\beta_2,\beta_3)$ is the speed  of the Whitham modulation equations \cite{Driscol} for $\beta_1, \beta_2$ and $\beta_3$   referring to   the Riemann invariants $\beta_2$. 
 The  solution  \eqref{qp1}  is the  dispersive shock wave (see Figure~\ref{fig02})  and it was  first derived for  the Korteweg-de Vries equation \cite{GP}.  The Whitham equations for MKdV are strictly hyperbolic  and genuinely  nonlinear \cite{Levermore} only for $\beta_2\neq 0$ which implies that   the equation \eqref{Whitham}  can be inverted for $d$ as a function of $\xi$ only when $c_+>0$. The case $c_->-c_+>0$ can be studied in a similar way, by getting a slightly different evolution  of the wave  parameters (see the Appendix~ \ref{WhithamApp} for further comments).  A comprehensive set of cases arising in the long-time asymptotic solution for the MKdV equation with step initial data  for various values of the parameters $c_-$ and $c_+$ has been discussed in \cite{EHS, Leach, Marchant}.
 The phase $x_0$ of the travelling wave solution  \eqref{qp1}  is given explicitly in terms of the   scattering data associated to the  MKdV equation.
 
We consider general step-like initial data and  we subject the initial function to the following conditions.
  \begin{Assumption}
  \label{Assump1}
  The initial data  $q_0(x)$ is assumed to be locally a function of bounded variation $BV_{loc}(\R)$ and  satisfying the following conditions
%   \begin{equation}\quad \lim\limits_{x\to\pm\infty}q_0(x)=c_{\pm},
%      \end{equation}
\begin{equation}
\label{integral}
\int_{-\infty}^{+\infty} |x|^2 |\d q_0(x)| <\infty
\end{equation}
and
\begin{equation}\label{exponential}
 \int\limits_{\R^{\pm}}e^{2 \sigma |x|}|q_0(x)-c_{\pm}|\d x<\infty,
\end{equation}
where $\sigma>\sqrt{c_-^2-c_+^2}>0,$ and $\d q_0(x)$ is the corresponding signed measure (distributional derivative of $q_0(x)$).
  \end{Assumption}
\color{black}{This class includes the case of exact (discontinuous) step function \eqref{ic0exact}}. \color{black}{Note that for a function $q_0(x)$ which is locally of bounded variation and tends to $c_{\pm}$ as $x\to\pm\infty,$ the condition \eqref{integral} is equivalent to 
\begin{equation}\label{first_moment}\int_{-\infty}^{0}(1+|x|)|q_0(x)-c_-|\d x+\int_0^{+\infty}(1+|x|)|q_0(x)-c_+|\d x<\infty.\end{equation}}

%The condition (\ref{exponential})  is not essential to establish existence of a solution of the MKdV  equation. Indeed it is enough to assume 
%%  In order to establish the existence of a  solution of MKdV initial value problem, it is enough to require  the   existence of the first moments, namely:
%  \begin{equation}
%  \label{first_moment}
%  \int_{\R^{\pm}}(1+|x|)|q_0(x)-c_{\pm}|\d x<\infty.
%  \end{equation}
%  We require condition (\ref{exponential}) to make the asymptotic analysis $t\to \infty$ easier.
  
  \begin{theorem}
  Under Assumption~\ref{Assump1}, the initial value problem of the MKdV equation (\ref{MKdV}) has a classical solution  for all $t>0$.
  \end{theorem}
Under Assumption~\ref{Assump1}, the  inverse of the transmission coefficient $a(k)$  is analytic for  $k\in\mathbb{C}_+ \setminus[\i c_{\l},0]$, and it has continuous limits to %\clr{can be extended continuously up to} 
the boundary, 
with the exception of the points $\i c_{\l},$ $\i c_{\r},$ where $a(k)$ may have at
most a fourth root singularity, namely $a(k)\sqrt[4]{k-\i c_{\pm}}$ is bounded  (see  Lemma~\ref{lem_abr}). 
The zeros of $a(k)$
%, that we assume to be generically  simple, 
form the point spectrum and by analiticity, the number of zeros is finite.
   We fix the number of zeros in  the quarter plane $\Im k\geq0$ and $\Re k\geq 0$  equal to $N$.
In order to formulate our results, we  enumerate the zeros   of $a(k)$ in $\Im k\geq0,$ $\Re k\geq 0$ in the decreasing order of the speed of the corresponding solitons or breathers,
namely the points 
$$\kappa_1,\dots,\kappa_N,\quad \Im \kappa_j\geq0,\;\;\Re \kappa_j\geq 0$$
correspond to the speeds
\begin{equation*}+\infty>V_1\geq V_2\geq \ldots \geq V_N>-\infty.\end{equation*} We recall that a soliton with the point spectrum $\kappa$ on a constant background $c$ has the speed $2c^2+4|\kappa|^2$, while a breather with the point spectrum $\kappa$ 
on a constant background $c$ has the speed specified in \eqref{breather_speed}. {\color{black}The speed of a breather on a elliptic background is specified in \eqref{g_def}}. In our case for each  breather  there are three options for large times:  it travels in  either the left constant background, right constant background, or dispersive shock wave background.
Note that it follows from properties 11  and 12 of Lemma \ref{lem_abr} below, that $a(k)$ cannot have zeros  on $(\i c_-,  0).$  Therefore a soliton can travel only to the right of the dispersive shock vave.
\color{black}
% The way to determine the speed of soliton/breather is to use an auxiliary $g-$function defined below in \eqref{g_def}; the speed $V_j$ of the soliton/breather corresponding to a point $\kappa_j$ in the discrete spectrum is the solution of the equation $\Im g(\kappa_j,\frac{1}{12}V_j)=0.$ 

\color{black}
We make further (generic) assumptions on the potential $q_0(x)$ which are formulated in terms of  the  associated spectral function $a(k).$
\begin{Assumption}\label{Assump2}
The spectral function $a(k)$ (see Definition~\ref{T_a_b}) satisfies the following generic conditions:
\begin{itemize}
\item the zeros of $a(k)$ are simple;
\item  the  zeros of $a(k)$ do not lie on $\mathbb{R};$
%the zero on $(\i c_+,0]$ is understood in one-sided
\item all the speeds $V_j$ of breathers and solitons are distinct,  namely  
\begin{equation*}+\infty>V_1> V_2> \ldots > V_N>-\infty;
\end{equation*} 
  %and do not coincide with the borders of the regions, all the corre 
\item the behavior of $a(k)$ at the points $\i c_{\pm}$ is as follows: for some nonzero constants $C_1, C_{2+}, C_{2-},$
\begin{equation*}\begin{split}
&a(k)\sim C_1(k-\i c_-)^{-1/4}\mbox{ as } k\to\i c_-;
\\
&
\mbox{ in the case $c_+>0,$\quad }
\color{black}
\quad
a(k)\sim C_{2,\pm}(k-\i c_+)^{-1/4}\mbox{ as } k\to\i c_+\pm0,
\end{split}
\end{equation*}
where $k\to\i c_+\pm0$ stands for the non tangential limit to the point $\i c_+$ from the left ($+$) and right $ (-)$  side of the oriented segment $[\i c_-,0]$, where the orientation is downward.
\end{itemize}
\end{Assumption}
\color{black}

Similarly to Beals and Coifman \cite{BC}, one can show that the set of potentials satisfying Assumption \ref{Assump2} form a dense open set in the space of all potentials satisfying Assumption \ref{Assump1} with respect to  the  topology induced by \eqref{first_moment}
\color{black} but it is beyond the scope of this manuscript to prove this  here.

%\color{black}{We will now make an assumption that the initial function $q_0(x)$ is a generic potential. Precisely, this means the following.}

%\begin{Assumption}
%Zeros of $a()$
%\end{Assumption}

\color{black}
\medskip

\color{black}
 The  question  we address  is: how,  under Assumptions \ref{Assump1} and \ref{Assump2}, %in the generic case
the solitons, breathers  and  the dispersive shock wave interact for large values of time (see Figure~\ref{fig0}).
    Theorem~\ref{thrm:asymp:rl}
  characterises these interactions in the general setting and \eqref{eq_long1}, \eqref{eq_long2} and \eqref{eq_long3}   below show explicitly how these interactions affect the asymptotic phase shifts of individual solitons  breathers and the dispersive shock wave.
  %{\color{black} From Assumption~\ref{Assump2}  and Lemma~\ref{lem_abr} (see below)  it follows that  generically,  there are no solitons in the dispersive shock wave regions but only  breathers}.
  
\begin{figure}[htb]
 \begin{center}
 \includegraphics[width=6.5cm]{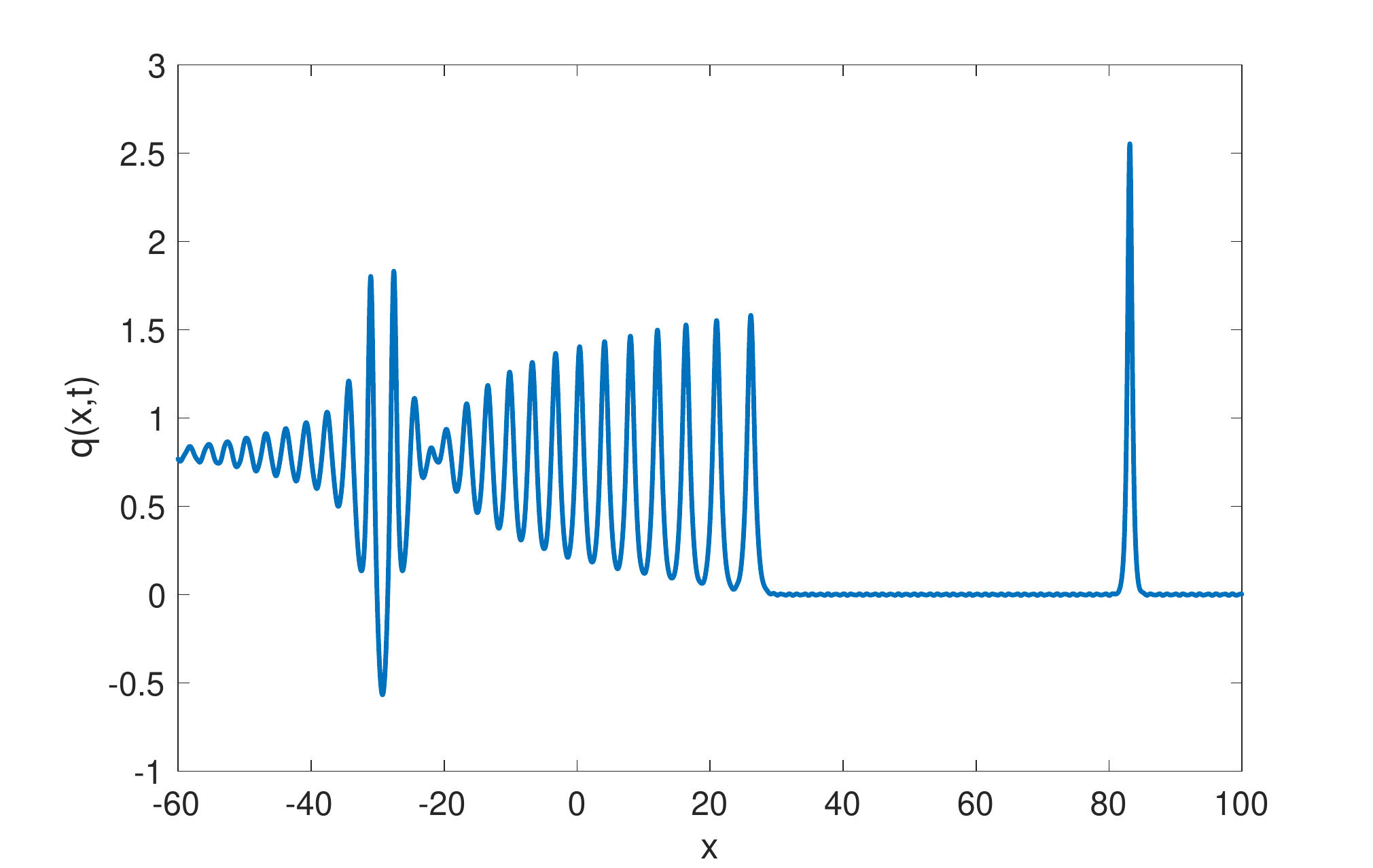}
\includegraphics[ width=6cm]{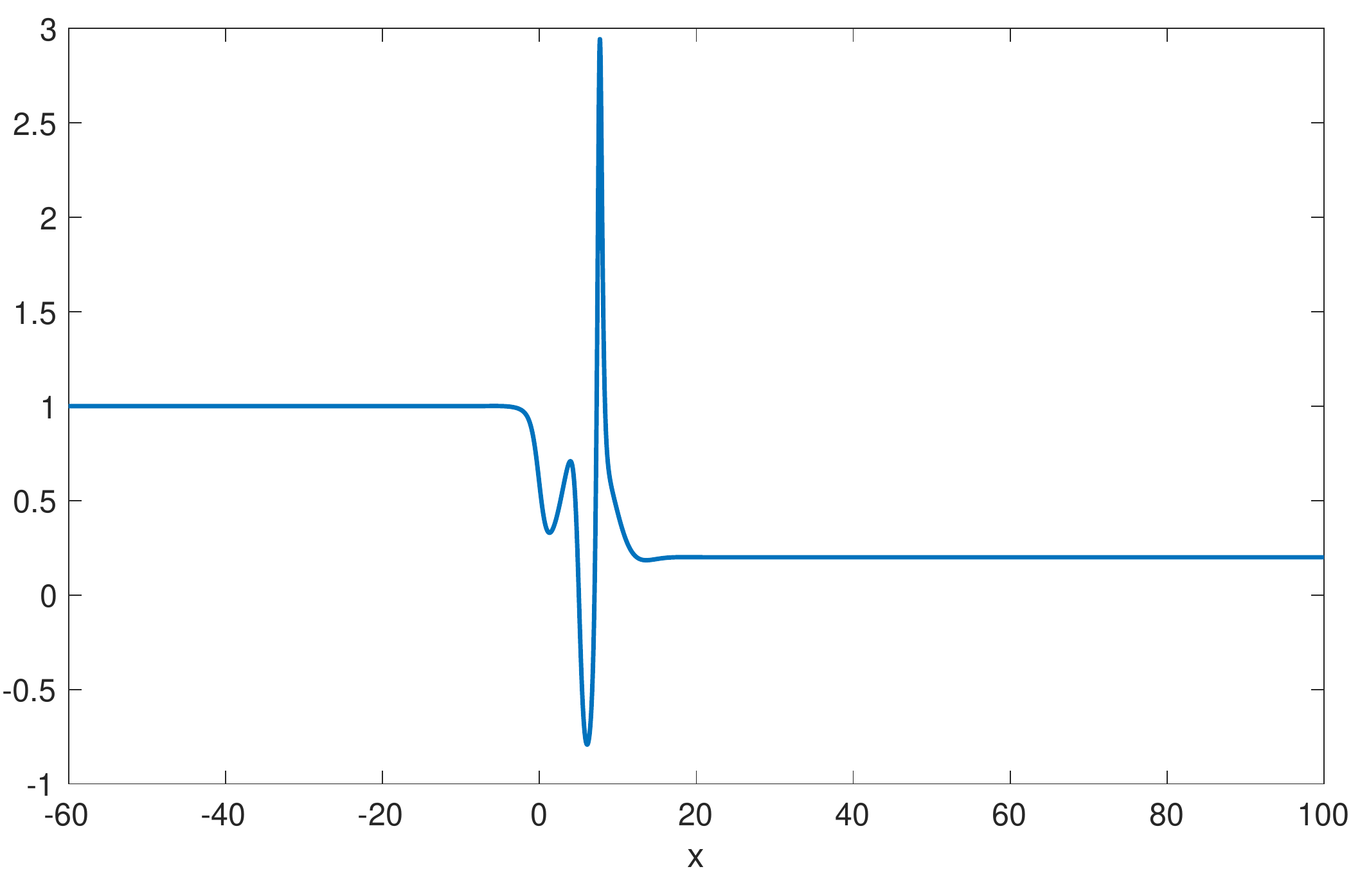}
\end{center}
\caption{ The  evolution  into a breather with negative speed, a dispersive shock wave and a soliton (left) of a steplike initial data  (right). }
 \label{fig0}
 \end{figure}

\noindent 
Now   we  define
\begin{equation}
\label{T_tilde}\widetilde T_j(k)=\prod\limits_{l<j,\Re \kappa_l>0}\frac{k-\ol\kappa_l}{k-\kappa_l}\frac{k+\kappa_l}{k+\ol\kappa_l}\cdot\prod\limits_{l<j, \Re\kappa_l=0}\frac{k-\ol\kappa_l}{k-\kappa_l},\quad j\geq 2, \qquad \widetilde T_1(k)=1.
\end{equation}

  \begin{theorem}\label{thrm:asymp:rl}
Let an initial datum satisfy 
%  For a step-like initial datum satisfying  
Assumptions  \ref{Assump1}, \ref{Assump2} and 
\color{black}{ let $\delta>0$ be a sufficiently small positive number such that all the breather/soliton speeds $V_j$ satisfy $|V_j-V_l|>4\delta$ for all $j\neq l,$ $1\leq j,l\leq N.$ 
}

Then the solution of the Cauchy problem for  the  MKdV   equation \eqref{MKdV}
  behaves for large $t$ in the following way:

  \begin{itemize}
  \item[(a)] Soliton and breather region:  $\frac{x}{t}>4c_{-}^2+2c_{+}^2+\delta_1$, $\delta_1>0$.
  
For $\frac{x}{t}$ such that $|\frac{x}{t}-V_j|>\delta$ for all $j,$
	\begin{equation*}q(x,t)=c_++\mathcal{O}(\e^{-C t}),\end{equation*} with some $C>0;$
\\
for $|\frac{x}{t}-V_j|<\delta$ for some $j,$  %$j=1,\dots, n$,
  \begin{equation}
  \label{eq_long1}
  q(x,t)=\left\{
  \begin{array}{ll}
  q_{soliton}(x,t; c_{\r}, \kappa_j, x_j) + \mathcal{O}(\e^{-C t}),\quad \mbox{ if }\quad \Re\kappa_j=0,\\
    &\\
 q_{breather}(x,t; c_{\r}, \kappa_j,\hat\nu_j) +\mathcal{O}(\e^{-C t}),\quad \mbox{ if }\quad \Re\kappa_j>0,
 \end{array}\right.
  \end{equation}
where  $q_{soliton},$   and $q_{breather}$ are defined in \eqref{soliton_intro}, \eqref{q_breath_intro}  respectively and 
   \begin{equation*}\hat\nu_j = \frac{\nu_j}{T_j^2(\kappa_j)},\qquad \quad
    x_j=\ln\frac{(|\kappa_j|^2-c_+^2)T_j^2(\i|\kappa_j|)}{|\nu_j| |\kappa_j|}
   \end{equation*} and 
   \begin{equation*}
   T_j(k)=\widetilde T_j(k)\cdot\exp\left[
   \frac{\sqrt{k^2+c_+^2}}{2\pi\i}\int\limits^{\i c_+}_{-\i c_+}\frac{\ln\widetilde T_j^2(s)\ \d s}{(s-k)\left(\sqrt{s^2+c_+^2}\right)_+}
   \right],
   \end{equation*}
   with $\widetilde T_j(k)$ as in  \eqref{T_tilde} and $\left(\sqrt{s^2+c_+^2}\right)_+$ is the boundary value on the positive side of the oriented segment $[\i c_+,-\i c_+]$.

  \item[(b)] Dispersive shock wave region: $-6c_{-}^2+12c_{+}^2+\delta_1<\frac{x}{t}<4c_{-}^2+2c_{+}^2-\delta_1$,  $\delta_1>0.$ 
  
  For $\frac{x}{t}$ such that $|\frac{x}{t}-V_j|>\delta$ for all $j,$
  one has
  \begin{equation}\label{eq_long2}
  q(x,t)=q_{per}(x,t;c_{-},d,c_{+},x_0)+O(t^{-1}), 
  \end{equation}
  where  the travelling wave  $q_{per}$ has been defined in (\ref{periodic_intro}),  the quantity  $d=d(x,t)$   depends on $x$ and $t$ as in (\ref{Whitham}) and  the phase 
  \begin{equation*}
x_0=-\dfrac{K(m)\Delta}{\pi}+K(m)
  \end{equation*}
  depends on the discrete and continuous spectrum via the relation 
\begin{equation}
\label{Delta_intro}
\begin{split}
&\Delta(x,t) =-\dfrac{\sqrt{c_-^2-c_+^2}}{K(m)}\left[\int\limits^{\i c_{\l}}_{\i d(x,t)}\frac{\ln(|a(s)|^2\widetilde{T}_n^2(s))s\d s}{R_+(s)}+
\int\limits^{\i c_{\r}}_{0}\frac{\(\ln \widetilde{T}_n^2(s)\)s\d s}{R_+(s)}\right]\\
&R_+(s)=\(\sqrt{(s^2+c_+^2)(s^2+c_-^2)(s^2+d^2)}\)_+
\end{split}
\end{equation}
\color{black}{where the quantity $\widetilde{T}_n(s)$ is as in \eqref{T_tilde}, where $n=n(\frac{x}{t})$ is  the number of soliton/breather  speed  $V_j$  satisfying $V_j>\frac{x}{t},$ }and $a(k)$ is the  inverse of the  transmission coefficient associated to the  Zakharov-Shabat  spectral problem.
{\color{black} For $\frac{x}{t}$ such that $|\frac{x}{t}-V_j|<\delta$ for some $j,$ one has  $q(x,t)=q_{be}(x,t)+o(1)$ where $q_{be}(x,t)$ is the breather solution on the  elliptic background and it  is specified by the solution of the RH  problem~\ref{RH_problem_W2} in Section~\ref{lenses}. 
The  corresponding  speed  $V_j$  is given  in \eqref{ellbreather_speed}.}

     \item[(c)] Breathers on a constant background: $\frac{x}{t}<-6c_{-}^2-\delta_1$
     or   $-6c_{-}^2+\delta_1 < \frac{x}{t} < -6c_{-}^2+12c_{+}^2-\delta_1,$     
      $\delta_1>0$. \\
For $t$  large and  such that $|\frac{x}{t}-V_j|>\delta$ for all $V_j,$
	\begin{equation*}q(x,t)=c_-+\mathcal{O}(t^{-1/2}),\end{equation*}
	and for $x/t$ such that $|\frac{x}{t}-V_j|<\delta$ for some $j,$ %=n+1,\dots,N,
   \begin{equation}
   \label{eq_long3}
   q(x,t)= q_{breather}(x, t; c_{\l}, \kappa_j, \hat\nu_j) + \mathcal{O}(t^{-1/2}),
   \end{equation}
   where $q_{breather}$ is defined in \eqref{q_breath_intro} and the phase $\hat\nu_j$ 
      \begin{equation*}\hat\nu_j = \frac{\nu_j
   %\cdot\e^{2\i(tg(\kappa_j, x/12t)-\theta((x,t,\kappa_j)))}
   }{T_j^2(\kappa_j, \xi)},\end{equation*}
  with 
  \begin{equation*}
  \begin{split}
  T_j(k,\xi)=\widetilde T_j(k) 
\exp&\left[\frac{-\sqrt{k^2+c_{\l}^2}}{2\pi\i}\left\{\int\limits_{\i c_{\l}}^{\i d_0(\xi)}-\int\limits^{-\i c_{\l}}_{-\i d_0(\xi)}\frac{\ln|a(s)|^2\d s}{(s-k)\left(\sqrt{s^2+c_{\l}^2}\right)_+}+\right.\right.\\
&+\left.\left.
\int\limits_{\i c_-}^{-\i c_-}\frac{\(\ln\widetilde T_j^2(s)\)\d s}{(s-k)\left(\sqrt{s^2+c_{\l}^2}\right)_+}
\right\}\right]
\end{split}
\end{equation*}
for $\frac{-c_{\l}^2}{2}<\xi<\frac{-c_{\l}^2}{2}+c_{\r}^2$ with  $d_0(\xi)=\i\sqrt{\xi+\frac{c_-^2}{2}}$, 
and
 \begin{equation*}
 \begin{split}
  T_j(k,\xi)&=  \widetilde T_j(k) 
\exp\left[\frac{-\sqrt{k^2+c_{\l}^2}}{2\pi\i}\left\{\int\limits_{\i c_{\l}}^{0}-\int\limits^{-\i c_{\l}}_{0}\frac{\ln|a(s)|^2\d s}{(s-k) \left(\sqrt{s^2+c_{\l}^2}\right)_+}\right.\right.\\
&\left.\left.+\int\limits_{\i c_-}^{-\i c_-}\frac{\ln\widetilde T_j^2(s)\d s}{(s-k)\left(\sqrt{s^2+c_{\l}^2}\right)_+}-\int\limits_{-k_0(\xi)}^{k_0(\xi)}\frac{\ln\(1+|r(s)|^2\)\d s}{(s-k)\sqrt{s^2+c_{\l}^2}}\right\}
\right],
\end{split}
 \end{equation*}
 for $\xi<-\frac{c_-^2}{2}$ and $\color{black}{k_0}(\xi)=\sqrt{-\xi-\frac{c_-^2}{2}}$  and $ \widetilde T_j(k) $ as in \eqref{T_tilde}.

 %  \begin{equation*}\begin{split}
%&T_j(k,\xi)=\widetilde T_j(k) 
%\exp\left[\frac{\chi_{\l}(k)}{2\pi\i}\left\{\int\limits_{\i c_{\l}}^{\i d_0}\frac{\(-\ln|a(s)|^2-\ln \widetilde T_j^2(s) \)\d s}{(s-k)\, \chi_{\l}(s+0)}\right\}\right]\times
%\\
%&
%\times\exp\left[\frac{\chi_{\l}(k)}{2\pi\i}\left\{\int\limits^{-\i c_{\l}}_{-\i d_0}\frac{\(\ln|a(s)|^2-\ln \widetilde T_j^2(s)\)\d s}{(s-k)\, \chi_{\l}(s+0)}
%+
%\int\limits_{\i d_0}^{-\i d_0}\frac{\(-\ln \widetilde T^2_j(s)\) \d s}{(s-k)\, \chi_{\l}(s+0)}
%\right\}\right],
%\mbox{ if } d_0:=\sqrt{\xi+\frac{c_-^2}{2}}>0,
%\end{split}\end{equation*}
%and 
%\begin{equation*}
%\begin{split}&T_j(k,\xi)=\widetilde T_j(k)
%\exp\left[\frac{\chi_{\l}(k)}{2\pi\i}\left\{\int\limits_{\i c_{\l}}^{0}\frac{\(-\ln|a(s)|^2-\ln \widetilde T_j^2(s) \)\d s}{(s-k)\, \chi_{\l}(s+0)}\right\}\right]\times
%\\
%&
%\times\exp\left[\frac{\chi_{\l}(k)}{2\pi\i}\left\{\int\limits_{0}^{-\i c_{\l}}\frac{\(\ln|a(s)|^2-\ln \widetilde T_j^2(s)\)\d s}{(s-k)\, \chi_{\l}(s+0)}
%+
%\int\limits_{-k_0}^{k_0}\frac{\ln\(1+|r(s)|^2\)\d s}{(s-k)\chi_{\l}(s)}
%\right\}\right],\ \mbox{ if } k_0:=\sqrt{-\xi-\frac{c_-^2}{2}}>0.
%\end{split}
%\end{equation*}
    \end{itemize}
  \end{theorem} 
  
\noindent 
The  subleading term  of order ${\mathcal O}(t^{-\frac{1}{2}}) $ of the expansion of $q(x,t)$ as $t\to\infty$  in the   left constant region is oscillatory and is described by the theorem below.
\begin{theorem}\label{theorem_error}
Away from breathers, the subleading term of the expansion of $q(x,t)$ as $t\to\infty$  in the regions
$x<(-6c_-^2-\delta_1)t$ and $(-6c_-^2+\delta_1)t<x<(-6c_-^2+12c_+^2-\delta_1)t$
is  given by the formula 
%$q(x,t)=c_-+q_{err}(x,t),$ where
\begin{equation}\label{q_errIntro}
\begin{split}
q(x,t)&=c_-+
\sqrt{\frac{|\nu(\xi)|\sqrt{-\xi+\frac{c_-^2}{2}}}{3t\,|\xi+\frac{c_-^2}{2}|}}
\times\\
&\times\cos
\left[
16t\(-\xi+\frac{c_-^2}{2}\)^{\frac32}
+\nu(\xi)\ln\(\frac{192\, t\, (\xi+\frac{c_-^2}{2})^2}{\sqrt{-\xi+\frac{c_-^2}{2}}}\)
+
\phi(\xi)
\right]
+\mathcal{O}(t^{-1}),
\end{split}
\end{equation}
where $\xi=\frac{x}{12t}$ and the phase shift and parameter $\nu(\xi)$  are different for different regions of $\xi=\frac{x}{12t}:$ 
\\
$\bullet $ for $\xi<-\frac{c_-^2}{2}$ they are given by the formulas
\begin{equation}\label{phaseUtmostIntro}
\begin{split}
&\nu(\xi)=\frac{1}{2\pi}\ln(1+|r(k_0)|^2),\quad k_0=\sqrt{-\xi-\frac{c_-^2}{2}}, \\
&\phi(\xi):=\frac{\pi}{4}-\arg r(k_0)-\arg\Gamma(\i\nu(\xi))+\arg\chi^2(k_0),
\end{split}
\end{equation}
and $\chi(k_0)=\lim\limits_{k\to k_0}T(k)\(\frac{k-k_0}{k+k_0}\)^{\i\nu(\xi)},$
\\$\bullet $ and for $-\frac{c_-^2}{2}<\xi<-\frac{c_-^2}{2}+c_+^2$ they are given by the formulas
\begin{equation}\label{phaseMiddleIntro}
\begin{split}
&\nu(\xi)=\frac{1}{2\pi}\ln(1-|r_+(\i d_0)|^2)<0,
\quad \i d_0 = \i \sqrt{\xi+\frac{c_-^2}{2}},\\
&\phi(\xi)=-\frac{\pi}{4}-\arg r_+(\i d_0)-\arg\Gamma(\i\nu(\xi))+\arg\chi^2(\i d_0),
\end{split}
\end{equation}
and $\chi(\i d_0)=\lim\limits_{k\to \i d_0+0}T(k)\(\frac{k-\i d_0}{-(k+\i d_0)}\)^{\i\nu(\xi)},$
$\quad |\chi(\i d_0)|=1.$
\end{theorem}
\color{black}
\begin{remark}Note that the term $\frac{\pi}{4}$ in the phase shifts has different signs for different values of $\xi,$ and that the parameter $\nu(\xi)$ has different signs in different regions.
We also  observe  that the amplitude in the formula \eqref{q_errIntro} does not blow up when $\xi$ approaches $-\frac{c_-^2}{2}$ if $c_+>0,$ since  $\nu(\xi)$ has a first order zero as $\xi=-\frac{c_-^2}{2}$  for $c_+>0$ while it blows up at $c_+=0$. This suggests the existence of a non trivial transition zone.
\end{remark}

  Our analysis is obtained by formulating the inverse scattering problem for the MKdV equation with step initial data as a RH  problem and then we implement 
   the long-time asymptotic analysis via the Deift-Zhou steepest descent method \cite{DZ93}. The proof that the oscillations in the  oscillatory region that were obtained in \cite{KM2} via genus two theta functions, can be reduced to the dispersive shock wave solution for the MKdV equation, is obtained by ``folding" the genus two Riemann Hilbert problem to a genus one RH  problem with poles. Further  analysis permits us to get rid of the poles and  to reduce the solution to the travelling wave solution of the MKdV equation.
   The estimate of the errors and the calculation of the subleading terms of the expansion   in Theorems~\ref{thrm:asymp:rl} and Theorem~\ref{theorem_error} is obtained via the construction of local parametrices in terms of 
Airy functions and parabolic cylinder functions.

      We illustrate our results with the following examples.
   \begin{Example}\label{Ex1}
   For the exact step \eqref{ic0exact} one has   that the  reflection coefficient and the inverse of the transmission coefficients are  \cite[formula (3.3)]{KM2}
\begin{equation*}
r(k)=\dfrac{\gamma(k)^2-1}{\gamma(k)^2+1},\quad 
a(k)=\frac{1}{2}\left(\gamma(k)+\frac{1}{\gamma(k)}\right),\ 
\gamma(k)=\left(\dfrac{(k-\i c_-)(k+\i c_+)}{(k+\i c_-)(k-\i c_+)}\right)^{\frac{1}{4}}\,.
\end{equation*}
The dispersive shock wave is given by the  relation \eqref{eq_long2}  with phase shift
\begin{equation*}
x_0=-\dfrac{\sqrt{c_-^2-c_+^2}}{\pi}\int\limits^{\i c_{\l}}_{\i d(x,t)}\frac{\ln(|a(s)|^2)s\d s}{\sqrt{(s^2+c_+^2)(s^2+c_-^2)(s^2+d^2)}}+K(m)\,,
\end{equation*}
The evolution for such initial data is illustrated in Figure~\ref{fig2}. The  leading edge of the oscillatory region has been studied in 
 \cite{KK}, where it has been identified with a train of asymptotic solitons, 
and it has been shown that the amplitude of the first soliton   is approximately described
\begin{equation*}
q(x,t)\simeq q_{soliton}(x,t;c_+,c_-,\tilde{x}_0)
\end{equation*}
where $\tilde{x}_0=\frac{3}{2}\log t+\alpha$ for some constant $\alpha$.  Here $q_{soliton}(x,t;c_+,c_-,\tilde{x}_0)$ is the soliton solution  \eqref{soliton_intro}
 on the constant background $c_+$ with spectrum $c_-$. The highest peak of the first soliton is  approximately 
 located at  the position $x_s(t)=2(c_+^2+2c_-^2)t-\frac{1}{2}\frac{\frac{3}{2}\log t+\alpha}{\sqrt{c_-^2-c_+^2}}$.
The transition region between the dispersive shock wave and asymptotic solitons turned out to be very rich, and has been studied recently in \cite{BM}.

 \begin{figure}[htb]
 \begin{center}
\includegraphics[width=7cm]{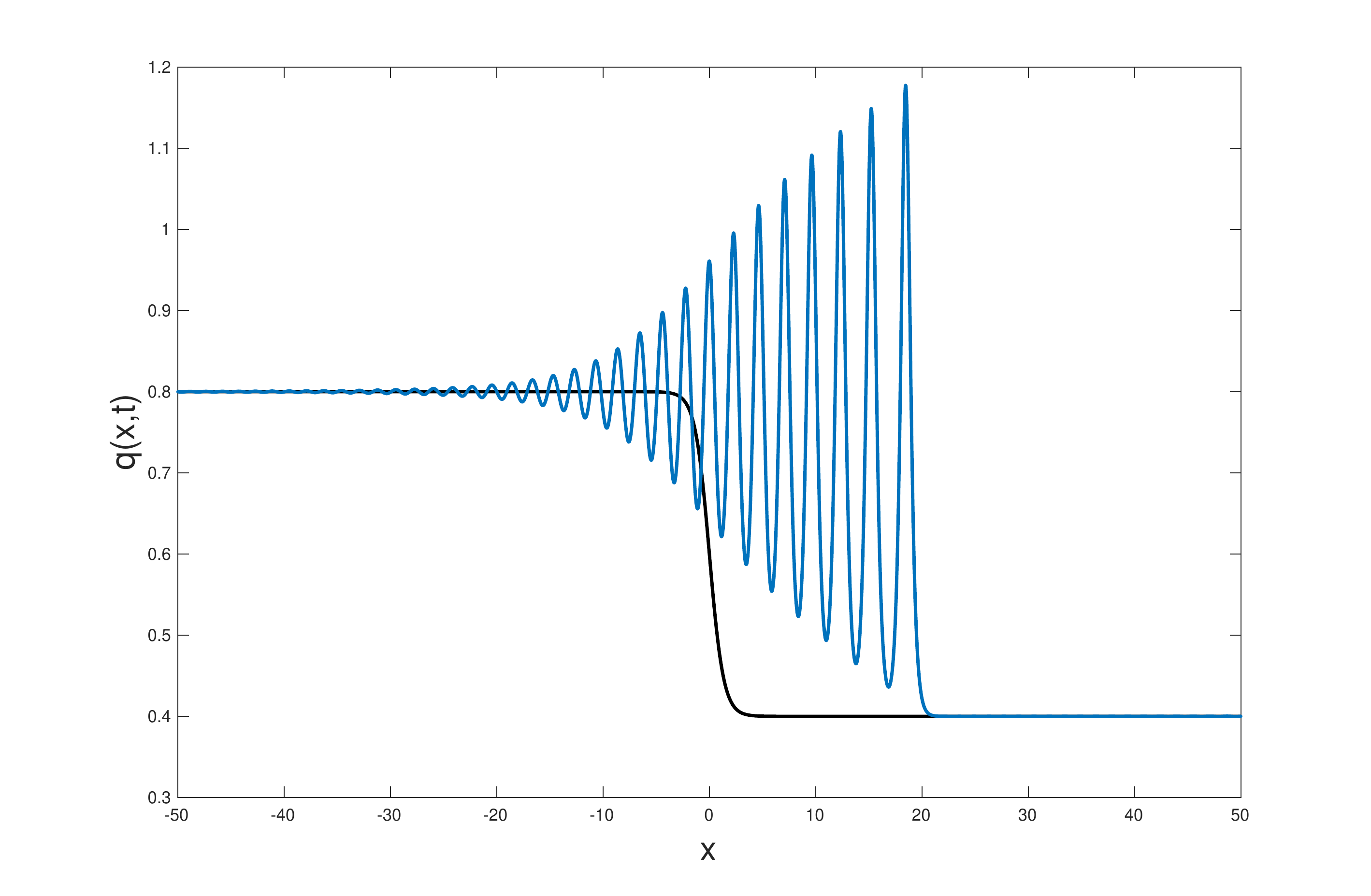}
\end{center}
\caption{The  (smoothed) step  initial data  \eqref{ic0exact}  in black for $c_-=0.8$  and $c_+=0.4$. The blue  line shows the evolution at time $t=15$.  }
 \label{fig2}
 \end{figure}
\end{Example}
\begin{Example}\label{Ex2}
 For an initial function in the form of a soliton on a constant background \eqref{soliton_intro} on the left, and a constant on the right, i.e. 
\begin{equation*}
q_0(x) = 
\begin{cases}
q_{soliton}(x,0;c_{\l}, \kappa_0,x_0),\quad x<0,\qquad \mbox{ where } \quad \kappa_0>c_{\l}>0, 
\\
c_{\r},\quad x>0,\quad \mbox{where } c_{\r}\in\mathbb{R}.
\end{cases}
\end{equation*}
Here $x_0=\log\dfrac{2(\kappa_0^2-c_-^2)}{|\nu|\kappa_0}\in\mathbb{R}$  and $\nu\neq 0$.

 The   evolution is shown in Figure~\ref{fig}, where it is shown that the initial  approximate "soliton" decomposes in two solitons that pass rapidly through the dispersive shock wave.
 
\begin{figure}[htb]
 \begin{center}
\includegraphics[ width=7.5cm]{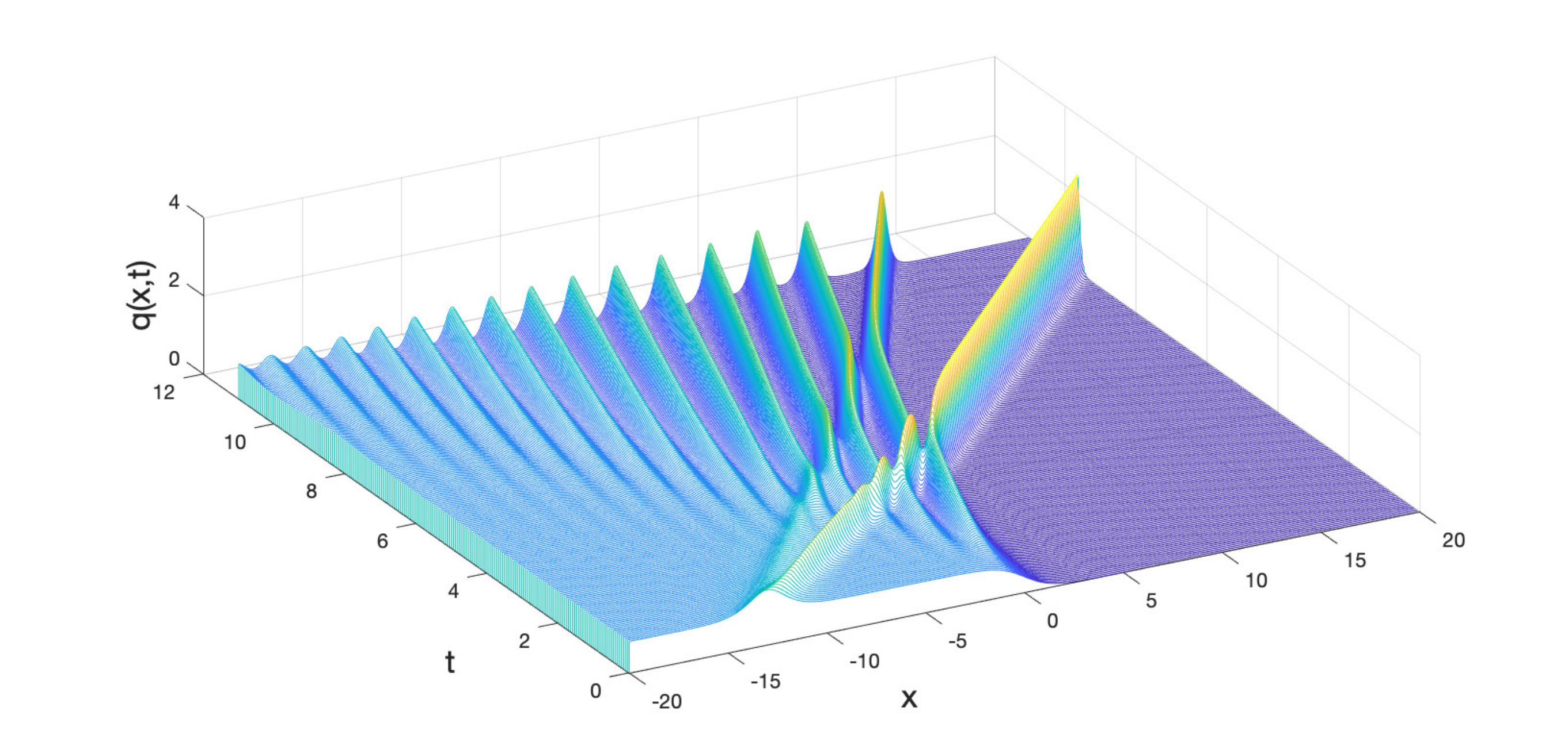}
\includegraphics[ width=5.0cm]{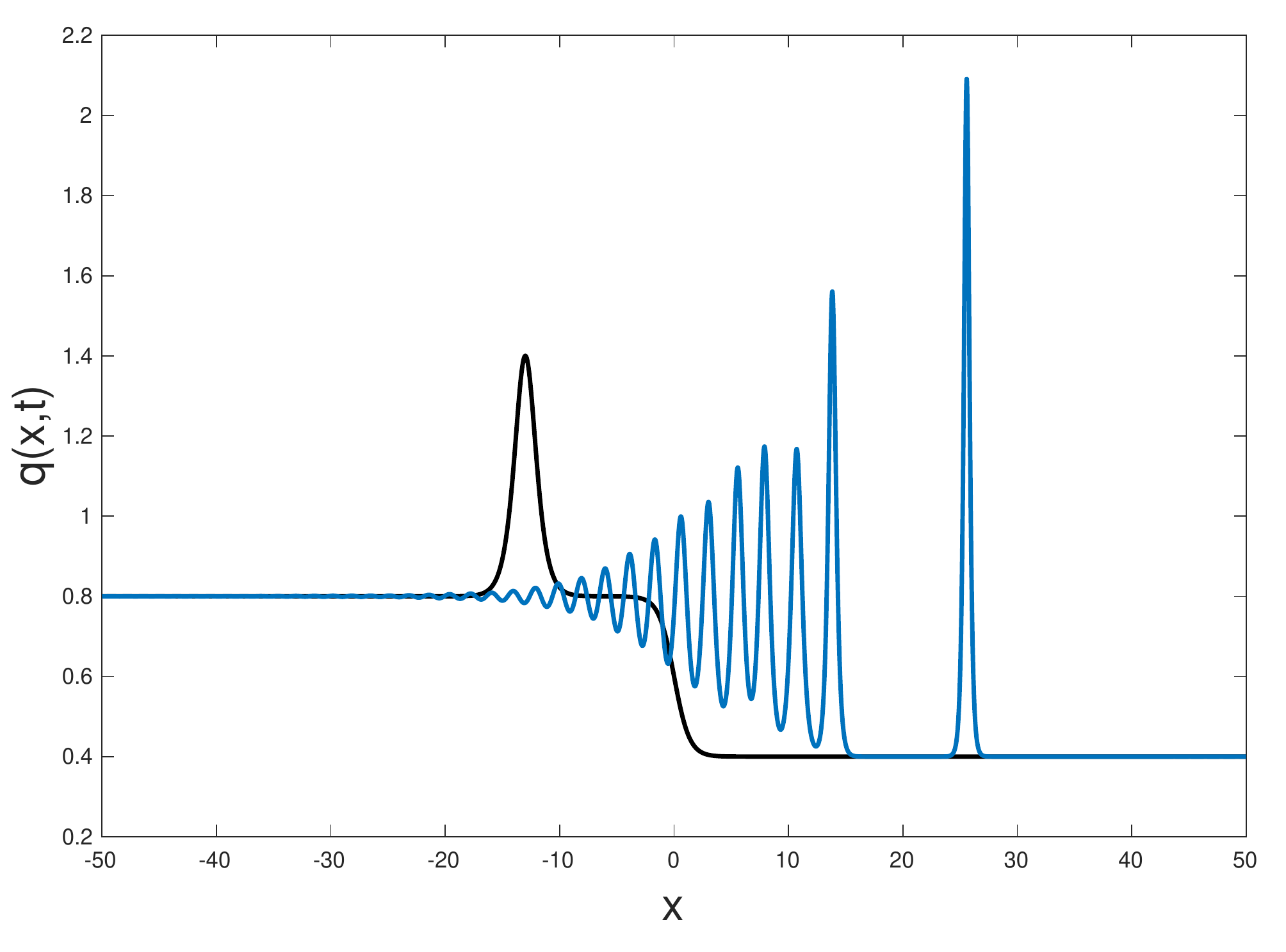}
\end{center}
\caption{The  (smoothed) initial data  of Example~\ref{Ex2}  for $c_-=0.8$  and $c_+=0.4$, $\kappa_0=1$ and $x_0=10$,  and $\nu>0$. On the left figure the evolution of $q(x,t)$ in the $(x,t)$ plane.
On the right  figure,  the black  line shows the  initial data and the blue line the evolution at time  $t=11$.  }
 \label{fig}
 \end{figure}

\end{Example}
 \begin{Example}\label{Ex3}
\color{black}{ Consider the initial function in the form of a  constant $c>0$ on the left, and a soliton on  a zero background on  the right, i.e. }
\begin{equation}
\label{solitonstep}
q_0(x)=\begin{cases}c>0,\quad x<0, \\ 
\frac{-2\kappa_0 \mathrm{sgn}(\nu)}{\cosh\left[2\kappa_0(x-x_0)\right]}, \quad x>0,\quad
\nu\in\mathbb{R}\setminus\left\{0\right\}, \ c>\kappa_0>0, \quad x_0=\frac{1}{2\kappa_0}\ln\left|\frac{\nu}{2\kappa_0}\right|.\end{cases}
\end{equation}
This initial data does not satisfy the decay at infinity of Assumption \ref{Assump1}. Below we  show that the reflection coefficient is singular on the imaginary axis $[\i c,0]$.
The reflection coefficient  is  $r(k)=\frac{b(k)}{a(k)}$  with 
\begin{equation*}
\begin{split}
&a(k) =\frac12 \left[\(\gamma(k)+\gamma(k)^{-1}\)\(1-\frac{\i\alpha}{k+\i\kappa_0}\)-\(\gamma(k)-\gamma(k)^{-1}\)\frac{\i\beta}{k+\i\kappa_0}\right],
\\
&b(k) = \frac12\left[\(\gamma(k)-\gamma(k)^{-1}\)\(1+\frac{\i\alpha}{k-\i\kappa_0}\)-\(\gamma(k)+\gamma(k)^{-1}\)\frac{\i\beta}{k-\i\kappa_0}\right],
\end{split}
\end{equation*}
where 
\begin{equation*}
\gamma(k)=\sqrt[4]{\frac{k-\i c}{k+\i c}},\quad \alpha = \frac{2\nu^2\kappa_0}{4\kappa_0^2+\nu^2},
\quad
\beta = \frac{4\nu\kappa_0^2}{4\kappa_0^2+\nu^2}.
\end{equation*}
The coefficient  $b(k)$ has  a pole at $k=\i \kappa_0$ and the  function $a(k)$ has two zeros $k_1, k_2$ in the half-plane  $\Im k\geq0$ :
\begin{equation*}
k_{1,2} = \dfrac{\pm\beta\sqrt{c^2+2c\beta-(\alpha-\kappa_0)^2}
+\i(\alpha-\kappa_0)(\beta+c)}{c+2\beta}
\end{equation*}
and they are symmetric w.r.t. imaginary axis: $k_1=-\ol{k_2}.$
Let us notice that when $x_0>0$ is not small, and hence $\nu>0$ is big, we have $\alpha\sim2\kappa_0,$ $\beta\sim 0,$ and hence 
\begin{equation*}k_{1,2}\sim \pm0+\i\kappa_0.\end{equation*}

Therefore the soliton  part of the initial data is  a breather from the spectral point of view,  with point spectrum very close to the imaginary axis.
The velocity of the breather is approximately $4\kappa_0^2<4c^2$   while the velocity of the dispersive shock wave is  $-6c^2<\frac{x}{t}<4c^2$.
Namely the breather has a  positive velocity that is smaller then the velocity of the leading front of the dispersive shock wave and it will remain trapped by the dispersive shock wave. Furthermore,   since $b(k)$ has a pole at $k=\i \kappa_0,$ also  the reflection coefficient has a pole at $k=\i \kappa_0$, and this pole requires a quite delicate asymptotic analysis
that is beyond the scope of the present article.
In the physical literature \cite{ElHoefer} this phenomenon received the  name of `soliton trapping' inside the dispersive shock wave  (see Figure~\ref{fig_trapped}).

\begin{figure}[htb]
 \begin{center}
\includegraphics[width=6.3cm]{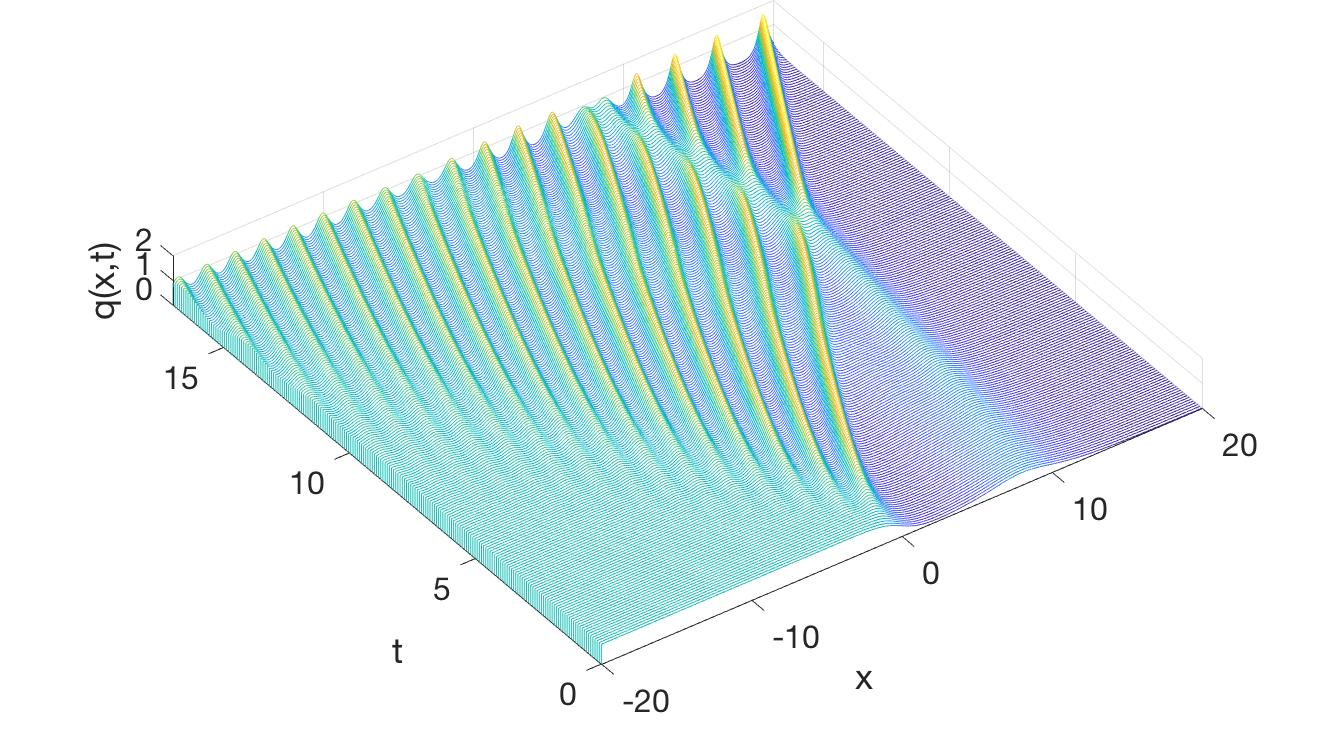}
\includegraphics[width=6.3cm]{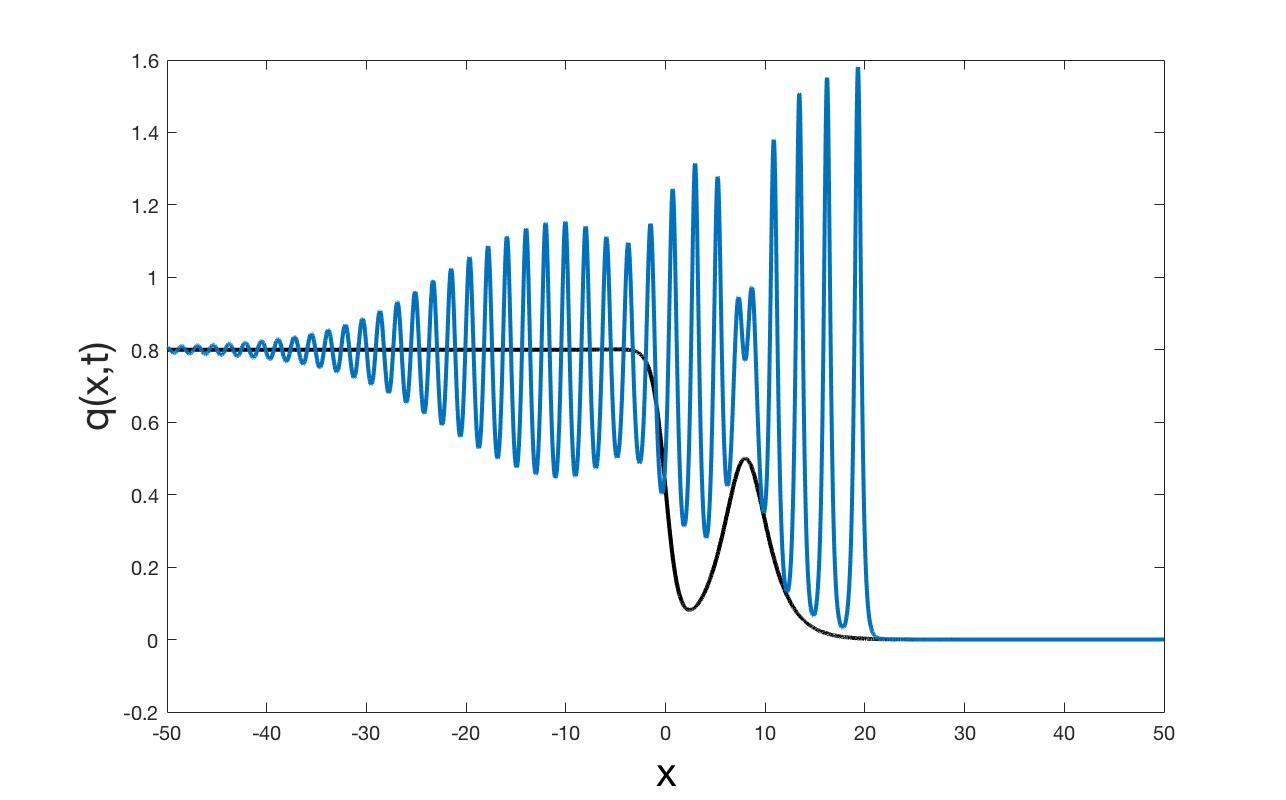}
\end{center}
\caption{The  initial data \eqref{solitonstep}  for $c_-=0.8$  and $c_+=0.0$, $\kappa_0=0.25 $. On the left, the 2D plot in the $(x,t)$ plane. On the right figure,   the black  line shows  the initial data, while the blue line  the evolution at time $t=17$.  }
 \label{fig_trapped}
 \end{figure}
%\begin{Remark}
% A similar data for the KdV equation gives a quite different result since the  point spectrum for the KdV equation  lies only the imaginary axis and a zero of $a(k)$ on the complex plane is simply a singularity of the reflection coefficient. When the zeros of $a(k)$ are very close to the imaginary axis, it follows that $a(k)$ is also very small on the   imaginary axis.   This implies  that the asymptotic analysis of the dispersive shock
%wave can be divided in two regimes: a regime when the singularities  near the imaginary axis $[0,\i c]$  of the reflection  coefficient play a role in the asymptotic analysis of the solution and a regime where these singularities have no effect.
%
%\end{Remark}
\end{Example}

This manuscript is organized as follow. 
We study the direct scattering problem and the main properties of the scattering data and the solvability of the associated RH  problem in section 2.
In section 3 we formulate  and solve  the RH   problems associated to the model problems that will be used in asymptotic analysis,
namely the  RH  problem for a breather and a soliton on a constant background and  then the model problems for the periodic 
 solution that is expressed via  hyperelliptic curves. We will then show how to reduce the  hyperelliptic solution to a travelling wave solution of MKdV.

In section 4 we start the asymptotic analysis introducing the $g$-function and performing the contour deformation and asymptotic analysis 
in the elliptic region and in section 5 we perform the contour deformation   and asymptotic analysis in the left and right constant regions,
to arrive at our main result, namely the 
proof of theorem~\ref{thrm:asymp:rl}.  In section 6 we calculate the subleading order correction of the left constant region thus proving  theorem~\ref{theorem_error}.
Several appendices are used  to prove the most technical results.

\section{Preliminaries}\label{sect_prel}

\color{black}

Let us briefly describe what will happen in this Section.

In section~\ref{sect_Lax} we review the Lax pair formulation of the MKdV equation and we construct the Jost solutions at time $t=0$.
In section~\ref{sect_scat}  we analyze the direct scattering problem following and  extending the derivation in \cite{KM} for our class of initial data.
In Section \ref{sect_RHP}
 we introduce the RH problem \ref{RH_problem_1}, where we imposed a time dependence of the  form $e^{i4k^3 t}.$ 
The way we arrive to the RH problem \ref{RH_problem_1} is heuristic, nevertheless, after the RH problem \ref{RH_problem_1} is stated, it serves as a basis for the rest of the section.
\color{black} 
In Section \ref{sect_solv} we follow \cite{KMshort} and, first, we prove that the solution of the RH problem exists and admits differentiation with respect to the parameters $x,t,$ and, second, that the columns of the solution of the RH problem \ref{RH_problem_1} are solutions of the overdetermined system which constitutes the Lax pair \eqref{x-eq}, \eqref{t-eq}. This allows to define a function $q(x,t),$ which satisfies the MKdV equation \eqref{MKdV} and equals to the initial datum at $t=0.$

\medskip

%Since the heuristic part is not a part of the rigorous analysis, but rather a suggestive consideration, in what follows we focus mostly on the rigorous part.
%Namely, Section \ref{sect_Lax} contains the Lax pair representation for the MKdV equation in the form of an overdetermined system of ordinary differential equations.
%This is the core of the analysis, however, it can not be used immediately, and will be used only after a preparatory work made in subsequent sections.
%
%In  Section \ref{sect_scat} we introduce the scattering data for the $x-$equation \eqref{x-eq} at the initial time $t=0,$ ignoring for a moment the time evolution equation \eqref{t-eq}. Then in 

\medskip

Note that the existence of RH problem's solution is guaranteed a priori by the Zhou's Schwartz reflection argument \cite{Zhou}, but in order to guarantee the differentiability in the parameters $x,t$ we need either to require analytic continuation of the reflection coefficient off the real axis, or to require sufficiently fast decay of the reflection coefficient at the infinite parts of the RH problem's contour. The former can be achieved by imposing the exponential decay \eqref{exponential}, and  the latter by requiring sufficient smoothness of the initial function. 
Neither of these function spaces is embedded in each other, but we chose to impose the exponential decay of the initial datum  rather than the smoothness condition. This allows us 
 to treat  the  important cases of discontinuous initial functions described in Examples \eqref{Ex1}-\eqref{Ex2}.

\medskip

%\todo{T. I would skip all the part below we can write it in the reply letter}
%Note also that the logics described above is not the only possible one, and we could have architectured the analysis in the following different way. We could start by assuming that the solution of the initial value problem \eqref{MKdV}, \eqref{ic0} exists for all values of time $t,$ and that for all values of time the solution $q(x,t)$ satisfies the first moment condition \eqref{integral}. This would allow to construct the scattering data for all values of $t,$ and to establish the evolution of the scattering data in time. We would thus nevertheless arrive at the RH problem \ref{RH_problem_1}, and the forthcoming steps of the analysis would coincide with the steps from the logics 1.
%
%These two logics are in fact inextricable parts of one big problem solving process, where the second logics constitutes the preliminary analysis, and the first logics is the justification of the results obtained by the latter. But of course because of lack of the space we can not afford to list in detail both parts, and we restrict ourselves to the justification part only, leaving in the air the question of how we arrive at the constructions which we use.

\medskip
\begin{remark}
%In the end let us note also that the `justification' logics allows to step outside of the restrictions set by the `analysis' logics.
Results similar to that of Theorem \ref{theorem_error} but written for large $x\to-\infty$ and  fixed $t$ allow to conclude that $|q(x,t)|\sim\frac{|r(-\frac{x}{12t})|}{\sqrt[4]{x}}\sim|x|^{-\frac{j}{2}-\frac34}$ for an initial datum   $q(x,0)$ whose $j$-th derivative is locally a function of bounded variation.  In the above the symbol $\sim$ means of the same order. Thus, the function $q(x,t)$ at a later time $t>0$ satisfies the first moment condition \eqref{first_moment} only when $j\geq 3,$ which is a much stronger  condition than the  conditions set in Assumption \ref{Assump1}.
Note that the classical direct analysis of the $x-$equation \eqref{x-eq} allows to guarantee the existence of the Jost solutions only under the convergence of the first moment \eqref{first_moment}.
\end{remark}
%We thus obtain a new class of potentials not satisfying \eqref{first_moment}, for which we still can guarantee the existence of `well-behaved' Jost solutions. This class can be characterised as functions, which under the evolution of MKdV flow transform at some value of time into functions satisfying \eqref{first_moment}.

%In subsequent section we will introduce the direct scattering for the equation \eqref{x-eq}. 
\color{black}

\color{black}
\subsection{Lax pair and  Jost solutions}
\label{sect_Lax}
\color{black}
The MKdV equation \eqref{MKdV} admits the Lax pair representation in the form 
\color{black}
of an over-determined system of linear ordinary equations for a $2\times2$ matrix-valued function $\Phi(x,t;k),$ where $k\in\C$ is the spectral variable
\color{black}
\cite{Wadati}, %\cite{Shabat}
\begin{eqnarray}\label{x-eq}
&&
\Phi_x(x,t;k)=(-\i k\sigma_3+Q(x,t))\Phi(x,t;k),\\ \label{t-eq}
&&\Phi_t(x,t;k)=(-4\i k^3\sigma_3+\widehat Q(x,t;k))\Phi(x,t;k),
\end{eqnarray}
where \color{black} 
$\Phi_x$ and $\Phi_t$ stand for partial derivative with respect to $x$ and $t$ respectively, $\sigma_3=\begin{pmatrix}1&0\\0&-1\end{pmatrix}$ and \color{black}
\begin{equation}\label{Q_Qhat}
       Q(x,t)=\begin{pmatrix}
          0 & q(x,t) \\ -q(x,t) & 0
         \end{pmatrix},\quad \widehat Q(x,t;k)=4k^2Q-2\i k(Q^2+Q_x)\sigma_3+2Q^3-Q_{xx}.
\end{equation}
\color{black} The compatibility condition $\Phi_{xt}=\Phi_{tx}$ is equivalent to the MKdV equation. \color{black}
If we substitute $q(x,t)=c$ (constant)  in \eqref{Q_Qhat}, then the Lax pair equations \eqref{x-eq}, \eqref{t-eq} admit  explicit solutions  \cite[page 3]{KM}
\begin{equation}\nonumber
 E_0(x,t;k)=\e^{-(i k x+4ik^3t)\sigma_3}, 
\end{equation}
for $c=0$ and 
\begin{equation}
\label{E_c}
\begin{split}
 E_c(x,t;k)&=\dfrac{1}{2}\begin{pmatrix}
\gamma(k)+\gamma^{-1}(k)&\gamma(k)-\gamma^{-1}(k) \\
\gamma(k)-\gamma^{-1}(k)&\gamma(k)+\gamma^{-1}(k)
 \end{pmatrix}\e^{-g_c(x,t;k)\sigma_3}, \quad
 % \color{black}{k\in\mathbb{C}\setminus[\i c,-\i c],}
 \\
 \mbox{ where }&\gamma(k)=\sqrt[4]{\frac{k-\i c}{k+\i c}},
  \quad g_c(x,t;k)=\i( x +2t(2k^2-c^2)\sqrt{k^2+c^2},
 \end{split}
 \end{equation}
 for  $c\neq 0;$ \color{black} the root functions  are defined in such a way that they have branch cut across the interval $[\i c,-\i c].$  In particular
 $\sqrt{k^2+c^2}$ is real for $k\in \R\cup[\i c,-\i c]$  and positive for $k= 0+$ where $0+$ means the non tangential  limit from the  right with respect to the imaginary axis. 
Here and in the rest of the manuscript  the imaginary axis is oriented downwards and for this reason we use the notation $[\i c,-\i c]$  for the oriented segment with endpoints $\pm\i c_-$.
\color{black}

 \subsection{Direct scattering}\label{sect_scat}
 Below we give a quick review of the direct scattering problem for step-like initial data  $q(x,0)$  \eqref{ic0}.

 \color{black} We denote $E^{\r}(x;k)\equiv E_{c_{\r}}(x,0;k)$   and $E^{\l}(x;k)\equiv E_{c_{\l}}(x,0;k)$  where $E_{c_{\pm}}(x,t;k)$ is the matrix defined \eqref{E_c}.
Now we   consider equation \eqref{x-eq} for $t=0$ and  define the {\it Jost solutions} $\Phi^{\pm}(x;k):=\Phi^{\pm}(x,0;k),$  which satisfy the equation \eqref{x-eq}, and have the (defining) property that 
\begin{equation}
\label{Jost0}
\Phi^{\pm}(x;k)=E^{\pm}(x;k)(1+o(1))\quad \mbox{as $x\to\pm \infty$},\;\; \color{black}{ k\in\mathbb{R}\cup[\i c_{\pm}, -\i c_{\pm}].}
\end{equation}
The Jost solutions  have the following integral representation
\begin{equation}
\label{Jostp}
\Phi^{\pm}(x;k)=E^{\pm}(x;k)+\int^{x}_{\pm\infty}L^\pm(x,y)E^\pm(y;k)dy,
\end{equation}
where the kernels $L^{\pm}(x,y)$ are independent from $k$ and they are studied  in the Appendix~\ref{sect_appendC}.
From the above expression  it is  clear that the Jost functions are continuous  in the variable $k$ for $k\in\mathbb{R}\cup(\i c_{\pm}, -\i c_{\pm}).$
The following two Lemmas appeared before in \cite[section II]{KM}, \cite[section 2]{KM2}  under the condition $c_+=0$ or  $c_+>0$ for  the exact step,  ( see also \cite{AK91}) but their proof was skipped there.%, but   with much stronger assumptions on the initial data.
\color{black}
%
%.in different perspective. Namely, there it was supposed that the solution of the MKdV equation exists for all $t,$ and that for all $t$ the solution $q(x,t)$ is convergent with the first moment to its background constants,
%\begin{equation*}\int_{-\infty}^0(1+|x|)|q(x,t)-c_-|\d x+
%\int^{+\infty}_0(1+|x|)|q(x,t)-c_+|\d x<\infty.\end{equation*}
%This assumption was made in order to ensure existence of Jost solutions for all $t,$ and hence the solvability of the underlying RHP.
%However, this assumption is quite restrictive, and is not satisfied for example in the case of pure  step initial data \cite{KM}. Hence, here we change the strategy. We require the convergence of the first moment only for the initial datum,
%\begin{equation*}\int_{-\infty}^0(1+|x|)|q_0(x)-c_-|\d x+
%\int^{+\infty}_0(1+|x|)|q_0(x)-c_+|\d x<\infty,\end{equation*}
%show that this is enough to define all the spectral coefficients $a(k), b(k), r(k)$ and hence to state the associated RHP. Then, using Schwartz symmetries of the latter, we obtain its solvability, smoothness, and thus we construct a solution of the initial value problem for the MKdV. 
%Hence, the properties of the Jost solutions are studied in the following two Lemmas only for $t=0.$
%
\begin{lemma}\label{lem_prop_Jost}
 Suppose the the initial data $q(x,0)$ satisfies Assumption~\ref{Assump1}  and Assumption~\ref{Assump2}. Then the  Jost solutions  \eqref{Jostp}  their columns $\Phi^{\pm}_{j}(x;k)$,
 $j=1,2$ and their entries $\Phi^{\pm}_{ij}(x;k)$ have the following
properties.
\begin{enumerate}\label{propPhiPsi}
  \setcounter{enumi}{0}
  \item   $\det\Phi^{\pm}(x;k)=1$.
  \item
  $\Phi_{1}^+(x;k)$ is analytic in $k\in
   \Db{C}^-\setminus[0,-\i c_{\r}]$,  and continuous  up to the boundary\\
  $\Phi_{2}^+(x;k)$ is analytic in $k\in
   \Db{C}^+\setminus[\i c_{\r},0]$  and continuous  up to the boundary,\\
  $\Phi_{1}^-(x;k)$ is analytic in $k\in
   \Db{C}^+\setminus[\i c_{\l},0]$ and continuous   up to the boundary\\
  $\Phi_{2}^-(x;k)$ is analytic in $k\in
    \Db{C}^-\setminus[0,-\i c_{\l}]$ and continuous   up to the boundary.
    \\
        \quad here $\mathbb{C}^{\pm} = \left\{k:\ \pm\Im k>0\right\};$
    \item    Let  $\Phi(x;k)$ denote  either $\Phi^{\l}(x;k)$ or $\Phi^{\r}(x;k)$. By the reality of $q(x)$ we have the symmetries:
  \begin{equation*}
 \begin{split}
 \widebar{\Phi_1(x;\widebar{k})  }&=\begin{pmatrix}0&1\\-1&0\end{pmatrix} \Phi_2(x;k),\\
    \Phi_1(x;-k)&=\begin{pmatrix}0&1\\-1&0\end{pmatrix} \Phi_2(x; k),\\
     \widebar{\Phi(x;-\widebar{k})}&= \Phi(x,k).
           \end{split}
       \end{equation*}
   
  \item Large $k$ asymptotics:
   \begin{equation*}
\begin{pmatrix}\Phi^{\r}_{ 1}(x;k)e^{+ikx},&
   \Phi^{\l}_{2}(x;k)e^{-ikx}\end{pmatrix}
   =I+{\mathrm{O}}\left(\displaystyle\frac{1}{k}
       \right),\quad k\rightarrow\infty, \quad \Im k\leq 0,
  \end{equation*}
 \begin{equation*}
\begin{pmatrix}\Phi^{\l}_{1}(x;k)e^{+ikx},&
  \Phi^{\r}_{ 2}(x;k)e^{-ikx}\end{pmatrix}
  =I+{\mathrm{O}}\left(\displaystyle\frac{1}{k}
      \right),\quad k\rightarrow\infty, \quad \Im k\ge0,
 \end{equation*}
 where $I$ stands for the $2\times 2$ identity matrix.
  \item Jump conditions:  \\
  $\Phi_-(x;k)=\Phi_+(x;k)\left(
  \begin{array}{ccc}0&i\\i&0\end{array}\right),\quad
  k\in(\ii c,-\ii c)$,\\
  where $\Phi(x;k)$ and $c$ denote $\Phi^{\l}(x;k)$ and $c_{\l}$
  or $\Phi^{\mathfrak{\r}}(x;k)$ and $c_{{\r}}$, respectively,
  and $\Phi_\pm(x; k)$      denotes the non-tangential limits of the matrix
  $\Phi(x;k)$  from the left ($+$) and from the right ($-$) of the  segment $(\ii c, -\ii c)$  which is oriented from $\i c$ to $-\i c$ (i.e., $\Phi_{\pm}(x;k) = \lim\limits_{\epsilon\to+0}\Phi(x;k\pm\epsilon)$).
\end{enumerate}
\end{lemma}
\begin{proof}\color{black}{The full proof of the Lemma is 
%very similar to the case considered in \cite{KM} and is 
given in the Appendix~C.}
Here we will only prove properties 1,  3  and 5.\\
{\bf Property 1. }
The functions $\Phi^{\pm}$ satisfy the $x-$equation \eqref{x-eq} for $t=0$; since the trace of the function $-\i k \sigma_3+Q$ equals 0, by Liouville's formula for determinants, it follows that   $\det\Phi^{\pm}(x;k)$ are independent of $x.$ From the large $x$ asymptotics   \eqref{Jost0} we then obtain that the determinants are equal to $1.$ \\
{\bf Property 3.}
Let us write the  $x-$equation \eqref{x-eq} element-wise,
\begin{equation*}
\begin{split}
&\begin{cases}
\Phi_{11,x}(x;k) + \i k \Phi_{11}(x;k) = q(x) \Phi_{21}(x;k),
\\
\Phi_{21,x}(x;k) - \i k \Phi_{21}(x;k) = -q(x) \Phi_{11}(x;k),
\end{cases}\\
&\begin{cases}
\Phi_{12,x}(x;k) + \i k \Phi_{12}(x;k) = q(x) \Phi_{22}(x;k),
\\
\Phi_{22,x}(x;k) - \i k \Phi_{22}(x;k) = -q(x) \Phi_{12}(x;k).
\end{cases}
\end{split}
\end{equation*}
Changing in the last two equations $k\to- k$ and using the large $x$-asymptotics of $\Phi_{ij}, i,j=1,2$, we get that $\Phi_{12}(x;-k)=-\Phi_{21}(x;k),\quad \Phi_{11}(x;k)=\Phi_{22}(x;-k)$ and $\Phi_{ij}(x;-\overline{k}) = \overline{\Phi_{ij}(x;k)},\  i,j=1,2.$ \\
Indeed, the two pairs of vector-valued functions,
$(\Phi_{22}(x;-k), -\Phi_{12}(x;-k))$ and \\ $(\ol{\Phi_{11}(x;-\ol k)},  \ol{\Phi_{21}(x;-\ol k)})$ satisfy the same systems of equations as $(\Phi_{11}(x;k), \Phi_{21}(x;k))$ and the same large $x$ asymptotics, and hence coincide with the latter.\\
 {\bf Property 5}. We use the defining property \eqref{Jost0} of the functions $\Phi^{\pm}$ and note that the limiting values of $E^\pm$ on the two  sides of the  oriented  interval $(\i c_\pm,-\i c_{\pm})$ satisfy
$E^\pm_-(x;k)=E^\pm_+(x;k)\begin{pmatrix}0&\i \\ \i & 0\end{pmatrix}.$
Thus, both functions $\Phi^\pm_-(x;k)$ and $\Phi^\pm_+(x;k+0)\begin{pmatrix}0&\i \\ \i & 0\end{pmatrix}$ satisfy the same differential equation \eqref{x-eq} (at $t=0$), and the same asymptotics as $x\to-\infty,$ and thus they are equal.

\end{proof}
\color{black}

% {\color{orange}[M: is this an assumption or it is a theorem?]}
% {\color{black}[S: it is an assumption. We need this in order to be able to construct Jost solutions for all times, and hence to formulate RHP. However, (\ref{first_moment}) does not hold in the case of
% discontinuous initial data, when we take $q_0=c$ for $x<0$ and $q_0(x)=0$ for $x>0.$ But in this case we still have RHP, which in turn produces solution of MKdV with given initial data $q_0(x)$.
% I wrote about this in the end of comment 2.1 on p. 7. It would be nice to have a characterization of initial data in terms of $a(k)$, $r(k)$, but I don't know such a characterization.
% By the way, since the Lax pair operator is not self-adjoint, maybe $a(k)$ might have non simple zeros. I do not know existence theorems for Cauchy problem for MKdV(+) with step-like initial data.
% For other MKdV(-) equation, and also for KdV, there is an existence theorem by Egorova, Teschl, etc.]}
Since the matrix-valued func\-tions $\Phi^{\pm}(x ;k)$ 
 are solutions of the first order
dif\-fer\-en\-tial equation (\ref{x-eq})%and (\ref{t-eq})
, they are
linearly dependent, i.e. there exists an $x$-independent {\it transition matrix}
$T(k)$, such that
%\todo{A: I removed the $t$-dependence from here}
\begin{equation}\label{Tr}
T(k)=(\Phi^{{\r}}(x;k))^{-1}\Phi^{{\l}}(x;k),\quad
k\in\mathbb{R} \cup(\i c_{\r},-\i c_{\r}).
\end{equation}
Due to symmetries of the $x$-equation \eqref{x-eq} the transition matrix has the following structure:
\begin{equation}
\label{T_a_b}T(k)=\begin{pmatrix}
        a(k)& -\ol{b(\ol{k})} \\ b(k) & \ol{a(\ol{k})}
       \end{pmatrix},\qquad a(k)=\det\(\Phi^{\l}_{1},\Phi^{\r}_{ 2}\),\ b(k)=\det\(\Phi^{\r}_{1},\Phi^{\l}_{ 1}\).
\end{equation}
\color{black}{Note that the representation \eqref{T_a_b}  and Lemma~\ref{lem_prop_Jost}  allow to extend the domains of definition of the spectral functions $a,b$: the  function $a(k)$ is analytic for $k\in\mathbb{C}\setminus[\i c_-,0],$ and is defined also on both sides of the interval $(0,-\i c_{\r})$, and $b(k)$ is defined for $k\in\mathbb{R}\cup(\i c_+, -\i c_-).$}\\
{\color{black} By property 1  of Lemma~\ref{lem_prop_Jost}, we have that $\det T(k)=1$ so that
\begin{equation}
\label{det1}
a(k) \ol{a(\ol{k})}+b(k)\ol{b(\ol{k})}=1.
\end{equation}
}
\noindent
The functions $a^{-1}(k)$ and $$r(k):=\frac{b(k)}{a(k)},\quad k\in\mathbb{R} \cup(\i c_+,-\i c_+)$$ 
%\todo{Check the domain where $r(k)$ is defined}
are called the (right) {\it transmission} and {\it reflection} coefficients, respectively.
Below we sumarize the analytical properties of the functions $a(k)$, $b(k)$ and $r(k)$.  
We denote by 
$a_\pm(k)$, $b_\pm(k)$ and $r_\pm(k)$ the non-tangential limits of  $a(k)$, $b(k)$ and $r(k)$ 
   from the left ($+$) and from the right ($-$) of the  segment $(\ii c_-, -\ii c_-)$  which is oriented from $\i c_-$ to $-\i c_-.$

%\color{black} We observe that because of Lemma~\ref{lem_prop_Jost}, property 2. the quantity $a(k)$ is analytic in $\mathbb{C}^+\backslash[\i c_-,0]$ and it is continous up to the boundary and 
%$b(k)$ is analytic in $\mathbb{C}^-\backslash[0,-\i c_-]$ and it is continuous up to the boundary.

%We denote by $a_{\pm}, r_{\pm}$ the limiting values of the functions $a, r$ on the banks $[\i c_-,-\i c_-]\pm0$ of the segment $[\i c_-,-\i c_-],$ respectively.
 
\color{black}
\begin{lemma}\label{lem_abr}
Suppose that  the initial data $q(x,0)$ satisfies Assumption~\ref{Assump1}  and Assumption~\ref{Assump2}. 
Then the spectral functions $a(k), b(k), r(k)$ have the following properties enumerated below.
\begin{enumerate}\label{prop_a_r}
\setcounter{enumi}{5}
\item Analyticity: \\
$a(k)$ is analytic for  $k\in\mathbb{C}^+ \setminus[\i c_{\l},0]$, and it can be extended continuously up to the boundary, 
with the exception of the points $\i c_{\l},$ $\i c_{\r},$ where $a(k)$ may have at
most a fourth root singularity $(k-\i c_{\pm})^{-1/4}$ \color{black}{(i.e. the function $a(k)(k-\i c_{\pm})^{1/4}$ is bounded as $k\to\i c_{\pm}$).}
The function $a(k)$ might have at most a finite number of zeros for  $k\in\mathbb{C}_+ \setminus[\i c_{\l},0]$.\\
The function $b(k)$   is defined for  $k\in\mathbb{R}\cup(\i c_{{\r}},-\i c_+)$.
%analytic for  $k\in\mathbb{C}^- \setminus[0,-\i c_{+}]$ and it can be extended continuously up to the boundary.
The function $r(k)$ is defined for  $k\in\mathbb{R}\cup(\i c_{{\r}},-\i
c_{{\r}})$ except for the points where $a(k)=0$.
\\
 Under the condtion \eqref{exponential},
 the function $a(k)$  and $b(k)$ have an analytic extension  in the $\delta-$neighbor\-hood of $\Sigma=\mathbb{R}\cup[\i c_{\l},-\i c_{\l}]$, where $\delta= \sqrt{\sigma^2+c_+^2}-c_-$ and $\sigma>\sqrt{c_-^2-c_+^2}>0$, with $\sigma$ as in Assumptions~\ref{Assump1}. The function $r(k)$ is meromorphic in the same domain  with poles at the zeros of $a(k).$

\item Asymptotics:\\
\begin{equation*}\begin{split}&a(k)=1+\mathcal{O}(k^{-1}) \quad \textrm{ as }\quad k\to\infty, \Im k\geq 0;
\\&r(k)=\mathcal{O}(k^{-1}).\end{split}\end{equation*}
%\todo{?}

 \item Symmetries: in their domains of  definition $$\ol{a\(-\ol{k}\)}=a(k),\quad \ol{b\(-\ol{k}\)}=b(k), \quad \ol{r\(-\ol{k}\)}=r(k).$$

% \item $\forall k\in(\i c,0] \ \ \ a_{\pm}(k)\neq0,$ \\where
% $a_{+}(k)$ and $a_-(k)$ are the boundary values of the function $a(k)$
% from the right and from the left of the interval $(\i c,0)$,
% respectively.
\item On
 $(\i c_{\r},-\i c_{\r}):$\\
$$a_-(k)=\ol{a_+\(\ol{k}\)},\quad b_-(k)=-\ol{b_+\(\ol{k}\)},\quad r_-(k)=-\ol{r_+\(\ol{k}\)},$$
and on $\color{black}{(\i c_{\l},\i c_{\r})}:$
$$a_-(k)=-\i\ \ol{b_+\(\ol{k}\)},\quad a_+(k)=\i\ \ol{b_-\(\ol{k}\)}.$$
\item \label{prop10} Let \begin{equation*}f(k):=\displaystyle\frac{i}{a_-(k)a_+(k)},\quad k\in (\i
c_{{\l}},0),\end{equation*}
then  on $(\i c_{\l},\i c_{\r}):$
$$f(k)=\frac{\i}{a_-(k)a_+(k)}=r_-(k)-r_+(k)=\frac{-1}{a_+(k)\ol{b_+\(\ol{k}\)}}=\frac{1}{a_-(k)\ol{b_-\(\ol{k}\)}},$$
and on $(\i c_{\r},0):$
$$f(k)=\frac{\i}{a_-(k)a_+(k)}=\i (1-r_-(k)r_+(k)).$$

  %and  $(\Phi^{\l}_{2})_{\pm}(k)=\lim_{\varepsilon\to 0}\Phi^{\l}_{2}(k\pm \varepsilon)$  where $\Re\varepsilon>0$ and $k\in [0,-\i c_-]$. 
\end{enumerate}
\begin{enumerate}
\setcounter{enumi}{10}
%\noindent Spectral coefficients $a(k), b(k), r(k)$ have the following properties (in some of the properties listed below we employ the analyticity of $b(.)$ in a vicinity of $[\i c_{\l},-\i c_{\l}]$):
\item If $c_+>0,$ then $|a(0)|=1,$ and $|a(k)|\leq 1$ for $k\in\mathbb{R}\setminus\left\{0\right\}$ and $|a(k)|\geq1$ for $k\in(\i c_+,0).$ 
%\todo{Is the last statement  in 11. true?}
\item \label{prop_nonvanish} Nonvanishing. 
The coefficients $a(k)$ and $\ol{b(\ol k)}$ do not vanish on the segment $(\i c_{\l},\i c_{\r});$

\item 
$\displaystyle\frac{(\Phi^-_{1}(x;k))_-}{a_-(k)}-
      \displaystyle\frac{(\Phi^{\l}_{1}(x;k))_+}{a_+(k)}=f(k)
      \Phi^{\r}_{2}(x;k),
      \quad k\in(\i c_{{\l}},\i c_{{\r}}),$\\
   and  $\displaystyle\frac{(\Phi^{\l}_{2}(x;k))_-}{\overline{a_-
  (\overline{k})}}-
      \displaystyle\frac{(\Phi^{\l}_{2}(x;k))_+}{\overline
      {a_+(\overline{k})}}=
      -\ol{f(\ol{k})}\Phi^{\r}_{1}(x;k),\quad k\in(-\i c_{{\r}},-\i
      c_{{\l}})\,.$

\end{enumerate}
\end{lemma}
\begin{proof}\color{black}
\textbf{Property 6.}  The analytic properties of $a(k)$ and $b(k)$ follows from the analytic properties of the Jost solutions derived in Lemma~\ref{lem_prop_Jost}.  The spectral coefficients $a(k)$ and $b(k)$ are defined by the Jost solutions by formulas \eqref{T_a_b}, and in turn, Jost solutions have representation \eqref{Jostp}. Since $E^{\pm}(x;k)$ have fourth-root type singularity at the points $k=\i c_{\pm}$ (i.e. they are bounded after multiplication by $(k-\i c_{\pm})^{1/4}$), the functions $a(k)$ and $ b(k)$ have {\it at most} fourth-root type singularities at $\i c_{\pm}.$
The analytic extension of $a(k)$, $b(k)$ and $r(k)$ is discussed in the Appendix~\ref{sect_appendC}, Corollary~\ref{cor_analyticity_Jost}.\\
\textbf{Property 7.}  It can be obtained by writing   the definition \eqref{T_a_b}  of   the functions $a(k)$  and  $b(k)$ and then using   property 1 and  the 
 large $k$ behaviour of  the  Jost solutions in   property 4 of Lemma \ref{lem_prop_Jost}. \\
\textbf{Property 8.} It  follows from the corresponding symmetries of Jost solutions (property 3 of Lemma \ref{lem_prop_Jost}).\\
 \textbf{Property 9}. We use property 5 of Lemma \ref{lem_prop_Jost}, and this allows to interrelate the limiting values  $T_\pm(k)$ using the definition 
$\Phi_\pm^-(x;k) = \Phi^+_{\pm}(x;k)T_\pm(k).$

For $k\in(\i c_-, -\i c_-),$ we have $\Phi^-_-(x;k)=\Phi^-_+(x;k)\begin{pmatrix}0&\i\\\i&0\end{pmatrix}$ and 
for $k\in(\i c_+, -\i c_+),$   we have   $\Phi^+_-(x;k)=\Phi^+_+(x;k)\begin{pmatrix}0&\i\\\i&0\end{pmatrix}$. 
Hence 
 we obtain $$T_-(k)=\begin{pmatrix} 0 & -\i \\ -\i & 0 \end{pmatrix}T_+(k)\begin{pmatrix} 0 & \i \\ \i & 0 \end{pmatrix},\quad k\in(\i c_+,-\i c_+)$$
and 
$$T_-(k)=T_+(k)\begin{pmatrix} 0 & \i \\ \i & 0 \end{pmatrix},\quad k\in(\i c_-, -\i c_+)\cup (-\i c_+,-\i c_-).$$
Considering the matrix entries of the above relations  completes the proof of property 9.\\
{\bf Property 10}. It follows  in a straightforward way from property 9.\\
{\bf Property 11}. Note that the absolute values of $a(k)$ and $b(k)$ do not have jump across the interval $(\i c_-,0)$, i.e. $|a_{+}(k)|=|a_-(k)|$, $|b_{+}(k)|=|b_-(k)|$.
The relation \eqref{det1} implies $|a(k)|^2+|b(k)|^2=1$   for $k\in \R$ while from property 9 we have that $|a(k)|^2= |b(k)|^2+1$ for $k\in (\i c_+,-\i c_+)$.
Combining the two relations together,  we obtain that for $c_->c_+> 0$,  $|a(k)|<1 $ for $k\in\R\backslash\{0\}$  while $|a(k)|>1$ for $k\in (\i c_+,0)$.
For $c_->c_+>0,$ the functions $a(k)$ and $b(k)$ do not have root-type singularities at the origin, and thus by continuity we obtain $|a(0)|=1$.\\
\textbf{Property 12.} By property 8, $a(k)\ol{a(\ol k)}+b(k)\ol{b(\ol k)}=1.$ Taking here $k$ on the positive side  of the  oriented  segment $(\i c_-, \i c_+),$ we obtain $a_+(k)\ol{a_+(\ol k)}+b_+(k)\ol{b_+(\ol k)}=1.$ By property 9, on the interval $(\i c_-,\i c_+)$, we can express the limiting values of the spectral coefficient $b$ in terms of the limiting values of the spectral coefficient $a;$ thus $-\i a_+(k)b_-(k) + \i a_-(k) b_+(k) = 1.$ By property 8, we have that the limiting values of $a,b$ from different sides of the interval $(\i c_-,\i c_+)$ are complex conjugates of each other, i.e. $a_+(k)=\ol{a_-(k)},$ $b_+(k)=\ol{b_-(k)},$ and hence $-\i a_+(k)\ol{b_+(k)} + \i \ol{a_+(k)}b_+(k)=1$ and $\Im(a_+(k)\ol{b_+(k)}) = \frac12.$ Hence, neither  $a(k)$ nor $b(k)$ can vanish on the interval $(\i c_-,\i c_+).$\\
%Note that this proof uses the existence of analytic extensions of the functions $a, b$ in a neighborhood of the segment $[\i c_-,0]$ but the result can also be proved without assumption on analytic extension (cf. \cite{BET})
{\bf Property 13.} We first express $\Phi^+(x;k)$ as  $\Phi^+(x;k) = \Phi^-(x;k)T(k)^{-1},$  by  formula \eqref{T}, and then express $\Phi^+_2(x;k)$ from the latter,
\begin{equation*}\Phi^+_2(x;k) = \ol{b(\ol k)}\Phi^-_1(x;k) + a(k)\Phi^-_2(x;k).\end{equation*}
The expression on the left-hand-side does not have a jump on the contour $(\i c_-, \i c_+),$ but every term on the right-hand-side has. Let us take the limit of the r.h.s. from the positive side of the contour $(\i c_-, \i c_+),$ and then substitute $\ol{b_+(\ol k)} = \i a_-(k)$ using property 11, and $(\Phi_2^-(x,t;k))_+=-\i (\Phi_1^-(x,t;k))_-$ from property 5 of Lemma \ref{lem_prop_Jost}. This concludes the proof of the first relation. The second   relation is obtained in a similarl way.

\end{proof}
\color{black}

%\begin{proof}
%Properties  6 follow from the definition of the functions $a(k),b(k), r(k)$ and the corresponding properties of $\Phi_{\pm}(x;k).$ Furthermore, the asymptotics 7 follow from Corollary \ref{cor_analyticity_Jost}. 
%Symmetries 8 follow from symmetries of $\Phi_{\pm}$ and the determinantal property  follows from the fact that $\det T \equiv 1.$
%Properties 10,11, 12,13 follow from property 6 of Lemma \ref{lem_prop_Jost}.
%Finally, property 9 follows after using the relations 
%\begin{equation}
%a(k)\ol{a(\ol{k})}+b(k)\ol{b(\ol{k})}=1,\quad
%\ol{a(-\ol{k})}=a(k), \quad \ol{b(-\ol{k})}=b(k), \quad a_-(k)=-\i
%\ol{b_+(\ol{k})},
%\label{extrasym}
%\end{equation}
%which imply that $a_{\pm}(k)\neq 0, b_{\pm}(k)\neq 0$ on $(\i c_{\l},\i c_{\r})\cup(-\i
%c_{\r},-\i c_{\l}).$
%\end{proof}
%\QED
Let us denote the first and second columns of the Jost solution $\Phi^-$ and $\Phi^+$ in \eqref{Jostp} as
\begin{equation}
\label{psipm}
\Phi^{-}_{1}(x;k)=\begin{pmatrix}
\varphi^{\l}(x;k)\\
\psi^{\l}(x;k)
 \end{pmatrix},\quad \Phi^{+}_{2}(x;k)=\begin{pmatrix}
 \varphi^{\r}(x;k)\\
 \psi^{\r}(x;k)\end{pmatrix}\,,
 \end{equation}
 so that 
\begin{equation*}
\Phi^{-}(x;k)=\begin{pmatrix}
\varphi^{\l}(x;k) & -\ol{\psi^{\l}(x;\ol k)} \\
\psi^{\l}(x;k) & \ol{\varphi^{\l}(x;\ol k)}
 \end{pmatrix},
 \quad
 \Phi^{+}(x;k)=\begin{pmatrix}
\ol{\psi^{\r}(x;\ol k)} & \varphi^{\r}(x;k)\\
-\ol{ \varphi^{\r}(x;\ol k)} &  \psi^{\r}(x;k)\end{pmatrix}\,.
 \end{equation*}

\begin{lemma}\label{lem_zeros}
Denote the zeros of $a(k)$ in the \color{black} quarter-plane $\left\{k\in\mathbb{C}:\ \Im k> 0\\ \mbox{ and } \Re k\geq 0\right\}$ \color{black} by $\kappa_1, \dots,\kappa_N.$ We have $\kappa_j\notin(\i c_-,0]$ for $j=1,\dots,N$. \color{black} If the zeros of $a(k)$ are simple then the residues of $a^{-1}(k)$ are given by
\begin{equation}
\label{resa}
\mathrm{Res}_{k=\kappa_j}a^{-1}(k) = \i\frac{\nu_j}{\mu_{j}},\quad j=1,\dots,N
\end{equation}
where $$(2\nu_j)^{-1}:=\int\limits_{-\infty}^{+\infty}\varphi^{+}(x;\kappa_j)\psi^{+}(x;\kappa_j)\d x$$ and
$$\varphi^{\l}(x;\kappa_j)=\mu_j\varphi^{\r}(x;\kappa_j),\quad \psi^{\l}(x;\kappa_j)=\mu_j\psi^{\r}(x;\kappa_j),$$
with $\psi^{\pm}$ and $\varphi^\pm$ as in \eqref{psipm}.
%
%\item 
% Let $\kappa_j,$ $\mathrm{Im}\kappa_j>0,$ $\kappa_j\notin[\i c_{\l},\i c_{\r}]$ be a zero of $a(.)$ of the $n_j^{th}$ order, $n_j\geq1,$ i.e.
%$$a(\kappa_j)=\ldots=a^{(n_j-1)}(\kappa_j)=0,\quad a^{(n_j)}(\kappa_j)\neq0;$$
%let the coefficients $\mu^{(j)}_0,\ldots,\mu^{(j)}_{n-1}$  (independent of $x,t$!) be such that 
%\begin{equation}\label{f_l_in_f_r_intr}\partial^m_k \Phi_{-,1}(\kappa_j)=\sum\limits_{p=0}^{m}C_{m}^{p}\mu^{(j)}_p \partial^{m-p}_k \Phi_{\r, 2}(\kappa_j),\quad m=0,\ldots,n-1, \qquad\textrm{ where }C_{m}^{p}:=\frac{m!}{p!(m-p)!}
%\end{equation}
%are the binomial coefficients,
%and let $T_1(x,t),\ldots,T_{n}(x,t)$ be the coefficients in the Taylor expansion of the function 
%${a(k)^{-1}} \e^{2\i\theta(x,t;k)-2\i\theta(x,t;\kappa_j)}$ at the point $k=\kappa_j,$ i.e.
%\begin{equation}\label{a_decomposition_intr}
%\frac{\e^{2\i\theta(x,t;k)}}{a(k)}=\sum\limits_{q=1}^{n}\frac{T_{q}(x,t)\ \e^{2\i\theta(x,t; \kappa_j)}}{(k-\kappa_j)^q}+\mathcal{O}(1),\quad k\to \kappa_j. 
%\end{equation}
%Then 
%\begin{equation}\label{regular_expression_intr}
%\frac{1}{a(k)} \Phi_{\l,1}(x,t;k)\e^{\i\theta(x,t;k)}-\left[\sum\limits_{p=1}^{n_j}
%\frac{A_p(x,t)\ \e^{2\i\theta(x,t;\kappa_j)}}{(k-\kappa_j)^p}\right] \Phi_{\r, 2}(x,t;k)\e^{-\i\theta(x,t;k)}
%\end{equation}
%is regular in a vicinity of the point $k=\kappa_j$ provided that 
%\begin{equation}\label{A_p_intr}A_{p}(x,t)=\sum\limits_{m=p}^{n}T_{m}(x,t)\frac{\mu_{m-p}}{(m-p)!}=\sum\limits_{m=0}^{n-p}T_{m+p}(x,t)\frac{\mu_{m}}{m!},\quad p=1,\ldots,n.\end{equation}

%\end{itemize}
\end{lemma}

\color{black}

%The lemma is proven in Appendix \ref{sect_nongen}
\begin{proof}
Combining \eqref{T_a_b} and \eqref{psipm} we can write 
 the coefficient $a(k)$ as
\begin{equation}
\label{defap}
a(k) = \det \begin{pmatrix}\varphi^-(x;k) &\varphi^+(x;k)      \\ \psi^-(x;k) &\psi^+(x;k)  \end{pmatrix},
\end{equation}
 Let $a(\kappa_j)=0$ and $\dot{a}(\kappa_j)\neq0$, \color{black} where $\dot{a}(k):=\dfrac{d}{dk}a(k)$. Then, by \eqref{defap}
$$\begin{cases}
\varphi^{\l} = \mu \varphi^{\r},\\ \psi^{\l} = \mu\psi^{\r},\quad \mu\neq0\,.
\end{cases}
$$

If $\begin{pmatrix}\varphi^-\\\psi^-\end{pmatrix}$ is a solution of the Lax equation \eqref{x-eq}, then $\begin{pmatrix}-\ol{\psi^-(\ol{k})} \\ \ol{\varphi^-(\ol{k})}\end{pmatrix}$ is also a solution of the equation \eqref{x-eq}.
Due to above mentioned symmetry, all the zeros $\kappa_1,\dots\kappa_N$  of $a(k)$  either belong to the imaginary axis or go in pairs, $k$ and $-\ol{k}.$
 According to property 11 and 12 of Lemma \ref{lem_abr},  the zeros of $a(k)$ are outside the segment $(\i c_-, 0]$.

In order to prove the relation \eqref{resa}, we differentiate in $k$ the relation \eqref{defap}, thus obtaining 
\begin{equation*}
\dot{a}(\kappa_j)=\det\begin{pmatrix}
   \dot{\varphi}^- (\kappa_j) &   \frac{1}{\mu}\varphi^- (\kappa_j)\\ \dot\psi^-(\kappa_j)&\frac{1}{\mu}\psi^-(\kappa_j)
      \end{pmatrix}+\det\begin{pmatrix}
     \mu\varphi^+(\kappa_j) & \dot{\varphi}^+(\kappa_j)   \\\mu\psi^+(\kappa_j)  & \dot\psi^+(\kappa_j) 
      \end{pmatrix}.
      \end{equation*}
In order to evaluate the above two determinants we  define the quantity
$$W^{\pm}:=\det\begin{pmatrix}
     \varphi^{\pm} & \dot{\varphi}^{\pm}  \\\psi^{\pm} & \dot\psi^{\pm}
      \end{pmatrix},
$$
where the dot means differentiation with respect to $k$.
From the Lax equation \eqref{x-eq} we have that
$$
\begin{cases}
 \varphi^{\pm}_{x}+\i k\varphi^{\pm} = q\psi^{\pm},
\\
\psi^{\pm}_x-\i k\psi^{\pm}=-q\varphi^{\pm},
\end{cases}
\qquad 
\begin{cases}
 \dot\varphi^{\pm}_{x}+\i k\dot\varphi^{\pm} + \i\varphi^{\pm} = q\dot\psi^{\pm},
\\
\dot\psi^{\pm}_x-\i k\dot\psi^{\pm}-\i\psi^{\pm}=-q\dot\varphi^{\pm}\,.
\end{cases}
$$

Hence, differentiating $W^{\pm}$ in $x,$ we obtain 
$$W^{\pm}_x=2\i\varphi^{\pm}\psi^{\pm},\quad \mbox{ and hence }W^{\pm}(B;\kappa_j)-W^{\pm}(A;\kappa_j)=2\i\int\limits_{A}^B\varphi^{\pm}(x;\kappa_j)\psi^{\pm}(x;\kappa_j) \d x$$
for any $A<B.$ Choosing $A=-\infty, B=0$ for $W^-$ and $A=0,B=+\infty$ for $W^{+},$ from the above relations  and \eqref{Jost0} we obtain 
\begin{equation*}\begin{split}\dot{a}(\kappa_j) &=
\frac{-1}{\mu}W^-(0;\kappa_j) + \mu W^+(0;\kappa_j)\\
&= \frac{-1}{\mu}\int\limits_{-\infty}^02\i\varphi^-(x;\kappa_j)\psi^-(x;\kappa_j)\d x -\mu\int\limits_{0}^{+\infty}2\i \varphi^+(x;\kappa_j)\psi^+(x;\kappa_j)\d x 
\\
&= -2\i\mu \int\limits_{-\infty}^{+\infty}\varphi^+(x;\kappa_j)\psi^+(x;\kappa_j)\d x =\frac{-2\i}{\mu}\int\limits_{-\infty}^{+\infty}\varphi^-(x;\kappa_j)\psi^-(x;\kappa_j)\d x,
\end{split}\end{equation*}
which concludes the proof of Lemma~\ref{lem_zeros}.
\end{proof}

\color{black}
\begin{remark}
Note that the $x-$equation \eqref{x-eq} of the Lax pair is not self-adjoint, and hence  the  zeros of $a(k)$ are not necessarily simple, and may lie everywhere in the  region  $\left\{k:\ \Im k\geq 0\right\}$ as long as the symmetry relation $\ol{a(-\ol k)}=a(k)$ is satisfied.

However, considerations similar to that of \cite{BC} allow to conclude that there is a dense set in the space of the step-like potentials \eqref{ic0}  such that for any initial data  $\tilde{q}(x)$
satisfying \eqref{first_moment} there is an initial data $q(x)$ close to $\tilde{q}(x)$ 
in the topology induced by \eqref{first_moment}, for which the 
 generic Assumptions  \ref{Assump2} are satisfied.
 \color{black}
% , namely all the zeros of $a(k)$ are simple, do not lie on $\mathbb{R},$ and all the corresponding speeds $V_j$ are different.
%We do not  prove this statement in the present manuscript.
%\todo{A: I inserted here a word ``Assumption \ref{Assump2}''.}

%We don't this issue here.
\end{remark}
Summarizing we arrive to the following set of scattering data for the initial data  $q(x,0)$ satisfying the Assumptions~\ref{Assump1}  and  \ref{Assump2}:
\begin{equation*}
{\mathcal S}=\{r(k), \{\kappa_j,\i \nu_j\}_{j=1}^N\}
\end{equation*}
with $r(k)$ meromorphic   in a $\delta-$neighbor\-hood of $\Sigma=\mathbb{R}\cup[\i c_{\l},-\i c_{\l}]$, where $\delta = \sqrt{\sigma^2+c_+^2}-c_-$ and  $\sigma>\sqrt{c_-^2-c_+^2}>0$
and where  $\{\kappa_j\}_{j=1}^N$  are  simple zeros of $a(k)$ with  $\kappa_j\in\C^+\backslash\{ \R\cup(\i c_-,0]\}$  and $\Re\kappa_j\geq 0$.
If $\kappa_j$ corresponds to a soliton, then $\bar{\kappa}_j$ is also a point of the discrete spectrum, while if $\kappa_j$ corresponds to a breather, then $-\kappa_j$, $\widebar{\kappa}_j$ and $-\widebar{\kappa}_j$ belong to the discrete spectrum.
%
%The  evolution in time of  the spectral data, is described  by \cite{Wadati}
%\begin{equation}
%\label{spectral_t}
%r(k;t)=r(k)e^{8itk^3},\ \nu_j(t)=\nu_je^{8it\kappa_j},\ j=1,\dots,N.
%\end{equation}

 %The second part of the lemma (for multiple zeros) is also proven  in Appendix, see lemmas \ref{lem_zeros_a_multiple}, \ref{lem_reg_multi}.

\color{black}

\subsection{Riemann-Hilbert problem in the generic case}\label{sect_RHP}

\color{black}

In this section we introduce the RH  problem \ref{RH_problem_1}, which allows to reduce the (nonlinear) initial value problem \eqref{MKdV}, \eqref{ic0} into a (linear) matrix conjugation problem.

\color{black}

%\subsubsection{Heuristic part.}\label{sect_heur}

{\color{black}In order to set up  a  RH problem for  $t\geq0$  we proceed as follows.
We first assume that the solution $q(x,t)$ of the initial value problem \eqref{MKdV}, \eqref{ic0} exists, and, moreover, that 
 the Jost solutions $\Phi^{\pm}(x,t;k)$  with the defining property 
 \begin{equation}
 \label{Phi_t}
\Phi^{\pm}(x,t;k)=E^{\pm}(x,t;k)(1+o(1))\quad \mbox{as $x\to\pm \infty$},\;\; \color{black}{ k\in\mathbb{R}\cup[\i c_{\pm}, -\i c_{\pm}],}
\end{equation}
  exist  for $t\geq 0$ (here, the matrix $E^{\pm}(x,t;k)$ has been defined in \eqref{E_c}).  
Our assumption will be dropped off later, after deriving the RH problem \ref{RH_problem_1}.
We will show the solvability of the RH  problem ~\ref{RH_problem_1}  and thus we  justify a posteriori the assumption of existence of $q(x,t)$ and  the  Jost solutions $\Phi^{\pm}(x,t;k).$
\color{black}

 To start, \color{black}we observe   that since $\Phi^{\pm}(x,t;k)$ are solutions to the  first order equation \eqref{x-eq},
  these solutions are related by the  linear transformation
  \begin{equation*}
 \Phi^{{\l}}(x,t;k)=\Phi^{{\r}}(x,t;k)T(t;k),\quad
k\in\mathbb{R} \cup(\i c_{\r},-\i c_{\r}).
\end{equation*}
Using the evolution equation \ref{t-eq}, it follows that
\begin{equation*}
\dfrac{d}{dt}T(t;k)=0,
\end{equation*}
namely $\dfrac{d}{dt}a(k,t)=0$ and $\dfrac{d}{dt}b(k,t)=0$.  The fact that the  scattering data are constant in time  is due to our choice of normalization 
of  the  Jost  solutions in \eqref{Phi_t}.
}

The scattering relations (\ref{T_a_b}) between the matrix-valued
functions $\Phi^{\l}(x,0;k)$ and
$\Phi^{{\r}}(x,0;k)$,  and the jump conditions  of Lemma~\ref{lem_prop_Jost} and Lemma~\ref{lem_abr} can
be written as a matrix Riemann -- Hilbert problem (RHP). Namely, let us notice that the matrix-valued function
\begin{equation}\label{M}M(x,0;k)=
\begin{cases}\left(
\displaystyle\frac{\Phi^{{\l}}_{1} (x,0;k)}{a(k)} e^{\i
kx},\Phi^{\r}_{2}(x,0;k)e^{-\i kx}\right),
\  k\in \mathbb{C}_+ \backslash [\ii c_{{\l}},0],  \\\\
\left(\Phi^{{\r}}_{1}(x,0;k)e^{\i kx},\displaystyle
\frac{\Phi^{{\l}}_{2}(x,0;k)} {\overline{a(\overline{k})}}
e^{-\i kx}\right),\ k\in\mathbb{C}_-\backslash[0, -\ii
c_{{\l}}],
\end{cases}\end{equation} 
satisfies the jump conditions 
\begin{equation*}
M_-(x,0;k)=M_+(x,0;k)J(x,0;k),\quad k\in \Sigma,
\end{equation*}
where  the matrix $J(x,0;k)$ and the oriented contour  $\Sigma=\R\cup[\i c_-,-\i c_-]$
are specified in  Figure~\ref{Graphic1}. 
{\color{black} Here $M_{\pm}(x,0;k)$ are the limiting values of the matrix $M(x,0;k)$ as $k$ approaches the contour from the positive/negative direction (the positive direction is on the left, the negative is on the right as one traverses the contour in the direction of orientation). Using Lemmas~\ref{lem_prop_Jost} and \ref{lem_abr} the matrix $J(x,0;k)$ can be obtain in a straightforward way from the definition \eqref{M} \cite[section 3]{KM2}.
 Further, let us assume that $\kappa_j$ with $\Re\kappa_j\geq 0$ and $\Im\kappa_j>$ is a zero of $a(k)$. Then taking the residue at $k=\kappa_j$ of the matrix $M(x,0;k)$  and using  Lemma~\ref{lem_zeros}
 we obtain
 $$\mathrm{Res}_{\kappa_j}M(x,0;k)=\left(
\displaystyle\i\frac{\nu_j}{\mu_j}\Phi^{{\l}}_{1} (x;\kappa_j) e^{\i kx},0\right)=\lim\limits_{k\to \kappa_j}M(x,0;k)\begin{pmatrix}0&0\\ \i\nu_{j}\e^{2\i \theta(x,0;\kappa_j)} & 0\end{pmatrix},$$
and similarly for $\bar{\kappa}_j$, $-\kappa_j$ and $-\bar{\kappa}_j$.
}

%\begin{figure}[ht]
%\begin{center}
%\epsfig{width=100mm,figure=Graphic1.eps}
%\end{center}
%\caption{The oriented contour $\Sigma$} \label{Graphic1}
%\end{figure}
%Then the func\-tion $M(x,t;k)$ (\ref{M}) solves
%the following Rie\-mann -- Hilbert problem
%:
\begin{figure}[h]
\begin{tikzpicture}
\draw[fill=black] (0,2.5) circle [radius=0.05];
\draw[fill=black] (0,0) circle [radius=0.05];
\draw[fill=black] (0,-2.5) circle [radius=0.05];
\draw[fill=black] (0,1.5) circle [radius=0.05];
\draw[fill=black] (0,-1.5) circle [radius=0.05];
\draw[->,thick] (0,-1.5)--(0,-2.0);
\draw[->,thick] (0,0.0)--(0,-0.7);
\draw[->,thick] (0,2.5)--(0,2.0);
\draw[->,thick] (0,1.2)--(0,0.6);
%segment [i c_l,-i c_l]
\draw[thick] (0,2.5) to (0,-2.5);
%lenses L_5, L_7, L_6, L_8
%\draw[thick, postaction = decorate, decoration = {markings, mark = at position 0.5 with {\arrow{>}}}](0,3.5) to [out=0, in=90] (1,3) [out=-90,in=30] to (0,2.2);
%\draw[thick, postaction = decorate, decoration = {markings, mark = at position 0.5 with {\arrow{>}}}](0,3.5) to [out=180, in=90] (-1,3) [out=-90,in=150] to (0,2.2);
%\draw[thick, postaction = decorate, decoration = {markings, mark = at position 0.5 with {\arrow{<}}}](0,-3.5) to [out=0, in=-90] (1,-3) to[out=90,in=-30](0,-2.2);
%\draw[thick, postaction = decorate, decoration = {markings, mark = at position 0.5 with {\arrow{<}}}](0,-3.5) to [out=180, in=-90] (-1,-3)to[out=90, in=-150] (0,-2.2);
\node at (-0.4,2.5) {$\i c_{\l}$};\node at (-0.35,1.5) {$\i c_{\r}$};
\node at (-0.15,0.18) {$0$};
\node at (-0.4,-2.5) {$-\i c_{\l}$};\node at (-0.4,-1.5) {$-\i c_{\r}$};%\node at (1.3,3.) {$L_7$}; \node at (-1.3, 3.) {$L_5$};\node at (1.3,-3.) {$L_8$}; \node at (-1.3, -3.) {$L_6$};
%lines L_1, L_2, L_3, ...
%\draw[thick, postaction = decorate, decoration = {markings, mark = at position 0.5 with {\arrow{>}}}](-3,1) to [out = 0, in =-110] (0,2.2);
%\draw[thick, postaction = decorate, decoration = {markings, mark = at position 0.5 with {\arrow{<}}}](3,1) to[out = 180, in =-70] (0,2.2);
%\draw[thick, postaction = decorate, decoration = {markings, mark = at position 0.5 with {\arrow{>}}}](-3,-1) to[out = 0, in =110] (0,-2.2);
%\draw[thick, postaction = decorate, decoration = {markings, mark = at position 0.5 with {\arrow{<}}}](3,-1) to[out = 180, in =70] (0,-2.2);
%real line\begin{pmatrix}\i \ol{r_{+}(\ol{k})}& -\ol{f(\ol{k})} \e^{-2\i t\theta(k,\xi)} \\ \i \e^{2\i t\theta(k,\xi)} & -\i \ol{r_-(\ol{k})}\end{pmatrix}
\node[right] at (-6,.8) {\small{$ \begin{pmatrix}1 & -\overline{r(k)}\e^{-2i\theta(x,t;k)}\\-r(k)\e^{2\i \theta(x,t;k)} & 1+|r(k)|^2\end{pmatrix}$}};
\node[right] at (-5,-1.2) {\small{$\theta(x,t;k)=kx+4k^3t$}};
\node[right] at (.5,1.) {$\begin{pmatrix}\i r_{-}(k)&\i \e^{-2\i \theta(x,t;k)}\\f(k) \e^{2\i \theta(x,t; k)} & -\i r_+(k)\end{pmatrix}$};
\draw[->,thick] (0.6,1)--(0.1,0.8);
\node[right] at (0.5,-1.0) {$\begin{pmatrix}\i \ol{r_{+}(\ol{k})}& -\ol{f(\ol{k})} \e^{-2\i \theta(x,t;k)} \\ \i \e^{2\i \theta(x,t;k)} & -\i \ol{r_-(\ol{k})}\end{pmatrix}$};
\draw[->,thick] (0.6,-1)--(0.1,-0.8);
\node[right] at (0.5,-2.2) {$\begin{pmatrix}1&-\overline{f(\ol{k})}\e^{-2\i \theta(x,t;k)} \\0& 1\end{pmatrix}$};
\draw[->,thick] (0.6,-2.)--(0.1,-1.8);
\node[right] at (0.6,2.2) {$\begin{pmatrix}1&0\\f(k)\e^{2\i \theta(x,t;k)} & 1\end{pmatrix}$};
\draw[->,thick] (0.7,2.4)--(0.1,2.0);
\draw[thick,dashed](0,3)--(0,-3);
\draw[thick, postaction = decorate, decoration = {markings, mark = at position 0.25 with {\arrow{>}}}, decoration = {markings, mark = at position 0.75 with {\arrow{>}}}] (-6,0) to (7,0);
\end{tikzpicture}
\caption{The oriented contour $\Sigma=\mathbb{R}\cup[\i c_-,\i c_-]$ and the corresponding jump matrix $J(x,t;k)$.} 
\label{Graphic1}
%\label{Fig_FinalRHP}
\end{figure}
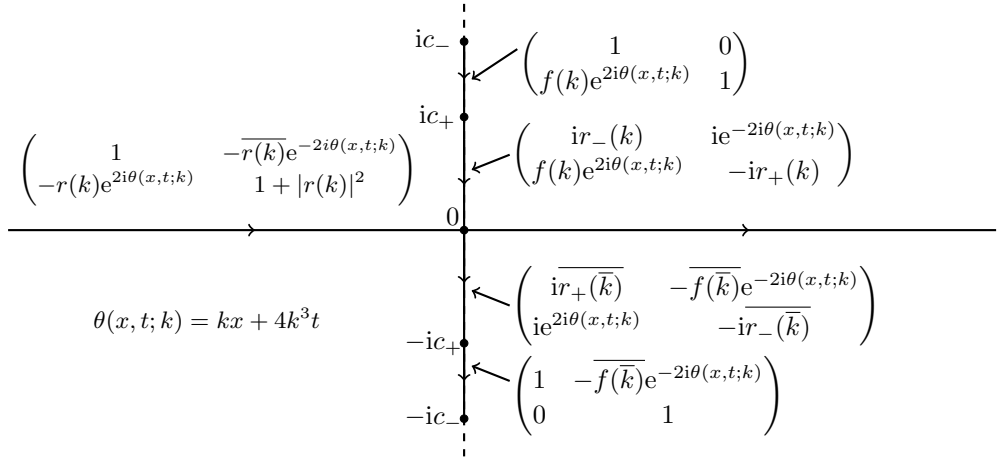

The jump properties of the matrix $M(x,0;k)$ prompts us to consider the following  matrix valued function
\begin{equation}\label{Mt}M(x,t;k)=
\begin{cases}\left(
\displaystyle\frac{\Phi^{{\l}}_{1} (x,t;k)}{a(k)} e^{\i \theta(x,t;k)},\Phi^{\r}_{2}(x,t;k)e^{-\i \theta(x,t;k)}\right),
\  k\in \mathbb{C}_+ \backslash [\ii c_{{\l}},0],  \\\\
\left(\Phi^{{\r}}_{1}(x,t;k)e^{\i \theta(x,t;k)},\displaystyle
\frac{\Phi^{{\l}}_{2}(x;k)} {\overline{a(\overline{k})}}
e^{-\i\theta(x,t;k)}\right),\ k\in\mathbb{C}_-\backslash[0, -\ii
c_{{\l}}],
\end{cases}\end{equation} 
where $$\theta(x,t;k)=xk+4k^3t.$$ 
Using  the properties of the Jost solutions described in section \ref{sect_scat}, we can check that the function $M(x,t;k)$, (assuming it  exists),  satisfies the jump, analyticity, and normalization conditions described below in the RH problem \ref{RH_problem_1}.

%
%
%\subsubsection{Rigorous part}\label{sect_rig}
%
% time-dependent RH problem:
%
\begin{RHP}\label{RH_problem_1}
To find a $2\times 2$ matrix-valued function $M(x,t;k)$ with the following properties:
\begin{enumerate} 
 \item $M(x,t;k)$ is meromorphic  for  $k\in \mathbb{C}\setminus\Sigma,$  where  $\Sigma=\R\cup[\i c_-,-\i c_-]$   (see Figure~\ref{Graphic1}) and  it has at most fourth root singularities at the points $\pm\i c_{\pm}$. 
  \item The boundary values $M_{\pm}(x,t;k)$  on the oriented contour $\Sigma$ satisfy the jump condition
   $$M_-(x,t;k)=M_+(x,t;k)J(x,t;k), \quad k\in\Sigma$$
    and  
\begin{equation}
\label{J_M0}
J(x,t;k)=\left\{
\begin{array}{lll}
\begin{pmatrix}1 & -\overline{r(k)}\e^{-2i\theta(x,t;k)}\\-r(k)\e^{2\i \theta(x,t;k)} & 1+|r(k)|^2\end{pmatrix},& k\in\mathbb{R}\setminus\left\{0\right\},&
\\
\begin{pmatrix}1&0\\f(k)\e^{2\i \theta(x,t;k)} & 1\end{pmatrix},& k\in(\i c_{\l},\i c_{\r}),&\\
\\
\begin{pmatrix}\i r_{-}(k)&\i \e^{-2\i \theta(x,t;k)}\\f(k) \e^{2\i \theta(x,t;k)} & -\i r_+(k)\end{pmatrix},& k\in(\i c_{\r},0),&\\
\\
\begin{pmatrix}\i \ol{r_{+}(\ol{k})}& -\ol{f(\ol{k})} \e^{-2\i \theta(x,t;k)} \\ \i \e^{2\i \theta(x,t;k)} & -\i \ol{r_-(\ol{k})}\end{pmatrix},& k\in(0,-\i c_{\r}),&
\\
\begin{pmatrix}1&-\overline{f(\ol{k})}\e^{-2\i \theta(x,t;k)} \\0& 1\end{pmatrix},& k\in(-\i c_{\r},-\i c_{\l}),&
\end{array}
\right.
\end{equation}
where $r(k)$ is the reflection coefficient  of the  spectral problem of the  MKdV  equation, and
 \begin{equation}
 \label{f_funct} f(k)=r_-(k)-r_+(k)=\frac{\i}{a_-(k)a_+(k)},\quad k\in(\i c_{\l},\i c_{\r})\cup (-\i c_{\r},-\i c_{\l}),
 \end{equation} 
  $a^{-1}(k)$ is the transmission coefficient and 
 \begin{equation}
 \label{theta}
 \theta(x,t;k)=kx+4k^3t\,.
\end{equation}
\item Simple poles: residue condition at $\kappa_j$  and $-\ol{\kappa_j}$ for $j=1,...,N$ with $\Re\kappa_j\geq 0$ and $\Im\kappa_j>0:$
$$\mathrm{Res}_{\kappa_j}M(x,t;k)=\lim\limits_{k\to \kappa_j}M(x,t;k)\begin{pmatrix}0&0\\ \i\nu_{j}\e^{2\i \theta(x,t;\kappa_j)} & 0\end{pmatrix},$$
$$\mathrm{Res}_{-\ol{\kappa_j}}M(x,t;k)=\lim\limits_{k\to -\ol{\kappa_j}}M(x,t;k)\begin{pmatrix}0&0\\ \i\ol{\nu_{j}}\e^{2\i \theta\(x,t;-\ol{\kappa_j}\)} & 0\end{pmatrix};$$
 residue conditions  in the lower half-plane:
$$\mathrm{Res}_{\ol{\kappa_j}}M(x,t;k)=\lim\limits_{k\to \ol{\kappa_j}}M(x,t;k)\begin{pmatrix}0&\i\ol{\nu_{j}}\e^{-2\i t\theta\(x,t; \ol{\kappa_j}\)} \\ 0 & 0\end{pmatrix},$$
$$\mathrm{Res}_{-\kappa_j}M(x,t;k)=\lim\limits_{k\to-\kappa_j}M(x,t;k)\begin{pmatrix}0&\i\nu_{j}\e^{-2\i \theta(x,t;-\kappa_j)} \\ 0 & 0\end{pmatrix};$$

%\todo{$\gamma/2->\nu_j?$}
%{\color{black} \begin{framed}
%       Check the sign in the above formula.
%
%What about zeros at $\mathbb{R}?$
%      \end{framed}
%}

\item asymptotics: $M(x,t;k)\to I$ as $k\to\infty.$
\end{enumerate}
\end{RHP}

This is the point at which we drop off the assumption of existence of $q(x,t)$ and $\Phi^{\pm}(x,t;k),$
 we study the RH problem \ref{RH_problem_1} and its solvability. \color{black}
\color{black}
Note that  the  jump matrix $J(x,t;k)$ is such that the  matrix $M(x,t;k)$ in  \eqref{M} has a trivial monodromy   at the origin.
Further we observe that the jump matrix $J(x,t;k)$ satisfies the  Schwartz symmetry $J^{-1}(k)=\(\ol{J^{T}(\ol{k})}\)$ for $k\in\Sigma\setminus\mathbb{R}.$\footnote{We observe that the usual definition of Schwartz symmetry
refers to a contour $\Sigma$ such that $\overline{\Sigma}=\Sigma$ and the corresponding jump matrices satisfy the relation $\(\ol{J^{T}(\ol{k})}\)=J(k)$.
In our case, we orient the contour off the real axis in such a way that it is symmetric up to orientation, this  implies that $\(\ol{J^{T}(\ol{k})}\)=J^{-1}(k)$ for $k$ off the real axis.
}
\color{black}

The solution $M(x,t;k)=M(k)$ of the RH problem \ref{RH_problem_1} automatically satisfies the following symmetries:
\begin{equation}\label{Symmetries}
M(k) = \ol{M(-\ol k)} = \begin{pmatrix}0&-1\\1&0\end{pmatrix}
M(-k)
\begin{pmatrix}0&1\\-1&0\end{pmatrix}=\begin{pmatrix}0&-1\\1&0\end{pmatrix}
\ol{M(\ol k)}
\begin{pmatrix}0&1\\-1&0\end{pmatrix}\,.
\end{equation}

\color{black}The existence and smoothness of a solution to  RH problem \eqref{RH_problem_1} for the case $c_+=0$ was established in \cite[Theorems 2.1, 2.2]{KMshort}. The  case $c_+>0$ can be treated similarly; below we give the details.

The existence of solution of the RH problem \ref{RH_problem_1} is based on the vanishing lemma for Schwartz symmetric RH problems, \cite[Theorem 9.3]{Zhou}. In our case, because of the choice of orientation of the parts of the contour off the real axis, it reads as 
$J^{-1}(k)=\ol{J^{T}(\ol{k})}$ for $k\in\Sigma\setminus\mathbb{R},$ and $J(k)+\ol{J^T(\ol k)}$ is positive definite for $k\in\mathbb{R}.$
\color{black}

%The jump matrices in the RH problem \ref{RH_problem_1} satisfy the  Schwartz symmetry $J^{-1}(k)=\(\ol{J^{T}(\ol{k})}\)$ for $k\in\Sigma\setminus\mathbb{R}.$ Hence, it follows from the vanishing lemma \cite{Zhou}
%that the solution of the RH problem \ref{RH_problem_1} exists. 
Further, from   our assumption on the initial data it follows the analyticity of the reflection coefficient  in a small neighbourhood of the contour $\Sigma$, (see Lemma~\ref{lem_abr}). Therefore  we can deform the contour $\Sigma$ into  a   new contour such that
the corresponding singular integral equation, which is equivalent to the RH problem \ref{RH_problem_1}, 
admits differentiation with respect to $x$ and $t$.
Then, in the spirit of the well-known result of Zakharov -- Shabat \cite{Zakharov-Shabat}, one can prove that 
the solution of the initial value problem (\ref{MKdV}), (\ref{ic0}) \color{black} exists and \color{black} can be reconstructed by the following formula (see \cite{M_disser}, chapter 2 for details):
$$ q(x,t)=\lim\limits_{k\to\infty}\(2\i k M(x,t;k)\)_{21}=\lim\limits_{k\to\infty}\(2\i k M(x,t;k)\)_{12}. $$

\color{black}
\begin{remark}
%Note that along with the framework of RH  problem there is a related framework of decomposing algebras. However, the problem formulations are different, and as a result, somehow confusingly, the uniqueness of solution, trivial for RHPs, is a complicated issue for decomposing algebras, and vice versa, the existence issue, complicated for RHPs, is simple for decomposing algebras.

Note that in the framework of the RH problem with unimodular jump matrix and with locally $L_2$ integrable boundary values of solutions, \color{black} one fixes the value of  the  solution at a given point ($\infty$ in our case), and specifies all the possible poles of the solution. This fixes the solution uniquely. Indeed, first of all, the determinant of a solution does not have jump over the contour, tends to the identity at infinity, and has at most removable singularities at the points $\pm\i c_{\pm}$; hence it equals 1 identically. Second, \color{black} assuming that there are two solutions of the RH problem, their ratio would have trivial jumps and will tend to identity at infinity, and hence would be equal to identity (see, for instance \cite[Theorem 7.18]{Deift99}).

%In the framework of decomposing algebras, one usually does not fix the value of a solution at a given point, and hence does not have uniqueness
%
%In the case when the contour is a simple curve, one usually requires the jump matrix to be continuous on the contour. This results in the solution taking its limiting values on the contour continuously.

In the case when the contour has points of self-intersections, the condition of continuity of the jump matrix is replaced at the points of self-intersection by the condition that the product of jump matrices equals identity matrix at the points of self-intersection. This guarantees that the solution takes its limiting values continuously from within each of the domains separated by the contour. See \cite[Section 7.1]{Deift99} and \cite[Appendices A,B]{KMM} for more detail.
\end{remark}
\color{black}

In order to make the asymptotic analysis  of the above RH problem as $t\to\infty$ 
 it is more convenient to transform the residue conditions at the poles to a jump conditions as in \cite{GT09}.
%As a first step of our transformations is to reduce the  residue expression to a jump condition as in \cite{Teschl}.

Let $\kappa_j$ be a pole 
of $a(k).$  Let us encircle this pole with a  circle $C_j$  of radius $\varepsilon>0$.
There are several options:

\begin{enumerate}
 \item $\Re\kappa>0,\ \Im\kappa>0.$ In this case $C_{j}$ is a  circle  of radius $\varepsilon$ and center   $\kappa_j$  and oriented anticlockwise  with respect to the center. Let us also define the three other
circles $\ol{C}_j,$ $-C_j,$ $-\ol{C}_j$  with center  the points $\ol{\kappa_j},$ $-\kappa_j$, $-\ol{\kappa_j},$  and radius $\varepsilon$ and oriented anticlockwise.
%\item $\Im \kappa=0,$ $\Re\kappa>0.$ In this case $C_{j}$ is a semicircle in the upper half-plane with  center at $\kappa_j$  and radius $\varepsilon$ and oriented anti-clockwise.
% We denote by $\ol{C}_j$  the  semicircle in the lower half-plane  centred at $\kappa_j$, with  anti-clockwise  orientation   with respect to the centre.
% Further, $-C_j$ and $-\ol{C}_j$ are the semicircles around the point $-\kappa_j,$ with the same agreement on orientation.
%\item $\kappa_j=0.$ This case is impossible for $c_{+}=0.$ We define $C_j,$ $\ol{C}_j,$ $-C_{j},$ $-\ol{C}_j$ to be quarter-circles, with the analogous agreement on the orientation as above.
\item $\Re\kappa_j=0, \ \Im\kappa_j>0.$ 
  In this case we denote $C_j$ a semicircle in the  plane $\Re k\geq0$ around the point $\kappa_j,$
 while $-\ol{C}_j$ is another semicircle in
$\Re k\leq0$ around the point $\kappa_j.$ Semicircles $\ol{C}_{j}$ and $-C_j$ are the corresponding semicircles in the lower half-plane, with the above agreement on the orientation.
\end{enumerate}
%
%%%%%%%%%%%%%%%% BEGIN COMMENT DOUBLE TRIPLE MULTIPLE
%
% Now we are ready to define the regular RHP.
% $$\widetilde{M}(k)=M(k)\displaystyle\begin{pmatrix}
%             1 & 0\\\sum\limits_{l=1}^{n_j}\frac{-A_l}{(k-\kappa_j)^l}& 1
%            \end{pmatrix}, \quad k \textrm{ is inside } C_j,\quad
% \widetilde{M}(k)=M(k)\begin{pmatrix}
%             1 & 0\\\sum\limits_{l=1}^{n_j}\frac{-(-1)^lA_l}{(k+\ol{\kappa_j})^l}& 1
%            \end{pmatrix}, \quad k \textrm{ is inside } -\ol{C}_j,
% $$
% $$\widetilde{M}(k)=M(k)\begin{pmatrix}
%             1 & 0\\\sum\limits_{l=1}^{n_j}\frac{\ol{A}_l}{(k-\ol{\kappa_j})^l}& 1
%            \end{pmatrix}, \quad k \textrm{ is inside } \ol{C}_j,\quad
% \widetilde{M}(k)=M(k)\begin{pmatrix}
%             1 & 0\\\sum\limits_{l=1}^{n_j}\frac{-(-1)^{l-1}A_l}{(k+\kappa_j)^l}& 1
%            \end{pmatrix}, \quad k \textrm{ is inside } -C_j,$$
% and $\widetilde{M}=M$ elsewhere. $\widetilde{M}$ solves the following RHP:
%%%%%%%%%%%%%%%% END COMMENT DOUBLE TRIPLE MULTIPLE
We replace the residue condition with a jump condition having only upper triangular matrices. For the purpose we    redefine the matrix $M(k)$ as  
$$M(k)\leadsto M(k)\displaystyle\begin{pmatrix}
            1 & 0\\\frac{-\color{black}{\i}\nu_j\e^{2\i\theta(x,t;\kappa_j)}}{(k-\kappa_j)}& 1
           \end{pmatrix}, \quad k \textrm{ is inside } C_j,$$
$$M(k)\leadsto M(k)\begin{pmatrix}
            1 & 0\\\frac{\color{black}{-\i}\,\ol{\nu_j} \e^{2\i\theta(x,t;-\ol{\kappa_j})}}{(k+\ol{\kappa_j})}& 1
           \end{pmatrix}, \quad k \textrm{ is inside } -\ol{C}_j,
$$
$$M(k)\leadsto M(k)\begin{pmatrix}
            1 & \frac{\color{black}{-\i}\,\ol{\nu}_j \e^{-2\i\theta(x,t;\ol{\kappa_j})}}{(k-\ol{\kappa_j})} \\ 0& 1
           \end{pmatrix}, \quad k \textrm{ is inside } \ol{C}_j,
           $$
$$M(k)\leadsto M(k)\begin{pmatrix}
            1 & \frac{-\color{black}{\i}\nu_j \e^{-2\i\theta(x,t;-\kappa_j)}}{(k+\kappa_j)} \\ 0 & 1
           \end{pmatrix},\quad  k \textrm{ is inside } -C_j.$$
Then the  jump matrix $J_M(k)$ for the RH  problem   for $M$ becomes
\begin{equation}
\label{J_M}
J_M(k)=\begin{cases}
\begin{pmatrix}
            1 & 0\\\frac{\color{black}{\i}\nu_j\e^{2\i\theta(x,t;\kappa_j)}}{(k-\kappa_j)}& 1
           \end{pmatrix}, \quad k \in C_j,\quad\\
\begin{pmatrix}
            1 & 0\\\frac{\color{black}{\i}\,\ol{\nu_j} \e^{2\i\theta(x,t;-\ol{\kappa_j})}}{(k+\ol{\kappa_j})}& 1
           \end{pmatrix}, \quad k \in-\ol{C}_j,\\
\begin{pmatrix}
            1 & \frac{\color{black}{\i}\,\ol{\nu}_j \e^{-2\i\theta(x,t;\ol{\kappa_j})}}{(k-\ol{\kappa_j})} \\ 0& 1
           \end{pmatrix}, \quad k \in \ol{C}_j,\\
\begin{pmatrix}
            1 & \frac{\color{black}{\i} \nu_j \e^{-2\i\theta(x,t;-\kappa_j)}}{(k+\kappa_j)} \\ 0 & 1
           \end{pmatrix}, \quad k \in -C_j,\\
           J(k),\quad \mbox{elsewhere,}
          \end{cases}
          \end{equation}
where $J(k)$ has been defined in \eqref{J_M0}.

\subsection{Solvability of RH problem and existence of solution for the MKdV}\label{sect_solv}
\color{black}{Denote by $C^{\infty}(\mathbb{R}_x\times\mathbb{R}_t^+)$ 
the set of functions, which are infinitely many times differentiable at any point $(x,t)\in\mathbb{R}_x\times\mathbb{R}_t^+,$
where $\mathbb{R}_t^+=\left\{t:\ t\in(0,+\infty)\right\}$.
}
\begin{theorem}\label{teor_exist_gener}
Let the  initial data $q_0(x)$ \eqref{ic0} satisfiy  Assumption \ref{Assump1}.
%\begin{itemize}
%\item [(a)] \eqref{first_moment_init_data}, \eqref{initial_funct_deriv_bounded}, \eqref{exp_decreasing}
%\item [(b)] \eqref{first_moment_init_data}, \eqref{initial_funct_deriv_bounded}, third moment.
%\end{itemize}
Then the initial value problem \eqref{MKdV}, \eqref{ic0} has a unique classical solution $q(x,t)$, which \color{black} belongs to the class
$C^{\infty}(\mathbb{R}_x\times\mathbb{R}^+_t).$ 

\color{black}
\end{theorem}
%\begin{remark}
%\noindent\color{black}{This theorem shows an instant regularisation of the initial datum within given class.
%At the same time we see that the condition of exponential convergence of the initial function to its background is crucial in establishing existence of the classical solution: this is not satisfied for the initial function $q(x, t=-1),$ since it will be equal to $q_0(x)$ as $t=1,$ and $q_0$ is not necessarily continuous at $t=0.$ A similar situation occurs in \cite{Murray78}.}
%\end{remark}
%\begin{teor}\label{teor_exist_partic}
% Let initial data be of one of the form \eqref{pure_step_ic}, \eqref{pure_step_ic_sol}, \eqref{pure_step_ic_3_step}.
%Then the ivp \eqref{MKdV}, \eqref{ic} has a unique classical solution, which is $C^{\infty}(\mathbb{R}_x)\times C^{\infty}(\mathbb{R}_t\setminus\left\{0\right\}).$
%\end{teor}
\begin{proof}

%\todo{T. Put a more precise reference as requested by the referees, like the number of the Theorem}
\color{black} The proof is similar to the one in \cite[section 2]{KMshort}, \cite{M_disser}, except that we need to make an extra step in order to treat the discontinuity of the jump matrices at the points $\pm\i c_{\pm}.$ \color{black} For the convenience of the reader, we will sketch the main steps.
 
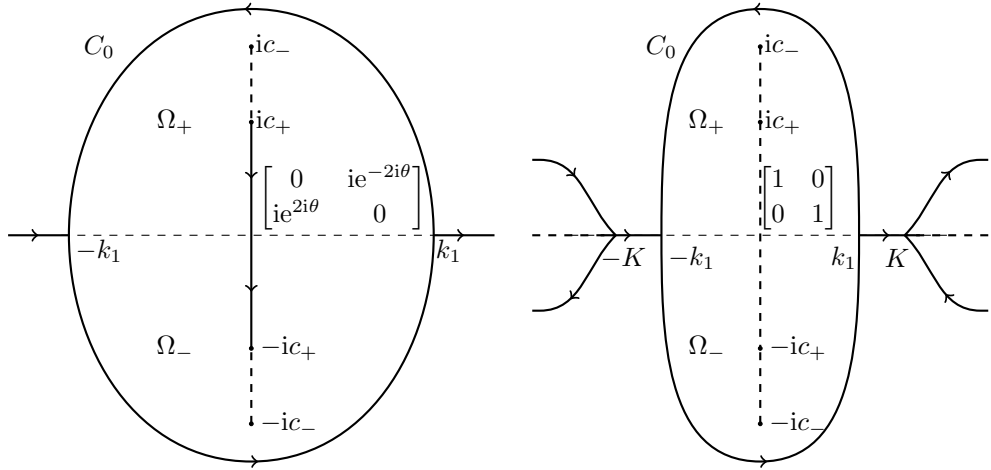
\begin{figure}[ht!]
\begin{tikzpicture}
\draw[thick, decoration={markings, mark=at position 0.5 with {\arrow{>}}}, postaction={decorate}] (-3.2,0) to (-2.4,0);
\draw[thick, decoration={markings, mark=at position 0.5 with {\arrow{>}}}, postaction={decorate}] (2.4,0) to (3.2,0);
\draw[thick,
decoration={markings, mark=at position 0.25 with {\arrow{<}}}, 
decoration={markings, mark=at position 0.75 with {\arrow{<}}}, 
postaction={decorate}
] (-2.4,0) [out=90,in=180] to (0,3)
[out=0,in=90] to (2.4,0)
[out=-90,in=0] to (0,-3)
[out=180,in=-90] to (-2.4,0);
\filldraw (0,2.5) circle (0.7pt);
\filldraw (0,-2.5) circle (0.7pt);
\filldraw (0,1.5) circle (0.7pt);
\filldraw (0,-1.5) circle (0.7pt);
\draw[thick, dashed] (0,2.5) to (0,1.5);
\draw[thick,decoration={markings, mark=at position 0.25 with {\arrow{>}}}, 
decoration={markings, mark=at position 0.75 with {\arrow{>}}}, 
postaction={decorate}](0,1.5) to (0,-1.5);
\draw[thick, dashed] (0,-2.5) to (0,-1.5);
\node at (1.2, 0.5) {$\begin{bmatrix}0 & \i \e^{-2\i\theta} \\ \i \e^{2\i\theta} & 0\end{bmatrix}$};
\draw[dashed](-2.5,0) to (2.5,0);
\node at (0.3, 2.5) {$\i c_-$};
\node at (0.3, 1.5) {$\i c_+$};
\node at (0.5, -1.5) {$-\i c_+$};
\node at (0.5, -2.5) {$-\i c_-$};
\node at (-2., 2.5) {$C_0$};
\node at (-1, 1.5) {$\Omega_+$};
\node at (-1, -1.5) {$\Omega_-$};
\node at (-2.0, -0.2){$-k_1$};
\node at (2.6, -0.2){$k_1$};
\end{tikzpicture}
\quad
%%%%
%%%%%
%%%%%
%%%%
\begin{tikzpicture}
\draw[thick, decoration={markings, mark=at position 0.5 with {\arrow{>}}}, postaction={decorate}] (-2.1, 0) to (-1.3, 0);
\draw[thick, decoration={markings, mark=at position 0.5 with {\arrow{>}}}, postaction={decorate}] (1.3,0) to (2.1,0);
\draw[thick,
decoration={markings, mark=at position 0.25 with {\arrow{<}}}, 
decoration={markings, mark=at position 0.75 with {\arrow{<}}}, 
postaction={decorate}
] (-1.3, 0) [out=90,in=180] to (0,3)
[out=0, in=90] to (1.3, 0)
[out=-90,in=0] to (0, -3)
[out=180,in=-90] to (-1.3, 0);
\filldraw (0, 2.5) circle (0.7pt);
\filldraw (0, -2.5) circle (0.7pt);
\filldraw (0, 1.5) circle (0.7pt);
\filldraw (0, -1.5) circle (0.7pt);
\draw[thick, dashed] (0, 2.5) to (0, 1.5);
\draw[thick,dashed](0, 1.5) to (0,-1.5);
\draw[thick, dashed] (0,-2.5) to (0,-1.5);
\node at (0.5, 0.5) {$\begin{bmatrix}1 & 0 \\ 0 & 1\end{bmatrix}$};
\draw[dashed](-2.5,0) to (2.5,0);
\node at (0.3, 2.5) {$\i c_-$};
\node at (0.3, 1.5) {$\i c_+$};
\node at (0.5, -1.5) {$-\i c_+$};
\node at (0.5, -2.5) {$-\i c_-$};
\draw[thick, decoration = {markings, mark = at position 0.4 with {\arrow{>}}},postaction = {decorate}] (-3.0, 1) to (-2.9, 1) [out=0, in =135] to (-1.9, 0);
\draw[thick, decoration = {markings, mark = at position 0.4 with {\arrow{<}}},postaction = {decorate}] (-3.0,-1) to (-2.9, -1) [out=0, in =-135] to (-1.9, 0);
\draw[thick, decoration = {markings, mark = at position 0.4 with {\arrow{<}}},postaction = {decorate}] (3.0, 1) to (2.9, 1) [out=180, in =45] to (1.9,0);
\draw[thick, decoration = {markings, mark = at position 0.4 with {\arrow{>}}},postaction = {decorate}] (3.0,-1) to (2.9,-1) [out=180, in =-45] to (1.9, 0);
\draw[thick, dashed](-3.0, 0) to (-1.9, 0);
\draw[thick, dashed](3.0, 0) to (1.9, 0);
\node at (-1.3, 2.5) {$C_0$};
\node at (-0.7, 1.5) {$\Omega_+$};
\node at (-0.7, -1.5) {$\Omega_-$};
\node at (-1.8, -0.3){$-K$};
\node at (1.8, -0.3){$K$};
\node at (-0.9, -0.3){$-k_1$};
\node at (1.1, -0.3){$k_1$};
\end{tikzpicture}

\caption{On the left, the contour of the  RH problem for the  function $M^{(1)}$.  We observe that the contours off the real axis are symmetric with respect to the transformation $k\to\bar{k}$  up to orientation.  For this reason the  corresponding jump matrix $J^{(1)}(k)$  satisfies the condition $\(\ol{J^{(1)}(\ol{k})}\)^T=(J^{(1)}(k))^{-1}$ for $k$ off the real  axis. On the right, the contour of the RH problem for $\widetilde{M}(x,t;k)$.}
\label{Fig_RHcircle}
\end{figure}
 
\color{black}
\noindent {\textbf{Step 1. Transforming the RH problem to another one without points of discontinuities.}}

Let $C_0$ be a contour which encloses the segment $[\i c_-,-\i c_-],$ do not pass through any poles of $M,$ and lies in the domain of analyticity of the reflection coefficient $r(k)$ (see the left figure  in
 Figure \ref{Fig_RHcircle}). Let $C_0\cap\mathbb{R}=\left\{-k_1, k_1\right\}.$
Let $\Omega_+, \Omega_-$ be the domains inside $C_0$ which lie inside $\mathbb{C}_{\pm}=\left\{k:\pm\Im k>0\right\},$ respectively.
We transform the RH problem \ref{RH_problem_1} for the  function $M(x,t;k)$ to an equivalent RH problem for a matrix function $M^{(1)}(x,t;k)$, where the functions $M$ and $M^{(1)}$ are related as follows: $M^{(1)}(x,t;k)=M(x,t;k)\begin{bmatrix}1 & 0 \\ -r(k)\e^{2\i\theta(x,t;k)} & 1\end{bmatrix}$
for $k\in \Omega_+,$ 
$M^{(1)}(x,t;k)=M(x,t;k)\begin{bmatrix}1 & \ol{r(\ol k)}\e^{-2\i\theta(x,t;k)} \\ 0 & 1\end{bmatrix}$ for $k\in\Omega_-,$
and $M^{(1)}(x,t;k)=M(x,t;k)$ elsewhere.
Property \ref{prop10} of Lemma \ref{lem_abr} allows to verify that the jump for $M^{(1)}$ on the intervals $(\i c_-,\i c_+)$ and $(-\i c_+, -\i c_-)$ becomes identity, and on the interval $(\i c_+,-\i c_+)$ it becomes $\begin{bmatrix}0 & \i \e^{-2\i\theta(x,t;k)} \\ \i\e^{2\i\theta(x,t;k)} & 0\end{bmatrix}.$
The next transformation removes also the jump on $[\i c_+,-\i c_+];$ this is done by passing to a function $M^{(2)}(x,t;k)$, which is related to $M^{(1)}$ in the following way:
$M^{(2)}(x,t;k) = M^{(1)}(x,t;k)M_C(x,t;k)^{-1}$ for $k\in\Omega_+\cup\Omega_-,$
and 
$M^{(2)}(x,t;k)=M^{(1)}(x,t;k)$ elsewhere.
Here $$M_C(x,t;k) = \frac12\begin{bmatrix}
\gamma+\gamma^{-1}
&
(\gamma-\gamma^{-1})\e^{-2\i\theta(x,t;k)}
\\
(\gamma-\gamma^{-1})\e^{2\i\theta(x,t;k)}
&
\gamma+\gamma^{-1}
\end{bmatrix}$$ with $\gamma=\gamma(k)=\sqrt[4]{\frac{k-\i c_+}{k+\i c_+}}.$
The matrix $M^{(2)}(x,t;k)$ is a solution of a RH problem with the jump matrix $J^{(2)}(x,t;k),$ $M^{(2)}_-(x,t;k)=M^{(2)}_+(x,t;k)J^{(2)}(x,t;k),$  for 
$k\in \Sigma_2 = \mathbb{R}\cup C_0\backslash(-k_1,k_1).$
 Note that since the singularities of $M$ and $M_C$ at the points $\pm\i c_{\pm}$ are bounded by $(k\mp \i c_{\pm})^{-1/4},$ and since $\det M\equiv \det M_C \equiv 1,$ the  matrix  $M^{(2)}(k)$ has singularities at these points which are bounded by 
$(k\mp \i c_{\pm})^{-1/2}$. Since  $M^{(2)}(k)$  does not have a jump on  $[\pm\i c_-,\pm\i c_+] $ it follows  that the points $\pm\i c_{\pm}$  are removable singularities.

\color{black}
\noindent {\textbf{Step 2. Solvability of RH problem for $M^{(2)}.$}}
The solvability of the RH problem \ref{RH_problem_1} is thus reduced to solvability of the RH problem for the function $M^{(2)}.$ 
We introduce the operator 
\begin{equation}
\label{FK}
\mathcal{K}[f](k)= \frac{1}{2\pi\i}\int\limits_{\Sigma_{tot}}
 \frac{f(s)(I-J^{(2)}(x,t;s))\d s}{(s-k)_{+}},\quad f\in L_2(\Sigma_{tot};\mathbb{C}^{2\times 2})
 \end{equation}
 where the symbol $(s-k)_{+}$  stands for the  limiting value of the integral on the positive side of the oriented contour 
 $$\Sigma_{tot}=\mathbb{R}\cup C_0 \cup_{j} C_j\cup_j\ol{C}_j\cup_j- \ol{C}_j\cup_{j}-C_j\backslash(-k_1,k_1).$$
 The operator $\mathcal{K}$ is a  bounded linear   operator from $ L_2(\Sigma_{tot};\mathbb{C}^{2\times 2})$ to itself.
 %and $J^{(2)}$ is the jump matrix defined in \eqref{J_M0} and \eqref{J_M}.
\color{black}
 It is standard \cite{Deift99}, that the RH problem %~\ref{RH_problem_1}
 for $M^{(2)}$ is equivalent to the following singular integral equation for the  matrix  function $\mu=M^{(2)}_+-I:$
 \begin{equation}\label{SIE}
\mu(x,t;k)=\mathcal{K}[\mu](x,t;k)+\mathcal{K}[I](x,t;k).\end{equation} 
 % where 
% $$\mathcal{K}[\mu](k)=
% \lim\limits_{\tiny\begin{array}{l}k'\to k
% \\
% k'\in+\mbox{side}
% \end{array}}
% \frac{1}{2\pi\i}\int\limits_{\Sigma}
% \frac{\mu(x,t;s)\(I-J(x,t;s)\)\d s}{s-k'},
% \
%  \mathcal{F}(x,t;k)=
%  \lim\limits_{\tiny\begin{array}{l}k'\to k
% \\
% k'\in+\mbox{side}
% \end{array}}\frac{1}{2\pi\i}\int\limits_{\Sigma}
% \frac{\(I-J(x,t;s)\)\d s}{s-k'},$$
%where the limit is on the positive side of the contour $\Sigma.$
% Further, the contour $\Sigma$ and the jump matrix $J(x,t;k)$ satisfy the Schwartz reflection principle \cite{Zhou}:
 
\noindent Further, the \color{black}{jump matrix $J^{(2)}(x,t;k)$ off the real axis is Schwartz reflection invariant} \cite[Theorem 9.3]{Zhou}
\color{black}{(note the inverse power of the jump matrix in the second property below, which is due to  the fact that the orientation of the contour off the real axis changes to the opposite one under the map $k\mapsto\ol{k}$, see Figure \ref{Fig_RHcircle}}):
\color{black}
 \begin{itemize}
  \item the contour $\Sigma_{tot}$ is symmetric with respect to the real axis $\mathbb{R}$ \color{black}{(up to the orientation)},
  \item $J^{(2)}(x,t;k)^{-1}=\ol{J^{(2)}(x,t;\ol{k})}^T$ for $k\in \Sigma_{tot}\setminus\mathbb{R}$
%   which does not preserve orientation, and 
%  $J^{(2)}(x,t;k)=\ol{J^{(2)}(x,t;\ol{k})}^T$ for $k$ \mbox{ in parts of} $\Sigma_{tot}\setminus\mathbb{R}$ which  preserve orientation, 
  \item $J^{(2)}(x,t;k)$ has a positive definite real part for $k\in\mathbb{R}.$
 \end{itemize}
 {\color{black}
Then Theorem 9.3 from \cite{Zhou}(p.984) guarantees that the operator 
$Id-\mathcal{K}$, where $Id $ is the identity operator,   is invertible as an operator acting from  $L_2(\Sigma_{tot};\mathbb{C}^{2\times 2})$ to itself.  \color{black}
  Invertibility is guaranteed from the fact that $Id-\mathcal{K}$  is a Fredholm integral operator with zero index and the kernel
of $Id-\mathcal{K}$ is  the zero $2\times 2$ matrix. \color{black} In particular this last point is obtained by applying the vanishing lemma. Indeed suppose that exists $\mu\in L_2(\Sigma_{tot};\mathbb{C}^{2\times 2})$ such that
 $(Id-\mathcal{K})\mu=0$. Then the quantity
 $$M_0(x,t;k)=\frac{1}{2\pi\i}\int\limits_{\Sigma_{tot}}\frac{\mu(s)(I-J^{(2)}(x,t;s))\d s}
{s-k}$$
solves the  following  RH problem:\begin{itemize}
\item[(1)] $M_0(x,t;k)$  is analytic in $\C\backslash \Sigma_{tot}$,  
\item[(2)] $M_{0-}(k)=M_{0_+}(k)J^{(2)}(k)$ for $k\in\Sigma_{tot}$,
\item[(3)] $M_0(k)={\mathcal O}(k^{-1})$ as $k\to\infty$.  
\end{itemize}
We define  $H(k)=M_0(k)\ol{M_0(\ol{k})}^T$, then clearly $\ol{H(\ol{k})}^T=H(k)$, $H(k)={\mathcal O}(k^{-2})$  as $k\to\infty$  and $H(k)$ is analytic in $\C^+\backslash\{C_0\cup_jC_j\cup_j-\ol{C}_j\}$.
 Following \cite{Zhou}, we integrate the function $H(k)$ over the boundary of every closed component in $\mathbb{C}_+\setminus\Sigma_{tot}$ (those components are separated by contours $C_0$ and $C_j$; each component is integrated in  the  positive direction), and then add them to each other. Integrals over each part of the boundary except for the real axis is taken twice, and 
we thus obtain that by analyticity 
\begin{align*}
&  \int\limits_{-\infty}^{+\infty}H_+(k)dk
-
\int\limits_{C_0\cap\mathbb{C}^+}(H_-(k)-H_+(k))dk
-
\sum_{j}\left(\int_{C_j}
+
\int\limits_{-\ol{C}_j}\right)(H_-(k)-H_+(k))dk
\\
&=\int\limits_{-\infty}^\infty M_{0+}(k)\ol{J^{(2)}(k)}^T\ol{M_{0+}(k)}^Tdk
-
\int\limits_{C_0\cap\mathbb{C}^+}M_{0+}(k)(J^{(2)}(k)\ol{J^{(2)}(\ol{k})}^T-I)\ol{M_{0+}(\ol{k})}^Tdk\\
&\qquad-\sum_{j}\left(\int_{C_j}+\int_{-\ol{C}_j}\right)M_{0+}(k)\(J^{(2)}(k)\ol{J^{(2)}(\ol{k})}^T-I\)\ol{M_{0+}(\ol{k})}^Tdk=0\,.
\end{align*}
%which is equal to zero by analyticity. 
Here we defined $J^{(2)}=I$ for $k\in (-k_1, k_1)$, the part of the real axis not in $\Sigma_{tot}.$
\color{black}
To obtain the second identity we use the jumps $J^{(2)}(x,t;k)$ of the function $M_0$.
The symmetry properties of the jump matrix $J(x,t;k)$ imply that $J^{(2)}(k) = \ol{J^{(2)}(\ol k)}^{-1}$ for $k\in\Sigma_{tot}\backslash\mathbb{R}$, and hence the integrals over the  contour $C_0\cap\mathbb{C}^+$ and the circles $C_j$ and $-\ol{C}_j$ do not give any contribution, so that one obtains
 \begin{equation*}
\int_{-\infty}^\infty M_{0+}(k)(J^{(2)}(k)+\ol{J^{(2)}(k)}^T)\ol{M_{0+}(k)}^Tdk=0\,.
\end{equation*}
Since $J^{(2)}(k)+\ol{J^{(2)}(k)}^T$ is positive definite, it follows that $M_{0+}(k)$ is identically zero.  This shows that the operator  $Id-\mathcal{K}$  has a trivial kernel.   \color{black}
}
%the $L_2$ invertibility of the operator . 
%
Therefore, the singular integral equation \eqref{SIE} has a unique solution $M^{(2)}_+(x,t;k)-I\in L_2(\Sigma_{tot};\mathbb{C}^{2\times 2})$ for any fixed
$x,t\in\mathbb{R},$ and the matrix $M^{(2)}(x,t;k)$
%solution of the above RH problem~\ref{RH_problem_1} 
can be obtained by the formula $$M^{(2)}(x,t;k)=I+\frac{1}{2\pi\i}\int\limits_{\Sigma_{tot}}\frac{M^{(2)}_+(x,t;s)(I-J(x,t;s))\d s}
{s-k}.$$
\color{black}
Finally, inverting the transformations that led from the matrix $M(x,t;k)$ to the matrix $M^{(2)}(x,t;k)$ in Step 1, we obtain the 
solution of the  RH problem~\ref{RH_problem_1}. 
\color{black}

{\textbf{Step 3.} Differentiability of $M(x,t;k)$ with respect to $x,t.$}
First of all we notice that it is impossible to differentiate the equation \eqref{SIE} with respect to $x,t,$ since the 
function $r(k),$ as well as the matrix $I-J(x,t;k),$ vanishes as $k^{-1}$ as $k\to\infty$ along the real axis. 
To avoid a weak rate of decreasing of the matrix $I-J(x,t;k)$ for large real $k,$ we use an equivalent RH problem where the jump matrix $J(x,t;k)$ for large $k$ becomes exponentially close to $I.$

Let take a positive $K>0.$ Then for $k>K$  we use the following factorization 
of the jump matrix:
\color{black}{$$J(x,t;k)=\begin{pmatrix}1 & 0 \\ -r(k)\e^{2\i t \theta(k,\xi)} & 1\end{pmatrix}
\begin{pmatrix}1 & -\ol{r(\ol{k})}\e^{-2\i t \theta(k,\xi)}\\0 & 1\end{pmatrix},$$}
\noindent which allows to transform the RH problem \ref{RH_problem_1}  for the matrix $M(x,t,;k)$  into the  one for the  matrix $\widetilde{M}(x,t;k)$, 
where the parts of the contour $(-\infty, -K),$ $(K,+\infty)$ are splitted into two lines in the upper and lower complex planes (see the right figure of Figure \ref{Fig_RHcircle}).
As a result, the jump $\widetilde{J}(x,t;k)$ of the transformed RH problem is exponentially close to $I$ for large $k$ on the transformed contour 
$\widetilde{\Sigma}_{tot}.$
The corresponding singular integral equation is as follows for the matrix $\widetilde{\mu}=\widetilde{M}_+-I$:
\begin{equation}\label{SIE2}\widetilde{\mu}(x,t;k)=\widetilde{\mathcal{K}}[\widetilde{\mu}](x,t;k)+\widetilde{\mathcal{K}}[I](x,t;k),
\end{equation}
where $\widetilde{\mathcal{K}}$ is as $\mathcal{K}$ in \eqref{FK} with  $J$  replaced by $\widetilde{J}$,
Following the  reasoning  as for \eqref{SIE}, the above integral equation has a unique solution in  $L_2(\widetilde{\Sigma}_{tot};\mathbb{C}^{2\times 2}).$  
\color{black}
Equation 
\eqref{SIE2} has the advantage, that we can differentiate it with respect to $x,t$ as many times as we wish.
Indeed, while the function $I-J(x,t;k)$ vanishes as $\frac{1}{k}$ on $\mathbb{R}$, the  function $I-\widetilde{J}(x,t;k)$ decays exponentially  fast on the infinite parts of the contour 
$\widetilde\Sigma.$ The singular integral equation, obtained by differentiating with respect to $t$ \eqref{SIE2}, has the same form as 
\eqref{SIE2} (only the right-hand-side of this equation changes). It provides unique solvability of the partial derivatives of 
$\widetilde{\mu}(x,t;k)$ 	with respect to $x,t.$ Hence, the same is true for $M(x,t;k).$

\noindent {\textbf{Step 4. Zakharov-Shabat scheme.}}
It is a standard Lax-pair argument (see, for instance, \cite{M_disser}) to show, that the function 
$$\widetilde{\Phi}(x,t;k)= M(x,t;k) \e^{(\i k x+4\i k^3 t)\sigma_3}$$
satisfies the Ablowitz-Kaup-Newell-Segur system of equations  (see e.g. \cite{Clarkson})
\begin{equation*}
\begin{split}
\widetilde{\Phi}_x(x,t;k) +\i k\sigma_3\widetilde{\Phi}(x,t;k)=Q(x,t)\widetilde{\Phi}(x,t;k),
\\
\widetilde{\Phi}_t(x,t;k) +4\i k^3\sigma_3\widetilde{\Phi}(x,t;k)=\widehat{Q}(x,t;\color{black}{k})\widetilde{\Phi}(x,t;k),
\end{split}\qquad
x\in\mathbb{R}, t>0,
\end{equation*}
where $$Q(x,t)=\begin{pmatrix}0&q(x,t)\\-q(x,t) & 0\end{pmatrix},$$
$$\widehat{Q}(x,t;k)=4k^2Q(x,t)-2\i k\(Q^2(x,t)+Q_x(x,t)\)\sigma_3+2Q^3(x,t)-Q_{xx}(x,t)$$
with the function $q(x,t)$ given by 
$$q(x,t)=2\i\lim\limits_{k\to\infty}k[M(x,t;k)]_{21}=\frac{-1}{\pi}\int\limits_{\Sigma}(I+\mu(x,t;k))\(I-J(x,t;k)\)\d k.$$
\end{proof}

\section{Model problems}\label{sect:model}
The asymptotic analysis of the Cauchy problem as $t\to \infty$ consists of several steps:
\begin{itemize}
\item the first step is to change the original phase function, which is present in the exponents of  the RH problem~\ref{RH_problem_1}, with an appropriate  function $g(k,\xi)$ that will be determined later; 
\item the second step consists in performing a chain of` ``exact" transformations of the RH problem;
\item the third step consists in approximating  the new RH problem  to some  model problem;
\item the fourth step consists in solving the model problems.
\end{itemize}
Before starting our asymptotic  analysis, we  introduce the RH  model problems that will   be obtained  from such analysis.
Namely the  RH problem for the soliton and breather solution on a constant background and the RH problem for the travelling wave solution \eqref{periodic_intro}. In particular for the travelling wave solution, we show
that  it can be obtained from a model problem solvable in terms of elliptic theta functions but also from a model problem that is solvable via hyperelliptic theta functions
that are defined on a genus $2$ hyperelliptic Riemann surface with symmetries. The first case arises when the step-like initial data is such that $c_{\r}=0$, while the second case occurs when $c_{\r}>0$.
 Then  we show that the genus $ 2$  solutions can nevertheless be written in a genus $1$ form.

\subsection{One-soliton solution on a constant background $c>0$}
\label{sect_sol_background}
Here we derive a one-soliton solution on a constant background. We use the notation $[\i c, -\i c]$ to denote the interval oriented dowward.

\noindent
\begin{RHP}\label{RHP_sol_const}
 Find a $2\times 2$ matrix $M(k)=M(x,t;k)$  meromorphic for $k\in\mathbb{C}\setminus[\i c, -\i c]$ with simple poles at $k=\pm\i\kappa_0, \quad \kappa_0>c>0,$  with the following properties:
\begin{enumerate}
\item the boundary values $M_{\pm}(k)$ for $k\in (\i c, -\i c)$ satisfy the  jump relations
 \begin{equation}
\label{RHs21}
M_-(k)=M_+(k)\begin{pmatrix}0&\i\\\i&0\end{pmatrix},\quad k\in(\i c, -\i c);
\end{equation}
 \item pole conditions: 
 \begin{equation}
\label{RHs22}
\begin{split}
&\mathrm{Res}_{\i\kappa_0}M(k)=\lim\limits_{k\to \i\kappa_0}M(k)\begin{pmatrix}0&0\\\i \nu\e^{2\i  g_c(x,t;k)} & 0\end{pmatrix},\\
&\mathrm{Res}_{-\i\kappa_0}M(k)=\lim\limits_{k\to -\i\kappa_0}M\begin{pmatrix}0&\i \nu\e^{-2\i  g_c(x,t;k)}\\0&0\end{pmatrix},
\end{split}
\end{equation}
where  $\nu$ is a non zero real constant and  
\begin{equation}
\label{gc}
g_c(k; x,t)=(2(2k^2-c^2)t+x)\sqrt{k^2+c^2};
\end{equation}
\item asymptotics: 
\begin{equation}
\label{RHs23}
M(k)= I+O\left(\frac{1}{k}\right),\quad \mbox{as $k\to\infty.$}
\end{equation}
\end{enumerate}
\end{RHP}
\noindent The solution of the MKdV equation    is obtained from   the matrix $M(k;x,t)$ by the relation 
\begin{equation}
\label{q_sol_0}
q(x,t)=2\i\lim\limits_{k\to\infty}k M_{21}(x,t;k)=2\i\lim\limits_{k\to\infty}k M_{12}(x,t;k).
\end{equation}
\begin{lemma}
The solution of   the MKdV $q_t+6q^2q_x+q_{xxx}=0$ obtained from the solution of the  RH problem ~\ref{RHP_sol_const}  is equal to  a soliton on a constant background $c$, namely
\begin{equation}
\label{q_sol_const}
q_{soliton}(x,t;\ c,\kappa_0,\nu)=c-\frac{2\,\mathrm{sgn}(\nu)(\kappa_0^2-c^2)}{\kappa_0\cosh\left[2\sqrt{\kappa_0^2-c^2}(x-(2c^2+4\kappa_0^2)t)+x_0\right]-\mathrm{sgn}(\nu)c},
\end{equation}
where
\begin{equation*}
x_0=\ln\frac{2(\kappa_0^2-c^2)}{|\nu|\kappa_0}\,.
\end{equation*}
\end{lemma}
\begin{proof}
%Here  $\kappa_0>0$, $\rho$  real,  are parameters of the soliton.
We first obtain the solution of  the RH problem  described by equations \eqref{RHs21} and \eqref{RHs23} in the form
\begin{equation}
\label{Mreg}
M_{reg}(k)=\begin{pmatrix} \psi_2(k)& \psi_1(k)\\
\psi_1(k)& \psi_2(k)
\end{pmatrix},
\end{equation}
where 
 $$ \psi_2=\dfrac{1}{2}\left(\sqrt[4]{\frac{k-\i c}{k+\i c}}+\sqrt[4]{\frac{k+\i c}{k-\i c}}\right),\quad 
 \psi_1=\dfrac{1}{2}\left(\sqrt[4]{\frac{k-\i c}{k+\i c}}-\sqrt[4]{\frac{k+\i c}{k-\i c}}\right).$$
Since $M(k) M_{reg}(k)^{-1}$ does not have jumps on $\C$ but only poles, 
 the solution  $M(k)$  can be found in the form
$$M(k;x,t)=
\begin{pmatrix}
 1+\frac{\i\alpha(x,t)}{k-\i\kappa_0}+\frac{\i\gamma(x,t)}{k+\i\kappa_0} & \frac{\i\beta(x,t)}{k+\i\kappa_0}+
\frac{\i\delta(x,t)}{k-\i\kappa_0}
\\
\frac{\i\beta(x,t)}{k-\i\kappa_0}+\frac{\i\delta(x,t)}{k+\i\kappa_0} & 1-\frac{\i\alpha(x,t)}{k+\i\kappa_0}-\frac{\i\gamma(x,t)}{k-\i\kappa_0}
\end{pmatrix}
M_{reg}(k),
$$
where $\alpha, \beta, \gamma, \delta$ are real parameters to be determined.
The solution of the MKdV equation    is obtained from   the matrix $M(k;x,t)$ through the formula
\begin{equation}
\label{q_sol_01}
q(x,t)=2\i\lim\limits_{k\to\infty}k M_{21}(x,t;k)=2\i\lim\limits_{k\to\infty}k M_{12}(x,t;k)=c-2\beta-2\delta.
\end{equation}
We observe that   due the symmetry of the problem, it is enough to consider the residue condition only at one of the poles $k=\pm\i\kappa_0$.   For example the condition \ref{RHs22} at   $k=\i\kappa_0$  gives 
\begin{equation}
\begin{split}
&\alpha\psi_1(\i\kappa_0)+\delta\psi_2(\i\kappa_0)=0,\\
&\beta \psi_1(\i\kappa_0)-\gamma\psi_2(\i\kappa_0)=0,\\
&\dfrac{1}{\nu}(\alpha\psi_2(\i\kappa_0)+\delta\psi_1(\i\kappa_0))\e^{-2\i g_c(\i\kappa_0)}=\psi_1(\i\kappa_0)+\dfrac{\gamma}{2\kappa_0}\psi_1(\i\kappa_0)+\\
&+
\dfrac{\beta}{2\kappa_0}\psi_2(\i\kappa_0)+\i\alpha\psi_1'(\i \kappa_0)+\i\delta\psi_2'(\i\kappa_0),\\
&\dfrac{1}{\nu}(\beta\psi_2(\i\kappa_0)-\gamma\psi_1(\i\kappa_0))e^{-2\i g_c(\i\kappa_0)}=\psi_2(\i\kappa_0)+\dfrac{\delta}{2\kappa_0}\psi_1(\i\kappa_0)-\\
&-\dfrac{\alpha}{2\kappa_0}\psi_2(\i\kappa_0)+\i\beta\psi_1'(\i \kappa_0)-\i\gamma\psi_2'(\i\kappa_0),
\end{split}
\end{equation}
where $'$ stands for derivative with respect to $k$.  Solving the above system of equations we obtain 
\begin{equation}
\alpha=-\delta\dfrac{\psi_2(\i\kappa_0)}{\psi_1(\i\kappa_0)},\quad \gamma=\beta\dfrac{\psi_1(\i\kappa_0)}{\psi_2(\i\kappa_0)}
\end{equation}
%$$\alpha(x,t)=\frac{-c\rho \e^{2 i g_c}\left\{2(k_0^2-c^2)(2k_0(k_0^2-c^2)-(2k_0^2-c^2)\sqrt{k_0^2-c^2}) \ + \ \e^{2\i g_c}c k_0\rho(c^2-k_0^2+k_0\sqrt{k_0^2-c^2}) \right\}}
%{(c^2-2k_0^2+2k_0\sqrt{k_0^2-c^2})\left\{4(k_0^2-c^2)^2-4(k_0^2-c^2)\rho c\ \e^{2\i g_c}+k_0^2\rho^2\e^{4\i g_c}\right\}},
%$$
and 

%$$\alpha(x,t)=\frac{-c\rho \e^{2 i g_c}\left\{2(k_0^2-c^2)(2k_0(k_0^2-c^2)-(2k_0^2-c^2)\sqrt{k_0^2-c^2}) \ + \ \e^{2\i g_c}c k_0\rho(c^2-k_0^2+k_0\sqrt{k_0^2-c^2}) \right\}}
%{(c^2-2k_0^2+2k_0\sqrt{k_0^2-c^2})\left\{4(k_0^2-c^2)^2-4(k_0^2-c^2)\rho c\ \e^{2\i g_c}+k_0^2\rho^2\e^{4\i g_c}\right\}},
%$$

$$\beta(x,t)=\sqrt{\kappa_0^2-c^2}\frac{\nu \e^{2 i g_c}\left\{-c\nu \kappa_0 \e^{2\i g_c}+2(\kappa_0^2-c^2)(\kappa_0+\sqrt{\kappa_0^2-c^2} ) \right\}}
{\left\{4(\kappa_0^2-c^2)^2-4(\kappa_0^2-c^2)\nu c\ \e^{2\i g_c}+\kappa_0^2\nu^2\e^{4\i g_c}\right\}},$$

$$\delta(x,t)=\sqrt{\kappa_0^2-c^2}\frac{\nu \e^{2 i g_c}\left\{c\nu \kappa_0 \e^{2\i g_c}-2(\kappa_0^2-c^2)(\kappa_0-\sqrt{\kappa_0^2-c^2} ) \right\}}
{\left\{4(\kappa_0^2-c^2)^2-4(\kappa_0^2-c^2)\nu c\ \e^{2\i g_c}+\kappa_0^2\nu^2\e^{4\i g_c}\right\}},$$
where $g_c=g_c(i\kappa_0)$.

%$$\delta(x,t)=\frac{-c^2\nu \e^{2 i g_c}\left\{2(k_0^2-c^2)(2k_0(k_0^2-c^2)-(2k_0^2-c^2)\sqrt{k_0^2-c^2}) \ + \ \e^{2\i g_c}c k_0\nu(c^2-k_0^2+k_0\sqrt{k_0^2-c^2}) \right\}}
%{(k_0+\sqrt{k_0^2-c^2})(c^2-2k_0^2+2k_0\sqrt{k_0^2-c^2})\left\{4(k_0^2-c^2)^2-4(k_0^2-c^2)\nu c\ \e^{2\i g_c}+k_0^2\nu^2\e^{4\i g_c}\right\}}.$$

%Let us notice that $$\beta(x,t)+\delta(x,t)=\frac{4(k_0^2-c^2)^2\nu \e^{2\i g_c}}{4(k_0^2-c^2)^2-4(k_0^2-c^2)\nu c\ \e^{2\i g_c}+k_0^2\nu^2\e^{4\i g_c}}$$

Plugging the above expressions for $\beta$ and $\delta$ into \eqref{q_sol_01} we obtain the statement of the lemma.
\end{proof}
%\begin{equation}\label{q_sol_const}
%q(x,t;\ c,\i\kappa_0,\i\nu)=c+\frac{\frac{2(\kappa_0^2-c^2)}{\kappa_0}}{\ \frac{c}{\kappa_0}-\mathrm{sgn}(\nu)\cdot\cosh\left[2\sqrt{\kappa_0^2-c^2}(x-(2c^2+4\kappa_0^2)t+
%\frac{\ln\frac{2(\kappa_0^2-c^2)}{|\nu|\kappa_0}}{2\sqrt{\kappa_0^2-c^2}})\right]}.
%\end{equation}
%
We observe that when $c=0$, the formula \eqref{q_sol_const} coincides with the one-soliton solution \eqref{soliton_intro}.

Let us notice that the denominator in the  formula \eqref{q_sol_const} is always nonzero.
%Formula \eqref{q_sol_const} coincides with \eqref{soliton}, \eqref{soliton_x0} for $c=0.$
%
%Function $q(x,t)-c$ \eqref{q_sol_const} varies
%over the interval
%$$(0, -2\mathrm{sgn}(\nu)(\kappa+\mathrm{sgn}(\nu)\, c)).$$
%
Since $\kappa_0>c$, in the case $\nu>0$ we have antisoliton, oriented downward, and for $\nu<0$ we have a  soliton, oriented upward.

We see that for $\kappa_0-c\to+0$ and $\nu<0$ the amplitude of the soliton tends to 0, while for $\nu>0$ the amplitude oantisolitonntisoliton tends to a constant $2\kappa.$

\textbf{Degenerate case of antisoliton.}\\
When $\nu>0$ and
% $\nu\equiv\frac{2(k_0^2-c^2)}{k_0}\cdot \e^{-x_0\sqrt{k_0^2-c^2}},$ and
 we let $k_0\to c+0,$ we obtain the special case of an antisoliton,  namely a  rational solution.
In this case 
%$$\alpha=\frac{c(1-2c(x-x_0)+12c^3t)}{1+(2c(x-x_0)-12c^3t)^2},$$
%
%$$\gamma=\frac{-c(1+2c(x-x_0)-12c^3t)}{1+(2c(x-x_0)-12c^3t)^2},$$
%
%$$\beta=\frac{c(1+2c(x-x_0)-12c^3t)}{1+(2c(x-x_0)-12c^3t)^2},$$
%$$\delta=\frac{c(1-2c(x-x_0)+12c^3t)}{1+(2c(x-x_0)-12c^3t)^2}.$$

$$q(x,t;\ c,c,\nu)=c-\frac{4c}{1+(2c(x-x_0)-12c^3t)^2}.$$

%Let us notice, that this degenerate antisoliton is not within the class of  initial data we are considering and for this reason it won't appear in the asymptotic analysis.

\begin{remark}
Let us observe that the MKdV admits a RH problem with poles of higher order \cite{WO82}. This case  is  non-generic and it is considered in  Appendix~\ref{sect_nongen}.
 \end{remark}

\subsection{Simple breathers on a constant  background}
\label{sect_breath}
In this section we consider a breather on a constant  background.
%\todo{Done: It remains to change the notations here and to make sure the final formulas are still correct.}
\begin{RHP}\label{RHP_breather}
\noindent
 Find a $2\times 2 $ matrix  $M(x,t;k)$  meromorphic in $k\in\mathbb{C}\backslash[-\i c,\i c]$ with simple poles at $\kappa\equiv\kappa_1+i\kappa_2,$ $\ol{\kappa}\equiv\kappa_1-i\kappa_2,$ $-\kappa\equiv-\kappa_1-i\kappa_2,$
$-\ol{\kappa}\equiv-\kappa_1+i\kappa_2,$ and such that 
\begin{enumerate}
\item  $M_-(k)=M_+(k) J(k),\quad k\in(\i c,-\i c),$ where $$J(k)=\begin{pmatrix}0&\i \\\i & 0\end{pmatrix},\quad k\in(\i c,-\i c);$$
\item pole conditions: \color{black}{$$\mathrm{Res}_{\kappa}M(k)=\lim\limits_{k\to \kappa}M(k)\begin{pmatrix}0&0\\ \i\nu \e^{2\i   g_c(x,t;k)} & 0\end{pmatrix},$$
$$\mathrm{Res}_{-\ol{\kappa}}M(k)=\lim\limits_{k\to -\ol{\kappa}}M(k)\begin{pmatrix}0&0\\ \i\,\widebar{\nu} \e^{2\i   g_c(x,t;k)} & 0\end{pmatrix},$$
$$\mathrm{Res}_{\ol{\kappa}}M(k)=\lim\limits_{k\to \ol{\kappa}}M(k)\begin{pmatrix}0&\i\,\widebar{\nu}\e^{-2\i  g_c(x,t;k)}\\0&0\end{pmatrix},$$
$$\mathrm{Res}_{-\kappa}M(k)=\lim\limits_{k\to -\kappa}M(k)\begin{pmatrix}0&\i\nu\e^{-2\i g_c(x,t;k)}\\0&0\end{pmatrix},$$}
where $\nu$ is a non zero complex number and 
$g_c(x,t;k)$ as in \eqref{gc};\item asymptotics: $M(k)=I+O(\frac{1}{k})$ as $k\to\infty.$
\end{enumerate}
\end{RHP}
%
%The solution of the above RH problem satisfies the differential equation
%$$M_x+\i k\gamma_3M-\i X(k)M\gamma_3=\i[\gamma_3,m_1]M,\quad Q=\begin{pmatrix}0&q\\-q&0\end{pmatrix}=\i[\gamma_3,m_1],\ \mbox{ where }\ M=I+\dfrac{m_1}{k}+\ldots, k\to\infty,$$

The solution of MKdV $q_t+6q^2q_x+q_{xxx}=0$ is obtained from the solution of this RH problem by one of the following formulas:
\begin{equation}
\label{q_breath1}
q_{breather}(x,t)=2\i\lim\limits_{k\to\infty}k M_{21}(x,t;k)=2\i\lim\limits_{k\to\infty}k M_{12}(x,t;k),
\end{equation}
or
$$q_{breather}^2(x,t)=c^2+2\i \partial_x \(\lim\limits_{k\to\infty} k(M-I)_{11}\).$$
\begin{theorem}
\label{theo_breather}
The solution \eqref{q_breath1} of the MKdV equation  obtained from the solution of the RH problem ~\ref{RHP_breather} corresponds to a breather 
 on a constant background $c$   with discrete spectrum $\kappa=\kappa_1+\i \kappa_2$ and complex parameter $\nu$  and  takes the form  \begin{equation}
  \label{q_breath}
\begin{split}
&q_{breather}(x,t;c,\kappa,\nu) =c+\\
&2\partial_x\arctan\left[\dfrac{{|\chi|}\cos\left( 2\Re g_{c}(x,t)+\theta_1-\theta_2\right)+\frac{c|\nu| \chi_1^2}{2|\chi|^2 \chi_2}\e^{- 2\Im g_{c}(x,t)}}{\frac{|\chi|^2}{|\nu|}\e^{2\Im g_{c}(x,t)}+\frac{\chi_1^2(|\chi|^2-c^2)}{4|\chi|^2\chi_2^2}|\nu|\e^{ -2\Im g_{c}(x,t)}+c \sin\left( 2\Re g_{c}+\theta_1-2\theta_2\right)}\right]
\end{split}
\end{equation}
%\todo[inline]{Here are some discrepancies with formulas (50), page24, from the file Reduction\_WITH\_MULTIPLE\_POLES\_g2\_g1\_20180604 .
%First of all, it seems that denominator and nominator in \eqref{q_breath} should be swapped, otherwise there is another sign in front of $\partial_x\arctan.$
%Second, formulas for $\Re g_c$ and $\Im g_c$ 2 lines below are indeed formulas for $2 \Re g_c$ and $2\Im g_c.$ Function $g_c$ is a stable notation, $g_c=2(2k^2-c^2)\chi t+\chi x.$
%}

\noindent 
where  $\chi=\chi_1+\i \chi_2:==\sqrt{\kappa^2+c^2}$,  with  $\chi_1>0$, $\chi_2>0$ and  real and 
 \begin{equation*}
 \begin{split}
& \Re g_{c}(x,t)=\chi_1\left(4\left(\chi_1^2-3\chi_2^2-\frac{3}{2}c^2\right)t+x\right),\\
& \Im g_{c}(x,t)=\chi_2\left(4\left(3\chi_1^2-\chi_2^2-\frac{3}{2}c^2\right)t+x\right),
\end{split}
 \end{equation*}
 and  phases \color{black}{$\theta_1=\mbox{arccos}\dfrac{-\nu_2}{|\nu|}$} and $\theta_2=\mbox{arccos}\dfrac{\chi_1}{|\chi|}$.
 Formula \eqref{q_breath} coincides with formula     \eqref{q_breath_intro}   for the breather provided in the introduction. 
 \end{theorem}
 \begin{remark}
 \label{Remark_b}
Let us introduce the quantities
 \begin{equation*}
 Z=x+4t(3\chi_1^2-\chi_2^2-\frac32c^2),\,\quad  \varphi=2(Z-8t|\chi|^2)\chi_1+\theta_1-\theta_2.
 \end{equation*}
Then  the breather \eqref{q_breath} on the constant background $c$ can be written in the form
 \begin{equation}\nonumber
  \begin{split}
&q_{breather}(x,t;c,\kappa,\nu)=c+\\
&+2\begin{cases}
\partial_x\arctan\left[\dfrac{|\chi|\cos\left(\varphi\right)+\frac{c|\nu| \chi_1^2}{2|\chi|^2 \chi_2}\e^{- 2 Z \chi_2}}
{\sqrt{|\chi|^2-c^2}\dfrac{\chi_1}{\chi_2}\cosh ( 2 Z \chi_2+\theta_3) +c\sin\left(\varphi-\theta_2\right)}\right],\;\;
\hfill{ if \,} |\chi|>c,\\
\partial_x\arctan\left[\dfrac{|\chi|\cos\left(\varphi\right)+\frac{c|\nu| \chi_1^2}{2|\chi|^2 \chi_2}\e^{- 2 Z \chi_2}}
{\sqrt{c^2-|\chi|^2}\dfrac{\chi_1}{\chi_2}\sinh ( 2 Z \chi_2+\theta_3) +c\sin\left(\varphi-\theta_2\right)}\right],\;\;
\hfill{ if \,\,} |\chi|<c,\\
\partial_x\arctan\left[\dfrac{\cos\left(\varphi\right)+\frac{|\nu| \chi_1^2}{2|\chi|^2 \chi_2}\e^{- 2 Z \chi_2}}
{\frac{c}{|\nu|}\e^{2Z\chi_2}+\sin\left(\varphi-\theta_2\right)}\right],\hfill{ if \,\,} |\chi|=c,
\end{cases}
\end{split}
\end{equation}
where $\theta_3=\log\frac{2\chi_2|\chi|^2}{|\nu|\chi_1\sqrt{|\chi|^2-c^2}}$.
 \end{remark}
 \color{black}
 \begin{proof}[Proof of Theorem~\ref{theo_breather}]
The solution  of the RH problem~\ref{RHP_breather}
 can be obtained  in the form
$$M(k)=M_{pol}(k)M_{reg}(k),$$ where 
$M_{reg}(k)$   has been defined in \eqref{Mreg} and $M_{pol}(k)=M_{pol}(x,t;k)$ admit an ansatz of the form 
\begin{equation}\label{Mpol}M_{pol}(k)=I+
\begin{pmatrix}
 \frac{A(x,t)}{k-\kappa} - \frac{\ol{A(x,t)}}{k+\ol{\kappa}} + \frac{C(x,t)}{k+\kappa} - \frac{\ol{C(x,t)}}{k-\ol{\kappa}}&  \frac{B(x,t)}{k+\kappa}-\frac{\ol{B(x,t)}}{k-\ol{\kappa}} -\frac{\ol{D(x,t)}}{k+\ol{\kappa}} + \frac{D(x,t)}{k-\kappa}
\\
\frac{B(x,t)}{k-\kappa} - \frac{\ol{B(x,t)}}{k+\ol{\kappa}} + \frac{D(x,t)}{k+\kappa} - \frac{\ol{D(x,t)}}{k-\ol{\kappa}} & \frac{\ol{A(x,t)}}{k-\ol{\kappa}} - \frac{A(x,t)}{k+\kappa} + \frac{\ol{C(x,t)}}{k+\ol{\kappa}} - \frac{C(x,t)}{k-\kappa}
\end{pmatrix},
\end{equation}
where $A=A(x,t)$, $B=B(x,t)$, $C=C(x,t)$ and $D=D(x,t)$ are unknown functions to be determined.
%The solution to the MKdV equation is given by the formula
%$$q_{breath}=c-4\Im(B)-4\Im(D).$$
%$$q_{breath}^2(x,t)=c^2-4(\Im A_x+\Im C_x).$$
%Writing down the pole conditions then lead to a linear system of equations for unknown parameters of $M_{sol}.$  
%The  Schwartz symmetry of the RH problem (poles are not problem here, since they can be 
%reformulated as jumps conditions on circles) guarantees existence of solution for all $x,t.$

Writing down the pole conditions at the point $\kappa,$ $\Re\kappa>0,$ $\Im\kappa>0,$ which is enough due to symmetries of the problem, we obtain 
 the following system of equations:
\begin{equation}
\label{sys_breath1}
\begin{cases}
 A\psi_1+D\psi_2=0,\\
B\psi_1-C\psi_2=0, \qquad (\textrm{pole free condition for the 2nd row at }\kappa)
\\
\frac{A\psi_2+D\psi_1}{{\i}\nu\e^{2\i g_c}}=\psi_1(1-\frac{\ol{A}}{\kappa+\ol{\kappa}} +\frac{C}{2\kappa} -\frac{\ol{C}}{\kappa-\ol{\kappa}} ) + 
\psi_2(\frac{-\ol{B}}{\kappa-\ol{\kappa}}+\frac{B}{2\kappa}-\frac{\ol{D}}{\kappa+\ol{\kappa}})+\psi_1'A+\psi_2'D,
\\
\frac{B\psi_2-C\psi_1}{{\i}\nu\e^{2\i g_c}}=\psi_1(\frac{-\ol{B}}{\kappa+\ol{\kappa}}+\frac{D}{2\kappa}-\frac{\ol{D}}{\kappa-\ol{\kappa}})+\psi_2(1+\frac{\ol{A}}{\kappa-\ol{\kappa}}-\frac{A}{2\kappa}+
\frac{\ol{C}}{\kappa+\ol{\kappa}})+\psi_1' B-\psi_2' C,
\end{cases}
\end{equation}
where   we use the compact notation $\psi_j=\psi_j(\kappa)$  and  $\psi_j'=\psi_j'(\kappa),$  where $\psi_j'(\kappa)=\dfrac{d}{dk}\psi_j(k)|_{k=\kappa}$ 
and
 $$ \psi_1(k)=\dfrac{1}{2}\left(\sqrt[4]{\frac{k-\i c}{k+\i c}}-\sqrt[4]{\frac{k+\i c}{k-\i c}}\right),\quad \psi_2(k)=\dfrac{1}{2}\left(\sqrt[4]{\frac{k-\i c}{k+\i c}}+\sqrt[4]{\frac{k+\i c}{k-\i c}}\right).$$
Let us introduce 
\begin{equation*}
\mathcal{R}(k)=\frac{\psi_1(k)}{\psi_2(k)}=\dfrac{\i}{c}(k-\sqrt{k^2+c^2}).
\end{equation*}
With the above notation it is straighforwad to obtain the solution  \eqref{q_breath1} of the MKdV equation 
 in the form
 \begin{equation}
\label{q_breath2}
 q_{breath}(x,t)=c+4\Im(A(x,t)\mathcal{R}(\kappa)-B(x,t)).
 \end{equation}
 We need to determine the quantities $A(x,t)$ and $B(x,t)$ in \eqref{q_breath2}. These constants are obtained by solving the system of 
equations \eqref{sys_breath1} that can be written in the form 
\color{black}
(we use here the determinantal property $\psi_2^2-\psi_1^2=1$)
\color{black}
\color{black}{
\begin{equation}\label{eq_sys_AB}
\begin{cases}
\(\dfrac{\e^{-2\i g_c(\kappa)}}{\i\nu\psi_2^2(\kappa)}-\mathcal{R}'(\kappa)\)A+\frac{\mathcal{R}(\kappa)-\ol{\mathcal{R}(\kappa)}}{\kappa+\ol{\kappa}}\ol{A}
-\frac{\mathcal{R}^2(\kappa)+1}{2\kappa}B+\frac{1+|\mathcal{R}(\kappa)|^2}{\kappa-\ol\kappa}\ol{B}=\mathcal{R}(\kappa),
\\
\(\dfrac{\e^{-2\i g_c(\kappa)}}{\i\nu\psi_2^2(\kappa)}-\mathcal{R}'(\kappa)\)B+\frac{\mathcal{R}(\kappa)-\ol{\mathcal{R}(\kappa)}}{\kappa+\ol{\kappa}}\ol{B}
+\frac{\mathcal{R}^2(\kappa)+1}{2\kappa}A-\frac{1+|\mathcal{R}(\kappa)|^2}{\kappa-\ol\kappa}\ol{A}=1.
\end{cases}
\end{equation}}
The system  of equations \eqref{eq_sys_AB} for $A,B$ is clearly a system of four  linear equations for the four real variables $A=A_1+\i A_2, B=B_1+\i B_2.$
Introducing \color{black} the variables $Z=A+\i B$ and $W=A-\i B$ the system of equations \eqref{eq_sys_AB} can be recast in the form

\color{black}{\begin{equation}\label{eq_sys_AB1}
\begin{cases}
\(\dfrac{\e^{-2\i g_c(\kappa)}}{\i\nu\psi_2^2(\kappa)}-\mathcal{R}'(\kappa)+\i \frac{\mathcal{R}^2(\kappa)+1}{2\kappa}\)Z+
\left(\frac{\mathcal{R}(\kappa)-\ol{\mathcal{R}(\kappa)}}{\kappa+\ol{\kappa}}
-\i \frac{1+|\mathcal{R}(\kappa)|^2}{\kappa-\ol\kappa}\right)\widebar{W}=\mathcal{R}(\kappa)+\i ,\\
\\
\left(\i\frac{1+|\mathcal{R}(\kappa)|^2}{\kappa-\ol\kappa}-\frac{\mathcal{R}(\kappa)-\ol{\mathcal{R}(\kappa)}}{\kappa+\ol{\kappa}}\right)Z
+\overline{\(\dfrac{\e^{-2\i g_c(\kappa)}}{\i\nu\psi_2^2(\kappa)}-\mathcal{R}'(\kappa)\)}{  \widebar W}
+\i\frac{\overline{\mathcal{R}^2(\kappa)}+1}{2\widebar{\kappa}}\widebar{W}=\widebar{R(\kappa)}+\i.
\end{cases}
\end{equation}}
Defining the quantities 
\begin{equation}
\label{E-G}
\begin{split}
&\color{black}{E=\dfrac{\e^{-2\i g_c(\kappa)}}{\i\nu\psi_2^2(\kappa)}-\mathcal{R}'(\kappa),}\\
&F=\frac{\mathcal{R}^2(\kappa)+1}{2\kappa},\\
&H=\frac{\mathcal{R}(\kappa)-\ol{\mathcal{R}(\kappa)}}{\kappa+\ol{\kappa}},\\
&G=\frac{1+|\mathcal{R}(\kappa)|^2}{\kappa-\ol\kappa},
\end{split}
\end{equation}
the solution of \eqref{eq_sys_AB1} is obtained as
\begin{equation*}Z=A+\i B=\dfrac{\left|\begin{array}{ccc}\mathcal{R}+\i  & H-\i G\\\widebar{\mathcal{R}}+\i&\ol{E}+\i \ol{F}\end{array}\right|}
{\left|\begin{array}{ccc}E+\i F & H-\i G\\-H+\i G\;&\ol{E}+\i\;\; \ol{F}\end{array}\right|}=:\frac{r+\i s}{f+\i h}=\frac{f r+ s h }{f^2+h^2}+\i \frac{s f-r h}{f^2+h^2},\end{equation*}
and  similarly $W=A-\i B=\frac{f r+s h }{f^2+h^2}-\i \frac{s f-r h}{f^2+h^2},$
where 
\begin{equation}
\label{rs}
\begin{split}
&r=\mathcal{R}\ol{E}-\ol{\mathcal{R}}H-\ol{F}-G,\\
&s=\ol{E}-H+\mathcal{R}\ol{F}+\ol{\mathcal{R}}G,\\
& f=|E|^2-|F|^2-G^2+H^2,\\
&h=2\Re\(E\ol{F}-HG\).
\end{split}
\end{equation}
Let us notice  that $f,h$ are real, while $r,s$ are complex-valued functions.

\noindent The solution to MKdV is given by the formula
\begin{equation}
\label{q_b_part}
q_{breather}=c+4\Im(A\mathcal{R}-B)=c+4\frac{f\(\mathcal{R}_2 r_1+\mathcal{R}_1 r_2-s_2\)+h(\mathcal{R}_2 s_1+\mathcal{R}_1 s_2+
r_2)}{f^2+h^2},
\end{equation}
where we denote $r=r_1+\i r_2$ with $r_1,r_2\in\R$ and similarly for the other quantities.
%We observe that we can write the above expression in the form
%\begin{equation*}
%q_{breath}=c+2\frac{f_x h-h_x f}{f^2+h^2}=c+2\partial_x\arctan\frac{f}{h}\quad\mbox{ provided that }\quad \frac{f_x-(\mathcal{R}_2s_1+\mathcal{R}_1s_2-r_2)}{f}=
%\frac{h_x+(\mathcal{R}_2r_1+\mathcal{R}_1r_2-s_2)}{h}.
%\end{equation*}
To obtain the expression for $f$  defined  in \eqref{rs} we first observe that
\begin{equation}
\label{H2}
\begin{split}
&H=\frac{\mathcal{R}(\kappa)-\ol{\mathcal{R}(\kappa)}}{\kappa+\ol{\kappa}}=\i\dfrac{\kappa_1-\chi_1}{c\kappa_1}=\i\dfrac{\kappa_1\kappa_2-\chi_1\kappa_2}{c\kappa_1\kappa_2}=
\i\dfrac{\chi_2-\kappa_2}{c\chi_2},
%=\dfrac{\chi-\kappa-\widebar{\chi}+\widebar{\kappa}}{2c\chi_2}
\end{split}
\end{equation}
where we use the identity $\kappa_1\kappa_2=\chi_1\chi_2$,
and
\begin{equation}
\label{G2}
\begin{split}
&G=\frac{1+|\mathcal{R}(\kappa)|^2}{\kappa-\ol\kappa}=\dfrac{c^2+(\kappa_1-\chi_1)^2+(\kappa_2-\chi_2)^2}{2\i\kappa_2c^2}=
\dfrac{\color{black}{2}\kappa_1(\kappa_2^2+\chi_1^2-\kappa_1\chi_1-\kappa_2\chi_2)}{2\i\kappa_1\kappa_2c^2}\\
&\quad=\dfrac{\color{black}{2}\(\kappa_2\chi_2+\chi_1\kappa_1-\kappa_1^2-\chi_2^2\)}{2\i\chi_2c^2}=\dfrac{c^2-|\kappa-\chi|^2}{\color{black}{2}\i\chi_2 c^2},
\end{split}
\end{equation}
where we use the identity $\chi_1^2+\kappa_2^2-c^2=\chi_2^2+\kappa_1^2$  repeatedly.
Plugging into \eqref{rs} the  expressions for $E,F$ from \eqref{E-G},  and $G$ and $H$ from \eqref{H2} and \eqref{G2} respectively,  we obtain
\color{black}{\begin{equation}
\label{def_f}
\begin{split}
& f|\psi_2|^4= f\dfrac{|\kappa+\chi|^2}{4|\chi|^2}={\left|\dfrac{\e^{-2\i g_c(\kappa)}}{\i\nu}-\color{black}\frac{\i c}{2\chi^2}\right|}^2-\dfrac{1}{c^4}\dfrac{|\kappa+\chi|^2}{4|\chi|^2}|\kappa-\chi|^2|+\\
&\quad\quad\quad+\dfrac{1}{c^4\chi_2^2}\left(\frac{1}{4}(|\kappa-\chi|^2-c^2)^2-c^2(\kappa_2-\chi_2)^2\right)\dfrac{|\kappa+\chi|^2}{4|\chi|^2}\\
&=\frac{\e^{4 \Im g_{c}(\kappa))}}{|\nu|^2}-\i\frac{c}{2}\e^{2 \Im g_{c}(\kappa)}\left(\dfrac{1}{-\i\,\bar{\nu}\chi^2}\e^{2 \i \Re g_{c}(\kappa)}-\dfrac{1}{\i\,\nu\widebar\chi^2}\e^{-2 \i \Re g_{c}(\kappa)}
\right)+\dfrac{\chi_1^2(|\chi|^2-c^2)}{4\chi^4\chi_2^2},
\end{split}
\end{equation}}
where we use the relation $(\kappa+\chi)(\kappa-\chi)=-c^2$ and 
\begin{equation*}
\begin{split}
&\Re g_{c}(\kappa)=\chi_1\left(4\left(\chi_1^2-3\chi_2^2-\frac{3}{2}c^2\right)t+x\right),\\
& \Im g_{c}(\kappa)=\chi_2\left(4\left(3\chi_1^2-\chi_2^2-\frac{3}{2}c^2\right)t+x\right).
\end{split}
\end{equation*}
In a similar way,
\color{black}{\begin{equation}
\label{def_h}
\begin{split}
 h|\psi_2|^4&=\e^{2\Im g_{c}(\kappa)}\left( \dfrac{ \e^{-2 \i \Re g_{c}(\kappa)}}{ 2\i\nu\widebar\chi}+\dfrac{ \e^{2 \i \Re g_{c}(\kappa)} }{-2\i\,\widebar{\nu}\chi}
\right)-\dfrac{c \chi_2}{2|\chi|^4}+\dfrac{c}{2|\chi|^2\chi_2}\\
&=\e^{2\Im g_{c}(\kappa)}\left( \dfrac{ \e^{-2 \i \Re g_{c}(\kappa)}}{ 2\i\nu\widebar\chi}+\dfrac{ \e^{2 \i \Re g_{c}(\kappa)} }{-2\i\,\widebar{\nu}\chi}
\right)+\dfrac{2c \chi_1^2}{2|\chi|^4\chi_2}.
\end{split}
\end{equation}}
Long and involved algebraic manipulations  give the identities
\begin{equation*}
( h_x-4\chi_2h)|\psi_2|^4=2(\mathcal{R}_2 r_1+\mathcal{R}_1 r_2-s_2)|\psi_2|^4
\end{equation*}
and 
\begin{equation*}
( f_x-4\chi_2 f)|\psi_2|^4=-2(\mathcal{R}_2 s_1+\mathcal{R}_1 s_2+r_2)|\psi_2|^4.
\end{equation*}
Combining the above two expressions with \eqref{q_b_part}, \eqref{def_f} and \eqref{def_h} we can write the breather solution on a constant background in the form
\begin{align*}
q_{breather}(x,t)&=c+2\dfrac{fh_x-f_x f}{f^2+h^2}=c+2\partial_x\arctan\frac{h}{f},\\
\end{align*}
with $f$ and $h$ as in \eqref{def_f} and \eqref{def_h},  which coincides with the statement of the theorem.
%ere  the phases $\theta_1=\mbox{arccos}\dfrac{\nu_1}{|\nu|}-\mbox{arccos}\dfrac{u}{\sqrt{u^2+v^2}}$ and $\theta_2=\mbox{arccos}\dfrac{u}{\sqrt{u^2+v^2}}$.
\end{proof}
\subsection{Model problem for the periodic travelling wave solution: the elliptic case}
\noindent 
We consider a model problem that can be solved via elliptic functions. This model problem was first solved in \cite{DKMVZ} in the context of asymptotic analysis  for orthogonal polynomials related to Hermitian matrix models. The same problem appeared  in the long-time asymptotic analysis of the MKdV solution 
with step initial data when $c_+=0$  \cite[formula (4.18)]{KM}.
We introduce two real constant parameters $\widetilde{c}>\widetilde{d}>0$. The RH problem is as follows.
\begin{RHP}\label{RHP_elliptic}
To find a $2\times 2$ matrix  $M=M(x,t,;k)$ such that
\begin{enumerate}
 \item $M(x,t;k)$ is analytic in $k\in\mathbb{C}\setminus[\i \widetilde c,-\i \widetilde c]$,
 \item $M_-(k)=M_+(k)J(k)$,
$J(k)=\begin{cases}
    \begin{pmatrix}
     0&\i\\\i&0
    \end{pmatrix},\quad k\in(\i\widetilde  c,\i \widetilde d)\cup(-\i\widetilde  d,-\i \widetilde c),\\\\
    \e^{\i( x U+t V+ \Delta)\sigma_3},\quad k\in(\i \widetilde d,-\i \widetilde d),
   \end{cases}
$

where  $\Delta$ is a real constant and 
\begin{equation}
\label{UW}
U=-\dfrac{\pi\widetilde  c}{K(m)},\quad V={\color{black}-}2(\widetilde c^2+\widetilde d^2)U,\quad m=\dfrac{\widetilde d}{\widetilde c}\,,
\end{equation}
where $K(m)$ is the complete elliptic integral of the first kind;
\item $M(k)\to I$  as $k\to\infty.$
\end{enumerate}
\end{RHP}

It follows from the standard scheme introduced by Zakharov, Shabat \cite{Zakharov-Shabat}, that 
\begin{equation}\label{q_fromRHP}
 q(x,t) = 2\i\lim\limits_{k\to\infty}kM_{21}(x,t;k) = 2\i\lim\limits_{k\to\infty}kM_{12}(x,t;k)
\end{equation}
satisfies MKdV equation \eqref{MKdV}. 
The  explicit formula for $M(k)$ from \cite[pages 24-26]{KM}, \cite{DKMVZ} is   constructed as follows.
We introduce the normalized holomorphic differential
\begin{equation}
\label{omega}
\omega=\dfrac{-\widetilde c}{4\i K(m)}\dfrac{d z}{\sqrt{(z^2+ \widetilde d^2)(z^2+\widetilde c^2)}},\quad 2\int_{\i \tilde{d}}^{-\i \tilde{d}}\omega=1\,,
\end{equation}
with $K(m)$ as in \eqref{UW}. 
%\todo{The homology is not consisten}
For our purpose we fix the function $\sqrt{(z^2+ \widetilde d^2)(z^2+\widetilde c^2)}$ by   the condition that it is analytic off the intervals $(\i \widetilde c,\i \widetilde{d})\cup(-\i \widetilde d,-\i \widetilde c)$ and positive at $z=0$. The intervals $(\i \widetilde c,\i \widetilde{d})$ and $(-\i \widetilde d,-\i \widetilde c)$  are oriented downwards.
We define the  $\beta$-cycle the 
counterclockwise loop encircling $(\i \widetilde c,\i \widetilde d)$ and 
the $\alpha$-cycle the path starting on the cut $(\i \widetilde c,\i \widetilde d)$ on the left, going to the cut $(-\i\widetilde  d,-\i \widetilde c)$  on the left and passing to the second sheet and reaching the cut $(\i \widetilde c,\i \widetilde d)$ from the right on the second sheet.
%\todo[inline]
%{A: The choice of $\alpha,\beta$ cycles is not standard.}
We define the quantity
\begin{equation*}
\tau=\int_{\beta}\omega.
\end{equation*}
It follows that
\begin{equation}
\label{tau_ell}
 \tau=\dfrac{\i }{2}\dfrac{K'(m)}{K(m)},\quad m=\dfrac{\widetilde d}{\widetilde c},
\end{equation}
where  $K'(m)=K(\sqrt{1-m^2})$ and $K(m)=\int\limits_0^{\frac{\pi}{2}}\frac{ds}{\sqrt{1-m^2\sin^2s}}$.
Using  the relations \cite[165.05, 162.01]{BF},  the quantity $\tau$ can also be written in the form
 \begin{equation}\label{K_m_mtilde}
\tau=\i\dfrac{ K'(\widetilde m)}{ K(\widetilde m)},\quad  K(\widetilde m)=(1+m)K(m),\quad \widetilde{m}^2=\dfrac{4m}{(1+m)^2}.
\end{equation}
 Let us introduce the Jacobi $theta$-function with modulus $\tau$
\begin{equation*}
\theta(\zeta)\equiv\theta(\zeta,\tau)=\sum\limits_{m=-\infty}^{\infty}\exp\left\{\pi \i \tau m^2+2\pi \i \zeta m\right\}.
\end{equation*}
It is an even function of $\zeta$ and it has the following periodicity properties
\begin{equation*}
\theta(\zeta+1)=\theta(\zeta),\quad \theta(\zeta+\tau)=\e^{-\i\pi\tau-2\pi \i  \zeta}\theta(\zeta).
\end{equation*} 
Next we introduce the Abel map with base point $\i \widetilde c$
\begin{equation}
\label{Abel}
A(k):=\int_{\i \widetilde c}^{k}\omega,
\end{equation}
where $\omega$ is the holomorphic differential \eqref{omega}
and we observe that $A(\i \infty)=\frac{1}{4}$.
Let us also introduce  the  function
 $\gamma(k)=\sqrt[4]{\frac{k-\i \widetilde c}{k-\i \widetilde d}}\sqrt[4]{\frac{k+\i \widetilde d}{k+\i \widetilde c}}$, \color{black}{which is analytic off the intervals $[\i  \widetilde c,\i  \widetilde d]\cup[-\i \widetilde  d, -\i  \widetilde c],$ and tends to $1$ at infinity.}

Then the solution for the RH  problem~\ref{RHP_elliptic}  is given by \cite{KM}

\begin{equation}
\label{M_sol}
\begin{split}
&M(k)=\dfrac{\theta(0)}{2\theta(\Omega)}\times\\
&\times
\begin{pmatrix}
(\gamma(k)+\gamma^{-1}(k))\dfrac{\theta(A(k)-\Omega-\frac{1}{4})}{\theta(A(k)-\frac{1}{4})}
&
(\gamma(k)-\gamma^{-1}(k))\dfrac{\theta(-A(k)-\frac{1}{4}-\Omega)}{\theta(-A(k)-\frac{1}{4})}
\\\\
(\gamma(k)-\gamma^{-1}(k))\dfrac{\theta(A(k)+\frac{1}{4}-\Omega)}{\theta(A(k)+\frac{1}{4})}
&
(\gamma(k)+\gamma^{-1}(k))\dfrac{\theta(-A(k)+\frac{1}{4}-\Omega)}{\theta(-A(k)+\frac{1}{4})}
\end{pmatrix},
\end{split}
\end{equation}
where 
\begin{equation*}
\Omega=\dfrac{ xU+tV+\Delta}{2\pi}.
\end{equation*}
Using the expression for $M(k)$ above, it  follows from \eqref{q_fromRHP} that  the MKdV solution $q(x,t)$ is given by
\begin{equation}\label{MKdV solution theta}
 q(x,t) = (\widetilde c-\widetilde d)\ \frac{\theta(\Omega+\frac{1}{2}; \tau)}{\theta( \Omega; \tau)}\dfrac{\theta(0;\tau)}{\theta(\frac{1}{2};\tau)}.
\end{equation}
We recall  the relation between the Jacobi $\theta$-function and the elliptic function $\mathrm{dn}$, \cite{Lawden},
namely
\begin{equation}
\label{dn}
\mathrm{dn}(2 K(\widetilde m)z|\widetilde m)=\dfrac{\theta(\frac{1}{2};\tau)}{\theta(0;\tau)}\dfrac{\theta(z;\tau)}{\theta(z+\frac{1}{2};\tau)}\quad 
\end{equation}
with $ \widetilde m$ as in \eqref{K_m_mtilde} and $\tau$ as in \eqref{tau_ell}.

Using the  above identities, the solution \eqref{MKdV solution theta} can be written in the form
\begin{equation}\label{MKdV solution dn theta}
 q(x,t) = (\widetilde c+\widetilde d)\mathrm{dn}\(2K(\widetilde m)(\Omega+\frac{1}{2}) \big| \ \widetilde m\).
\end{equation}
%Taking into account expression for $U, V$ (\ref{UV bilinear relations}), we come to formula %(\ref{eqn_eta_6}) with specified phase
%\begin{equation}\label{eqn_eta_7}\eta=(c+d)\mathrm{dn}\((c+d)(x-2(c^2+d^2)t)+K(\widetilde m)\ |\ \frac{4cd}{(c+d)^2}\).\end{equation}
% 
\noindent
According to \cite{BF}, formulas (162.01), p.38, and (165.05), p.41,
\begin{equation}\label{identity_sn2_dn_1}
\mathrm{dn}\(u(1+m)| \widetilde m\)
=\frac{1-m\ \mathrm{sn}^2(u|m)}{1+m\ \mathrm{sn}^2(u|m)},\quad K(\widetilde m)=K(m)(1+m),
\end{equation}
with $m$  as in \eqref{tau_ell}.
Hence,  the expression  \eqref{MKdV solution dn theta} can be rewritten as 
%\begin{equation}\label{MKdV solution dn theta}
% q(x,t) = \sqrt{c^2-d^2}\mathrm{dn}\(K(\widetilde m)\ \frac{x U+t V+ \pi}{\pi}\ | \ \widetilde m\),
%\end{equation}
%

\begin{equation}\label{eqn_eta_4} 
\begin{split}
q(x,t)&=(\widetilde c+\widetilde d)\frac{1-\frac{\widetilde d}{\widetilde c}\ \mathrm{sn}^2\(2K( m) \Omega+K(m)|m\)}{1+\frac{\widetilde d}{\widetilde c}\ \mathrm{sn}^2\(2K( m)\Omega+K(m)|m\)}\\
&=-\widetilde c-\widetilde d+\dfrac{2\widetilde c(\widetilde c+\widetilde d)}{\widetilde c+\widetilde d-\widetilde d\, \mathrm{cn}^2\(\widetilde c(x-2(\widetilde c^2+\widetilde d^2)t)-\dfrac{\Delta}{\pi}K(m)+K(m) |m\)},
\end{split}
\end{equation}
where to obtain the second expression we use the relation $ \mathrm{cn}^2(u|m)+ \mathrm{sn}^2(u|m)=1$  and the explicit expressions of $U$ and $V$ as in \eqref{UW}.
The second expression in \eqref{eqn_eta_4}   coincides with the travelling wave solution \eqref{periodic_intro} identifying $\beta_3=\widetilde c$, $\beta_2=\widetilde d$ and $\beta_1=0$
and $x_0=-\dfrac{\Delta}{\pi}K(m)+K(m)$ (see also Appendix~\ref{Appendix_travelling}).

\subsection{Model problem for the periodic  travelling wave solution: the hyperelliptic case}\label{sect_hyperModel}
The model problem we are considering below   is obtained from the longtime asymptotic analysis of the MKdV RH problem in the oscillatory region when the step $c_+>0$.
It can be solved using hyperelliptic theta-functions.  The goal of this section is to show that such model problem still gives the periodic travelling 
wave solution  \eqref{periodic_intro} of the MKdV equation.  To reach our goal we introduce a conformal transformation  of the complex plane 
 and an auxiliary RH problem that is  going to reduce the hyperelliptic RH to an elliptic RH problem.  
\begin{RHP}\label{RH_problem_W}
 Find a $2\times 2$ matrix-valued  function $W(k)$ analytic in $\C\backslash [\i c_-,-\i c_-]$ such that
 \begin{enumerate} 
 
\item $W_-(k)=W_+(k)J_W(k),\quad k\in (\i c_-,-\i c_-)$
with 
\begin{equation}
\label{RH_W}
\begin{split}
&J_W(x,t;k)=\begin{pmatrix}0&\i\\\i&0\end{pmatrix}
\,\quad k\in(\i c_{-},\i d)\cup (\i c_+,-\i c_+)\cup(-\i d, -\i c_-)\\
&J_W(x,t;k)=\begin{pmatrix} \e^{\i (xU+t V + \Delta)}  & 0 \\ 0 &  \e^{-\i (xU+t V + \Delta)}  \end{pmatrix},\quad k\in(\i d,\i c_+)\cup (-\i c_+,-\i d);
\end{split}
\end{equation}
where $c_->d>c_+$,   and  
\begin{equation}
\label{UVh}
U=-\pi \dfrac{\sqrt{c_{\l}^2-c_{\r}^2}}{K(m)},\qquad V=-2(c_{\l}^2+c_{\r}^2+d^2)U,
\end{equation}
with  $m^2=\dfrac{d^2-c_+^2}{c_-^2-c_+^2}$   and $\Delta$  real constant.
\item $W(k)= I+O\left(\frac{1}{k}\right)\quad \mbox{as}\;k\to\infty$;
\item $W(k)$  has at most fourth root singularities at the  points $\pm \i c_-$, $\pm \i c_+$ and $\pm \i d$.
\end{enumerate}
\end{RHP}
Then the quantity
\begin{equation}
\label{q_W}
q_{hel}(x,t)=\lim_{k\to\infty}2ikW_{12}(k;x,t)=\lim_{k\to\infty}2ikW_{21}(k;x,t)
\end{equation}
is a solution of the  MKdV  equation.
\begin{theorem}
\label{theorem_periodic}
The   solution of the MKdV equation \eqref{q_W}  obtained from the RH problem  \ref{RH_problem_W}   is the travelling wave solution \eqref{periodic_intro},  namely 
\begin{equation}
\begin{split}
\label{q_periodic1}
q_{hel}(x,t)&=q_{per}(x,t, c_+,d,c_-,x_0)\\
&=-c_--d-c_++2\frac{(c_-+d)( c_++ c_-)}{d+c_--(d-c_+)\mathrm{cn}^2\(\sqrt{c_-^2-c_+^2} (x-{\mathcal V}t)+x_0|m\)},
\end{split}
\end{equation}
with ${\mathcal V}=2(c_-^2+c_+^2+d^2)$,
and the phase $x_0$ takes the form
\begin{equation}
\label{phase_x0}
x_0=-\dfrac{K(m)\Delta}{\pi}+K(m).
\end{equation}
Here   $\mathrm{cn}(u |m)$ is the Jacobi elliptic function of modulus $m^2=\dfrac{d^2-c_+^2}{c_-^2-c_+^2}$ and $K(m)$ is the complete elliptic integral of the first kind of modulus $m$.

\end{theorem}
%Namely for $\xi\in[ \frac{-c_{-}^2}{2}+c_{+}^2+\delta_1, \frac{c_{-}^2}{3}+\frac{c_{+}^2}{6}-\delta_1]$, with $\delta_1>0$ a small constant
%\begin{equation*}
%q(x,t)=q_{ell}(x,t)+O(t^{-\alpha_0}),\quad \alpha_0>0.
%\end{equation*}
The RH  problem \ref{RH_problem_W}    has been considered in \cite{KM2} where it was solved in terms of hyperelliptic theta-function defined on the Jacobi variety of the surface 
$\Gamma:=\{(k,y)\in\C^2\;|\; y^2=(k^2+d^2)(k^2+c_-^2)(k^2+c_+^2)\}$. Such a surface has two automorphisms $\tau_1:(y,k)\to (y,-k)$  and $\tau_2:(y,k)\to (-y,k)$.
Therefore the curve $\Gamma$ covers two elliptic curves $\Gamma_+:=\Gamma/\tau_1$ and $\Gamma_-=\Gamma/(\tau_1\tau_2)$.
The corresponding genus 2 theta-function can be factorized as a product of Jacobi theta-function. However, pursuing this strategy, we did not see a simple way at arrive 
to the travelling wave solution \eqref{q_periodic1}.
For this reason, we change  our  strategy and we formulate an auxiliary RH-problem that produces the desired solution and we connect such a  problem to our RH problem \ref{RH_problem_W}.
\subsubsection{Auxiliary RH  problem}

We consider the  two real numbers $\widetilde c>\widetilde d>0$  with $\widetilde c^2=c_-^2-c_+^2$ and $\widetilde d^2=d^2-c_+^2$ and construct the following RH  problem for a $2\times 2$ matrix $M_{el}=M_{el}(\lambda)$: 
\begin{enumerate}
 \item $M_{el}(\lambda)$ is analytic in $\lambda\in\mathbb{C}\setminus[\i \widetilde c,-\i \widetilde c],$
 \item $M_{el,-}(\lambda)=M_{el,+}(\lambda)J_{el}(\lambda),$ where
$$J_{el}(\lambda)=\begin{pmatrix}
          0&\i\\\i&0
         \end{pmatrix},\quad \lambda\in(\i\widetilde c,\i\widetilde d)\cup(-\i\widetilde d,-\i\widetilde c),
$$
$$J_{el}(\lambda)=\e^{(\i xU+\i tV +\i\Delta-\Delta_4)\sigma_3},\quad \lambda\in(\i \widetilde d,-\i\widetilde d),
$$
where $\Delta_4$ is a constant to be determined and $U$, $V$ and $\Delta$, as in the RH  problem  \ref{RH_problem_W}.
\item $M_{el}(\lambda)= I+O\left(\frac{1}{\lambda}\right)$ as $\lambda \to\infty.$
\end{enumerate}
The  explicit formula for $M_{el}(\lambda)$   can be obtained from the solution of the RH problem \ref{RHP_elliptic} as in \eqref{M_sol}  with $\Omega= \frac{1}{2\pi}(xU+ tV +\Delta+\i \Delta_4).$%
\noindent 
%The solution of this RH problem possesses the symmetry

%$$M_{el}(\lambda\pm\i 0)=\begin{pmatrix}
%                      0 & -1 \\ 1 & 0
%                     \end{pmatrix}M_{el}(-\lambda\pm\i0)
%\begin{pmatrix}
%                      0 & 1 \\ -1 & 0
%                     \end{pmatrix}\,,
%$$
%for $\lambda\notin[\i \widetilde c,-\i \widetilde c]$.

To proceed further, we make a transformation of the complex plane to reduce  the  RH  problem  for $M_{el}(\lambda)$  to the one in RH problem \ref{RH_problem_W}. 
We introduce a change of variable $\lambda\to k$ defined as 
$$\lambda=\sqrt{k^2+c_{\r}^2},\qquad \textrm { and denote }\quad
c_{\l}=\sqrt{\widetilde c^2+c_{\r}^2},\quad d=\sqrt{\widetilde
d^2+c_{\r}^2}.$$  
The function $\lambda=\lambda(k)$ is analytic for $k\in \mathbb{C}\backslash [\i c_+,-\i c_+]$.
Next we introduce the matrix 
$$\Lambda(k)=\dfrac{1}{2}\begin{pmatrix}
                     a(k)+\dfrac{1}{a(k)} & -\i (a(k)-\dfrac{1}{a(k)})\\ \i( a(k)-\dfrac{1}{a(k)} )& a(k)+\dfrac{1}{a(k)}
                    \end{pmatrix} M_{el}(\lambda(k)),
$$ where
$$\qquad a(k)=\sqrt[4]{\frac{k^2+c_{\r}^2}{k^2}}.
$$
Then the matrix $\Lambda(k)$ is analytic for $k\in \mathbb{C}\backslash [\i c_-,-\i c_- ]$ and satisfies the following conditions:
  $$\quad\Lambda_-(k)=\Lambda_+(k)J_{\Lambda}(k),\quad k\in[\i c_{\l},-\i c_{\l}]$$
  with 
$$J_{\Lambda}(k)=\left\{
\begin{array}{lll}
\begin{pmatrix}
   0 & 1\\-1 & 0
  \end{pmatrix}, &k\in(\i c_{\r},0),&\\
  \begin{pmatrix}
   0 & -1\\1 & 0
  \end{pmatrix},&k\in(0,-\i c_{\r}),&\\
\begin{pmatrix}
   0 & \i\\\i & 0
  \end{pmatrix}, &k\in(\i c_{\l},\i d)\cup(-\i d,-\i c_{\l}),&\\
\e^{\i (xU+tV+\Delta+\i\Delta_4)\sigma_3}, &
k\in(\i d,\i c_{\r})\cup(-\i c_{\r},-\i d).&
\end{array}
\right.
$$
This is not exactly the RH problem \ref{RH_problem_W}  with jumps as in  \ref{RH_W}.  We need to do some extra work.
For the purpose  we introduce  a scalar function $F=F(k)$ analytic in $k\in\mathbb{C}\setminus[\i c_{\l},-\i c_{\l}],$ which satisfies the following conditions:
$$F_+(k)F_-(k)=1,\quad k\in(\i c_{\l},\i d)\cup (-\i d, -\i c_{\l}),$$
$$F_+(k)F_-(k)=\i,\quad k\in(\i c_{\r},0),\qquad F_+(k)F_-(k)=-\i,\quad k\in(0,-\i c_{\r}),$$
$$\frac{F_+(k)}{F_-(k)}=\e^{\Delta_4},\quad k\in(\i d,\i c_{\r})\cup(-\i c_{\r},-\i d).$$
The quantity $\Delta_4$   is independent from $k$ and it has to be chosen in such a way that  $F(k)$ is  bounded as $k\to\infty.$
The function $F(k)$ can be represented in the following way:
\begin{equation*}
\begin{split}
F(k)=\exp&\left\{\frac{R(k)}{2\pi\i}\(\int\limits_{\i c_{\r}}^0\frac{\frac{\pi\i}{2}\d s}{(s-k)R_+(s)}
- \int\limits^{-\i c_{\r}}_0\frac{\frac{\pi\i}{2}\d s}{(s-k)R_+(s)} +\right.\right.\\
&+\left. \left.\int\limits_{\i d}^{\i c_{\r}}\frac{\Delta_4\d s}{(s-k)R(s)} + \int\limits^{-\i d}_{-\i c_{\r}}\frac{\Delta_4\d s}{(s-k)R(s)} \) \right\},
\end{split}
\end{equation*}
where $R(k)=\sqrt{(k^2+c_+^2)(k^2+c_-^2)(k^2+d^2)}$. The function $R(k)$ is analytic for $k\in \C\backslash \{[\i c_-,\i d]\cup[\i c_+,-\i c_+]\cup[-\i d,-\i c_-]\}$ and positive and real at $k=+0$.  
The function $F(k)$ is bounded at infinity provided that 
 \begin{equation}
 \label{Delta4}
 \Delta_4 =
\frac{-\pi\i}{2}\frac{\int\limits_{\i c_{\r}}^0\frac{s\d
s}{R_+(s)}} {\int\limits^{\i c_{\r}}_{\i d}\frac{s\d
s}{R(s)}}=-2\pi\i \int\limits^{c_+ }_0\omega,
\end{equation}
where the one-form $\omega$ has been defined in \eqref{omega}.
Hence $$F(k)=1+\mathcal{O}(k^{-1})\quad
\textrm{ as }\quad k\to\infty.$$

% Let us mention that in the
%corresponding Riemann surface associated with
%$\sqrt{(k^2+\widetilde c^2)(k^2+\widetilde d^2)}$ the Abel map at
%$k=c_{\r}$ is equal to
%$$A(c_{\r})=\frac{-\tau}{2}+\frac{\pi\i}{2}-\Delta_4.$$
%
%
\noindent
 Denote $\Lambda^{(1)}(k)=\Lambda(k)
F^{-\sigma_3}(k).$  The  jump conditions for the matrix  $\Lambda^{(1)}$ are as follows:  
 $\Lambda^{(1)}_-(k)=\Lambda^{(1)}_+(k)J^{(1)}(k)$  with 
$$J^{(1)}(k)=\left\{
\begin{array}{lll}
\begin{pmatrix}
           0 & \i \\\i & 0
          \end{pmatrix},& k\in(\i c_{\l},\i d)\cup(\i c_{\r},-\i c_{\r})\cup(-\i d, -\i c_{\l}),&\\
\e^{\i (xU+tV+\Delta)\sigma_3},& k\in(\i d,\i c_{\r})\cup(-\i c_{\r},-\i d).&
\end{array}
\right.$$

The matrix $\Lambda^{(1)}(k)$ is not exactly  the solution of the hyperelliptic model problem \ref{RH_W}, since it has poles at the point  $k=0.$
This is because the function $F(k)$ is vanishing as $\sqrt{k}$ as $k\to0$ with $\Re k>0,$ and is growing as $\frac{1}{\sqrt{k}}$ as $k\to0$ with $\Re k<0.$
Hence, the second column of $\Lambda^{(1)}$ has a pole of the  first order when $k\to 0$ with $\Re k<0,$ and the first column of $\Lambda^{(1)}$ has a pole when $k\to0$ with $\Re k>0.$

Direct analysis of  the  behavior of $\Lambda^{(1)}(k) $ at $k\to0$ shows
that the matrix  function
\begin{equation}
\label{Lambda2}
\Lambda^{(2)}(k):=\begin{pmatrix}1+\frac{\alpha}{k} & \frac{\i\alpha}{k}\\\frac{\i\alpha}{k} & 1-\frac{\alpha}{k}\end{pmatrix}\Lambda^{(1)}(k)
\end{equation}
does not have pole at $k=0$ provided that
\begin{equation}
\label{alpha}
\alpha=\frac{-c_{\r}}{2}\frac{M_{el,11}(c_{\r})-\i M_{el,21}(c_{\r})}{{M_{el,11}(c_{\r})+\i M_{el,21}(c_{\r})}},
\end{equation}
 where $M_{el}(k)$ is as in \eqref{M_sol} by taking care of replacing $\Delta$ by $\Delta+\i \Delta_4$. 
We arrive at  the following lemma.
\begin{lemma}
The $2\times2 $  matrix  $\Lambda^{(2)}(k)$  defined in \eqref{Lambda2}  with $\alpha$ as in \eqref{alpha}  is the unique solution to the
hyperelliptic model problem in \eqref{RH_W}.
\end{lemma}
\noindent
Further, the solution of the MKdV equation is given by the
 formula:
$$q_{hel}(x,t)=\lim\limits_{k\to\infty}2\i k\Lambda^{(2)}_{21}(k)=2\i\lim\limits_{k\to\infty}kM_{el21}(k)-2\alpha=
2\i\lim\limits_{k\to\infty}kM_{el12}(k)-2\alpha,$$
so that, plugging into the above expression the explicit expression of $M_{ell}(k)$  and $\alpha$ we obtain 
\begin{equation}
\label{qhyper0}
q_{hel}(x,t)=(\widetilde
c-\widetilde d)\frac{
\theta(\frac{1}{2}+\Omega-\frac{\Delta_4}{2\pi \i};\tau)}{\theta(\Omega-\frac{\Delta_4}{2\pi \i};\tau)}\dfrac{\theta(0;\tau)}{\theta(\frac{1}{2};\tau)}+c_{\r}\frac{M_{el,11}(c_{\r})-\i M_{el,21}(c_{\r})}{{M_{el,11}(c_{\r})+\i M_{el,21}(c_{\r})}},
\end{equation}
with $\tau=\dfrac{\i }{2}\dfrac{K'(m)}{K(m)},$ $m=\dfrac{\widetilde d}{\widetilde c}$, and
\begin{equation}
\label{Omega_ell}
\Omega=\dfrac{xU+tV+\Delta}{2\pi},
\end{equation}
 with $U$, $V$ and $\Delta$ as in the  RH problem with jumps as in \eqref{RH_W}.
Summarizing we have obtained the solution of the hyperelliptic RH problem  with jumps as in \eqref{RH_W}   and therefore of the MKdV equation
in terms of elliptic functions.  We need to do some extra work to show that the expression \eqref{qhyper0}  coincides with the travelling wave solution of the MKdV equation.
\begin{proof} [Proof  of  Theorem~\ref{theorem_periodic}]
%\begin{theorem}
In order to prove Theorem~\ref{theorem_periodic}, we need to show that the quantity $q_{hel}(x,t)$  in \eqref{qhyper0} is equal to   the  travelling wave solution of the MKdV equation defined in \eqref{periodic_intro}, namely we have to prove the relation
\begin{equation}
\label{qhyper1}
q_{hel}(x,t)=-c_--d-c_++2\frac{(c_-+d)( c_++ c_-)}{d+c_--(d-c_+)\mathrm{cn}^2\(\sqrt{c_-^2-c_+^2} (x-{\mathcal V}t)+x_0|m\)},
\end{equation}
where the phase $x_0$ takes the form
\begin{equation}
\label{phase_x01}
x_0=\dfrac{K(m)\Delta}{\pi}+K(m),\quad m^2=\dfrac{d^2-c_+^2}{c_-^2-c_+^2}.
\end{equation}
%\end{theorem}
%\begin{proof}
For the purpose we need a series of identities among elliptic functions.
We first consider the term $\alpha$ in \eqref{alpha}. We  observe   from the relation \eqref{omega} and \eqref{Delta4} that 
$$A(c_{\r})=\int_{i\widetilde{c}}^{c_{\r}}\omega=-\frac{\tau}{2}+\frac{1}{4}-\dfrac{\Delta_4}{2\pi i }$$
so that the quantity $\alpha$ in \eqref{alpha} takes the form
 \begin{equation}
 \label{alpha0}
 \begin{split}
 \alpha&=\frac{-c_{\r}}{2}\frac{(\gamma(c_{\r})+\gamma^{-1}(c_{\r}))
 \frac{\theta(\frac{-\tau}{2}-\Omega;\tau)}{\theta(\frac{-\tau}{2}-\frac{\Delta_4}{2\pi \i};\tau)}-
 \i(\gamma(c_{\r})-\gamma^{-1}(c_{\r}))\frac{\theta(\frac{-\tau}{2}+\frac{1}{2}-\Omega;\tau)}{\theta(\frac{-\tau}{2}+\frac{1}{2}-\frac{\Delta_4}{2\pi \i};\tau)}}
 {(\gamma(c_{\r})+\gamma^{-1}(c_{\r}))\frac{\theta(\frac{-\tau}{2}\Omega;\tau)}{\theta(\frac{-\tau}{2}-\frac{\Delta_4}{2\pi \i};\tau)}+
 \i (\gamma(c_{\r})-\gamma^{-1}(c_{\r}))\frac{\theta(\frac{-\tau}{2}+\frac{1}{2}-\Omega;\tau)}{\theta(\frac{-\tau}{2}+\frac{1}{2}-\frac{\Delta_4}{2\pi \i};\tau)}}\\
&=-\frac{c_{\r}}{2}\frac{  (\gamma(c_{\r})+\gamma^{-1}(c_{\r}))\frac{\theta(\frac{-\tau}{2}-\Omega;\tau)}{\theta(\frac{-\tau}{2}+\frac{1}{2}-\Omega;\tau)}
- \i (\gamma(c_{\r})-\gamma^{-1}(c_{\r}))\ \frac{\theta(\frac{-\tau}{2}-\frac{\Delta_4}{2\pi \i};\tau)}{\theta(\frac{-\tau}{2}+\frac{1}{2}-\frac{\Delta_4}{2\pi \i};\tau)}}
 {(\gamma(c_{\r})+\gamma^{-1}(c_{\r}))\ \frac{\theta(\frac{-\tau}{2}-\Omega;\tau)}{\theta(\frac{-\tau}{2}+\frac{1}{2}-\Omega;\tau)}+
 \i (\gamma(c_{\r})-\gamma^{-1}(c_{\r}))\ \frac{\theta(\frac{-\tau}{2}-\frac{\Delta_4}{2\pi \i};\tau)}{\theta(\frac{-\tau}{2}+\frac{1}{2}-\frac{\Delta_4}{2\pi \i};\tau)}}.
 \end{split}
 \end{equation}
 In order to simplify the above expression we use the  identities  \cite[page 25, unnumbered formula before (4.34)]{KM}
% $$\frac{\theta(A(k)+\pi\i)}{\theta(A(k))}=\sqrt{\frac{k-\i \widetilde d}{k+\i \widetilde d}},$$
$$\frac{\theta(A(k)-\frac{1}{4};\tau)}
{\theta(A(k)+\frac{1}{4};\tau)}
=
\sqrt{\frac{\widetilde c+\widetilde d}{\widetilde c-\widetilde
d}}\cdot \frac{\gamma(k)+\gamma^{-1}(k)}{\eta(k)+\eta(k)^{-1}},
%=\sqrt{\frac{\widetilde c-\widetilde d}{\widetilde c+\widetilde d}}\cdot \frac{\lambda-\lambda^{-1}}{\gamma-\gamma^{-1}},
$$ 
where
$\eta(k)=\sqrt[4]{\frac{k-\i \widetilde c}{k+\i\widetilde
c}}\sqrt[4]{\frac{k-\i \widetilde d}{k+\i\widetilde d}}$ and  use  the following periodicity property of elliptic functions:
\begin{equation*}
\mathrm{dn}(u+iK'(\widetilde{m})\,|\widetilde m)=-\i\dfrac{\mathrm{cn}(u|\widetilde m)}{\mathrm{sn}(u|\widetilde m)},\quad K'=K(\sqrt{1-\widetilde m^2}),
\end{equation*}
\begin{equation*}
\mathrm{dn}^2(u\,|\widetilde m)=1-\widetilde{m}^2\mathrm{sn}^2(u\,|\widetilde m),\quad \mathrm{sn}^2(u\,|\widetilde m)+\mathrm{cn}^2(u\,|\widetilde m)=1,
\end{equation*}
where $\widetilde{m}$ is defined in \eqref{K_m_mtilde}.
Using the above  three identities  and \eqref{dn} we arrive at the following form for $q_{hel}(x,t)$ in \eqref{qhyper0}:
\begin{equation}
\label{q_par}
\begin{split}
q_{hel}(x,t)&=(\widetilde d+\widetilde c) \mathrm{dn}\(2 K(\widetilde m)(\frac{1}{2}+\Omega-\frac{\Delta_4}{2\pi \i})|\widetilde m\)+\\
&+c_{\r}\dfrac{\i\sqrt{\dfrac{1-m}{1+m}}  \dfrac{\mathrm{sn}(2 K(\widetilde m)\Omega+K(\widetilde m)|\widetilde m)}{\mathrm{cn}(2 K(\widetilde m)\Omega+K(\widetilde m)|\widetilde m)} -\sqrt{\dfrac{c_--d}{c_-+d}}}{ 
\i\sqrt{\dfrac{1-m}{1+m}}  \dfrac{\mathrm{sn}(2 K(\widetilde m)\Omega+K(\widetilde m)|\widetilde m)}{\mathrm{cn}(2 K(\widetilde m)\Omega+K(\widetilde m)|\widetilde m)} +\sqrt{\dfrac{c_--d}{c_-+d}}}\,.
\end{split}
\end{equation}
Next we use the addition formula for Jacobi elliptic function
 $\mathrm{dn}$ \cite[123.01, p.23]{BF}

\begin{equation}\label{sncndn_summation}\begin{split}
&\mathrm{dn}(u+v|\widetilde m)=\frac{\mathrm{dn}(u|\widetilde m) \mathrm{dn}(v|\widetilde m)-\widetilde m^2\ \mathrm{sn}(u|\widetilde m)\mathrm{cn}(u|\widetilde m)\mathrm{sn}(v|\widetilde m)\mathrm{cn}(v|\widetilde m)}
{1-\widetilde m^2\ \mathrm{sn}^2(u|\widetilde m)\mathrm{sn}^2(v|\widetilde m)},
%\\
%&\mathrm{sn(u+v|\widetilde m)}=\frac{\mathrm{sn}(u)\mathrm{cn}(v)\mathrm{dn}(v)+\mathrm{sn}(v)\mathrm{cn}(u)\mathrm{dn}(u)}
%{1-\widetilde m^2\ \mathrm{sn}^2(u)\mathrm{sn}^2(v)},
%\\
%&\mathrm{cn}(u+v|\widetilde m)=\frac{\mathrm{cn}(u|\widetilde m) \mathrm{cn}(v|\widetilde m)-\mathrm{sn}(u|\widetilde m)\mathrm{dn}(u|\widetilde m)\mathrm{sn}(v|\widetilde m)\mathrm{dn}(v|\widetilde m)}
%{1-\widetilde m^2\ \mathrm{sn}^2(u|\widetilde m)\mathrm{sn}^2(v|\widetilde m)}.
\end{split}\end{equation}
and the following relations  \cite[162.01, 165.05]{BF}
\begin{equation}\label{sncndn_transform}
\begin{cases}
\mathrm{sn}(u|\widetilde m) = (1+m)\frac{\mathrm{sn}\(\frac{u}{1+m}|m\)}{1+m\ \mathrm{sn}^2\(\frac{u}{1+m}|m\)},\\
 \mathrm{cn}(u|\widetilde m) = \frac{\mathrm{cn}\(\frac{u}{1+m}|m\)\ \mathrm{dn}\(\frac{u}{1+m}|m\)}{1+m\ \mathrm{sn}^2\(\frac{u}{1+m}|m\)},\\
 \mathrm{dn}(u|\widetilde m) = \frac{1-m\ \mathrm{sn}^2\(\frac{u}{1+m}|m\)}{1+m\ \mathrm{sn}^2\(\frac{u}{1+m}|m\)},
 \end{cases}\widetilde m^2=\frac{4m}{(1+m)^2},\;\; K(\widetilde m)=K(m)(1+m).
\end{equation}
 We also obtain the  following relations 
\begin{equation}\label{prom_ident_2}
\begin{cases}
&\mathrm{sn}(\frac{\i K[\widetilde m]\Delta_4}{\pi}|\widetilde m)=\frac{\i\ (\widetilde c+\widetilde d)\sqrt{(d-\widetilde d)(c_{\l}+\widetilde c)}}{(c_{\l}+\widetilde c-d+\widetilde d)\sqrt{\widetilde c\widetilde d\ }},\\
&\mathrm{cn}(\frac{\i K[\widetilde m]\Delta_4}{\pi}|\widetilde m)=
\frac{(c_{\l}\widetilde d+d \widetilde c)\sqrt{c_{\l}+\widetilde c}}{\sqrt{\widetilde c\widetilde d(d+\widetilde d)\ }\ (c_{\l}+\widetilde c-d+\widetilde d)},
\\
&
\mathrm{dn}(\frac{\i K[\widetilde m]\Delta_4}{\pi}|\widetilde m)=
\frac{c_{\l}+\widetilde c+d-\widetilde d}{c_{\l}+\widetilde c-d+\widetilde d}=\frac{\widetilde c-\widetilde d}{c_l-d}=\frac{c_l+d}{\widetilde c+\widetilde d},
\end{cases}
\end{equation}
and the further identity
\begin{equation*}
\frac{2(c_l+\widetilde c)(c_l\widetilde d+d\widetilde c)}{(c_l+\widetilde c-d+\widetilde d)(c_l+\widetilde c+d-
\widetilde d)(d+\widetilde d)}=1.
\end{equation*}
Substituting the above relations in \eqref{q_par}    and defining  $\widehat{\Omega}=2K(m)\Omega+K(m)$, we obtain
\begin{equation}
\begin{split}
q_{hel}(x,t)&=\dfrac{1-m^2 \mathrm{sn}^4(\widehat{\Omega}| m)+\dfrac{c_+}{\widetilde c^2}(d+c_-) \mathrm{sn}^2(\widehat{\Omega}| m)-
\frac{c_+ \mathrm{cn}^2(\widehat{\Omega}| m) \mathrm{dn}^2(\widehat{\Omega}| m)}{d+c_-}}
{\dfrac{(d+c_-) }{\widetilde c^2}\mathrm{sn}^2(\widehat{\Omega}| m)+\frac{\mathrm{cn}^2(\widehat{\Omega}| m) \mathrm{dn}^2(\widehat{\Omega}| m)}{d+c_-}}\\
&=-c_++(d+c_-)\dfrac{1-\frac{d-c_+}{c_-+c_+} \mathrm{sn}^2(\widehat{\Omega}| m)}{1+\frac{d-c_+}{c_-+c_+} \mathrm{sn}^2(\widehat{\Omega}| m)}\\
&=-c_+-c_--d+2\dfrac{(c_-+d)(c_++c_-)}{c_-+d-(d-c_+)\mathrm{cn}^2(\widehat{\Omega}| m)},
\end{split}
\end{equation}
where 
\begin{equation*}
\widehat{\Omega}=\sqrt{c_-^2-c_+^2}(x-2(c_-^2+c_+^2+d^2)t)+x_0,\;\;\;
x_0=-\dfrac{K(m)\Delta}{\pi}+K(m).
\end{equation*}
which concludes the proof of Theorem~\ref{theorem_periodic}.
\end{proof}

%%
%% Newsubsection 
%%
%\input{BreatherElliptic}
%%
%%

\section{Large time asymptotics: proof of Theorem \ref{thrm:asymp:rl} part (b)}\label{sect_asymp}
We  study the long-time asymptotics of the RH  problem~\ref{RH_problem_1} by  applying   the  Deift-Zhou steepest descend method \cite{DZ93} for oscillatory RH problems. The high oscillatory terms of  the matrix entries of $J_M(k)$  defined in \eqref{J_M} come from the exponential  factors 
 $\e^{\pm\i \theta(x,t;k)}$.
Since the stationary point of $\theta(x,t;k)=kx+4k^3t$ is $k=\sqrt{-x/12t},$ we introduce  a new independent variable
\begin{equation*}
\xi=\dfrac{x}{12t},
\end{equation*}
and the function $\widehat{\theta}(k,\xi)$ with  $t\widehat{\theta}(k,\xi)=\theta(x,t;k)$, namely
\begin{equation}
\label{widetheta}
\widehat{\theta}(k,\xi)=12k\xi+4k^3.
\end{equation}
The sign of the $\Im\widehat{\theta}(k,\xi)$ are plot in  Figure~\ref{Fig_theta}.
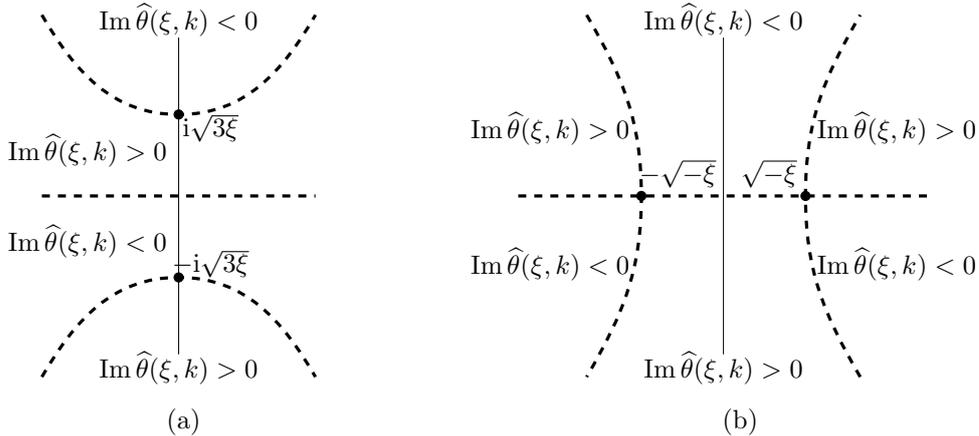
\begin{figure}[ht]
\begin{minipage}{0.4\linewidth}
\begin{tikzpicture}[scale=0.6]
%\draw[fill=black] (0,3.5) circle [radius=0.05];
%\draw[fill=black] (0,-3.5) circle [radius=0.05];
%\draw[fill=black] (0,1.5) circle [radius=0.05];
%\draw[fill=black] (0,-1.5) circle [radius=0.05];
\draw[fill=black] (0,1.8) circle [radius=0.1];
\draw[fill=black] (0,-1.8) circle [radius=0.1];
%segment [i c_l,-i c_l]
\draw[very thin] (0,-3.5) to (0,3.5);
%\draw[very thick,dashed] (0,2.2) to (0,3.5);
%\draw[very thick,dashed] (0,-2.4) to (0,-3.5);
%\draw[dashed,gray, very thin] (0,-2.2) to (0,-1.5);
%\draw[dashed, gray, very thin] (0,2.2) to (0,1.5);
%points $i c_l, i c_r
%\node at (0.4,3.1) {$\i c_{\l}$};\node at (0.35,1.4) {$\i c_{\r}$};\node at (0.3,2.1) {$\i d$};
%\node at (0.4,-3.1) {$-\i c_{\l}$};\node at (0.4,-1.2) {$-\i c_{\r}$};\node at (0.4,-2.1) {$-\i d$};
\node at (0.7,1.5) {$\mathrm{i} \sqrt{3\xi} $};
\node at (0.7,-1.5) {$-\mathrm{i}\sqrt{3\xi} $};
\draw[very thick,dashed] (-3,4) [out=-60, in=-180] to (0,1.8) [out=0,in=-120] to (3,4);
\draw[very thick,dashed] (-3,-4) [out=60, in=180] to (0,-1.8) [out=0,in=120] to (3,-4);
%pluses, minuses
\node at (0,3.9) {$\Im \widehat\theta(\xi,k)<0$};
\node at (-2,1.) {$\Im\widehat\theta(\xi,k)>0$};
%\node at (2,-1.5)  {$\Im\widehat\theta(\xi,k)<0$};
\node at (-2,-1.)  {$\Im\widehat\theta(\xi,k)<0$};
\node at (0,-3.8){$\Im \widehat\theta(\xi,k)>0$};
%%real line
\draw[very thick, dashed] (-3,0) to (3,0);
\end{tikzpicture}\\
\centering{(a)}
\end{minipage}
\begin{minipage}{0.53\linewidth}
\begin{tikzpicture}[scale=0.6]
\draw[fill=black] (1.8,0) circle [radius=0.1];
\draw[fill=black] (-1.8,0) circle [radius=0.1];
\draw[very thin] (0,-3.5) to (0,3.5);
\node at (1.,0.5) {$ \sqrt{-\xi} $};
\node at (-1.,0.5) {$-\sqrt{-\xi} $};
\draw[very thick,dashed] (-3,4) [out=-60, in=90] to (-1.8,0) [out=-90,in=60] to (-3,-4);
\draw[very thick,dashed] (3,-4) [out=120, in=-90] to (1.8,0) [out=90,in=-120] to (3.,4);
\node at (0,3.9) {$\Im \widehat\theta(\xi,k)<0$};
\node at (-3.8,1.5) {$\Im\widehat\theta(\xi,k)>0$};
\node at (3.8,1.5) {$\Im\widehat\theta(\xi,k)>0$};
\node at (-3.8,-1.5)  {$\Im\widehat\theta(\xi,k)<0$};
\node at (3.8,-1.5)  {$\Im\widehat\theta(\xi,k)<0$};
\node at (0,-3.8){$\Im \widehat\theta(\xi,k)>0$};
\draw[very thick, dashed] (-4.5,0) to (4.5,0);
\end{tikzpicture}\\
\centering{(b)}
\end{minipage}
 \caption{ Signs of   $\Im \widehat\theta(\xi,k)$ for $\xi>0$ in figure $(a)$ and for $\xi<0$ in figure $(b)$.}
\label{Fig_theta}
\end{figure}

To perform the asymptotic analysis  of the Riemann Hilbert problem \ref{RH_problem_1}, our first step is the introduction of a scalar  function $g=g(k,\xi)$  
which is asymptotic to $\widehat{\theta}(k,\xi)$, namely
\begin{equation}\label{gthetaas}
g(k,\xi)=\widehat{\theta}(k,\xi)+O(\frac{1}{k}),\quad |k|\to\infty.
\end{equation}

%This is done with the help of the following function $g(k,\xi),$ 
%
The function $g(k,\xi)$  takes different forms for different regions of the parameter $\xi$ (\cite[page 9]{KM2}):
\begin{equation}\label{g_def}
g(k,\xi)=\begin{cases}
\hat{g}_{c_+}(k,\xi)
\quad \mbox{ if }\xi\geq \frac{c_-^2}{3}+\frac{c_+^2}{6},
\\\\
\int\limits_{\i c_-}^{k}\frac{12 s(s^2+\xi+\frac{c_-^2+c_+^2-d^2}{2})\sqrt{s^2+d(\xi)^2}\, \d s}{\sqrt{(s^2+c_-^2)(s^2+c_+^2)}},
\quad \mbox{ if } 
\frac{-c_-^2}{2}+c_+^2 \leq \xi \leq \frac{c_-^2}{3}+\frac{c_+^2}{6},
\\\\
\hat{g}_{c_-}(k,\xi),\quad \mbox{ if } \xi\leq\frac{-c_-^2}{2}+c_+^2.
\end{cases}
\end{equation}
where the function  $\hat{g}_{c}(k,\xi)=t^{-1} g_{c}(x,t;k)=(4k^2-2c^2 + 12\xi)\sqrt{k^2+c^2},$ and  $g_{c}(x,t;k)$ has been defined in \eqref{E_c}. 
The quantity  $d(\xi)$  appearing in the function $g(k,\xi)$ in the middle region,  is a function of the parameter $\xi,$  and it is obtained from \eqref{d_mu_hyper} below and satisfies $c_+<d(\xi)<c_-$.
{\color{black}
We observe that the function $g(k,\xi)$ is  Schwartz symmetric. Indeed  $\overline{\hat{g}_{c}(\bar{k},\xi)}=\hat{g}_{c}(k,\xi)$, while in the middle region, we have  
\begin{align*}
g(k,\xi)&=\left(\int\limits_{\i c_-}^{k}+\frac{1}{2}\int\limits_{-\i c_-}^{\i c_-}\right)\frac{12 s(s^2+\xi+\frac{c_-^2+c_+^2-d^2}{2})\sqrt{s^2+d(\xi)^2}\, \d s}{\sqrt{(s^2+c_-^2)(s^2+c_+^2)}}\\
&=\frac{1}{2}\left(\int\limits_{\i c_-}^{k}+\int\limits_{-\i c_-}^{k}\right)\frac{12 s(s^2+\xi+\frac{c_-^2+c_+^2-d^2}{2})\sqrt{s^2+d(\xi)^2}\, \d s}{\sqrt{(s^2+c_-^2)(s^2+c_+^2)}}\\
\end{align*}
where the second integral on the first line is equal to zero because of   \eqref{d_mu_hyper} and   the residue theorem. Clearly the second line of the above expression is  Schwartz symmetric.
 }
\noindent Then we define the first transformation of the RH problem~\ref{RH_problem_1} \color{black}{(note that since the large time asymptotics is studied along the rays $x/t=const,$  the parameters $x,t$ in $M$ are changed to parameters $\xi=x/(12t),t$ in $Y$ )}
$$Y(\xi,t;k)=M(x,t;k)\e^{\i\(tg(k,\xi)-\theta(x,t;k)\)\sigma_3},$$
so that
\begin{equation*}
Y_{-}(\xi,t;k)=Y_{+}(\xi,t;k)J_Y(\xi,t;k), \quad k\in \Sigma,
\end{equation*}
where
 $$J_{Y}(\xi,t;k)=\e^{-\i\(tg_+(k,\xi)-\theta(x,t;k)\)\sigma_3}J_{M}(x,t;k)\e^{\i\(tg_-(k,\xi)-\theta(x,t;k)\)\sigma_3}$$
 with $J_M$ as in \eqref{J_M}.
%   It is shown in \cite{KM}  that the function $g(k,\xi)$ is %always
%   analytic in the  complex plane $\C\backslash[\i c_-,-\i c_-]$.
%Furthermore, the asymptotic analysis in \cite{KM2}  produces three $g$ functions that will be verified a-posteriori. 
\color{black}{The three different regions in the definition of $g$ will be referred to as }%Namely, \color{black}{different $g-$functions will be used in the following regions:}
\color{black}
\begin{itemize}  
\item  a dispersive shock wave region $\frac{-c_-^2}{2}+c_+^2 < \xi < \frac{c_-^2}{3}+\frac{c_+^2}{6}$, that can contain breathers on an elliptic background;

\item  right constant region  $\xi \geq \frac{c_-^2}{3}+\frac{c_+^2}{6}$, with possible solitons and breathers on the constant background $c_+$, travelling in positive direction;

\item left constant region $\frac{-c_-^2}{2}+c_+^2 \leq \xi \leq \frac{c_-^2}{3}+\frac{c_+^2}{6}$,  where possible breathers on the constant background $c_-$ travel in either positive or negative direction.
 \end{itemize}
 Furthermore, the left constant region is subdivided into two regions:
 \begin{itemize}
 \item utmost left constant region $\xi<\frac{-c_-^2}{2};$
 \item middle left constant region $\frac{-c_-^2}{2}<\xi<\frac{-c_-^2}{2}+c_+^2.$
 \end{itemize}
\color{black}

We start by performing the asymptotic analysis of the dispersive shock wave region.

\subsection{Proof of theorem~\ref{thrm:asymp:rl}  part (b): dispersive shock wave region with  $ \frac{-c_{\l}^2}{2}+c_{\r}^2<\frac{x}{12 t}<\frac{c_{\l}^2}{3}+\frac{c_{\r}^2}{6}$}
\label{sect_dispersive}
In this region we will verify a-posteriori that the $g$-function  is analytic in $\mathbb{C}\backslash [\i c_-, -\i c_-]$ and  takes the form
 \begin{equation}
 \label{g_ell0}
 g(k,\xi)=12\displaystyle\int\limits_{\i c_{-}}^{k}\frac{s\left(s^2+ \xi+\frac{c_{-}^2+c_{+}^2}{2}-\frac{d^2}{2}\right)\sqrt{s^2+d^2(\xi)}\d s}{\sqrt{(s^2+c_{-}^2)(s^2+c_{+}^2)}}\,,
 \end{equation}
 for
  $\quad \frac{-c_{-}^2}{2}+c_{+}^2<\xi<\frac{c_{-}^2}{3}+\frac{c_{+}^2}{6},$  where the quantity $d=d(\xi)$  is determined by 
  \begin{equation}
  \label{d_mu_hyper}
 \int\limits_{\i d}^{\i c_{+}}\frac{s\left(s^2+ \xi+\frac{c_{-}^2+c_{+}^2}{2}-\frac{d^2}{2}\right)\sqrt{s^2+d^2}\d s}{\sqrt{(s^2+c_{-}^2)(s^2+c_{+}^2)}}=0.
 \end{equation}
\begin{wrapfigure}{r}{4cm}
\begin{tikzpicture}[scale=0.4]
%\draw[fill=black] (0,3.5) circle [radius=0.05];
%\draw[fill=black] (0,-3.5) circle [radius=0.05];
%\draw[fill=black] (0,1.5) circle [radius=0.05];
%\draw[fill=black] (0,-1.5) circle [radius=0.05];
\draw[fill=black] (0,1) circle [radius=0.05];
\draw[fill=black] (0,-1) circle [radius=0.05];
\draw[fill=black] (0,2.5) circle [radius=0.05];
\draw[fill=black] (0,-2.5) circle [radius=0.05];
\draw[fill=black] (-1,1.5) circle [radius=0.1];
\draw[fill=black] (-1,-1.5) circle [radius=0.1];
\draw[fill=black] (1,-1.5) circle [radius=0.1];
\draw[fill=black] (1,1.5) circle [radius=0.1];
\draw[fill=red] (-2,2.5) circle [radius=0.1];
\draw[fill=red] (-2,-2.5) circle [radius=0.1];
\draw[fill=red] (2,-2.5) circle [radius=0.1];
\draw[fill=red] (2,2.5) circle [radius=0.1];
\draw[very thin] (0,-3.5) to (0,3.5);
\draw[very thick] (-3,5) [out=-60, in=-180] to (0,2.5) [out=0,in=-120] to (3,5);
\draw[very thick] (-3,-5) [out=60, in=180] to (0,-2.5) [out=0,in=120] to (3,-5);
\draw[very thick] (-3,3) [out=-60, in=-180] to (0,1.) [out=0,in=-120] to (3,3);
\draw[very thick] (-3,-3) [out=60, in=180] to (0,-1.) [out=0,in=120] to (3,-3);
%pluses, minuses
\node at (0.7,.5) {$\i c_+ $};
\node at (0.7,3.1) {$\i c_- $};
\node at (0.85,-.7) {$-\i c_+ $};
\node at (0.85,-2.3) {$-\i c_- $};
%%real line
\draw[very thick, dashed] (-3,0) to (3,0);
\end{tikzpicture}
\caption{The  zero set described by  \eqref{boundary} and the  spectra  of two  distinct  breathers  trapped inside the dispersive shock wave region.}
\label{fig_Bre}
\end{wrapfigure}
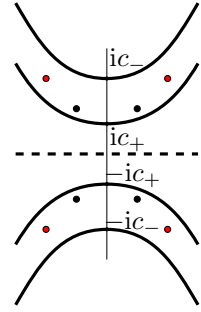
The solvability of this   equation  for $d=d(\xi)$  was established in \cite{KM2}.  Below we give a different derivation of \eqref{g_ell0} and we show the solvability  of $d=d(\xi)$ using the hyperbolic
nature of the Whitham modulation equations.
The dispersive shock wave region contains a trapped breather if  the discrete spectrum of  the breather is contained within the region bounded by the curves
\begin{equation}
\label{boundary}
\Im g_{c_+}(k, \frac{c_-^2}{3}+\frac{c_+^2}{6})=0,\quad \Im g_{c_-}(k,\frac{-c_-^2}{2}+c_+^2),
\end{equation}
where  $12(\frac{c_-^2}{3}+\frac{c_+^2}{6})$ and $12(\frac{-c_-^2}{2}+c_+^2)$ are the speeds   respectively of the leading edge and trailing edge of the dispersive shock wave (see Figure~\ref{fig_Bre}  for the level set of \eqref{boundary}).
A plot of the signs of the $\Im g(k,\xi)$ is shown in Figure~\ref{Fig_Elliptic}, which describes  the regions of the $k$-plane  where the quantity
$\e^{\i t g(k,\xi)}$ is exponentially small.
\newpage
\begin{figure}[ht]
\begin{minipage}{0.47\linewidth}
\begin{tikzpicture}
\draw[fill=black] (0,3.5) circle [radius=0.05];
\draw[fill=black] (0,-3.5) circle [radius=0.05];
\draw[fill=black] (0,1.5) circle [radius=0.05];
\draw[fill=black] (0,-1.5) circle [radius=0.05];
\draw[] (0,2.2) circle [radius=0.05];
\draw[] (0,-2.2) circle [radius=0.05];
%segment [i c_l,-i c_l]
\draw[very thick,dashed] (0,1.5) to (0,-1.5);\draw[very thick,dashed] (0,2.2) to (0,3.5);
\draw[very thick,dashed] (0,-2.4) to (0,-3.5);
\draw[dashed,gray, very thin] (0,-2.2) to (0,-1.5);
\draw[dashed, gray, very thin] (0,2.2) to (0,1.5);
%points $i c_l, i c_r
%\node at (0.4,3.1) {$\i c_{\l}$};\node at (0.35,1.4) {$\i c_{\r}$};\node at (0.3,2.1) {$\i d$};
%\node at (0.4,-3.1) {$-\i c_{\l}$};\node at (0.4,-1.2) {$-\i c_{\r}$};\node at (0.4,-2.1) {$-\i d$};
\node at (-0.4,3.4) {$\mathrm{i} c_-$};\node at (-0.35,1.23) {$\mathrm{i} c_+$};
\node at (-0.4,-3.4) {$-\mathrm{i} c_-$};\node at (-0.4,-1.23) {$-\mathrm{i} c_+$};
\node at (0.3,2.5) {$\mathrm{i} d $};
\node at (0.4,-2.5) {$-\mathrm{i} d$};
\draw[very thick,dashed] (-3,4) [out=-60, in=-120] to (0,2.2) [out=-60,in=-120] to (3,4);
\draw[very thick,dashed] (-3,-4) [out=60, in=120] to (0,-2.2) [out=60,in=120] to (3,-4);
%pluses, minuses
\node at (0,3.9) {$\Im g(k)<0$};
\node at (2,1.5) {$\Im g(k)>0$};\node at (-2,1.5) {$\Im g(k)>0$};
\node at (2,-1.5)  {$\Im g(k)<0$};\node at (-2,-1.5)  {$\Im g(k)<0$};
\node at (0,-3.8){$\Im g(k)>0$};
%%real line
\draw[very thick, dashed] (-3,0) to (3,0);
\end{tikzpicture}\\
\centering{(a)}
\end{minipage}
\hfill
\begin{minipage}{0.47\linewidth}
\begin{tikzpicture}
\draw[fill=black] (0,3.5) circle [radius=0.05];
\draw[fill=black] (0,-3.5) circle [radius=0.05];
\draw[fill=black] (0,1.5) circle [radius=0.05];
\draw[fill=black] (0,-1.5) circle [radius=0.05];
\draw[very thick,dashed] (-3,4) [out=-60, in=-120] to (0,2.2) [out=-60,in=-120] to (3,4);
\draw[very thick,dashed] (-3,-4) [out=60, in=120] to (0,-2.2) [out=60,in=120] to (3,-4);

%segment [i c_l,-i c_l]
\draw[thick,postaction = decorate, decoration = {markings, mark = at position 0.4 with {\arrow{>}}},
decoration = {markings, mark = at position 0.6 with {\arrow{>}}},
decoration = {markings, mark = at position 0.1 with {\arrow{>}}},
decoration = {markings, mark = at position 0.9 with {\arrow{>}}}] (0,3.5) to (0,-3.5);
%lenses L_5, L_7, L_6, L_8
\draw[thick, postaction = decorate, decoration = {markings, mark = at position 0.5 with {\arrow{>}}}](0,3.5) to [out=0, in=90] (1,3) [out=-90,in=30] to (0,2.2);
\draw[thick, postaction = decorate, decoration = {markings, mark = at position 0.5 with {\arrow{>}}}](0,3.5) to [out=180, in=90] (-1,3) [out=-90,in=150] to (0,2.2);
\draw[thick, postaction = decorate, decoration = {markings, mark = at position 0.5 with {\arrow{<}}}](0,-3.5) to [out=0, in=-90] (1,-3) to[out=90,in=-30](0,-2.2);
\draw[thick, postaction = decorate, decoration = {markings, mark = at position 0.5 with {\arrow{<}}}](0,-3.5) to [out=180, in=-90] (-1,-3)to[out=90, in=-150] (0,-2.2);
\node at (0.4,3.3) {$\i c_{\l}$};\node at (0.35,1.4) {$\i c_{\r}$};\node at (0.35,2.1) {$\i d$};
\node at (0.4,-3.2) {$-\i c_{\l}$};\node at (0.4,-1.2) {$-\i c_{\r}$};\node at (0.4,-2.1) {$-\i d$};
\node at (1.3,3.) {$L_7$}; \node at (-1.3, 3.) {$L_5$};\node at (1.3,-3.) {$L_8$}; \node at (-1.3, -3.) {$L_6$};
%lines L_1, L_2, L_3, ...
\draw[thick, postaction = decorate, decoration = {markings, mark = at position 0.5 with {\arrow{>}}}](-3,1) to [out = 0, in =-110] (0,2.2);
\draw[thick, postaction = decorate, decoration = {markings, mark = at position 0.5 with {\arrow{<}}}](3,1) to[out = 180, in =-70] (0,2.2);
\draw[thick, postaction = decorate, decoration = {markings, mark = at position 0.5 with {\arrow{>}}}](-3,-1) to[out = 0, in =110] (0,-2.2);
\draw[thick, postaction = decorate, decoration = {markings, mark = at position 0.5 with {\arrow{<}}}](3,-1) to[out = 180, in =70] (0,-2.2);
\node at (2,0.7) {$L_1$};\node at (-2,0.7) {$L_1$};\node at (2,-0.7) {$L_2$};\node at (-2,-0.7) {$L_2$};
\node at (-1,0.3) {$\Omega_1$};\node at (1,0.3) {$\Omega_1$};
\node at (-1,-0.3) {$\Omega_2$};\node at (1,-0.3) {$\Omega_2$};
\node at (0.0,3.9) {$\Omega_3$};\node at (0.0,-3.9) {$\Omega_4$};
\node at (0.5, 2.9){$\Omega_7$};\node at (-0.5, 2.9){$\Omega_5$};
\node at (0.5, -2.8){$\Omega_8$};\node at (-0.5, -2.8){$\Omega_6$};
%real line
\draw[thick, dashed] (-3,0) to (3,0);
\end{tikzpicture}
\centering{(b)}
\end{minipage}
 \caption{(a) Distribution of signs of  $\Im g =0$ for $c_{\r}^2-\frac{c_{\l}^2}{2}<\xi<\frac{c_{\l}^2}{3}+\frac{c_{\r}^2}{6}.$
 (b) Contour   deformation of  the  RH  problem   for $X$.}
\label{Fig_Elliptic}
%\label{Fig_FinalRHP}
\end{figure}
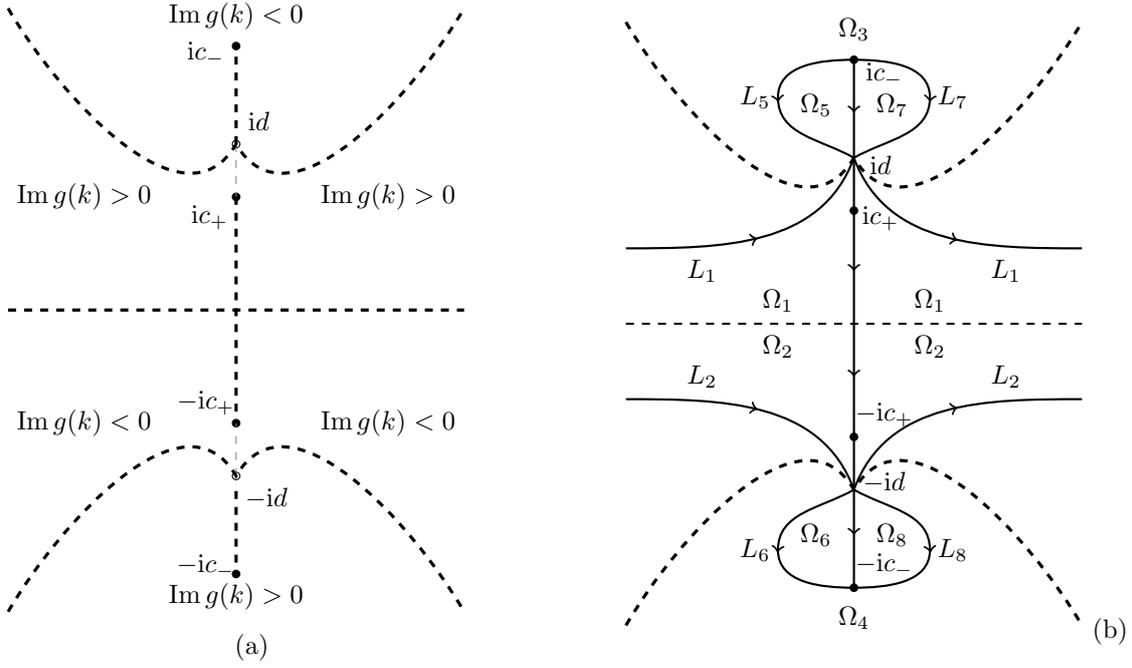
Our first step in the asymptotic analysis is to take care of the discrete spectrum.
We introduce the function $T(k,\xi)$ defined as
\begin{equation}
\label{T}
\begin{split}
&T(k,\xi)=\widetilde{T}(k)H(k),\\
&  \widetilde{T}(k)=\prod\limits_{\substack{\kappa_j\,|{\Im g(\kappa_j,\xi)<0} \\ \Re \kappa_j>0, \Im \kappa_j>0}}\left(\dfrac{k-\ol{\kappa_j}}{k-\kappa_j}\dfrac{k+\kappa_j}{k+\ol\kappa_j}\right)
\prod\limits_{\substack{\kappa_j \,|{\Im g(\kappa_j,\xi)<0}\\\Im \kappa_j>0\\ \Re \kappa_j=0}}\left(\dfrac{k-\ol{\kappa_j}}{k-\kappa_j}\right)
\end{split}
\end{equation}
where  $H(k)$ is analytic in $\C\backslash [ic_{\l},-ic_{\l}]$  and $H(k)= 1+O(k^{-1})$ as $|k|\to\infty$. Further properties  of $T(k)$   and $H(k)$  will be determined later.

 Then we define   the first transformation 

\begin{equation}
\label{Y_tilde}
\begin{split}
&\widetilde{Y}(\xi,t;k)=Y(\xi,t;k)T^{-\sigma_3}(\xi,t;k)\times\\
&\times\begin{cases}
\begin{pmatrix}1 & \frac{-\e^{-2\i t g(k,\xi)}}{\left[\frac{\i\nu_j}{T^2(k)(k-\kappa_j)}\right]} \\
\left[\frac{\i\nu_j}{T^2(k)(k-\kappa_j)}\right]\e^{2\i t g(k,\xi)} & 0 \end{pmatrix},\;
|k-\kappa_j|<\varepsilon,\quad \Im  g(\kappa_j,\xi)\leq-\delta\\
\begin{pmatrix}1 & \frac{-\e^{-2\i t g(k,\xi)}}{\left[\frac{\i\,\ol{\nu_j}}{T^2(k)(k+\ol\kappa_j)}\right]} \\
\left[\frac{\i\,\ol{\nu_j}}{T^2(k)(k+\ol\kappa_j)}\right]\e^{2\i t g(k,\xi)} & 0 \end{pmatrix},\;
|k+\ol\kappa_j|<\varepsilon,\; \Im g(-\ol\kappa_j,\xi)\leq-\delta\\
\begin{pmatrix} 0 & \left[\frac{\i\,\ol{\nu_j}\,T^2(k)}{(k-\ol\kappa_j)}\right]\e^{-2\i t g(k,\xi)}\\
\e^{2\i t g(k,\xi)}{\left[\frac{(k-\ol{\kappa_j})}{-\i\,\ol{\nu_j}\,T^2(k)}\right]} & 1
\end{pmatrix},\;  |k-\ol\kappa_j|<\varepsilon,\ \Im g(\ol{\kappa_j},\xi)\geq\delta,
\\
\begin{pmatrix} 0 & \left[\frac{\i\nu_jT^2(k)}{(k+\kappa_j)}\right]\e^{-2\i t g(k,\xi)}\\
\e^{2\i t g(k,\xi)}{\left[\frac{(k+\kappa_j)}{-\i\nu_j\,T^2(k)}\right]} & 1
\end{pmatrix},\;  |k+\kappa_j|<\varepsilon,\ \Im g(-\kappa_j,\xi)\geq\delta,
\\
I, \textrm{ elsewhere.}
\end{cases}
\end{split}
\end{equation}
We assume $\varepsilon$ sufficiently small so that $ \Im  g(\kappa,\xi)<-\frac{\delta}{2}$ for $|k-\kappa_j|\leq\varepsilon$ and similarly for the other cases.
In this way we are taking control of the exponentially big terms in the jump matrix related to all points  of the discrete  spectrum except for those  $\kappa_j$  for which 
{\color{black}    $\Im g(\kappa_j,\xi)=0,$   for some $\xi$ with $\frac{-c_-^2}{2}+c_+^2 < \xi<\frac{c_-^2}{3}+\frac{c_+^2}{6}$.  Following the   discussion in the introduction, this is possible only for points of the  spectrum corresponding to breathers.
%  The corresponding  speed   $V_{\ell}$  of the breather is obtained by solving the equation
%\begin{equation*}
%\Im g(\kappa_{\ell},V_{\ell})=0,
%\end{equation*}
%and $\frac{-c_-^2}{2}+c_+^2 < V_{\ell}<\frac{c_-^2}{3}+\frac{c_+^2}{6}$.
%In this case  the  breather is trapped inside the dispersive shock wave region.}

 The jump matrix $J_{\widetilde{Y}}$ (  associated to the RH problem $\widetilde{Y}_- = \widetilde{Y}_+ J_{\widetilde{Y}}$), is given by

 \begin{equation}
 \label{J_tilde}
J_{\widetilde{Y}}(\xi,t;k)=\begin{cases}
\begin{cases}
\begin{pmatrix} 1 & \frac{T^2(k,\xi)(k-\kappa_j)}{{\i\nu_j\, \e^{2\i t g(k,\xi)}} } \\ 0 & 1  \end{pmatrix}\quad\mbox{ if  $\Im g(k)<-\frac{\delta}{2}$  and}\\
\begin{pmatrix} 1 & 0 \\ \frac{\i\nu_j\, \e^{2\i t g(k,\xi)}}{(k-\kappa_j)\,T^2(k,\xi)}&1\end{pmatrix}
\quad\mbox{ if  $\Im g(k)>\frac{\delta}{2}$}, \;  k\in C_j;
\end{cases}\\
\begin{cases}
\begin{pmatrix} 1 & \frac{T^2(k,\xi)(k+\ol{\kappa_j})}{{\i\,\ol{\nu_j}\, \e^{2\i t g(k,\xi)}} } \\ 0 & 1  \end{pmatrix}\quad
\mbox{ if  $\Im g(k)<-\frac{\delta}{2}$  and}\\
\begin{pmatrix} 1 & 0 \\ \frac{\i\,\ol{\nu_j}\, \e^{2\i t g(k,\xi)}}{(k+\ol{\kappa_j})\,T^2(k,\xi)}&1\end{pmatrix}
\quad \mbox{ if }\Im g(k)>\frac{\delta}{2},\;\;k\in -\ol{C_j};
\end{cases}\\
\begin{cases}
\begin{pmatrix} 1 & 0\\
\frac{(k-\ol{\kappa_j})\e^{2\i t g(k,\xi)}}{\i\,\ol{\nu_j}\,T^2(k,\xi)}&1\end{pmatrix}\quad \mbox{ if  $\Im g(k)>\frac{\delta}{2}$  and }\\
\begin{pmatrix}1 &  
\frac{\i\,\ol{\nu_j}\,T^2(k,\xi)}{(k-\ol{\kappa_j})\e^{2\i t g(k,\xi)}}
\\0&1\end{pmatrix}\quad\mbox{ if  $\Im g(k)<-\frac{\delta}{2}$,}
\;\;\;k\in\ol{C}_j;
\end{cases}\\
\begin{cases}
\begin{pmatrix}1&0\\\frac{(k+{\kappa_j})\e^{2\i t g(k,\xi)}}{\i {\nu_j}T^2(k,\xi)}&1\end{pmatrix}\quad\mbox{ if $\Im g(k)>\frac{\delta}{2}$ and}\\
\begin{pmatrix}1 &  
\frac{\i\,{\nu_j}\,T^2(k,\xi)}{(k+{\kappa_j})\,\e^{2\i t g(k,\xi)}}
\\0&1\end{pmatrix}
\quad \mbox{ if $\Im g\(k\)<-\frac{\delta}{2}$}, \;\; k\in-{C}_j;
\end{cases}\\
\\
T_+^{\sigma_3}(k,\xi)J_Y(k)T_-^{-\sigma_3}(k,\xi),\quad\mbox{elsewhere.}
\end{cases}
\end{equation}

%\todo{ Why we use $\delta$ and $\delta/2$?}
%%%%%%%%%%%%%%%%%%%%%%%%%%
% \todo{Must be different triangularity for $C_j$ and $-C_j.$}
 It is clear from the form of the above jumps, that the matrix $J_{\widetilde{Y}}$ will be exponentially close to the identity  as $t\to\infty$ on the circles $\pm C_j$ and $\pm \ol{C}_j$ for all those   points  $\kappa_j$  of the discrete spectrum 
 for which $\Im g(\kappa_j,\xi)\neq 0$ when $\xi $ is in the region $\frac{-c_-^2}{2}+c_+^2 < \xi<\frac{c_-^2}{3}+\frac{c_+^2}{6}$. 

The next  step is to  take care of the continuous spectrum  on the real axis.  As a first step, we   reduce the jump $J_{\widetilde{Y}}(t,\xi,k)$  for $k\in\mathbb{R}\backslash\{0\}$ to a matrix exponentially close to the identity.
   For the purpose it  is  sufficient to  factorise  the matrix $J_{\widetilde{Y}}(t,\xi,k)$ to the form
  \begin{align*}
  J_{\widetilde{Y}}(\xi,t;k)&=\begin{pmatrix}1 & -\overline{r(k)}T^2(k)\e^{-2itg(k,\xi)}\\-\dfrac{r(k)}{T^2(k)}\e^{2\i tg(k,\xi)} & 1+|r(k)|^2\end{pmatrix},\quad k\in\mathbb{R}\setminus\left\{0\right\}\\
  &=\begin{pmatrix}1 & 0\\-\dfrac{r(k)}{T^2(k)}\e^{2\i tg(k,\xi)} & 1\end{pmatrix}\begin{pmatrix}1 & -\overline{r(k)}T^2(k)\e^{-2itg(k,\xi)}\\0 & 1\end{pmatrix}.
  &\end{align*}
Then  using Deift-Zhou contour deformation  method  we introduce the new matrix $X,$ 
\begin{equation*}
X(\xi,t;k)=\left\{
\begin{array}{ll}
&\widetilde{Y}(\xi,t;k)\begin{pmatrix}1 & 0\\-\dfrac{r(k)}{T^2(k)}\e^{2\i tg(k,\xi)}&1\end{pmatrix},\quad  k\in\Omega_1,\\
\\
&\widetilde{Y}(\xi,t;k)
\begin{pmatrix}1 & \overline{r(\ol{k})}T^2(k)\e^{-2itg(k,\xi)}\\0 & 1\end{pmatrix},\quad k\in\Omega_2,\\
&\widetilde{Y}(\xi,t;k),\quad\mbox{elsewhere,}
\end{array}
\right.
\end{equation*}
where the regions $\Omega_1$ and $\Omega_{2}$ are specified   in Figure~\ref{Fig_Elliptic} and do not contain points of the discrete spectrum.
In this way  we can reduce the jump of $X(k)$  on  $\mathbb{R}\backslash\{0\}$  to identity, while the jumps of $X(k)$ on the lines $L_{1,2}\backslash U_{\pm \i d}$
 are exponentially close to the identity, where $U_{\pm \i d}$ is a small neighbourhood of the point $\pm \i d$.  The jumps $J_X(k)$ for  $k\in L_{1,2}\cap U_{\pm \i d}$ give the subleading  contribution to the asymptotics, analysis.  \\
 We still need to determine the functions $g(k,\xi)$ and $F(k)$.
  The remaining jumps of the matrix $X(k)$ are  obtained using the identities  from Lemma \ref{lem_abr}, and take the following  form:
\begin{equation}
\label{J_X0}
J_X(\xi,t;k)=\begin{cases}
T^{\sigma_3}_+(k)\begin{pmatrix}\e^{-\i t(g_+-g_-)}&0\\f(k)\e^{\i t(g_++g_-)} & \e^{\i t(g_+-g_-)}\end{pmatrix}T^{-\sigma_3}_-(k)\,\quad k\in(\i c_{\l},\i d),\\
T^{\sigma_3}_+(k) \e^{-\i t (g_+-g_-)\sigma_3} T^{-\sigma_3}_-(k),\;\; k\in(id,ic_{\r})\cup (-ic_{\r},-id),\\
T^{\sigma_3}_+(k)\begin{pmatrix} 0 &\i  \e^{-\i t(g_++g_-)} \\ \i \e^{\i t(g_++g_-)}&0\end{pmatrix}T^{-\sigma_3}_-(k),\quad k\in(\i c_{\r},-ic_{\r}),\\
%$$J_X(x,t;k)=T^{\sigma_3}_+(k)\begin{pmatrix}\i \ol{r_{+}(\ol{k})}\e^{-\i t(g_+-g_-)}& -\ol{f(\ol{k})} \e^{-\i t(g_++g_-)} \\ \i  \e^{\i t(g_++g_-)}  & -\i \ol{r_-(\ol{k})}\e^{\i t(g_+-g_-)}\end{pmatrix},\quad k\in(0,-\i c_{+}),$$
T^{\sigma_3}_+(k)\begin{pmatrix}\e^{-\i t(g_+-g_-)}&-\overline{f(\ol{k})}\e^{-\i t(g_++g_-)}\\0& \e^{\i t(g_+-g_-)}\end{pmatrix}T^{-\sigma_3}_-(k),k\in(-\i d,-\i c_{\l}).
 \end{cases}
 \end{equation}
  
We require that the above matrix $J_X$  has oscillatory   diagonal terms and non oscillatory  off-diagonal terms as $t\to\infty$.
Therefore we need to  require  that 
\begin{equation}
\begin{split}
\label{eq_g1}
&g_+(k)+g_-(k)=0,\quad k\in[\i c_-,\i d]\cup[\i c_+,-\i c_+]\cup [-\i d,-i c_-], \\
&g_+(k)-g_-(k)\in\R,\quad k\in[\i c_-,-\i c_-],\\
& g_+(k)-g_-(k)=-B(\xi),\quad k\in (\i d,\i c_+)\cup (-\i c_+,-\i d), 
\end{split}
\end{equation}
where $B=B(\xi)$ is independent from $k$.  Furthermore, for reasons that will become clear later,  we chose the scalar function $T(k)$ in such a way that
\begin{equation}
\label{RH_T1}
T_+(k)T_-(k)=1,\quad k\in [\i  c_+,-\i c_+].
\end{equation}
Then the above jump matrices are reduced to the form
\begin{equation}
\label{J_X}
J_X(\xi,t;k)=\left\{
\begin{array}{ll}
&\begin{pmatrix} \frac{T_+(k)\e^{-\i t (g_+-g_-)}}{T_-(k)}  & 0 \\ \dfrac{f(k)}{T_-(k)T_+(k)} & \frac{T_-(k)\e^{\i t(g_+-g_-)}}{T_+(k)} \end{pmatrix},\,\quad k\in(\i c_{-},\i d),\\
&\begin{pmatrix} \frac{T_+(k)\e^{\i t B(\xi)}}{T_-(k)}  & 0 \\ 0 & \frac{T_-(k)\e^{-\i t B(\xi)}}{T_+(k)} \end{pmatrix},\quad k\in(\i d,\i c_+)\cup (-\i c_{\r},-\i d)\\
&\begin{pmatrix} 0 &\i   \\ \i &0\end{pmatrix},\quad k\in(\i c_{\r},-ic_{\r}),\\
%$$J_X(x,t;k)=T^{\sigma_3}_+(k)\begin{pmatrix}\i \ol{r_{+}(\ol{k})}\e^{-\i t(g_+-g_-)}& -\ol{f(\ol{k})} \e^{-\i t(g_++g_-)} \\ \i  \e^{\i t(g_++g_-)}  & -\i \ol{r_-(\ol{k})}\e^{\i t(g_+-g_-)}\end{pmatrix},\quad k\in(0,-\i c_{+}),$$
&\begin{pmatrix} \frac{T_+(k)\e^{-\i t (g_+-g_-)}}{T_-(k)}  & -\overline{f(\ol{k})}T_-(k)T_+(k)\\0& \frac{T_-(k)\e^{\i t(g_+-g_-)}}{T_+(k)} \end{pmatrix},\quad k\in(-\i d,-\i c_{\l}).\\
  \end{array}\right.
  \end{equation}

\subsubsection{Determination of   the scalar functions $T(k)$  and $g(k)$.}
In this subsection we determine the scalar function $T(k)$ and we derive the expression for the function $g(k)$  that satisfies \eqref{eq_g1} and \eqref{gthetaas}.

The function $T(k)$  satisfies the relation (\ref{RH_T1}) and  \eqref{RH_T2}. We still need to add some assumptions on the boundary values of $T(k)$ in the interval
$(\i d,\i c_+)$ and $(-\i c_+,-\i d)$. In order to obtain a constant jump  matrix $J_X$  in (\ref{J_X})  we assume that $T_+(k)=T_-(k)\e^{\i \Delta}$ for $k\in(\i d,\i c_+)\cup (-\i c_+,-\i d)$, where the function $\Delta=\Delta(\xi)$   will be independent from $k$ and needs to be determined.
Using the relations (\ref{RH_T2})  we   finally have  the following RH problem for the function $T(k):$
 \begin{equation}
\nonumber
 \begin{split}
& T_+(k)T_-(k)=\dfrac{1}{|a(k)|^2}\quad k\in (\i c_{\l},\i d),\\
& T_+(k)T_-(k)=|a(k)|^2\quad k\in \cup (-\i d,-\i c_{\l}),\\
&T_+(k)=T_-(k)\e^{\i \Delta(k)},\quad k\in(\i d,\i c_{\r})\cup (-\i c_{\r},-\i d),\\
&T_+(k)T_-(k)=1,\quad k\in (\i  c_{\r},-\i c_{\r}),\\
&T(k)=1+O(k^{-1}),\quad \mbox{as}\;|k|\to \infty.
\end{split}
\end{equation}
The solution is obtained passing to the logarithm and using the Plemelj formula. Let us introduce
\begin{equation}
\label{R}
R(k)=\sqrt{(k^2+c_{\l}^2)(k^2+d^2)(k^2+c_{\r}^2)},
\end{equation}
where $R(k)$ is analytic in $\C\backslash [\i c_{\l},\i d]\cup[\i c_{\r},-\i c_{\r}]\cup [-\i d,-i c_{\l}]$   and real positive for $k=+0$, where $+0$ means the limit to $0$ from the right. Then the expression
\begin{equation}\label{T_ell}
\begin{split}
&T(k)=\widetilde{T}(k)
\exp\left[\frac{R(k)}{2\pi\i}\left\{\left( \int\limits_{\i c_{\l}}^{\i d}+\int\limits^{-\i c_{\l}}_{-\i d}+\int\limits^{-\i c_{\r}}_{\i c_{\r}}\right)\dfrac{-\log\widetilde{T}^2(s)}{(s-k)R_+(s)}ds\right.\right.\\
&+\left.\left.\left( \int\limits^{-\i c_{\l}}_{-\i d}-\int\limits_{\i c_{\l}}^{\i d}\right)\frac{\ln|a(s)|^2}{(s-k)R_+(s)}ds+
\left(\int\limits_{\i d}^{\i c_{\r}}+\int\limits_{-\i c_{\r}}^{-\i d}\right)\frac{\i\Delta\d s}{(s-k)R(s)}\right\}\right]
\end{split}
\end{equation}
solves the scalar RH problem for $T(k)$ with the quantity $\widetilde{T}(k)$ defined in \eqref{T}.
The requirement that $T(k)=1+O(k^{-1})$ as $|k|\to\infty$ determines $\Delta(\xi)$. Indeed we have, using the symmetry of the problem, that
\begin{equation}
\begin{split}
\label{Delta}
\Delta &= \frac{\int\limits^{\i c_{\l}}_{\i d}\frac{\ln(|a(s)|^2\widetilde{T}^2(s))s\d s}{R_+(s)}+\int\limits^{\i c_{\r}}_{0}\frac{\(\ln \widetilde{T}^2(s)\)s\d s}{R_+(s)}}{-\i\int\limits_{\i d}^{\i c_{\r}}\frac{s\d s}{R(s)}}\\
&=-\dfrac{\sqrt{c_-^2-c_+^2}}{K(m)}\left[\int\limits^{\i c_{\l}}_{\i d}\frac{\ln(|a(s)|^2\widetilde{T}^2(s))s\d s}{R_+(s)}+\int\limits^{\i c_{\r}}_{0}\frac{\(\ln \widetilde{T}^2(s)\)s\d s}{R_+(s)}
\right].
\end{split}
\end{equation}

The scalar function $g(k)$ satisfies the conditions (\ref{eq_g1}) and (\ref{gthetaas}). This implies that the function $g'(k)$  is analytic in $\C\backslash [\i c_{\l},-\i c_{\l}]$ and on the interval $[\i c_{\l},-\i c_{\l}]$ satisfies the conditions
\begin{equation}
\begin{split}
\label{eq_g2}
&g'_+(k)+g'_-(k)=0,\quad k\in(\i c_{\l},\i d)\cup(\i c_{\r},-\i c_{\r})\cup (-\i d,-i c_{\l})\\
&g'_+(k)-g'_-(k)=0,\quad k\in (\i d,\i c_{\r})\cup (-\i c_{\r},-\i d)\\
&g'(k)=\hat{\theta}'(k)+O(k^{-2})\quad \mbox{as $ |k|\to\infty$}.
\end{split}
\end{equation}
From the above conditions it follows that 
\begin{equation}
\label{eq_gp}
g'(k)=12\dfrac{P(k) }{R(k)},\quad P(k)=k^5+k^3(\xi+\frac{1}{2}(d^2+c_{\l}^2+c_{\r}^2))+bk
\end{equation}
where $R(k)$ is  defined in \eqref{R}. The constant $b$  in \eqref{eq_gp}  is determined by  requiring that the integral
\begin{equation*}
g(k)=\int_{\i c_{\l}}^kg'(s)ds
\end{equation*}
satisfies the third relation in (\ref{eq_g1}). This   immediately  implies that 
\begin{equation}
\label{g_zero1}
\int^{\i c_{\r}}_{\i d}g'(k)dk=0,\quad \int^{-\i c_{\r}}_{-\i d}g'(k)dk=0
\end{equation}
which gives $$b=\dfrac{1}{3}(c_{\l}^2c_{\r}^2+c_{\l}^2d^2+c_{\r}^2d^2)+(\xi-\frac{1}{6}(c_{\l}^2+c_{\r}^2+d^2))\left[c_{\l}^2-(c_{\l}^2-c_{\r}^2)\dfrac{E(m)}{K(m)}\right],$$
where $E(m)=\int_0^{\pi/2}\sqrt{1-m^2\sin^2\theta}d\theta$ and $K(m)=\int_0^{\pi/2}\frac{d\theta}{\sqrt{1-m^2\sin^2\theta}}$ are the complete elliptic integrals of the second and first kind respectively with modulus $m^2=\dfrac{d^2-c_{\r}^2}{c_{\l}^2-c_{\r}^2}$.
 We also have
\begin{equation}
\label{B}
\begin{split}
B(\xi)&=2\int\limits^{\i c_{\l}}_{\i d}g_+'(k)dk=12\int\limits_{-d^2}^{-c^2_-}\dfrac{s^2+s(\xi+\frac{1}{2}(d^2+c_{\l}^2+c_{\r}^2))+b}{\sqrt{(s+d^2)(s+c_{\l}^2)(s+c_{\r}^2)}}ds\\
&=-12\pi \dfrac{\sqrt{c_{\l}^2-c_{\r}^2}}{K(m)}(\xi-\dfrac{1}{6}(c_{\l}^2+c_{\r}^2+d^2))\in\R.
\end{split}
\end{equation}
Let us observe that
\begin{equation}
\label{BUV}
tB(\xi)=xU+tV
\end{equation}
where $U$ and $V$ have been defined in \eqref{UVh}.
 We still need to determine the  quantity $d$. This is obtained by requiring that $g'(k)|_{k=\pm \i d}=0$ that implies that the polynomial $P(k)$ in \eqref{eq_gp}  has a zero at $k=\pm \i d$, namely
 \begin{equation}\begin{split}
 \label{Whitham_d}
&P(\pm\i d)=\pm\i[d^5-d^3(\xi+\frac{1}{2}(d^2+c_{\l}^2+c_{\r}^2))+bd]=0\\
&\mbox{or    }12 \xi=W_2(c_{\r},d,c_{\l})\,,
 \end{split}
 \end{equation}
 where $W_2(\beta_1,\beta_2,\beta_3)$, $\beta_1\leq\beta_2\leq\beta_3$ has been defined in (\ref{Whitham_d0}).
 
 We observe that $W_2(\beta_1,\beta_2,\beta_3)$ is the speed of the Whitham modulation equations for MKdV derived in \cite{Driscol}.
 The relation with the speed  $V_2(r_1,r_2,r_3)$  of the Whitham modulation equations for KdV  \cite{W} is as follows:
 \begin{equation*}
 W_2(\beta_1,\beta_2,\beta_3)=V_2(\beta^2_1,\beta^2_2,\beta^2_3)\,.
 \end{equation*}
 In particular it was shown in \cite{Levermore} that the Whitham modulation equations for KdV are strictly hyperbolic and  satisfy the relation $\partial_{r_2}V_2(r_1,r_2,r_3)>0$ for $r_1<r_2<r_3$ which 
 implies that
 \begin{equation*}
 \dfrac{\partial}{\partial \beta_2}W_2(\beta_1,\beta_2,\beta_3)=2\beta_2 \dfrac{\partial}{\partial r_2}V_2(\beta_1^2,r_2,\beta_3^2)|_{r_2=\beta_2^2}>0,\quad \beta_2>0.
 \end{equation*}
 The above relation shows that the equation \eqref{Whitham_d} is invertible for $d$ as a function of $\xi$ only when $d>0$ or equivalently when $c_+>0$.
  Further comments about the case $c_->-c_+>0$ are given in Appendix~\ref{WhithamApp}.
 
 Using the properties of the elliptic functions \cite{Lawden}, we get that as $m\to 0$ 
 \begin{equation}
\label{elliptictrail}
K(m)=\dfrac{\pi}{2}\left(1+\dfrac{m^2}{4}+\dfrac{9}{64}m^4+O(m^6)\right),\;
E(m)=\dfrac{\pi}{2}\left(1-\dfrac{m^2}{4}-\dfrac{3}{64}m^4+O(m^6)\right),
\end{equation}
and as $m\to 1$
\begin{equation}
\label{ellipticlead}
E(m)= 1+\dfrac{1}{2}(1-m)\left[\log\dfrac{16}{1-m^2}-1\right](1+o(1)),
\quad K(m)= \dfrac{1}{2}\log\dfrac{16}{1-m^2}(1+o(1)).
\end{equation}
Using the above expansions  we have that
 \begin{itemize}
 \item as $m\to 0$ or $d\to c_+,$ 
 \begin{equation*}
 W_2(c_{\r},c_{\r},c_{\l})=-6c_{\l}^2+12c_{\r}^2 \end{equation*}
 and 
 \item  as $m\to 1$ or $d\to c_{\l},$ 
 \begin{equation*}
 W_2(c_{\r},c_{\l},c_{\l})=4c_{\l}^2+2c_{\r}^2
 \end{equation*}
 \end{itemize}
 which implies that 
 \begin{equation*}
 \frac{-c_{\l}^2}{2}+c_{\r}^2<\xi<\frac{c_{\l}^2}{3}+\frac{c_{\r}^2}{6}.
 \end{equation*}
 
 Summarizing,  the function $g(k)=g(k,\xi)$ takes the form
 \begin{equation}
 \label{g_final}
 \begin{split}
 g(k)&=12\int_{\i c_-}^k\dfrac{s^5+s^3(\xi+\frac{1}{2}(d^2+c_{\l}^2+c_{\r}^2))+bs }{\sqrt{(s^2+c_-^2)(s^2+c_+^2)(s^2+d^2)}}ds,\\
  b&=\dfrac{1}{3}(c_{\l}^2c_{\r}^2+c_{\l}^2d^2+c_{\r}^2d^2)+(\xi-\frac{1}{6}(c_{\l}^2+c_{\r}^2+d^2))\left[c_{\l}^2-(c_{\l}^2-c_{\r}^2)\dfrac{E(m)}{K(m)}\right]\,.
 \end{split}
 \end{equation}
 The parameter $d=d(\xi)$ is determined by \eqref{Whitham_d}. The above 
 derivation  for the function $g(k)$  is equivalent to the one obtained in \cite[page 11]{KM2} and written in \eqref{g_ell0}  with $d=d(\xi)$ as in \eqref{d_mu_hyper}.
 The signature of $\Im g(k)$ is given in Figure~\ref{Fig_Elliptic}.

 \begin{remark}
 Signature of $\Im g(k)$ can be constructed from the following considerations.
 Let us look at a level line $\Im g(k) = const.$ It can be parametrised  as  $k = k(s),$ where $s\in\mathbb{R},$ with $|k'(s)| = 1.$  The function $g(k(s))$ equals
 \begin{equation*}
 g(k(s)) = g(k(s_0))+\int^{s}_{s_0} g'(k(\sigma)) k'(\sigma) \d \sigma.
 \end{equation*} 
 Since $\Im g(k(s)) = const,$ hence 
 \begin{equation*}k'(s) = \dfrac{\ol{ g'(k(s))}}{|g'(k(s))|}.\end{equation*}
 This immediately gives us the direction of the level line for every point $k,$ and there is  exactly one line passing through every such a point  as long as  $g'(k)\neq 0$.
 Except for those regular points, we have several singular points, where $g'(k)$ is either 0 or infinite.
 {\color{black} For example, near the point $k=\pm \i d$ one has  $g(k)=\int ^{\i d}_{\i c_-}g'(s)ds+\int_{\i d}^kg'(s)ds$  so that 
 \begin{equation*}
\Im  g(k)\sim \Im \frac{2}{3}(k-\i d)^{\frac{3}{2}}e^{\frac{\i \pi}{4}}\times \mbox{const},\mbox{ as $k\to \i d$}
 \end{equation*}
 where the constant is real. So the lines where $\Im g(k)=0$ coming out of the point $k=\i d$ have argument $\varphi=-\frac{\pi}{6}+\frac{2}{3}n\pi$ with $n\in\mathbb{Z}$.  This corresponds to three different lines  with angles $\frac{\pi}{2},$ $-\frac{\pi}{6},$ $\frac{7\pi}{6}.$
 In a similar way there are three lines emerging from the point $k=-\i d$.}
%
%For such a point $k_0,$ let us take a ray
%\begin{equation*}k(s) = k_0 + s \e^{\i \varphi},\end{equation*}
%which emanates from the point $k_0$ at an angle $\varphi.$
%The condition that on  that ray  $Im \,g(k)$ is constant, i.e. 
%\begin{equation*}
%\e^{\i\varphi} = \dfrac{\ol{ g'(k(s))}}{|g'(k(s))|},
%\end{equation*}
%gives us the values for $\varphi.$
%For the point $k_0 = \i c,$ there is only one angle, $-\frac{\pi}{2},$ and for the point $k_0 = \i d, $ there are three angles, $\frac{\pi}{2},$ $-\frac{\pi}{6},$ $-\frac{5\pi}{6}.$ For the point $k=\i\mu,$ there are 4 values for the angle. (Here $\pm\i\mu$ are the zeros of the numerator  in the integrand in formula \eqref{B}.)

Also,  there are 6 rays $\Im g = const,$  converging to the point $k\to\infty,$  along directions $\pi k/3,$ $k =0,1,2,3,4,5.$
These rays consist of the real axis and the two ray emerging from the points $\pm \i d$ (see  Figure~\ref{Fig_Elliptic}).

%The ray emanating from $k = \i c_{\l},$ comes into point $\i d,$ gluing with one of the three rays emanating from the point $\i d.$
%There is only one way to connect all the rays emanating from the singular points. Furthermore, among all the level lines $\Im g = const,$ we need to take only those, where $\Im g =0.$ They consist of the real line, segments $\pm[\i c_{\l}, \i d],$ and hence all the other level rays emanating from $\pm\i d.$ This gives us all 6 rays $\Im g=0$, coming to $\infty,$ and is plotted in Figure~\ref{Fig_Elliptic}
%
 \end{remark}

\subsubsection{Opening of the lenses}\label{lenses}
With our choice of the functions $T(k)$ and $g(k,\xi)$, the jump $J_X$ in \eqref{J_X}  reduces to the form
\begin{equation}
\label{J_X1}
J_X(\xi,t;k)=\left\{
\begin{array}{ll}
&\begin{pmatrix} \frac{T_+(k)\e^{-\i t (g_+(k)-g_-(k))}}{T_-(k)}  & 0 \\ \dfrac{f(k)}{T_-(k)T_+(k)} & \frac{T_-(k)\e^{\i t(g_+(k)-g_-(k))}}{T_+(k)} \end{pmatrix},\,\quad k\in(\i c_{-},\i d),\\
&\\
&e^{\i (t B(\xi)+\Delta(\xi))\sigma_3},\quad k\in(\i d,\i c_+)\cup (-\i c_{\r},-\i d),\\
&\\
&\begin{pmatrix} 0 &\i   \\ \i &0\end{pmatrix},\quad k\in(\i c_{\r},-ic_{\r}),\\
&\\
%$$J_X(x,t;k)=T^{\sigma_3}_+(k)\begin{pmatrix}\i \ol{r_{+}(\ol{k})}\e^{-\i t(g_+-g_-)}& -\ol{f(\ol{k})} \e^{-\i t(g_++g_-)} \\ \i  \e^{\i t(g_++g_-)}  & -\i \ol{r_-(\ol{k})}\e^{\i t(g_+-g_-)}\end{pmatrix},\quad k\in(0,-\i c_{+}),$$
&\begin{pmatrix} \frac{T_+(k)\e^{-\i t (g_+(k)-g_-(k))}}{T_-(k)}  & -\overline{f(\ol{k})}T_-(k)T_+(k)\\0& \frac{T_-(k)\e^{\i t(g_+(k)-g_-(k))}}{T_+(k)} \end{pmatrix},\quad k\in(-\i d,-\i c_{\l}).\\
  \end{array}\right.
\end{equation}
Finally  we  apply Deift-Zhou steepest descend method   to get rid of the highly oscillatory terms in $k$ in the  diagonal  exponents of the above matrix.
 For the purpose we open lenses in the intervals $(\i c_{\l},\i d)$ and $(-\i d,-\i c_{\l})$.
We first  need to define the analytic extension of the function $f(k)=\frac{\i}{a_+(k)a_-(k)}$  to a neighbourhood of the interval $[\i c_{\l},-\i c_{\l}]$.  Recalling that $a_-(k)=-\i \ol{b(\ol{k})}$ for $k\in (\i c_-,\i c_+)\cup (-\i c_+,-\i c_-)$ we define
%Using the fact that $a_-(k)=\ol{b\(\ol{k}\)$ for 
%$k\in(\ic_-,\i c_+)\cup(-\i c_+,-\ic_-)$ we define the function $\hat{f}(k)$ as 
\begin{equation}
\label{hatf}\widehat f(k) = \frac{-1}{a(k)\ol{b\(\ol{k}\)}}=\frac{1+r(k)\ol{r\(\ol{k}\)}}{-\ol{r\(\ol{k}\)}},\qquad k\in U_{\sigma}([\i c_-,-\i c_-]).
\end{equation}
Then  it is immediate to verify that 
\begin{equation*}
\widehat f_+(k)=f(k),\quad \widehat f_-(k)=-f(k),\quad k\in [\i c_{\l},\i c_{\r}]\cup [-\i c_{\r},-\i c_{\l}].
\end{equation*}
 In order to get a factorization of the matrices $J_X(x,t;  k)$ for $k \in (\i c_{\l},\i d)\cup (-\i d,-\i c_{\l})$ we assume that
 \begin{equation}
 \label{RH_T2}
 \begin{split}
& T_+(k)T_-(k)=-\i \widehat f_+(k)=-\i f(k),\quad k\in (\i c_{\l},\i d),\\
& T_+(k)T_-(k)=-\dfrac{\i}{\ol {f(\ol k)} },\quad k\in \cup (-\i d,-\i c_{\l}).
\end{split}
\end{equation}
 In this way we can factorize
 \begin{equation}\label{fac_1}J_{X}(\xi,t;k)=\begin{pmatrix}1 & \frac{T_+^2(k)\e^{-2\i t g_+(k)}}{\widehat{f}_+(k)} \\ 0 & 1\end{pmatrix}
\begin{pmatrix}0&\i\\\i&0\end{pmatrix}
\begin{pmatrix}1 & \frac{-T_-^2(k)\e^{-2\i t g_-(k)}}{\widehat{f}_-(k)} \\ 0 & 1\end{pmatrix},
\quad k\in\(\i c_{\l},\i d\),
\end{equation}
\begin{equation}\label{fac_2}J_{X}(\xi,t;k) = \begin{pmatrix}1&0\\\frac{-\e^{2\i t g_+(k)}}{\ol{\widehat{f}_+\(\ol{k}\)}T_+^2(k)} & 1\end{pmatrix}\begin{pmatrix}0&\i\\\i&0\end{pmatrix}
\begin{pmatrix}1&0\\\frac{\e^{2\i t g_-(k)}}{\ol{\widehat{f}_-\(\ol{k}\)}T_-^2(k)} & 1\end{pmatrix},
\quad k\in(-\i d,-\i c_{\l}).\end{equation}
Then we proceed with contour deformation  and we define a new matrix $W(k)$  as
\begin{equation*}
W(k)=X(k)\begin{cases}
\begin{pmatrix}1 & \frac{T^2(k)\e^{-2\i t g(k)}}{\widehat{f}(k)} \\ 0 & 1\end{pmatrix},\quad k\in\Omega_5\cup\Omega_7,\\
\begin{pmatrix}1&0\\\frac{-\e^{2\i t g(k)}}{\ol{\widehat{f}\(\ol{k}\)}T^2(k)} & 1\end{pmatrix},\quad k\in\Omega_6\cup\Omega_8,\\
I\quad \mbox{elsewhere.}
\end{cases}
\end{equation*}
We obtain that the jumps for $W(\xi,t;k)$ are
\color{black}{\begin{equation}
\label{J_W}
J_W(\xi,t;k)=
\begin{cases}
\begin{pmatrix}1 & \frac{T^2(k)\e^{-2\i t g(k)}}{\widehat{f}(k)} \\ 0 & 1\end{pmatrix},\quad k\in L_7,
\qquad
\begin{pmatrix}1 & \frac{T^2(k)\e^{-2\i t g(k)}}{-\widehat{f}(k)} \\ 0 & 1\end{pmatrix},\quad k\in L_5,\\
\begin{pmatrix}1&0\\\frac{-\e^{2\i t g(k)}}{\ol{\widehat{f}(k)\(\ol{k}\)}T^2(k)} & 1\end{pmatrix},\quad k\in L_8,
\qquad
\begin{pmatrix}1&0\\\frac{\e^{2\i t g(k)}}{\ol{\widehat{f}(k)\(\ol{k}\)}T^2(k)} & 1\end{pmatrix},\quad k\in L_6,\\
\begin{pmatrix}0&\i\\\i&0\end{pmatrix}
\,\quad k\in(\i c_{\l},\i d)\cup (\i c_{\r},-\i c_{\r})\cup(-\i d, -\i c_{\l}),\\
e^{\i (t B(\xi)+\Delta(\xi))\sigma_3},\quad k\in(\i d,\i c_+)\cup (-\i c_{\r},-\i d), 
\end{cases}
\end{equation}
}
and for those values of   $\xi$   for which  $-\delta<\Im g(\kappa_{\ell},\xi)<\delta$  for some point $\kappa_\ell$ of the discrete spectrum
$$\mathrm{Res}_{\kappa_{\ell}}W(\xi,t;k)=\lim\limits_{k\to \kappa_{\ell}}W(\xi,t;k)\begin{pmatrix}0&0\\ \frac{ \i\nu_\ell}{T^2(k)} \e^{2\i t  g(k,\xi)} & 0\end{pmatrix},$$
$$\mathrm{Res}_{-\ol{\kappa}_{\ell}}W(\xi,t;k)=\lim\limits_{k\to -\ol{\kappa}_{\ell}}W(\xi,t;k)\begin{pmatrix}0&0\\\frac{ \i\,\widebar{\nu}_\ell }{T^2(k)} \e^{2\i t  g(k,\xi)} & 0\end{pmatrix},$$
$$\mathrm{Res}_{\ol{\kappa}_{\ell}}W(\xi,t;k)=\lim\limits_{k\to \ol{\kappa}_{\ell}}W(\xi,t;k)\begin{pmatrix}0&\i\,\widebar{\nu_\ell}T^2(k)\e^{-2\i t  g(k,\xi)}\\0&0\end{pmatrix},$$
$$\mathrm{Res}_{-\kappa_{\ell}}W(\xi,t;k)=\lim\limits_{k\to -\kappa_{\ell}}W(\xi,t;k)\begin{pmatrix}0&\i\nu_\ell T^2(k)\e^{-2\i t g(k,\xi)}\\0&0\end{pmatrix}.$$

Because of the signature of $\Im g(k)$ given in Figure~\ref{Fig_Elliptic}, we have that the  matrix $J_W$ in \eqref{J_W} is exponentially close to the identity on $L_5\cup L_7\backslash\{ U_{\i d}\cup U_{\i c_{\l}}\}$ and  on $L_6\cup L_8\backslash \{U_{-\i d}\cup U_{-\i c_{\l}}\}$,
while on the  contours $L_5\cup L_7\cap \{U_{\i d}\cup U_{\i c_{\l}}\}$ and  on $L_6\cup L_8\cap\{U_{-\i d}\cup U_{-\i c_{\l}}\}$ the matrix $J_W$ is close to the identity  but not uniformly. 
A detailed analysis of the error term arising in this case has been obtained in a similar setting for the MKdV equation with $c_+=0$ \cite[Theorem 2]{BM} and also for the KdV equation in \cite{GGM} where it was shown that the error is of order $\mathcal{O}(t^{-1})$.

We arrive at the model problem for  the matrix $P^{\infty}(k)$. 
{\color{black}
\begin{RHP}\label{RH_problem_W2}
 Find a $2\times 2$ matrix-valued  function $P^{\infty}(\xi,t;k)$ analytic in $k\in\C\backslash [\i c_-,-\i c_-]$ such that
\begin{enumerate}
\item $P^{\infty}_-(\xi,t;k)=P^{\infty}_+(\xi,t;k)J_{P^{\infty}}(\xi,t;k),$   for $ k\in (\i c_-,-\i c_-)$
with 
\begin{equation}
\label{RH_W2}
J_{P^{\infty}}(\xi,t;k)=\left\{
\begin{array}{ll}
\begin{pmatrix}0&\i\\\i&0\end{pmatrix},
\,\quad k\in(\i c_{-},\i d)\cup (\i c_+,-\i c_+)\cup(-\i d, -\i c_-),\\
 \e^{\i( t B(\xi)+ \Delta(\xi))\sigma_3},\quad k\in(\i d,\i c_+)\cup (-\i c_+,-\i d);
\end{array}\right.
\end{equation}
with $B(\xi)$ and $\Delta(\xi)$ defined in \eqref{B}   and \eqref{Delta} respectively;
\item  for those values of   $\xi$   for which  $-\delta<Im\, g(\kappa_\ell,\xi)<\delta$, $0<\delta\ll1$, 
 $$\mathrm{Res}_{\kappa_{\ell}}P^{\infty}(\xi,t;k)=\lim\limits_{k\to \kappa_{\ell}}P^{\infty}(\xi,t;k)\begin{pmatrix}0&0\\ \frac{ \i\nu_\ell}{T^2(k)} \e^{2\i t  g(k,\xi)} & 0\end{pmatrix},$$
$$\mathrm{Res}_{-\ol{\kappa}_{\ell}}P^{\infty}(\xi,t;k)=\lim\limits_{k\to -\ol{\kappa}_{\ell}}P^{\infty}(\xi,t;k)\begin{pmatrix}0&0\\\frac{ \i\,\widebar{\nu_\ell} }{T^2(k)} \e^{2\i t  g(k,\xi)} & 0\end{pmatrix},$$
$$\mathrm{Res}_{\ol{\kappa}_{\ell}}P^{\infty}(\xi,t;k)=\lim\limits_{k\to \ol{\kappa}_{\ell}}P^{\infty}(\xi,t;k)\begin{pmatrix}0&\i\,\widebar{\nu_\ell}T^2(k)\e^{-2\i t  g(k,\xi)}\\0&0\end{pmatrix},$$
$$\mathrm{Res}_{-\kappa_{\ell}}P^{\infty}(\xi,t;k)=\lim\limits_{k\to -\kappa_{\ell}}P^{\infty}(\xi,t;k)\begin{pmatrix}0&\i\nu_\ell T^2(k)\e^{-2\i t g(k,\xi)}\\0&0\end{pmatrix};$$
\item $P^{\infty}(\xi,t;k)\to I\quad \mbox{as}\;k\to\infty$;
\item $P^{\infty}(\xi,t;k)$  has at most fourth root singularities at the  points $\pm \i c_-$, $\pm \i c_+$ and $\pm \i d$.
\end{enumerate}
\end{RHP}
Then the quantity
\begin{equation}
\label{breather_ell}
q(x,t)=\lim_{k\to\infty}2ik P^{\infty}_{12}(\xi,t;k)=\lim_{k\to\infty}2ikP^{\infty}_{21}(\xi,t;k),\quad \mbox{ where } \xi=\frac{x}{12t},
\end{equation}
approximates the MKdV solution for large times.   The   solution of the RH  problem \ref{RH_problem_W2},  excluding the pole condition 2,
  has been considered in \cite{KM2} where it was solved in terms of hyperelliptic theta-functions defined on the Jacobi variety of the surface 
$\Gamma:=\{(k,y)\in\C^2\;|\; y^2=R^2(k)\}$.  In Theorem~\ref{theorem_periodic}  we showed that the MKdV solution  obtained from   the RH problem \ref{RH_problem_W2} without the pole condition 2  corresponds to the travelling wave solution 
\eqref{periodic_intro}. In particular we obtain  that for $V_{\ell}+\tilde{\delta}<\xi< \frac{c_{\l}^2}{3}+\frac{c_{\r}^2}{6}$ and $\frac{-c_{\l}^2}{2}+c_{\r}^2<\xi<V_{\ell}-\tilde{\delta}$ for some $\tilde{\delta}>0$,   namely for $\xi$ distinct from the 
breather speed $V_\ell$ we have 
\begin{equation}
\label{q_per_f}
q(x,t)=q_{per}(x,t;c_-,d,c_+,x_0)+O(t^{-1}),
\end{equation}
where   $q_{per}(x,t;\beta_1,\beta_2,\beta_3,x_0)$  is the travelling wave defined in \eqref{q_periodic1} and 
\begin{equation*}
x_0=-\dfrac{K(m)\Delta}{\pi}+K(m)
\end{equation*}
with $K(m)$ the complete elliptic integral   with modulus $m^2=\dfrac{d^2-c_+^2}{c_-^2-c_+^2}$ and $\Delta$ defined in \eqref{Delta}.
The speed  $V_\ell$  of the breather is obtained by solving the equation $\Im g(\kappa_\ell,V_{\ell})=0$ and  by \eqref{eq_gp} and \eqref{g_zero1} is 
\begin{equation}
\label{ellbreather_speed}
V_{\ell}=-\dfrac{\Im\int\limits_{-c_-^2}^{\kappa_\ell^2}\frac{s^2+\frac{s}{2}(c_-^2+c_+^2+d^2)+b_1}{\sqrt{s+c_-^2)(s+c_+^2)(s+d^2)}} ds} {\Im\int\limits_{-c_-^2}^{\kappa_\ell^2}\frac{s+b_0}{\sqrt{s+c_-^2)(s+c_+^2)(s+d^2)}} ds},
\end{equation}
where $b_0=c_{\l}^2-(c_{\l}^2-c_{\r}^2)\dfrac{E(m)}{K(m)}$ and 
\begin{equation*}
b_1=\dfrac{1}{3}(c_{\l}^2c_{\r}^2+c_{\l}^2d^2+c_{\r}^2d^2)-\frac{1}{6}(c_{\l}^2+c_{\r}^2+d^2)\left[c_{\l}^2-(c_{\l}^2-c_{\r}^2)\dfrac{E(m)}{K(m)}\right].
\end{equation*}

 When  $|V_{\ell}-\xi|<\tilde{\delta}$  the solution  of the RH problem~\ref{RH_problem_W2}  with the pole condition 2 describes a breather on the elliptic background that we indicate as $q_{be}(x,t)$.
The determination of the explicit expression of the breather on the elliptic background  is  similar to the derivation obtained in Section~\ref{sect_breath} for the breather on the constant background, but the computation is  algebraically involved and it is omitted.
}
The error term in \eqref{q_per_f} is  similar  to the  case  computed in \cite[Theorem 2]{BM} where all the local parametrices are the same. We add the derivation in the section below.
The error term is valid strictly inside the dispersive shock wave region. On the boundary 
of the dispersive shock wave region a different error term is present (see \cite[Section 3]{BM}) for details).
We have thus concluded the proof of  the Theorem~\ref{thrm:asymp:rl}  part (b)  leading order term.

\color{black}
\subsubsection{Estimate of the error term  in the dispersive shock wave region.}
This section is fully analogous to Section 3 in \cite{BM}.
We calculate the error term in the dispersive shock wave region away from the breathers.
The jump matrix $J_W$ has the main jump over the segment $[\i c_-,-\i c_-],$ and on the rest of the contour is exponentially close to the identity matrix except for small neighbourhoods of the points $\pm\i c_-$ and $\pm\i d.$ This means that we need to construct local solutions to the RH problem for the function $W$ near the points $\pm\i c_-, \pm\i d.$

To deal with the point $\i c_-,$ we `push up' the contours $L_7$ and $L_5$, so that they join at a point $\i (c_-+ \varepsilon)$ (with some small $\varepsilon>0$) instead of the point $\i c_-.$ In this way we completely remove the problem at the point $\i c_-$ (this corresponds to the trivial case of the `Bessel' parametrix.)

The point $\i d$ requires more careful considerations. First, we deform  the contour $L_1$  downward, so that it  intersects the segment $[\i c_-,-\i c_-]$ at a point $\i d-\i\varepsilon$, $\varepsilon>0$ instead of the point $\i d.$ 
Then, we construct an explicit exact solution of the $W$-RH problem inside the disk $|k-\i d|<\varepsilon$ in terms of the Airy parametrix. This will give us an error term of the order $t^{-1}.$ 
Below are the details.\\
%
%First of all, with abuse of notations, we denote the RH problem which is obtained after pushing the contours $L_1, L_7, L_5,$ and the symmetrical to them contours $L_2, L_8, L_6,$ by the same letter $W.$
%\medskip\medskip
% 
\noindent \textbf{Point $\boldsymbol{k=\i d(\xi)}.$} Due to the symmetry $k\to\ol k,$ it is enough to consider the point $\i d(\xi)$ only, and then the construction near the point $-\i d(\xi)$ will follow by symmetry.
Note that the phase function $g(k,\xi)=\int_{\i c_-}^k g'(s,\xi)ds$   with $g'(k,\xi)$ as in \eqref{Whitham_d}.
Since $g'(k,\xi)$ vanishes as $\sqrt{k-\i d}$ when  $k\to \i d$ it follows that near the point $k=\i d$    we have  $g(k,\xi)- \int_{\i c_-}^{\i d} g'(s,\xi)ds\sim(k-\i d)^{\frac32}$.
{\color{black}
For $\delta>0$ we define the neighbourhood
\begin{equation*}
{\mathcal U}_{\i d}=\{k\in\C, s.t.\; |k-\i d|<\delta\}
\end{equation*}
and ${\mathcal U}_{\i d}^{\pm}$ are the right and left neighbourhood of $\i d$ with respect to the imaginary axis.
Next we define the conformal map
\begin{equation*}
\zeta=\left(-\frac{3}{2}\i [tg(k)- t\int_{\i c_-}^{\i d} g'_{\pm}(s,\xi)ds]\right)^{\frac{2}{3}},\;\;\; s,k\in {\mathcal U}_{\i d}^{\pm}.
\end{equation*}
 }
%\begin{equation*}
%\begin{split}
%&g_{\pm}(k,\xi)=\mp\frac{B(\xi)}{2}+
%\frac{\i\, 8\sqrt{2\,}\,d^{\frac32}(d^2-\mu^2)}{\sqrt{c_-^2-d^2}\sqrt{d^2-c_+^2}}
%\(
%\frac{k-\i d}{-\i}\)^{\frac32}
%\times
%\\
% &\times
%\(
%1+\frac{3\(5d^6+3d^4(\mu^2-3c_+^2-3c_-^2)+d^2(13c_-^2c_+^2+(c_-^2+c_+^2)\mu^2)-5c_-^2c_+^2\mu^2\)}
%{20 d (c_-^2-d^2)(d^2-c_+^2)(d^2-\mu^2)}
%\(\frac{k-\i d}{-\i}\)
%+
%\right.
%\\
%&\hskip12cm\left.+\mathcal{O}\(\frac{k-\i d}{-\i}\)^2
%\),
%\end{split}
%\end{equation*}
%%
%where the sign $+(-)$ is taken according to the orientation of the contour of the RH problem, in other words as the point $k$ approaches the point $\i d(\xi)$ with $\Re k>0$ ($\Re k<0$), respectively. This suggests to introduce a new local coordinate $z$ in such a way that
%\begin{equation*}
%g(k,\xi)=:\mp\frac{B}{2}+\i z^{3/2},
%\quad
%t g(k,\xi)=\mp\frac{tB(\xi)}{2}+\i t z^{\frac32}
%=:
%\mp\frac{tB(\xi)}{2}+\i \frac23\zeta^{\frac32},
%\quad |k-\i d|< r,
%\end{equation*}
%where the orientation of $z$ is such that $z>0$ for $k\in(\i d, \i (d-r))$ and $z$ is negative for $k\in(\i (d+r), \i d),$
%\begin{equation*}
%z=\frac{4\sqrt[3]{2\,}\,d(d^2-\mu^2)^{\frac23}}{\sqrt[3]{c_-^2-d^2\,}\,\sqrt[3]{d^2-c_+^2}}
%\(\frac{k-\i d}{-\i}\)\(1+\mathcal{O}\(\frac{k-\i d}{-\i}\)\),\quad
%k\to\i d,
%\end{equation*}
%%
%and we also introduce a scaled version of $z$ by the formula $\zeta \equiv z t.$
%A parameter $r>0$ here is a sufficiently small number such that $|d(\xi)-c_-|>r,$ $|d(\xi)-c_+|>r$
%for all $-\frac{c_-^2}{2}+c_+^2 + \delta\leq \xi \leq \frac{c_-^2}{3}+\frac{c_+^2}{6}-\delta.$
Furthermore, we introduce the function
\begin{equation*}
\phi(k, \xi)=\begin{cases}
\frac{T(k, \xi)\e^{\frac{\i t B(\xi)}{2}}}{\sqrt{-\i \widehat f(k)}},\quad  k\in {\mathcal U}_{\i d}^{+},
\\\\
\frac{T(k, \xi)\e^{\frac{-\i t B(\xi)}{2}}}{\sqrt{\i \widehat f(k)}},\quad  k\in {\mathcal U}_{\i d}^{-},
\end{cases}
\end{equation*}
which has a jump across $k\in(\i (d+\delta),-\i (d-\delta)).$   We recall that $\widehat f(k)$ is defined in \eqref{hatf}, $B(\xi)$ is defined in \eqref{B}  and $T(k,\xi)$ is defined in \eqref{T_ell}. We also observe   that $-\i \widehat f_+(k)>0$ and $\i\widehat f_-(k)>0$ for $k\in(\i c_-, \i c_+).$
Then the jump matrix $J_{W}(k)=J_{W}(\xi,t;k)$ in \eqref{J_W} in $ {\mathcal U}_{\i d}$  takes the following form:
\begin{equation*}
\begin{split}
&
J_{W}(k)=\phi^{\sigma_3}
\begin{pmatrix}
1 & -\i\e^{\frac23\zeta^{\frac32}} \\ 0 & 1
\end{pmatrix}
\phi^{-\sigma_3}, \, k\in (L_7\cup L_5)\cap {\mathcal U}_{\i d},
\\
&
J_{W}(k)=\phi_+^{\sigma_3}
\begin{pmatrix}
1 & 0 \\ \i \e^{-\frac23\zeta^{\frac32}} & 1
\end{pmatrix}
\phi_-^{-\sigma_3}, \, k\in (\i d, \i (d-\delta)),\\
&
J_{W}(k)=\phi_+^{\sigma_3}
\begin{pmatrix}
0 & \i \\ \i & 0
\end{pmatrix}
\phi_-^{-\sigma_3},\quad k\in(\i (d+\delta),\i d).
\end{split}
\end{equation*}
\textbf{Local parametrix at } $\boldsymbol{\i d(\xi)}.$  In order to mimic the above jumps matrix, we introduce the matrix function $P_{Ai}(k)$ obtained via the Airy function.
\begin{equation}\label{PAi}
P_{Ai}(k):=\Psi_{Ai}(\zeta)\phi(k,\xi)^{-\sigma_3}\e^{-\frac23\zeta^{\frac32}\sigma_3},
\end{equation}
where $\Psi_{Ai}(\zeta)$ is defined as follows:
\begin{equation*}\Psi_{\Ai}(\zeta)
=\begin{cases}
\begin{pmatrix}
v_1(\zeta) & v_0(\zeta)
\\
v'_1(\zeta) & v'_0(\zeta)
\end{pmatrix},\arg\zeta\in(0,\frac23\pi),
\\
\begin{pmatrix}
v_1(\zeta) & -\i v_{-1}(\zeta)
\\
v'_1(\zeta) & -\i v'_{-1}(\zeta)
\end{pmatrix},\arg\zeta\in(\frac23\pi,\pi),
\qquad\ \
\\
\begin{pmatrix}
v_{-1}(\zeta) & \i v_{1}(\zeta)
\\
v'_{-1}(\zeta) & \i v'_{1}(\zeta)
\end{pmatrix},\arg\zeta\in(-\pi, -\frac23\pi),
\\
\begin{pmatrix}
v_{-1}(\zeta) & v_0(\zeta)
\\
v'_{-1}(\zeta) & v'_0(\zeta)
\end{pmatrix},\arg\zeta\in(-\frac23\pi, 0),
\end{cases}
\end{equation*}
where the functions
$v_0, v_1, v_{-1}$ are defined in terms of the Airy function $\Ai$ as 
\begin{equation*}
v_0(\zeta)=\sqrt{2\pi}\Ai(\zeta),\quad 
v_1(\zeta)=\e^{-\frac{\pi\i}{6}}\sqrt{2\pi}\Ai(\zeta\e^{-\frac{2\pi\i}{3}}),
\quad
v_{-1}(\zeta)=\e^{\frac{\pi\i}{6}}\sqrt{2\pi}\Ai(\zeta\e^{\frac{2\pi\i}{3}}),
\end{equation*}
and $'$ denotes the derivative in $\zeta.$ The well-known relation $\Ai(\zeta)+\e^{\frac{2\pi\i}{3}}\Ai(\zeta\e^{\frac{2\pi\i}{3}})+\e^{\frac{-2\pi\i}{3}}\Ai(\zeta\e^{\frac{-2\pi\i}{3}})=0$ reads as
\begin{equation*}v_0(\zeta)-\i v_1(\zeta) + \i v_{-1}(\zeta)=0.\end{equation*}
The latter allows us to verify that the function $\Psi_{\Ai}$ has the following jumps $\Psi_{\Ai,-}(\zeta)=\Psi_{\Ai,+}(\zeta)J_{\Ai}(\zeta)$ across the contour $\Sigma_{\Ai}=\mathbb{R}\cup(\e^{2\pi\i/3}\infty, 0)\cup(\e^{-2\pi\i/3}\infty, 0),$ which is oriented according to the order in which the points are mentioned, that is from $\e^{\pm2\pi\i/3}\infty$ to $0$ and from $-\infty$ to $+\infty,$
\begin{equation*}
\begin{split}
&J_{\Ai}(\zeta) = \begin{pmatrix}1 & 0 \\ \i & 1\end{pmatrix},\;\arg\zeta=0,\;\;
J_{\Ai}(\zeta) = \begin{pmatrix}1 & -\i \\ 0 & 1\end{pmatrix},\;\arg\zeta=\pm\frac{2\pi}{3},\;\;\\
&J_{\Ai}(\zeta) = \begin{pmatrix}0 & \i \\ \i & 0\end{pmatrix},\arg\zeta=\pi.
\end{split}
\end{equation*}
Besides, the function $\Psi_{\Ai}$ admits the following asymptotics for large $\zeta$, which is uniform with respect to $\arg\zeta\in[-\pi,\pi]$:
\begin{equation}\label{PsiAiAsymp}
\Psi_{\Ai}(\zeta)=\frac{1}{\sqrt{2}}\zeta^{-\sigma_3/4}\begin{pmatrix}1 & 1 \\ 1 & -1\end{pmatrix}
\(I+\begin{pmatrix}
\frac{-1}{48} & \frac{-1}{8} \\ \frac{1}{8} & \frac{1}{48}
\end{pmatrix}\zeta^{-\frac32} + \mathcal{O}(\zeta^{-3})
\)\e^{\frac23\zeta^{\frac32}\sigma_3},\quad \zeta\to\infty.
\end{equation}

\medskip\noindent
\textbf{Approximation of $\boldsymbol{W.}$}
Now we define a function $W_{appr}(\xi,t;k),$ which will be shown to be a good approximation of the function $W(\xi,t;k),$
\begin{equation*}
W_{appr}(\xi,t;k)
=\begin{cases}
P^{(\infty)}(k),\quad |k\mp \i d(\xi)|>\delta,
\\
\mathcal{B}(k)\ \Psi_{\Ai}(\zeta(k))\phi(k,\xi)^{-\sigma_3}\e^{-\frac23(\zeta(k))^{\frac32}\sigma_3},
\quad |k-\i d(\xi)|<\delta,
\\\\
\ol{W_{appr}(\xi, t; \ol k)},\quad |k+\i d(\xi)|<\delta.
\end{cases}
\end{equation*}
Note that here we defined $W_{appr}$ for $|k+\i d|<\delta$ in terms of $W_{appr}$ for $|k-\i d|<\delta.$
Furthermore, here $\mathcal{B}$ is an unknown matrix  function  analytic inside the disk $|k-\i d|<\delta$. It is obtained by requesting that 
the error $$E(k) = W(k) W_{appr}(k)^{-1}$$ has jump $J_E(k)=I+o(1)$ as $t\to\infty$ and $ k\in \partial{\mathcal U}_{\i d}$.
Hence, using \eqref{PsiAiAsymp}, we obtain  
\begin{equation*}\mathcal{B}(k)
=
P^{(\infty)}(k)\phi(k,\xi)^{\sigma_3}\frac{1}{\sqrt{2}}\begin{pmatrix}
1 & 1 \\ 1 & -1
\end{pmatrix}\zeta^{\sigma_3/4}.
\end{equation*}
Note that $\mathcal{B}(k)$ is indeed holomorphic in the disk ${\mathcal U}_{\i d}$. This can be verified using the relations $\phi_+\phi_-=1$ for $k\in(\i (d+\delta), \i d),$ and $\phi_+=\phi_-\e^{\i t B+\i \Delta}$ for $k\in(\i d, \i (d-\delta)).$

Using the asymptotics   \eqref{PsiAiAsymp}, we see that on the circle  $\partial {\mathcal U}_{\i d}$ the jump matrix for the error function $E_-(k)=E_+(k)J_E(k)$   is given by
\begin{equation*}J_{E}(k)=
P^{(\infty)}(k)\phi(k,\xi)^{\sigma_3}\frac{1}{\sqrt{2}}\begin{pmatrix}
1 & 1 \\ 1 & -1
\end{pmatrix}\zeta^{\sigma_3/4}
\Psi_{\Ai}(\zeta)
\phi(k)^{-\sigma_3}\e^{-\frac23\zeta^{\frac32}\sigma_3}
P^{(\infty)}(k)^{-1},
\end{equation*}
is indeed of the order $I+\mathcal{O}(\zeta^{-3/2})=I+\mathcal{O}(t^{-1})$  as $t\to \infty$, uniformly on the contour.
 Hence the error matrix $E(k)$  is also of the same order, namely $E(k)=I+\mathcal{O}(t^{-1}).$
 Therefore for large $|k|$ we have $W(k)= P^{(\infty)}(k)+W_{err}(k)$ where $W_{err}(k)$ is of the order $\mathcal{O}(t^{-1}).$
The solution $q(x,t)$ of the MKdV equation equals
\begin{equation*}q(x,t)=\lim\limits_{k\to\infty}(2\i kW(k))_{21}
=
\lim\limits_{k\to\infty}(2\i kP^{(\infty)}(k))_{21}+
\lim\limits_{k\to\infty}(2\i kW_{err}(k))_{21},\end{equation*}
where, in the r.h.s.,  the second term is of the order $\mathcal{O}(t^{-1}),$ while the first term is $q_{ell}(x,t).$
We thus obtain
\begin{equation*}q(x,t)=q_{ell}(x,t)+\mathcal{O}(t^{-1})\end{equation*}
for $\xi\equiv\frac{x}{12t}$ strictly inside the elliptic region, $-\frac{c_-^2}{2}+c_+^2+\delta\leq\xi<\frac{c_-^2}{3}+\frac{c_+^2}{6}-\delta.$

\color{black}

\section{Large time asymptotics: proof of Theorem~\ref{thrm:asymp:rl},  parts (a) and (c)}

Here we continue the proof of Theorem \ref{thrm:asymp:rl} in the soliton-breather region and in the breather region.

\subsection{Breather region: Theorem \ref{thrm:asymp:rl}  part (c)}
In the breather region the function  $g(k,\xi)$ takes the form  \cite[page 46]{KM2}
\begin{equation}
\label{g_breather}
g(k,\xi)=2\(2k^2-c_{-}^2+6\xi\right)\sqrt{k^2+c_{-}^2}\,,\quad \quad \xi<\frac{-c_{-}^2}{2}+c_{+}^2,
\end{equation}
namely $g(k,\xi)$ is analytic in $\mathbb{C}\backslash [\i c_-,-\i c_-]$  and
\begin{equation}
\label{g_breather1}
g_+(k,\xi)+g_-(k,\xi)=0,\quad k\in(\i c_-,-\i c_-),\quad g(k,\xi)=\hat{\theta}(k,\xi)+O(k^{-1})\quad \mbox{as $ |k|\to\infty$}.
\end{equation}
The above two conditions defined $g(k,\xi)$ uniquely.  We chose $\sqrt{k^2+c_{-}^2}$ to be real on $(\i c_-,-\i c_-)$ and positive for $k=+0$.

We observe that $\Im g(k,\xi)=0$ for $k$ on the segment $[\i c_-,-\i c_-]$ and on the real line. 
For a given $\xi$, there is an extra   couple of  symmetric curves such that   $\Im g(k,\xi)=0$   (see  Figure~\ref{Fig_11}).  For  $\xi<-\frac{c_-^2}{2}$ this   couple of  curves  
 crosses the real  line  at the points $k=\pm k_0=\pm \sqrt{-\frac{c_-^2}{2}-\xi}$   (\color{black}{utmost} left constant region). For   $-\frac{c_-^2}{2}<\xi<-\frac{c_{-}^2}{2}+c_{+}^2$  the level lines
 $\Im g(k,\xi)=0$ 
are  a couple of  symmetric curves that  cross  the imaginary axis at the points $k=\pm\i d_0=\pm \i\sqrt{\xi+\frac{c_-^2}{2}}$   (\color{black}{middle} left  constant region).
The points $\pm k_0$ and $\pm \i d_0$  are the zeros of   the equation
\begin{equation*}
g'(k,\xi)=6\left(2k^2+c_{-}^2+2\xi\right)\frac{k}{\sqrt{k^2+c_{-}^2}}=0\,.
\end{equation*}
\begin{figure}[ht]
\begin{minipage}{0.48\linewidth}
\begin{tikzpicture}
\draw[fill=black] (0,3) circle [radius=0.05];
\draw[fill=black] (0,-3) circle [radius=0.05];
\draw[fill=black] (0,1.5) circle [radius=0.05];
\draw[fill=black] (0,-1.5) circle [radius=0.05];
\draw[] (0,1) circle [radius=0.05];
\draw[] (0,-1) circle [radius=0.05];
\draw[thick,dashed] (0,3) to (0,-3);
%points $i c_l, i c_r
\node at (-0.4,2.9) {$\i c_{\l}$};\node at (-0.35,1.5) {$\i c_{\r}$};
\node at (-0.4,-3.1) {$-\i c_{\l}$};\node at (-0.35,-1.5) {$-\i c_{\r}$};
\draw[thick,dashed] (-3,3.2) [out=-60, in=180] to (0,1) [out=0,in=-120] to (3,3.2);
\draw[thick,dashed] (-3,-3.2) [out=60, in=180] to (0,-1) [out=0,in=120] to (3,-3.2);
\node at (1.3,0.7) {$\i d_0=\i\sqrt{\frac{c_-^2}{2}+\xi}$};
\node at (0.4,-0.7) {$-\i d_0$};
%pluses, minuses
\node at (2.2,1.4) {\color{black}$+$};\node at (-2,1) {\color{black}$+$};
\node at (2,-1) {\color{black}$-$};\node at (-2,-1) {\color{black}$-$};
\node at (1.2,2.7) {\color{black}$-$};\node at (-1.2,2.7) {\color{black}$-$};
\node at (1.2,-2.7) {\color{black}$+$};\node at (-1.2,-2.7) {\color{black}$+$};
%real line
\draw[thick, dashed] (-3,0) to (3,0);
\end{tikzpicture}\\
(a) $\frac{-c_{\l}^2}{2}<\xi<\frac{-c_{\l}^2}{2}+c_{\r}^2$ (\color{black}{Middle} left constant region)
\end{minipage}
\hfill
\begin{minipage}{0.49\linewidth}
\begin{tikzpicture}
\draw[fill=black] (0,3) circle [radius=0.05];
\draw[fill=black] (0,-3) circle [radius=0.05];
\draw[fill=black] (0,1.5) circle [radius=0.05];
\draw[fill=black] (0,-1.5) circle [radius=0.05];
\draw[] (1,0) circle [radius=0.05];
\draw[] (-1,0) circle [radius=0.05];
%segment [i c_l,-i c_l]
\draw[thick,dashed] (0,3) to (0,-3);
%points $i c_l, i c_r
\node at (-0.4,2.9) {$\i c_{\l}$};\node at (-0.35,1.23) {$\i c_{\r}$};
\node at (-0.4,-3.1) {$-\i c_{\l}$};\node at (-0.35,-1.23) {$-\i c_{\r}$};
\draw[thick,dashed] (-3,3) [out=-60, in=90] to (-1,0) [out=-90,in=60] to (-3,-3);
\draw[thick,dashed] (3,3) [out=-120, in=90] to (1,0) [out=-90,in=120] to (3,-3);
\node at (1.9,-0.35) {$k_0=\sqrt{-\xi-\frac{c_{\l}^2}{2}}$};
\node at (-1.4,-0.3) {$-k_0$};
%pluses, minuses
\node at (2,1) {\color{black}$+$};\node at (-2,1) {\color{black}$+$};
\node at (2,-1) {\color{black}$-$};\node at (-2,-1) {\color{black}$-$};
\node at (1.2,2.7) {\color{black}$-$};\node at (-1.2,2.7) {\color{black}$-$};
\node at (1.2,-2.7) {\color{black}$+$};\node at (-1.2,-2.7) {\color{black}$+$};
%real line
\draw[thick, dashed] (-3,0) to (3,0);
\end{tikzpicture}\\
(b) $\xi<\frac{-c_{\l}^2}{2}$ (\color{black}{Utmost} left constant region)
\end{minipage}
\label{Fig_11}
\caption{Distribution of signs of $\Im g(k,\xi)$ in the breather region.}
\end{figure}
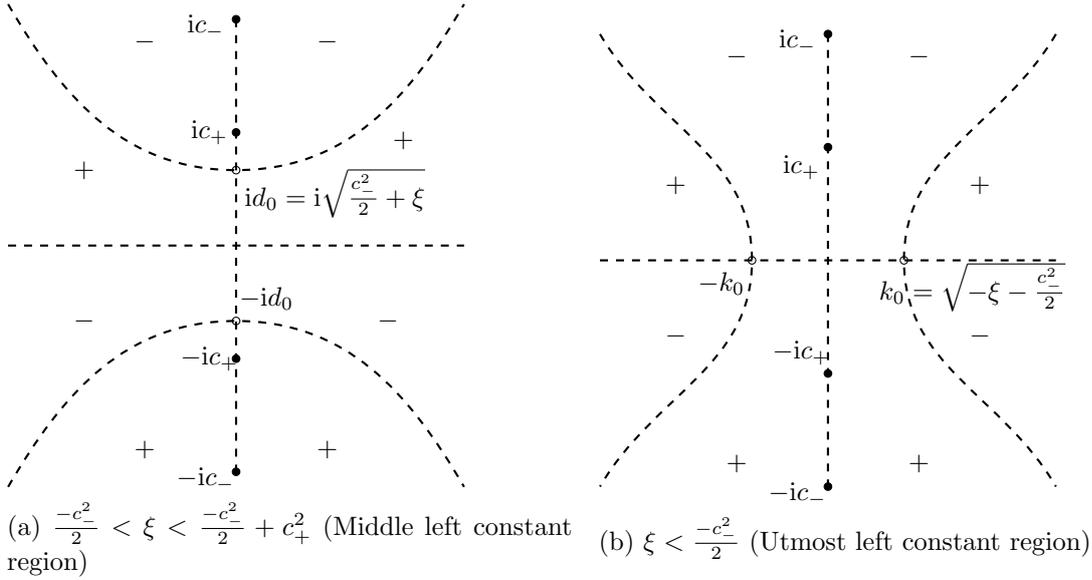

The asymptotic analysis in these two cases is slightly different. We start our analysis  with the middle left constant region.
\subsubsection{\color{black}{Middle} left constant region $\frac{-c_{\l}^2}{2}<\frac{x}{12t}\equiv\xi<\frac{-c_{\l}^2}{2}+c_{\r}^2$}
\label{middle_lcr}

We define  the first transformation of the RH problem  \ref{RH_problem_1}, $M(k)\to Y(k)$,  defined as 

$$Y (k)=M(k)\e^{\i\(tg(k,\xi)-\theta(x,t;k)\)\sigma_3}T^{-\sigma_3}(k,\xi),
$$
where now $g(k,\xi)$ is as in \eqref{g_breather}  and $ T(k)=T(k,\xi)$ is defined  for a given $\xi$ as follows:
\begin{equation}\label{Ttilde} 
\begin{split}
T(k,\xi)&=\widetilde T(k,\xi)H(k,\xi),\\
  \widetilde T(k,\xi) &= \prod\limits_{\begin{array}{ccc}\Re \kappa_j>0,\Im \kappa_j>0\\\Im g(\kappa_j,\xi)<0\end{array}} \frac{k-\ol{\kappa_j}}{k-\kappa_j}\frac{k+{\kappa_j}}{k+\ol{\kappa_j}} \cdot
\prod\limits_{\begin{array}{ccc}\Re \kappa_j=0,\Im\kappa_j>0\\\Im g(\kappa_j,\xi)<0\end{array}} \frac{k-\ol{\kappa_j}}{k-\kappa_j},
\end{split}
\end{equation}
and $H(k,\xi)$  is supposed to be analytic in $\mathbb{C}\backslash [\i c_-,-\i c_-]$ and it will be determined later. 
%$$T=\ldots$$

%
% For the values of  $\kappa_j$ such that 
%  $\Im g(\kappa_j)<0,$ $\Im \kappa_j>0,$  we have
%  \begin{equation*}
%\begin{split}
%&Y_2(k) - \frac{(k-\kappa_j)T^{2}(k)}{\nu_j\e^{2\i t g(k)}} Y_1(k) = \mathcal{O}(1), k \to \kappa_j,
%\quad 
%Y_2(k) + \frac{(k+\ol{\kappa_j})T^2(k)}{\ol{\nu_j}\e^{2\i t g(k)}} Y_1(k) = \mathcal{O}(1), k \to -\ol{\kappa_j},
%\\
%&Y_1(k) + \frac{(k-\ol{\kappa_j})\e^{2\i t g(k)}}{\ol{\nu_j}\ T^2(k)} Y_2(k) = \mathcal{O}(1), k \to \ol{\kappa_j},
%\qquad
%Y_1(k) - \frac{(k+{\kappa_j})\e^{2\i t g(k)}}{{\nu_jT^2(k)}} Y_2(k) = \mathcal{O}(1), k \to -{\kappa_j},
%\\
%&Y_1(k) = \mathcal{O}(1), k\to \kappa_j, k\to-\ol{\kappa_j}, \quad
%Y_2(k) = \mathcal{O}(1), k\to \ol{\kappa_j}, k\to-{\kappa_j},
%\end{split}
%\end{equation*}
%where $Y_1$ and $Y_2$ are the first and second columns of the matrix $Y$.
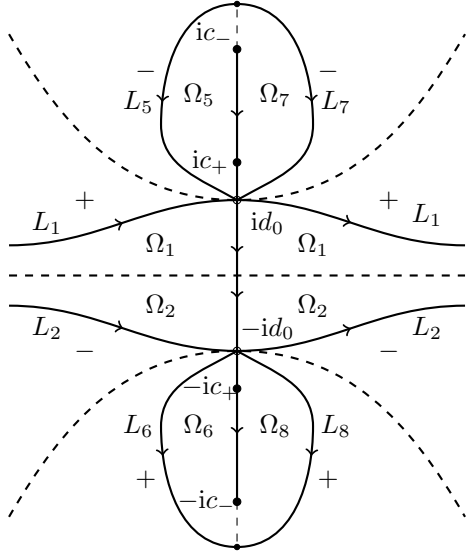
\begin{figure}
\centering
\begin{tikzpicture}
\draw[fill=black] (0,3) circle [radius=0.05];
\draw[fill=black] (0,-3) circle [radius=0.05];
\draw[fill=black] (0,1.5) circle [radius=0.05];
\draw[fill=black] (0,-1.5) circle [radius=0.05];
\draw[] (0,1) circle [radius=0.05];
\draw[] (0,-1) circle [radius=0.05];
\draw[thick, postaction = decorate, decoration = {markings, mark = at position 0.5 with {\arrow{>}}}](-3,0.4) to [out = 0, in =180] (0,1);
\draw[thick, postaction = decorate, decoration = {markings, mark = at position 0.5 with {\arrow{>}}}](0,1) to [out = 0, in =180] (3,0.4);
\draw[thick, postaction = decorate, decoration = {markings, mark = at position 0.5 with {\arrow{>}}}](-3,-0.4) to [out = 0, in =180] (0,-1);
\draw[thick, postaction = decorate, decoration = {markings, mark = at position 0.5 with {\arrow{>}}}](0,-1) to [out = 0, in =180] (3,-0.4);
\draw[thick, postaction = decorate, decoration = {markings, mark = at position 0.5 with {\arrow{>}}}](0,3.6) to [out=0, in=90] (1,2) [out=-90,in=30] to (0,1);
\draw[thick, postaction = decorate, decoration = {markings, mark = at position 0.5 with {\arrow{>}}}](0,3.6) to [out=180, in=90] (-1,2) [out=-90,in=150] to (0,1);
\draw[thick, postaction = decorate, decoration = {markings, mark = at position 0.5 with {\arrow{<}}}](0,-3.6) to [out=0, in=-90] (1,-2) to[out=90,in=-30](0,-1);
\draw[thick, postaction = decorate, decoration = {markings, mark = at position 0.5 with {\arrow{<}}}](0,-3.6) to [out=180, in=-90] (-1,-2)to[out=90, in=-150] (0,-1);%segment [i c_l,-i c_l]
\draw[dashed](0,3.) to (0,3.6);
\draw[dashed](0,-3.) to (0,-3.6);
\node at (2.5,0.8) {$L_1$};\node at (-2.5,0.7) {$L_1$};\node at (2.5,-0.7) {$L_2$};\node at (-2.5,-0.7) {$L_2$};
\node at (1.3,2.3) {$L_7$}; \node at (-1.3, 2.3) {$L_5$};\node at (1.3,-2.0) {$L_8$}; \node at (-1.3, -2.0) {$L_6$};
\node at (1,0.4) {$\Omega_1$};\node at (-1,0.4) {$\Omega_1$};
\node at (1,-0.4) {$\Omega_2$};\node at (-1,-0.4) {$\Omega_2$};
\node at (0.5,2.4) {$\Omega_7$}; \node at (-0.5, 2.4) {$\Omega_5$};\node at (0.5,-2.0) {$\Omega_8$}; \node at (-0.5, -2.0) {$\Omega_6$};
\draw[thick, 
decoration={markings, mark = at position 0.15 with \arrow{>}},
decoration={markings, mark = at position 0.85 with \arrow{>}},
decoration={markings, mark = at position 0.45 with \arrow{>}},
decoration={markings, mark = at position 0.55 with \arrow{>}}, postaction={decorate}] (0,3) to (0,-3);
%points $i c_l, i c_r
\node at (-0.3,3.2) {$\i c_{\l}$};\node at (-0.35,1.5) {$\i c_{\r}$};
\node at (-0.4,-3.) {$-\i c_{\l}$};\node at (-0.35,-1.5) {$-\i c_{\r}$};
\draw[thick,dashed] (-3,3.2) [out=-60, in=180] to (0,1) [out=0,in=-120] to (3,3.2);
\draw[thick,dashed] (-3,-3.2) [out=60, in=180] to (0,-1) [out=0,in=120] to (3,-3.2);
\node at (0.4,0.7) {$\i d_0$};
\node at (0.4,-0.7) {$-\i d_0$};
%pluses, minuses
\node at (2,1) {\color{black}$+$};\node at (-2,1) {\color{black}$+$};
\node at (2,-1) {\color{black}$-$};\node at (-2,-1) {\color{black}$-$};
\node at (1.2,2.7) {\color{black}$-$};\node at (-1.2,2.7) {\color{black}$-$};
\node at (1.2,-2.7) {\color{black}$+$};\node at (-1.2,-2.7) {\color{black}$+$};
%real line
\draw[thick, dashed] (-3,0) to (3,0);
\draw[fill] (0,3.6) circle (1pt);
\draw[fill] (0,-3.6) circle (1pt);
\end{tikzpicture}
\caption{The opening of the lenses in the \color{black}{middle} left constant region. \color{black}{The contours $L_5, L_7$ are separated by their highest point, and the contours $L_6, L_8$ are separated by their lowest point.}}
\label{breather_lens1}
\end{figure}

In order to perform our analysis on the continuum spectrum we assume
 the condition \eqref{g_breather1}  for $g(k)$ and 
\begin{equation}
\label{T_conj_2nd_constant}T_+(k)T_-(k)=
\begin{cases}\frac{1}{|a(k)|^2},\quad k\in(\i c_{\l}, \i d_0), \\ |a(k)|^2,\quad k\in(-\i d_0,-\i c_{\l}), \\ 1,\quad k\in(\i d_0,-\i d_0)\,. \end{cases}
\end{equation}
Then the  function $T(k)$ can be found  taking the logarithm  of the above expression  and applying Plemelj formula so that we obtain
%\begin{equation*}
%\begin{split}
%&T(k,\xi)=\widetilde T(k,\xi) 
%\exp\left[\frac{\sqrt{k^2+c_{\l}^2}}{2\pi\i}\left\{\int\limits_{\i c_{\l}}^{\i d_0}\frac{\(-\ln|a(s)|^2-\ln\(\widetilde T^2(s,\xi)\)\)\d s}{(s-k)\left(\sqrt{s^2+c_{\l}^2}\right)_+}\right\}\right]\times
%\\
%&
%\times\exp\left[\frac{\sqrt{k^2+c_{\l}^2}}{2\pi\i}\left\{\int\limits^{-\i c_{\l}}_{-\i d_0}\frac{\(\ln|a(s)|^2-\ln\(\widetilde T^2(s,\xi)\)\)\d s}{(s-k)\left(\sqrt{s^2+c_{\l}^2}\right)_+}
%+
%\int\limits_{\i d_0}^{-\i d_0}\frac{\(-\ln\(\widetilde T^2(s,\xi)\)\)\d s}{(s-k)\left(\sqrt{s^2+c_{\l}^2}\right)_+}
%\right\}\right],
%\end{split}\end{equation*}
%which can be written in the compact form
\begin{equation}\label{T_cl2}\begin{split}
T(k,\xi)=\widetilde T(k,\xi) 
\exp&\left[\frac{\sqrt{k^2+c_{\l}^2}}{2\pi\i}\left\{\int\limits^{-\i c_{\l}}_{-\i d_0}-\int\limits_{\i c_{\l}}^{\i d_0}\frac{\ln|a(s)|^2\d s}{(s-k)\left(\sqrt{s^2+c_{\l}^2}\right)_+}-\right.\right.\\
&\left.\left.-
\int\limits_{\i c_-}^{-\i c_-}\frac{\(\ln\widetilde T^2(s,\xi)\)\d s}{(s-k)\left(\sqrt{s^2+c_{\l}^2}\right)_+}
\right\}\right]\,.
\end{split}\end{equation}

%Here $$d_0=\sqrt{\xi+\frac{c_{\l}^2}{2}}.$$
 Let us notice, that $T(k,\xi)=1+\mathcal{O}\(\frac{1}{k}\),$ as $ k\to\infty.$

With our choice of $g(k)$ and $T(k)$  as in \eqref{g_breather} and \eqref{T_cl2},  respectively,  we consider the matrix $\widetilde{Y}(k)$ defined in \eqref{Y_tilde}. Then it is clear that the jumps  $J_{\tilde{Y}}$  as in 
\eqref{J_tilde} on the circles 
$C_j$, $\overline{C}_j$, $-C_j$ and $-\overline{C}_j$ are exponentially small
 for those values of $\kappa_j$ for which  $ \Im g(\kappa_j,\xi)>\delta$ or  $\Im g(\kappa_j)<-\delta$, $\Re\kappa_j\geq 0$ and $\Im \kappa_j>0$.
     The values of $\kappa_j$   and $\xi$  for which $-\delta<\Im g(\kappa_j,\xi)<\delta$  will be considered later.  Now we take care of the continuous spectrum.
      The jump matrix $J_{\tilde{Y}}(\xi,t;k)$   relative to the RH problem for $\widetilde{Y} (\xi,t;k)$, $\widetilde{Y}_-(\xi,t;k)=\tilde{Y}_+(\xi,t;k)J_{\tilde{Y}}(\xi,t;k)$  admits, for $k\neq \pm\i c_+$,  a  factorization   of the form 
\begin{equation}\label{fac_12}
J_{\tilde{Y}}(\xi,t;k)=
\begin{cases}
 \begin{pmatrix}1 & \frac{T_+^2(k)\e^{-2\i t g_+(k)}}{\widehat{f}_+(k)} \\ 0 & 1\end{pmatrix}
\begin{pmatrix}0&\i\\\i&0\end{pmatrix}
\begin{pmatrix}1 & \frac{-T_-^2(k)\e^{-2\i t g_-(k)}}{\widehat{f}_-(k)} \\ 0 & 1\end{pmatrix}, k\in(\i c_{\l},\i d_0)
\\
 \begin{pmatrix}1&0\\\frac{-\e^{2\i t g_+(k)}}{\ol{\widehat{f}_+\(\ol{k}\)}T_+^2(k)} & 1\end{pmatrix}\begin{pmatrix}0&\i\\\i&0\end{pmatrix}
\begin{pmatrix}1&0\\\frac{\e^{2\i t g_-(k)}}{\ol{\widehat{f}_-\(\ol{k}\)}T_-^2(k)} & 1\end{pmatrix},
\quad k\in(-\i d_0,-\i c_{\l})
\\
\begin{pmatrix}1 & 0 \\ \frac{-r(k)\e^{2\i t g(k,\xi)}}{T^2(k)} & 1\end{pmatrix}
\begin{pmatrix}1 & -\ol{r\(\ol{k}\)}T^2(k)\e^{-2\i t g(k,\xi)} \\ 0 & 1\end{pmatrix},\quad k\in\mathbb{R}\backslash \{0\}
\\
\begin{pmatrix}\frac{\i r_-(k)T_+(k)\e^{\i t (g_--g_+)}}{T_-(k)} & \i\, T_+(k)T_-(k)\e^{-\i t (g_-+g_+)}
\\
\frac{f\e^{\i t(g_-+g_+)}}{T_+(k)T_-(k)} & \frac{-\i r_+(k)T_-(k)\e^{\i t(g_+-g_-)}}{T_+(k)} \end{pmatrix},\quad k\in (\i d_0,0)
\\
\begin{pmatrix}\frac{\i \ol{r_+\(\ol k\)}T_+(k)\e^{\i t (g_--g_+)}}{T_-(k)} & -\ol{f\(\ol{k}\)} T_+T_-\e^{-\i t(g_-+g_+)}
\\
\frac{\i \e^{\i t (g_-+g_+)}}{T_+(k)T_-(k)} & \frac{-\i\ \ol{r_-\(\ol k\)}T_-(k)\e^{\i t(g_+-g_-)}}{T_+(k)}\end{pmatrix},\;
 k\in(0,-i d_0)
\end{cases}
\end{equation}
Here $$\widehat f(k) = \frac{-1}{a(k)\ol{b\(\ol{k}\)}}=\frac{1+r(k)\ol{r\(\ol{k}\)}}{-\ol{r\(\ol{k}\)}}.$$

Now we open the lenses as in Figure~\ref{breather_lens1} where we assume that we can deform all the curves $L_j$ so that there are no poles within the regions
$\Omega_j$.
We define a new matrix $X$ as
\begin{equation}
X(\xi,t;k)=\begin{cases}
\tilde{Y}(\xi,t;k)\begin{pmatrix}1 & 0 \\ \frac{-r(k)\e^{2\i t g(k,\xi)}}{T^2(k)} & 1\end{pmatrix},\quad k\in\Omega_1,\\
\tilde{Y}(\xi,t;k)\begin{pmatrix}1 & -\ol{r\(\ol{k}\)}T^2(k)\e^{-2\i t g(k,\xi)} \\ 0 & 1\end{pmatrix}^{-1},\quad k\in\Omega_2,\\
\tilde{Y}(\xi,t;k) \begin{pmatrix}1&0\\\frac{-\e^{2\i t g_+(k)}}{\ol{\widehat{f}_+\(\ol{k}\)}T_+^2(k)} & 1\end{pmatrix},\quad k\in\Omega_8,\\
\tilde{Y}(\xi,t;k)\begin{pmatrix}1&0\\\frac{\e^{2\i t g_-(k)}}{\ol{\widehat{f}_-\(\ol{k}\)}T_-^2(k)} & 1\end{pmatrix}^{-1},\quad k\in\Omega_6,\\
\tilde{Y}(\xi,t;k) \begin{pmatrix}1 & \frac{T_+^2(k)\e^{-2\i t g_+(k)}}{\widehat{f}_+(k)} \\ 0 & 1\end{pmatrix},\quad k\in\Omega_7,\\
\tilde{Y}(\xi,t;k)
\begin{pmatrix}1 & \frac{-T_-^2(k)\e^{-2\i t g_-(k)}}{\widehat{f}_-(k)} \\ 0 & 1\end{pmatrix}^{-1},\quad k\in\Omega_5.\\
\end{cases}
\end{equation}
With the above transformation the jump matrix on $(\i d_0, 0)$ transforms to 
\begin{equation}\label{fac_3}
\begin{split}
&
\begin{pmatrix}1 & 0 \\ \frac{r_+\e^{2\i t g_+}}{T_+^2} & 1\end{pmatrix}
\begin{pmatrix}\frac{\i r_-T_+\e^{\i t (g_--g_+)}}{T_-} & \i\, T_+T_-\e^{-\i t (g_-+g_+)}
\\
\frac{f\e^{\i t(g_-+g_+)}}{T_+T_-} & \frac{-\i r_+T_-(k)\e^{\i t(g_+-g_-)}}{T_+} \end{pmatrix}
\begin{pmatrix}1 & 0 \\ \frac{-r_-\e^{2\i t g_-}}{T_-^2} & 1\end{pmatrix}
=
\\
&
=\begin{pmatrix}0 & \i \\\i & 0\end{pmatrix}, k\in(\i d_0,0),
\end{split}
\end{equation}
and  similarly on $(0,-\i d_0)$.
%\begin{equation}
%\begin{split}
%&
%\begin{pmatrix}1 & -\ol{r_+\(\ol{k}\)}T_+^2(k)\e^{-2\i t g_+(k,\xi)} \\ 0 & 1\end{pmatrix}
%\begin{pmatrix}\frac{\i \ol{r_+\(\ol k\)}T_+\e^{\i t (g_--g_+)}}{T_-} & -\ol{f\(\ol{k}\)} T_+T_-\e^{-\i t(g_-+g_+)}
%\\
%\frac{\i \e^{\i t (g_-+g_+)}}{T_+T_-} & \frac{-\i\ \ol{r_-\(\ol k\)}T_-\e^{\i t(g_+-g_-)}}{T_+}\end{pmatrix}
%\cdot
%\\
%&\cdot\begin{pmatrix}1 & \ol{r_-\(\ol{k}\)}T_-^2(k)\e^{-2\i t g_-(k,\xi)} \\ 0 & 1\end{pmatrix}
%=\begin{pmatrix}0&\i\\\i&0\end{pmatrix}, k\in(0,-\i d_0).
%\end{split}
%\label{fac_4}
%\end{equation}
%
%
%\
%%
%
%After opening lenses around the real line and the intervals $(\i c_{\l}, \i d_0)\cup(-\i d_0, -\i c_{\l})$, which is realized by passing from $M^{(1)}$ to $X,$ 
We obtain that the matrix $X(k)$ solves the following RH problem.
\begin{RHP}\label{RHP_final_constant2}
\textbf{Final RH problem for the \color{black}{middle} left constant region $\frac{-c_{\l}^2}{2}<\xi<~\frac{-c_{\l}^2}{2}+~c_{\r}^2.$ }
\begin{enumerate}
\item Find a $2\times 2$ matrix $X(k)$ meromorphic in $\mathbb{C}\backslash  \Sigma$ with \\  $\Sigma= \{\cup_{j=1}^7 L_j\cup(\i c_-,-\i c_-)\}$  (see Figure \ref{breather_lens1})   such that 
\item Jumps:
$$
X_-(\xi,t;k)=X_+(\xi,t;k)J_X(\xi,t;k),\quad k\in\Sigma
$$
and 

$$J_{X}(\xi,t;k)=\begin{pmatrix}0&\i\\\i&0\end{pmatrix},\quad k\in(\i c_{\l}, -\i c_{\l}),$$

$$J_{X}(\xi,t;k)=\begin{cases}
\begin{pmatrix}1 & 0 \\ \frac{-r(k)\ \e^{2\i t g(k)}}{T(k)^2}&1\end{pmatrix}, k\in L_1,\\
\begin{pmatrix}1 & -\ol{r\(\ol k\)} T(k)^2 \e^{-2\i t g(k)}\\0&1\end{pmatrix}, k\in L_2, \\
\begin{pmatrix}1 & \frac{T(k)^2\e^{-2\i t g(k)}}{\widehat f(k)} \\ 0 & 1 \end{pmatrix}, k\in L_7,\quad
\begin{pmatrix}1 & \frac{-T^2(k)\e^{-2\i t g}}{\widehat f(k)} \\ 0 & 1 \end{pmatrix}, k\in L_5,\\
\begin{pmatrix}1 & 0 \\\frac{-\e^{2\i t g(k)}}{\ol{\widehat f\(\ol{k}\)}\ T(k)^2} & 1 \end{pmatrix}, k\in L_8,
\quad
 \begin{pmatrix}1 & 0 \\\frac{\e^{2\i t g(k)}}{\ol{\widehat f\(\ol{k}\)}\ T(k)^2} & 1 \end{pmatrix}, k\in L_6.
\end{cases}$$
\item Poles: for those $\kappa_j$ with $-\delta<\Im g(\kappa_j,\xi)<\delta,\;\Im \kappa_j>0,\,\Re\kappa_j\geq 0:$
{\color{black}
$$\mathrm{Res}_{\kappa_j}X(\xi,t;k)=\lim\limits_{k\to \kappa_j}X(\xi,t;k)\begin{pmatrix}0&0\\ \frac{ \i\nu_j}{T^2(k)} \e^{2\i t  g(k,\xi)} & 0\end{pmatrix},$$
$$\mathrm{Res}_{-\ol{\kappa}_j}X(\xi,t;k)=\lim\limits_{k\to -\ol{\kappa}_j}X(\xi,t;k)\begin{pmatrix}0&0\\\frac{ \i\,\widebar{\nu}_j }{T^2(k)} \e^{2\i t  g(k,\xi)} & 0\end{pmatrix},$$
$$\mathrm{Res}_{\ol{\kappa}_j}X(\xi,t;k)=\lim\limits_{k\to \ol{\kappa}_j}X(\xi,t;k)\begin{pmatrix}0&\i\,\widebar{\nu}_jT^2(k)\e^{-2\i t  g(k,\xi)}\\0&0\end{pmatrix},$$
$$\mathrm{Res}_{-\kappa_j}X(\xi,t;k)=\lim\limits_{k\to -\kappa_j}X(\xi,t;k)\begin{pmatrix}0&\i\nu_j T^2(k)\e^{-2\i t g(k,\xi)}\\0&0\end{pmatrix}.$$
}
%\begin{equation*}\begin{split}
%&X_1(k) - \frac{\i\nu_j\ \e^{2\i t g(k)}}{(k-\kappa_j)\ T^2(k)}X_2(k) = \mathcal{O}(1), \ X_2(k)=\mathcal{O}(1),\ k\to \kappa_j,\\&
%X_1(k) - \frac{\i\,\ol{\nu_j}\ \e^{2\i t g(k)}}{(k+\ol{\kappa_j})\ T^2(k)}X_2(k) = \mathcal{O}(1), \  X_2(k)=\mathcal{O}(1),\
%k\to -\ol{\kappa_j},
%\\
%&X_2(k) - \frac{\i\,\ol{\nu_j}\ T^2(k) }{(k-\ol{\kappa_j})\ \e^{2\i t g(k, \xi)}}X_1(k) = \mathcal{O}(1), \ X_1(k)=\mathcal{O}(1),\ k\to \ol{\kappa_j},\\&
%X_2(k) - \frac{\i{\nu_j}\ T^2(k)}{(k+{\kappa_j})\ \e^{2\i t g(k)}}X_1(k) = \mathcal{O}(1), \  X_1(k)=\mathcal{O}(1),\ 
%k\to -{\kappa_j},
%\end{split}\end{equation*}
%where $X_1(k)$ and $X_2(k)$ are the columns of the matrix $X(k)$ and 

%for those $\kappa_j$ with $\Im g(\kappa_j)<-\varepsilon,\ \Im \kappa_j>0:$
%\color{black}{\begin{equation*}\begin{split}
%&X_2(k) - \left[\frac{(k-\kappa_j)\ T^2(k)}{\i \nu_j\ \e^{2\i t g(k)}}\right] X_1 = \mathcal{O}(1), \ X_1(k)=\mathcal{O}(1),\ k\to \kappa_j,\\&
%X_2(k) - \left[\frac{(k+\ol{\kappa_j})\ T^2(k)}{\i\,\ol{\nu_j}\ \e^{2\i t g(k)}}\right]X_1 = \mathcal{O}(1), \  X_1(k)=\mathcal{O}(1),\ 
%k\to -\ol{\kappa_j},
%\\
%&X_1(k) - \left[\frac{(k-\ol{\kappa_j})\ \e^{2\i t g(k)}}{\i\,\ol{\nu_j}\ T^2(k)}\right]X_2(k) = \mathcal{O}(1), \ X_2(k)=\mathcal{O}(1),\ k\to \ol{\kappa_j},\\&
%X_1(k) - \left[\frac{(k+{\kappa_j})\ \e^{2\i t g(k)}}{\i{\nu_j}\ T^2(k)}\right]X_2(k) = \mathcal{O}(1), \  X_2(k)=\mathcal{O}(1),\ 
%k\to -{\kappa_j}.
%\end{split}\end{equation*}}
\item Asymptotics: $X(\xi,t;k)\to I$ as $k\to\infty.$
\end{enumerate}
\end{RHP}
%

%%\noindent \textbf{End.}
%\begin{Remark}
%Using the symmetry properties of the $g$ function, namely $\widebar{g(\bar{k}))}=g(k)$ and $g(-k)=-g(k)$ we observe that if  $\Im g(\kappa_j)>\varepsilon$ for $ \Im \kappa_j>0$
%then the first set of pole conditions becomes empty as $t\to \infty$ while for $\Im g(\kappa_j)<-\varepsilon$,  the second set if pole conditions become empty as  $t\to \infty$.
%The only pole conditions that remain are for those values of $\kappa_j$ such that $-\varepsilon<\Im g(\kappa_j)<\varepsilon$.
%\end{Remark}
%

\subsubsection{\color{black}{Utmost} left constant region $\xi<\frac{-c_{\l}^2}{2}$}
As in the other case we define
$$Y (\xi,t;k)=M(x,t;k)\e^{\i\(tg(k,\xi)-\theta(x,t;k)\)\sigma_3}T^{-\sigma_3}(k,\xi),
\quad \xi=\frac{x}{12t},$$
where the fuction  $g(k) $ is as  in  \eqref{g_breather}  while for the function  $T(k)$   defined in \eqref{Ttilde}  we require  the following conditions that are needed in order to 
apply the Deift-Zhou steepest descent method to deform the contours:
\begin{equation}
\label{T_condtion1}
T_+(k)T_-(k)=\begin{cases}-\i f(k)=\frac{1}{|a(k)|^2},\ k\in(\i c_{\l}, 0), \\ \frac{-\i}{\ol{f(k)}}=|a(k)|^2,\ k\in(0, -\i c_{\l}),\end{cases}
\end{equation}
and 
\begin{equation}
\label{T_condtion2}
\color{black}{\frac{T_+(k)}{T_-(k)}=1+|r(k)|^2,\quad k\in(-k_0, k_0),\quad k_0=\sqrt{-\xi-\frac{c_-^2}{2}},}
\end{equation}
with  $T(k)\to 1$ as $|k|\to\infty$.  Observe that in this case $T(k)$ \color{black}{ has a cut across } $[\i c_-,-\i c_-]$ and also on \color{black}{a finite interval $[-k_0,k_0]$ of} the real axis.
\noindent This function can be found in the form 
\begin{equation}\label{T_cl1}
\begin{split}&T(k,\xi)=\widetilde T(k,\xi) 
\exp\left[\frac{\sqrt{k^2+c_{\l}^2}}{2\pi\i}\int\limits_{\i c_{\l}}^{0}\frac{\(-\ln|a(s)|^2-\ln \widetilde T(s,\xi)^2
\)\d s}{(s-k)\left( \sqrt{s^2+c_{\l}^2}\right)_+}\right]\cdot
\\
&
\cdot\exp\left[\frac{\sqrt{k^2+c_{\l}^2}}{2\pi\i}\left\{\hskip-1mm\int\limits_{0}^{-\i c_{\l}}\frac{\(\ln|a(s)|^2-\ln\widetilde T(s,\xi)^2 \)\d s}{(s-k)\left(\sqrt{s^2+c_{\l}^2}\right)_+}
+
\hskip-1mm
\int\limits_{\color{black}{-k_0}}^{\color{black}{k_0}}\frac{\ln\(1+|r(s)|^2\)\d s}{(s-k)\left(\sqrt{s^2+c_{\l}^2}\right)_+}
\right\}\right]%, \chi_{\l}(k)=\sqrt{k^2+c_{\l}^2}
\end{split}\end{equation}
with  $\widetilde T(k,\xi) $ as  in formula \eqref{Ttilde}.

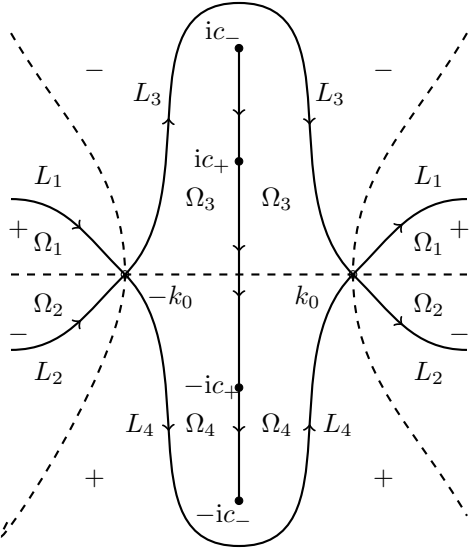
\begin{figure}[ht!]
\centering
\begin{tikzpicture}
\draw[fill=black] (0,3) circle [radius=0.05];
\draw[fill=black] (0,-3) circle [radius=0.05];
\draw[fill=black] (0,1.5) circle [radius=0.05];
\draw[fill=black] (0,-1.5) circle [radius=0.05];
\draw[] (1.5,0) circle [radius=0.05];
\draw[] (-1.5,0) circle [radius=0.05];
%\draw[thick, postaction = decorate, decoration = {markings, mark = at position 0.5 with {\arrow{>}}}](-3,0.1) to [out = 0, in =180] (0,1);
%\draw[thick, postaction = decorate, decoration = {markings, mark = at position 0.5 with {\arrow{>}}}](0,1) to [out = 0, in =180] (3,0.1);
%\draw[thick, postaction = decorate, decoration = {markings, mark = at position 0.5 with {\arrow{>}}}](-3,-0.1) to [out = 0, in =180] (0,-1);
%\draw[thick, postaction = decorate, decoration = {markings, mark = at position 0.5 with {\arrow{>}}}](0,-1) to [out = 0, in =180] (3,-0.1);

\draw[thick, postaction = decorate, decoration = {markings, mark = at position 0.5 with {\arrow{>}}}](0,3.6) to [out=0,in=135] (1.5,0);
\draw[thick, postaction = decorate, decoration = {markings, mark = at position 0.5 with {\arrow{>}}}] (1.5,0) to [out=45, in=180] (3,1);
\draw[thick, postaction = decorate, decoration = {markings, mark = at position 0.5 with {\arrow{<}}}](0,3.6) to [out=180,in=45] (-1.5,0);
\draw[thick, postaction = decorate, decoration = {markings, mark = at position 0.5 with {\arrow{<}}}] (-1.5,0) to [out=135, in=0] (-3,1);
\draw[thick, postaction = decorate, decoration = {markings, mark = at position 0.5 with {\arrow{>}}}](0,-3.6) to [out=0,in=-135] (1.5,0);
\draw[thick, postaction = decorate, decoration = {markings, mark = at position 0.5 with {\arrow{>}}}] (1.5,0) to [out=-45, in=180] (3,-1);
\draw[thick, postaction = decorate, decoration = {markings, mark = at position 0.5 with {\arrow{<}}}](0,-3.6) to [out=-180,in=-45] (-1.5,0);
\draw[thick, postaction = decorate, decoration = {markings, mark = at position 0.5 with {\arrow{<}}}] (-1.5,0) to [out=-135, in=0] (-3,-1);

%\draw[thick, postaction = decorate, decoration = {markings, mark = at position 0.5 with {\arrow{>}}}](0,3.0) to [out=180, in=90] (-1,2) [out=-90,in=150] to (0,1);
%\draw[thick, postaction = decorate, decoration = {markings, mark = at position 0.5 with {\arrow{<}}}](0,-3.) to [out=0, in=-90] (1,-2) to[out=90,in=-30](0,-1);
%\draw[thick, postaction = decorate, decoration = {markings, mark = at position 0.5 with {\arrow{<}}}](0,-3.) to [out=180, in=-90] (-1,-2)to[out=90, in=-150] (0,-1);%segment [i c_l,-i c_l]
\node at (2.5,1.3) {$L_1$};\node at (-2.5,1.3) {$L_1$};\node at (2.5,-1.3) {$L_2$};\node at (-2.5,-1.3) {$L_2$};
\node at (1.2,2.4) {$L_3$}; \node at (-1.2, 2.4) {$L_3$};\node at (1.3,-2.0) {$L_4$}; \node at (-1.3, -2.0) {$L_4$};
\node at (2.5,0.4) {$\Omega_1$};\node at (-2.5,0.4) {$\Omega_1$};
\node at (2.5,-0.4) {$\Omega_2$};\node at (-2.5,-0.4) {$\Omega_2$};
\node at (0.5,1.0) {$\Omega_3$}; \node at (-0.5, 1.0) {$\Omega_3$};\node at (0.5,-2.0) {$\Omega_4$}; \node at (-0.5, -2.0) {$\Omega_4$};
\draw[thick,
decoration={markings, mark = at position 0.15 with {\arrow{>}}},
decoration={markings, mark = at position 0.85 with {\arrow{>}}},
decoration={markings, mark = at position 0.45 with {\arrow{>}}},
decoration={markings, mark = at position 0.55 with {\arrow{>}}}, postaction={decorate}] (0,3) to (0,-3);
%points $i c_l, i c_r
\node at (-0.2,3.2) {$\i c_{\l}$};\node at (-0.35,1.5) {$\i c_{\r}$};
\node at (-0.2,-3.2) {$-\i c_{\l}$};\node at (-0.35,-1.5) {$-\i c_{\r}$};
\draw[thick,dashed] (-3,3.2) [out=-60, in=90] to (-1.5,0) [out=-90,in=-120] to (-3,-3.2);
\draw[thick,dashed] (3,3.2) [out=-120, in=90] to (1.5,0) [out=-90,in=120] to (3,-3.2);
\node at (0.9, -0.3) {$k_0$};
\node at (-0.9,-0.3) {$-k_0$};
%pluses, minuses
\node at (2.9,0.6) {\color{black}$+$};\node at (-2.9, 0.6) {\color{black}$+$};
\node at (2.9,-0.8) {\color{black}$-$};\node at (-2.9,-0.8) {\color{black}$-$};
\node at (1.9, 2.7) {\color{black}$-$};\node at (-1.9, 2.7) {\color{black}$-$};
\node at (1.9,-2.7) {\color{black}$+$};\node at (-1.9, -2.7) {\color{black}$+$};
%real line
\draw[thick, dashed] (-3,0) to (3,0);
\end{tikzpicture}
\caption{The opening of the lenses in the \color{black}{utmost} left constant region.}
\label{breather_lens}
\end{figure}

\noindent 
Because of \eqref{T_condtion1} and \eqref{T_condtion2}, the jump matrix $J_Y$ of the RH problem for $Y(k)$
has a factorization on the real axis of the form
\begin{equation}\label{fac_122}
J_{Y}(\xi,t;k)=
\begin{cases}
\begin{pmatrix}1 & 0 \\ \frac{-r(k)\e^{2\i t g(k,\xi)}}{T^2(k)} & 1\end{pmatrix}
\begin{pmatrix}1 & -\ol{r\(\ol{k}\)}T^2(k)\e^{-2\i t g(k,\xi)} \\ 0 & 1\end{pmatrix},\; k\in\mathbb{R}\backslash(-k_0, k_0)\\
\begin{pmatrix}1& \frac{-\ol{r\(\ol{k}\)}\ T_+^2\e^{-2\i t g}}{1+|r|^2} \\ 0 & 1\end{pmatrix}
\begin{pmatrix}1 & 0 \\ \frac{-r(k)\ \e^{2\i t g}}{(1+|r|^2)T_-^2} & 1
\end{pmatrix}, \quad k\in(-\color{black}{k_0, k_0})\,.
\end{cases}
\end{equation}
We use the above factorization  to  open the lenses and  define the matrix $X(\xi,t;k)$ as
\begin{equation}
X(\xi,t;k)=
\begin{cases}
&Y(\xi,t;k)\begin{pmatrix}1 & 0 \\ \frac{-r(k)\e^{2\i t g(k,\xi)}}{T^2(k)} & 1\end{pmatrix},\quad k\in\Omega_1, \\
&Y(\xi,t;k)\begin{pmatrix}1 & -\ol{r\(\ol{k}\)}T^2(k)\e^{-2\i t g(k,\xi)} \\ 0 & 1\end{pmatrix}^{-1},\quad k\in\Omega_2, \\
&Y(\xi,t;k)\begin{pmatrix}1& \frac{-\ol{r\(\ol{k}\)}\ T_+^2\e^{-2\i t g}}{1+|r|^2} \\ 0 & 1\end{pmatrix},\quad k\in\Omega_3, \\
&Y(\xi,t;k)\begin{pmatrix}1 & 0 \\ \frac{-r(k)\ \e^{2\i t g}}{(1+|r|^2)T_-^2(k)} & 1
\end{pmatrix}^{-1}, \quad k\in\Omega_4, \\
&Y(\xi,t;k),\quad \mbox{elsewhere}\,.
\end{cases}
\end{equation}
We arrive at  the following RH problem for the function $X(\xi,t;k)$.
\noindent
\begin{RHP}\label{RHP_final_constant1}\textbf{Final RH problem for the utmost left  constant region $\xi<\frac{-c_{\l}^2}{2}.$ }
\begin{itemize}
\item Find a $2\times 2$ matrix function $X(\xi,t;k)$ meromorphic  for $k\in\mathbb{C}\backslash \Sigma$ with $\Sigma=\cup_{j=1}^4L_j\cup(\i c_-,-\i c_-)$  (see Figure \ref{breather_lens}) 
and   such that
\item $ X_-(\xi,t;k)=X_+(\xi,t;k)J_X(\xi,t;k),\quad k\in \Sigma$
with 
$$J_X(\xi,t;k)=\begin{pmatrix}0&\i\\\i&0\end{pmatrix},\ k\in(\i c_{\l}, -\i c_{\l}),$$
$$J_X(\xi,t;k)=\begin{cases}
\begin{pmatrix}1 & 0 \\ \frac{-r(k)\ \e^{2\i t g(k)}}{T^2(k)} & 1\end{pmatrix}, k\in L_1,\\
\begin{pmatrix}1 & -\ol{r\(\ol k\)} T^2(k) \e^{-2\i t g(k)} \\ 0 & 1\end{pmatrix}, k\in L_2, \\
\begin{pmatrix}
1& \frac{-\ol{r\(\ol{k}\)}\ T^2(k)\e^{-2\i t g(k)}}{1+r(k)\ol{r\(\ol{k}\)}} \\ 0 & 1
\end{pmatrix}, k\in L_3,\\\begin{pmatrix}1 & 0 \\ \frac{-r(k)\ \e^{2\i t g(k)}}{\(1+r(k)\ol{r\(\ol{k}\)}\)T^2(k)} & 1
\end{pmatrix}, k\in L_4,
\end{cases}$$
and the pole conditions are the same as in the RH problem~\ref{RHP_final_constant2}.
%$$J_X(k)=\begin{cases}
%\begin{pmatrix} 1 & 0 \\ \left[\frac{\nu_j}{(k-\kappa_j)}\right]\frac{\e^{2\i t g(k)}}{T^2(k)}\end{pmatrix},
%\quad
%|k-\i\kappa_j|=\varepsilon,\quad \Im\kappa_{j}>0,\quad \Im g(\kappa_j)>\varepsilon,
%\\
%\begin{pmatrix} 1 & \frac{T^2(k)\e^{-2\i t g(k)}}{\left[\frac{\nu_j}{(k-\kappa_j)}\right] } \\ 0 & 1  \end{pmatrix},
%\quad
%|k-\i\kappa_j|=\varepsilon,\quad \Im\kappa_{j}>0,\quad \Im g(\kappa_j)<-\varepsilon,
%\\
%\begin{pmatrix}1 & \left[\frac{-\ol{\nu_j}}{(k-\ol{\kappa_j})}\right]T^2(k)\e^{-2\i t g(k)}\\0&1\end{pmatrix},
%\quad |k+\kappa_j|=\varepsilon,\quad \Im\(\ol{\kappa_j}\)<0,\quad \Im g\(\ol{\kappa_j}\)<-\varepsilon,
%\\
%\begin{pmatrix}1 & 0 \\ \frac{T^{-2}\e^{2\i t g(k)}}{\left[\frac{-\ol{\nu_j}}{(k-\ol{\kappa_j})}\right]}&1\end{pmatrix},
%\quad |k+\kappa_j|=\varepsilon,\quad \Im\(\ol{\kappa_j}\)<0,\quad \Im g\(\ol{\kappa_j}\)>\varepsilon,
%\end{cases}$$
%
%At the points $\kappa_j$ with $\Im \kappa_j>0,$ $-\varepsilon<\Im g(\kappa_j)<\varepsilon$ we keep the pole condition
%$$
%\begin{cases}
%X_1-\left[\frac{A_j}{k-\kappa_j}\right]\frac{\e^{2\i t g(k,\xi)}}{T^2(k,\xi)}X_2 = \mathcal{O}(1),\quad k\to\kappa_j,\quad \Im \kappa_j>0,$ $-\varepsilon<\Im g(\kappa_j)<\varepsilon,
%\\
%X_2-\left[\frac{-\ol{A_j}}{(k-\kappa_j)}\right]T^2\e^{-2\i t g(k,\xi)}X_1=\mathcal{O}(1),\ k\to\ol{\kappa_j},
%\quad \quad \Im \ol{\kappa_j}<0,$ $-\varepsilon<\Im g\(\ol{\kappa_j}\)<\varepsilon.
%\end{cases}
%$$
\item Asymptotics: $X(k)\to I$ as $k\to\infty.$
\end{itemize}
\end{RHP}

%%%%%%%%%%%%%%%%%
\subsubsection{Model problems for the regions $\xi<\frac{-c_{\l}^2}{2} $ and $\frac{-c_{\l}^2}{2}<\xi<\frac{-c_{\l}^2}{2}+c_{\r}^2$}
Looking at the RH problems \ref{RHP_final_constant2} and  \ref{RHP_final_constant1}, we observe that the jump matrix is exponentially close to $I$ everywhere, except for the jump on $(\i c_{\l}, -\i c_{\l}),$ and the parts of the curves $L_j,$ $j=1,\ldots,8,$ intersecting either the real line (i.e. the vicinities of the points $\color{black}{\pm k_0(\xi)}$) , or the interval $(\i c_{\l}, -\i c_{\l})$ (i.e. the vicinities of the points $\pm\i d_0(\xi)$). A careful analysis of the contribution of the points $\pm \color{black}{k_0(\xi)},$ to the asymptotics for $q(x,t)$ involves the construction  of a local (approximate) solution to the RH problem in the vicinities of those points, and is usually called \textit{ parametrix} analysis. The parametrices in the vicinities of the points $\color{black}{\pm k_0}$ can be constructed in terms of parabolic cylinder function, in a similar way as  in \cite[section 5]{M11}, and they give the contribution of the order $\mathcal{O}(t^{-1/2}).$ The parametrices in the vicinities of the points $\pm \i d_0$ can be constructed in terms of the Airy function and its derivative, in a similar way  as was done in \cite[section 3]{BM}, and they give a contribution of the order $\mathcal{O}(t^{-1}).$ We thus reduce the problem to analysis of the corresponding model problems.
The regular RH model problem in these regions gives the constant $c_{\l}$ as the main term in the asymptotic expansion  for $q(x,t).$
The meromorphic RH model  problems give us the breathers  on a constant background.
The model problem for this case is 
\begin{RHP}\label{RHP_model_const}\textbf{Model RH problem for the 
%$1^{st}$ and $2^{nd}$ 
\color{black}{middle and utmost} left
constant regions $\boldsymbol{\frac{-c_{\l}^2}{2}<\xi<\frac{-c_{\l}^2}{2}+c_{\r}^2}$   and $\boldsymbol{\frac{-c_{\l}^2}{2}<\xi}$. }

Find a $2\times 2$ matrix $M^{mod}(k,\xi)$ meromorphic for  $k\in \mathbb{C}\backslash[\i c_-,-\i c_-]$ and such that
\begin{enumerate}
\item Jump: $$M^{mod}_-(k)=M^{mod}_+(k)\begin{pmatrix}0&\i\\\i&0\end{pmatrix},\quad k\in(\i c_{\l}, -\i c_{\l}),$$
%$$J^{(2)}=\begin{cases}
%\begin{pmatrix}1 & 0 \\ \frac{-r\ \e^{2\i t g}}{T^2}\end{pmatrix}, k\in L_1,\quad
%\begin{pmatrix}1 & -\ol{r\(\ol k\)} T^2 \e^{-2\i t g}\end{pmatrix}, k\in L_2, \\
%\begin{pmatrix}1 & \frac{T^2\e^{-2\i t g}}{\widehat f} \\ 0 & 1 \end{pmatrix}, k\in L_7,\quad
%\begin{pmatrix}1 & \frac{-T^2\e^{-2\i t g}}{\widehat f} \\ 0 & 1 \end{pmatrix}, k\in L_5,\\
%\begin{pmatrix}1 & 0 \\\frac{-\e^{2\i t g}}{\ol{\widehat f\(\ol{k}\)}\ T^2} & 1 \end{pmatrix}, k\in L_8,
%\quad
%\begin{pmatrix}1 & 0 \\\frac{\e^{2\i t g}}{\ol{\widehat f\(\ol{k}\)}\ T^2} & 1 \end{pmatrix}, k\in L_6,
%\end{cases}$$
%
%$$J^{(2)}=\begin{cases}
%\begin{pmatrix} 1 & 0 \\ \left[\frac{A_j}{(k-\kappa_j)}\right]\frac{\e^{2\i t g(k,\xi)}}{T^2(k,\xi)}\end{pmatrix},
%\quad
%|k-\i\kappa_j|=\varepsilon,\quad \Im\kappa_{j}>0,\quad \Im g(\kappa_j)>\varepsilon,
%\\
%\begin{pmatrix} 1 & \frac{T^2(k,\xi)\e^{-2\i t g(k,\xi)}}{\left[\frac{A_j}{(k-\kappa_j)}\right] } \\ 0 & 1  \end{pmatrix},
%\quad
%|k-\i\kappa_j|=\varepsilon,\quad \Im\kappa_{j}>0,\quad \Im g(\kappa_j)<-\varepsilon,
%\\
%\begin{pmatrix}1 & \left[\frac{-\ol{A_j}}{(k-\ol{\kappa_j})}\right]T^2\e^{-2\i t g(k,\xi)}\\0&1\end{pmatrix},
%\quad |k+\kappa_j|=\varepsilon,\quad \Im\(\ol{\kappa_j}\)<0,\quad \Im g\(\ol{\kappa_j}\)<-\varepsilon,
%\\
%\begin{pmatrix}1 & 0 \\ \frac{T^{-2}\e^{2\i t g(k,\xi)}}{\left[\frac{-\ol{A_j}}{(k-\ol{\kappa_j})}\right]}&1\end{pmatrix},
%\quad |k+\kappa_j|=\varepsilon,\quad \Im\(\ol{\kappa_j}\)<0,\quad \Im g\(\ol{\kappa_j}\)>\varepsilon,
%\end{cases}$$

\item Pole: if there are points $\kappa_j$ with $\Im \kappa_j>0,$ $\Re\kappa_j>0$,  $-\delta<\Im g(\kappa_j,\xi)<\delta,$ we have the pole condition
{\color{black}
$$\mathrm{Res}_{\kappa_j}M^{(mod)}(k)=\lim\limits_{k\to \kappa_j}M^{(mod)}(k)\begin{pmatrix}0&0\\ \frac{ \i\nu_j}{T^2(k,\xi)} \e^{2\i t  g(k,\xi)} & 0\end{pmatrix},$$
$$\mathrm{Res}_{-\ol{\kappa}_j}M^{(mod)}(k)=\lim\limits_{k\to -\ol{\kappa}_j}M^{(mod)}(k)\begin{pmatrix}0&0\\\frac{ \i\,\widebar{\nu}_j }{T^2(k,\xi)} \e^{2\i t  g(k,\xi)} & 0\end{pmatrix},$$
$$\mathrm{Res}_{\ol{\kappa}_j}M^{(mod)}(k)=\lim\limits_{k\to \ol{\kappa}_j}M^{(mod)}(k)\begin{pmatrix}0&\i\,\widebar{\nu}_jT^2(k,\xi)\e^{-2\i t  g(k,\xi)}\\0&0\end{pmatrix},$$
$$\mathrm{Res}_{-\kappa_j}M^{(mod)}(k)=\lim\limits_{k\to -\kappa_j}M^{(mod)}(k)\begin{pmatrix}0&\i\nu_j T^2(k,\xi)\e^{-2\i t g(k,\xi)}\\0&0\end{pmatrix},$$
}
%
%\color{black}{\begin{equation*}
%\begin{cases}
%M^{(mod)}_1(k,\xi)-\dfrac{\i\nu_j}{k-\kappa_j}\frac{\e^{2\i t g(k,\xi)}}{T^2(k,\xi)}M^{(mod)}_2(k,\xi)= \mathcal{O}(1),\quad k\to\kappa_j,
%\\
%M^{(mod)}_1(k,\xi)-\dfrac{\i\,\ol{\nu_j}}{k+\ol{\kappa_j}}\frac{\e^{2\i t g(k,\xi)}}{T^2(k,\xi)}M^{(mod)}_2(k,\xi) = \mathcal{O}(1),\quad k\to-\ol{\kappa_j},
%\\
%M^{(mod)}_2(k,\xi)-\dfrac{\i\,\ol{\nu_j}\ T^2(k,\xi)}{k-\ol{\kappa_j}}\e^{-2\i t g(k,\xi)}M^{(mod)}_1(k,\xi)=\mathcal{O}(1),\ k\to\ol{\kappa_j},
%\\
%M^{(mod)}_2(k,\xi)-\dfrac{\i {\nu_j}\ T^2(k,\xi)}{k+\kappa_j}\e^{-2\i t g(k,\xi)}M^{(mod)}_1(k,\xi)=\mathcal{O}(1),\ k\to-{\kappa_j}.
%\end{cases}
%\end{equation*}}
\item $M^{(mod)}(k)=1+O(k^{-1})$ as $|k|\to\infty$.
\end{enumerate}
\end{RHP}
We observe that  $\Im g(\kappa_j,\xi)=\Im\left( \sqrt{\kappa_j^2+c_-^2}\right)(12\xi-V_j)$ where $$V_j=4\Im\left( \sqrt{\kappa_j^2+c_-^2}\right)^2+6c_-^2-12\Re\left( \sqrt{\kappa_j^2+c_-^2}\right)^2$$   is the speed of the breather corresponding to the spectrum $\kappa_j$ on the constant background $c_-$.

Therefore the condition  $-\delta<\Im g(\kappa_j,\xi)<\delta,$  is equivalent to requiring that 
\begin{equation*}
|V_j-12\xi|<\tilde{\delta},\quad \xi=\dfrac{x}{12 t}, \quad  \tilde{\delta}=\delta/\Im\left( \sqrt{\kappa_j^2+c_-^2}\right).
\end{equation*} 
Namely we can see the  breather if  we observe in the direction of the $(x,t)$   plane such that $|V_j-\frac{x}{t}|<\tilde{\delta}$.   
We number the  points of the discrete spectrum $\kappa_j$, according to their speed as specified in  the introduction.
%with $\Im \kappa_j>0$ and $\Re \kappa_j\geq 0$ 
% according to their velocities 
%%  and we assume that there are $n$ breathers or solitons moving with distinct velocities  $V_j$,  $j=1,\dots, n$  to the right of the dispersive shock wave, one breather moving in the dispersive shock wave region and $N-n-1$ breathers moving  with distinct velocities  $V_j$, $j=N-n-1,\dots, N$  to the left of the dispersive shock wave on the constant background $c_+$  so that
% \begin{equation}
% \label{numbering}
%-\infty< V_N<V_{N-1}<\dots< V_2 <V_{1}<+\infty
% \end{equation}
We further observe that if $\xi$ is such that  $\Im g(\kappa_j,\xi)=0$ then   $\xi=\frac{V_j}{12}$ and  with the above ordering we have that 
$$\Im g\left(\kappa_{l},\frac{V_j}{12}\right)<0,\;\mbox{for $l<j$ and } \Im g\left(\kappa_{l},\frac{V_j}{12}\right)>0, \;\; \mbox{for $l>j$}.$$
With this ordering the function $\widetilde T(k,\frac{V_j}{12})$ defined in \eqref{Ttilde} takes the form \eqref{T_tilde}.
The RH problem  \ref{RHP_model_const}  corresponds to the case of a single breather on a constant background  that   we considered in Section \ref{sect_breath}.
 This problem can be solved explicitly in terms of elementary functions, and 
$$2\i \lim\limits_{k\to\infty}M_{12}^{(mod)}(k,\xi) = q_{breath}(x,t; c_{\l},\kappa_j, \widehat{\nu_j}) ,$$
where
\begin{equation}\label{chi_j_hat}\widehat \nu_j = \dfrac {\nu_j\ %\e^{2\i (t g(\kappa_j,\xi)-\theta(12\xi t, t, \kappa_j))}
}{T^2(\kappa_j,\xi)},
\end{equation}
 with $T(k,\xi)$ as in \eqref{T_cl2} 
   for $\frac{-c_{\l}^2}{2}<\xi<\frac{-c_{\l}^2}{2}+c_{\r}^2$     and $d_0=\i\sqrt{\xi+\frac{c_-^2}{2}}$, and $T(k,\xi)$ as in \eqref{T_condtion2}  for 
 $\xi>\frac{-c_{\l}^2}{2}$, $\color{black}{k_0}=\sqrt{-\xi-\frac{c_-^2}{2}}$.
Summing up, we have shown that for  large values of time the solution of the MKdV equation 
in the domains $\xi<\frac{-c_{\l}^2}{2}$ and $\frac{-c_{\l}^2}{2}<\xi<\frac{-c_{\l}^2}{2}+c_{\r}^2$   is as follows:
\begin{itemize}
\item[(1)] for those $\xi, t$ such that  $|\Im g(\kappa_j,\xi)|>\delta$ for all $\kappa_j$ we have
$$q(x,t)=c_{\l}+{\mathcal O}(t^{-\frac{1}{2}}),$$
\item [(2)] For those $\xi, t$ such that there exists $\kappa_j$, $\Re\kappa_j>0,$ $\Im\kappa_j>0,$  with $|\Im g(\kappa_j, \xi)|<\delta,$
we have 
$$q(x,t)=q_{breath}(x,t; c_{\l}, \kappa_j, \widehat\nu_j)+{\mathcal O}(t^{-\frac{1}{2}}), 
$$
where $ q_{breath}$ has been defined in \eqref{q_breath} and the phase   $\widehat \nu_j$ in \eqref{chi_j_hat}.
\end{itemize}
Thus we have concluded the proof of  Theorem \ref{thrm:asymp:rl}   part (c).

%%%%%%%%%%%%%%%%%

\subsection{Proof of Theorem \ref{thrm:asymp:rl}   part (a):   soliton and breather region}
In this case we consider the right constant region $\xi>\frac{c_{\l}^2}{3}+\frac{c_{\r}^2}{6}$  where solitons and breathers are moving in the positive $x$ direction with a speed greater then the dispersive shock wave.
As in the  previous case we introduce the matrix function
$$Y (\xi,t;k)=M(x,t;k)\e^{\i\(tg(k,\xi)-\theta(x,t;k)\)\sigma_3}T^{-\sigma_3}(k,\xi),
\quad \xi=\frac{x}{12t},$$
where now the function  $g(k,\xi)$ takes the form  \cite[page 11]{KM2}
\begin{equation}
\label{g_breather2}
g(k,\xi)=2\(2k^2-c_{+}^2+6\xi\right)\sqrt{k^2+c_{+}^2}\,,\quad \quad \end{equation}
namely $g(k,\xi)$ is analytic in $k\in\mathbb{C}\backslash [\i c_+,-\i c_+]$  and
\begin{equation}
\label{g_soliton}
g_+(k,\xi)+g_-(k,\xi)=0,\; k\in(\i c_+,-\i c_+),\quad g(k,\xi)=4k^3+12\xi k
+O(k^{-1}),\;\mbox{as $ |k|\to\infty$}.
\end{equation}
 We chose $\sqrt{k^2+c_{+}^2}$ to be real on $(\i c_+,-\i c_+)$ and positive for $k=+0$.
The sign of $\Im g(k,\xi)$ are plotted in Figure \ref{soliton_region}.
\begin{figure}[ht]

\begin{minipage}{0.5\linewidth}

\begin{tikzpicture}
\draw[fill=black] (0,1.3) circle [radius=0.05];
\draw[fill=black] (0,-1.3) circle [radius=0.05];
\draw[fill=black] (0,0.5) circle [radius=0.05];
\draw[fill=black] (0,-0.5) circle [radius=0.05];
\draw[] (0,2.) circle [radius=0.05];
\draw[] (0,-2) circle [radius=0.05];
%segment [i c_l,-i c_l]
\draw[thick,dashed] (0,0.5) to (0,-.5);
%points $i c_l, i c_r
\node at (-0.4,1.3) {$\i c_{\l}$};\node at (0.4,0.5) {$\i c_{\r}$};
\node at (-0.4,-1.3) {$-\i c_{\l}$};\node at (0.4,-0.5) {$-\i c_{\r}$};
\draw[thick,dashed] (-2,3) [out=-60, in=180] to (0,2.) [out=0,in=-120] to (2,3);
\draw[thick,dashed] (-2,-3) [out=60, in=180] to (0,-2.) [out=0,in=120] to (2,-3);
\node at (1.7,1.8) {$\i\kappa_0=\i\sqrt{3\xi-c_{\r}^2/2}$};
\node at (0.3,-1.8) {$-\i\kappa_0$};
%\node at (0.3,-1.8) {$-\i\kappa_0$};
\node at (0.3,-1.8) {$-\i\kappa_0$};
%pluses, minuses
\node at (2,1.3) {\color{black}$+$};\node at (-2,1.3) {\color{black}$+$};
\node at (2,-1.5) {\color{black}$-$};\node at (-2,-1.5) {\color{black}$-$};
\node at (0,-2.8) {\color{black}$+$};\node at (0,2.8) {\color{black}$-$};
%real line
\draw[thick, dashed] (-3,0) to (3,0);
\end{tikzpicture}
\\
\centering{(a)} 
\end{minipage}
\begin{minipage}{0.49\linewidth}
\begin{tikzpicture}
\draw[fill=black] (0,1.3) circle [radius=0.05];
\draw[fill=black] (0,-1.3) circle [radius=0.05];
\draw[fill=black] (0,0.5) circle [radius=0.05];
\draw[fill=black] (0,-0.5) circle [radius=0.05];
\draw[] (0,2.) circle [radius=0.05];
\draw[] (0,-2) circle [radius=0.05];%segment [i c_l,-i c_l]
\draw[thick,postaction = decorate, decoration = {markings, mark = at position 0.25 with {\arrow{>}}},
decoration = {markings, mark = at position 0.75 with {\arrow{>}}}] (0,.5) to (0,-.5);
%points
\node at (-0.4,1.3) {$\i c_{\l}$};\node at (0.4,0.5) {$\i c_{\r}$};
\node at (-1.,-.5) {$\Omega_2$};\node at (1.4,0.5) {$\Omega_1$};
\node at (-0.4,-1.3) {$-\i c_{\l}$};\node at (0.4,-0.5) {$-\i c_{\r}$};
\node at (0.3,1.8) {$\i\kappa_0$};
\node at (0.3,-1.8) {$-\i\kappa_0$};
\node at (2.2,1.3) {$L_1$};
\node at (2.2,-1.2) {$L_2$};
%lines L_1, L_2, L_3, ...
\draw[thick, postaction = decorate, decoration = {markings, mark = at position 0.5 with {\arrow{>}}}](-3,0.7) to [out = 0, in =-180] (0,1.5) to [out=0, in = -180] (3,0.7);
\draw[thick, postaction = decorate, decoration = {markings, mark = at position 0.5 with {\arrow{>}}}](-3,-0.7) to [out = 0, in =-180] (0,-1.5) to [out=0, in = -180] (3,-0.7);
%real line
\draw[thick, dashed] (-3,0) to (3,0);
\end{tikzpicture}
\\
\\
\\
\\
\centering{ (b)} 
\end{minipage}
\caption{Right constant region. (a) Distribution of signs for $\Im g =0$ for 
$\xi>\frac{c_{\l}^2}{3}+\frac{c_{\r}^2}{6}.$  (b) Contours of the RH problem.}
\label{soliton_region}
\end{figure}
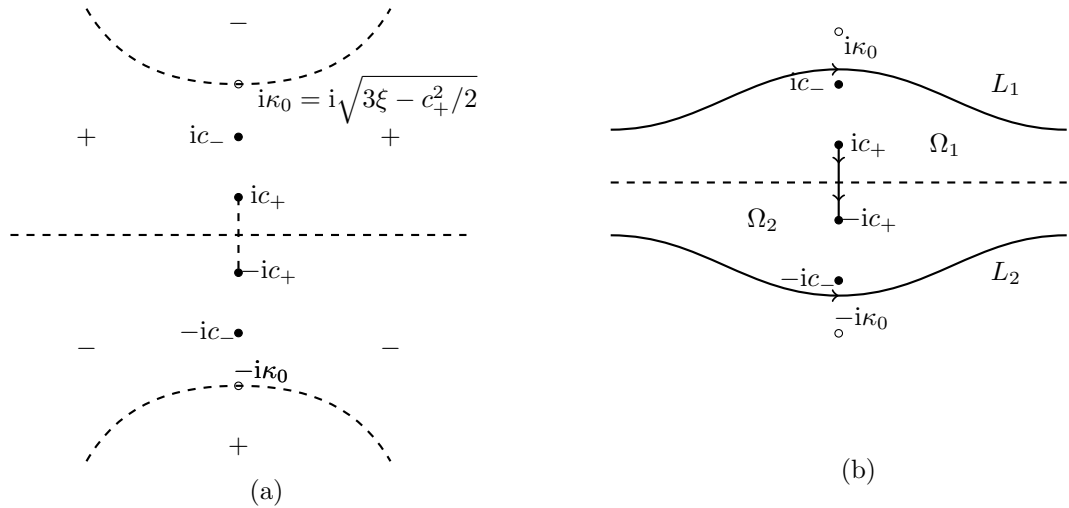
The function $T(k,\xi)$ is as in \eqref{Ttilde}  with $c_-$ replaced by $c_+$ and it satisfies the condition 
$$
T_-(k)T_+(k)=1,\quad k\in(\i c_{\l}, -\i c_{\l})\,.
$$
The solution is  explicitly given as follows:
\begin{equation}\label{T_cr} T(k,\xi)=\widetilde T(k,\xi)
\ \exp\left[\frac{\sqrt{k^2+c_{\r}^2}}{2\pi\i}\int\limits_{\i c_{\r}}^{-\i c_{\r}}\frac{\(-\ln\widetilde T^2(s,\xi) \)\d s}{(s-k)\left( \sqrt{k^2+c_{\r}^2}\right)_+}\right], 
\end{equation}
where $\widetilde T(k,\xi)$ is of the form \eqref{Ttilde}. 
This choice of the function $T(k,\xi)$ permits the factorization of the jump matrix $J_Y$ on the real axis in the form 
\begin{equation*}
J_Y(\xi,t;k)=\begin{pmatrix}1 & 0 \\ \frac{-r(k)\e^{2\i t g(k,\xi)}}{T^2(k)} & 1\end{pmatrix}
\begin{pmatrix}1 & -\ol{r\(\ol{k}\)}T^2(k)\e^{-2\i t g(k,\xi)} \\ 0 & 1\end{pmatrix},\quad k\in\mathbb{R}.
\end{equation*}
We use the above factorization  to  open the lenses and  define the matrix $X(\xi,t;k)$ as
\begin{equation}
X(\xi,t;k)=
\begin{cases}
&Y(\xi,t;k)\begin{pmatrix}1 & 0 \\ \frac{-r(k)\e^{2\i t g(k,\xi)}}{T^2(k)} & 1\end{pmatrix},\quad k\in\Omega_1, \\
&Y(\xi,t;k)\begin{pmatrix}1 & -\ol{r\(\ol{k}\)}T^2(k)\e^{-2\i t g(k,\xi)} \\ 0 & 1\end{pmatrix}^{-1},\quad k\in\Omega_2, \\
&Y(\xi,t;k),\quad \mbox{elsewhere}\,,
\end{cases}
\end{equation}
where the regions $\Omega_1$ and $\Omega_2$ are specified Figure~\ref{soliton_region}.

\noindent
\begin{RHP}\label{RHP_final_solitons}\textbf{Final RH problem for the soliton and breather region $\boldsymbol{\xi>\frac{c_{\l}^2}{3}+\frac{c_{\l}^2}{6}.}$ }
Find a $2\times 2$ matrix $X(\xi,t;k)$ meromorphic in $k\in\mathbb{C}\backslash \Sigma$ where the contour $\Sigma$ is specified below
and 
\begin{equation*}
X_-(\xi,t;k)=X_+(\xi,t;k)J_X(\xi,t;k),\quad k\in\Sigma,
\end{equation*}
and 
\begin{equation}
J_X(\xi,t;k)=
\begin{cases}
\begin{pmatrix}0&\i\\\i&0\end{pmatrix},\ k\in(\i c_{\r}, -\i c_{\r}),\\
 J^{(2)}=I, \ k\in(\i c_{\l}, \i c_{\r}) \cup (-\i c_{\r}, -\i c_{\l}),\\
\begin{pmatrix}1 & 0 \\ \frac{-r(k)\ \e^{2\i t g(k,\xi)}}{T^2(k,\xi)} & 1\end{pmatrix}, \quad k\in L_1,\\
\begin{pmatrix}1 & -\ol{r\(\ol k\)} T^2(k,\xi) \e^{-2\i t g(k,\xi)} \\ 0 & 1\end{pmatrix}, \quad k\in L_2,
\end{cases}
\end{equation}
with $L_1$ and $L_2$ as in figure~\ref{soliton_region}  and 
%$$J_X(k)=\begin{cases}
%\begin{pmatrix} 1 & 0 \\ \left[\frac{\nu_j}{(k-\kappa_j)}\right]\frac{\e^{2\i t g(k,\xi)}}{T^2(k,\xi)}&1\end{pmatrix},
%\quad
%|k-\i\kappa_j|=\varepsilon,\quad \Im\kappa_{j}>0,\quad \Im g(\kappa_j)>\varepsilon,
%\\
%\begin{pmatrix} 1 & \frac{T^2(k,\xi)\e^{-2\i t g(k,\xi)}}{\left[\frac{\nu_j}{(k-\kappa_j)}\right] } \\ 0 & 1  \end{pmatrix},
%\quad
%|k-\i\kappa_j|=\varepsilon,\quad \Im\kappa_{j}>0,\quad \Im g(\kappa_j)<-\varepsilon,
%\\
%\begin{pmatrix}1 & \left[\frac{-\ol{\nu_j}}{(k-\ol{\kappa_j})}\right]T^2\e^{-2\i t g(k,\xi)}\\0&1\end{pmatrix},
%\quad |k+\kappa_j|=\varepsilon,\quad \Im\(\ol{\kappa_j}\)<0,\quad \Im g\(\ol{\kappa_j}\)<-\varepsilon,
%\\
%\begin{pmatrix}1 & 0 \\ \frac{T^{-2}\e^{2\i t g(k,\xi)}}{\left[\frac{-\ol{\nu_j}}{(k-\ol{\kappa_j})}\right]}&1\end{pmatrix},
%\quad |k+\kappa_j|=\varepsilon,\quad \Im\(\ol{\kappa_j}\)<0,\quad \Im g\(\ol{\kappa_j}\)>\varepsilon,
%\end{cases}$$
at the points $\kappa_j$ with $\Im \kappa_j>0,$  we keep the pole condition  as in the RH problem \ref{RHP_final_constant2}.
%$$
%\begin{cases}
%X_1(k)-\left[\frac{\nu_j}{(k-\kappa_j)}\right]\frac{\e^{2\i t g(k,\xi)}}{T^2(k,\xi)}X_2(k) = \mathcal{O}(1),\quad k\to\kappa_j,\quad \Im \kappa_j>0,$ $-\varepsilon<\Im g(\kappa_j)<\varepsilon,
%\\
%X_2(k)-\left[\frac{-\ol{\nu_j}}{(k-\ol\kappa_j)}\right]T^2\e^{-2\i t g(k,\xi)}X_1(k)=\mathcal{O}(1),\ k\to\ol{\kappa_j},
%\quad \quad \Im \ol{\kappa_j}<0,$ $-\varepsilon<\Im g\(\ol{\kappa_j}\)<\varepsilon.
%\end{cases}
%$$
Finally the normalizing condition at infinity is  $X(k)=1+O(k^{-1})$ as $|k|\to\infty$.
\end{RHP}
We observe that the jump matrices  approach the identity exponentially fast on $L_1$ and $L_2$  and also for those values of $\xi$ and $\kappa_j$ such that $ \Im\kappa_{j}>0$  and 
$\Im g(\kappa_j)>\varepsilon$ or $ \Im g(\kappa_j)<-\varepsilon$.
 We finally arrive to the following model problem.
\begin{RHP}\label{RHP_model_soliton}\textbf{Model RH problem for the right region  $\boldsymbol{\xi>\frac{c_{\l}^2}{3}+\frac{c_{\r}^2}{6}.}$ }
Find a $2\times 2$ matrix $M^{mod}(\xi,t;k)$ meromorphic for  $k\in \mathbb{C}\backslash[\i c_+,-\i c_+]$ and such that
\begin{itemize}
\item Jump: $$M^{(mod)}_-(\xi,t;k)=M^{(mod)}_+(\xi,t;k)\begin{pmatrix}0&\i\\\i&0\end{pmatrix},\ k\in(\i c_{\r}, -\i c_{\r}).$$
\item Poles: if there are points $\kappa_j$ with $\Im \kappa_j>0,$ $\Re\kappa_j>0$,  $-\delta<\Im g(\kappa_j,\xi)<\delta,$ we have the pole condition
{\color{black}
$$\mathrm{Res}_{\kappa_j}M^{(mod)}(\xi,t;k)=\lim\limits_{k\to \kappa_j}M^{(mod)}(\xi,t;k)\begin{pmatrix}0&0\\ \frac{ \i\nu_j}{T^2(k)} \e^{2\i t  g(k,\xi)} & 0\end{pmatrix},$$
%$$\mathrm{Res}_{-\ol{\kappa}_j}M^{(mod)}(k)=\lim\limits_{k\to -\ol{\kappa}_j}M^{(mod)}(k)\begin{pmatrix}0&0\\\frac{ \i\,\widebar{\nu} }{T^2(k)} \e^{2\i t  g(k,\xi)} & 0\end{pmatrix},$$
$$\mathrm{Res}_{\ol{\kappa}_j}M^{(mod)}(\xi,t;k)=\lim\limits_{k\to \ol{\kappa}_j}M^{(mod)}(\xi,t;k)\begin{pmatrix}0&\i\,\widebar{\nu}_jT^2(k)\e^{-2\i t  g(k,\xi)}\\0&0\end{pmatrix},$$
%$$\mathrm{Res}_{-\kappa_j}M^{(mod)}(k)=\lim\limits_{k\to -\kappa_j}M^{(mod)}(k)\begin{pmatrix}0&\i\nu T^2(k)\e^{-2\i t g(k,\xi)}\\0&0\end{pmatrix},$$
}
and the  symmetric conditions at the points $-\kappa_j$ and  $-\ol{\kappa}$  when we are considering a breather.
\item $M^{(mod)}(\xi,t;k)=1+O(k^{-1})$ as $k\to\infty$.
\end{itemize}
\end{RHP}
This model problem is exactly the model problem, which gives solitons or breathers on the constant background $c_{\r}.$
Thus,
\begin{equation*}\begin{split}&\lim\limits_{k\to\infty}2\i k(M^{mod}(\xi,t;k))_{12} = q_{sol}(x,t; c_{\r}, \kappa_j, x_j),\quad \mbox{ if }\Re\kappa_j=0,
\\&
\lim\limits_{k\to\infty}2\i k(M^{mod}(\xi,t;k))_{12} = q_{breath}(x,t; c_{\r}, \kappa_j, \hat\nu_j),\quad \mbox{ if }\Re\kappa_j>0,
\end{split}\end{equation*}
where $\xi=x/(12t)$ and where  we use the  numbering of solitons and breathers  according to their velocities  so that the phase shift can be written in the form
\begin{equation}
\label{chihat_j}
\begin{split}
&\hat\nu_j = \dfrac{\nu_j }{T_j^2(\kappa_j)},\quad x_j=\log\dfrac{2(\kappa_j^2-c_+^2)T_j^2(\i|\kappa_j|)}{|\nu_j|\kappa_j}, 
\\
&T_j(k):=T(k,\frac{V_j}{12})=\widetilde T(k,\frac{V_j}{12})
\ \exp\left[\frac{\sqrt{k^2+c_{\r}^2}}{2\pi\i}\int\limits_{\i c_{\r}}^{-\i c_{\r}}\frac{\(-\ln\widetilde T^2(s,\xi) \)\d s}{(s-k)\left( \sqrt{k^2+c_{\r}^2}\right)_+}\right], 
\end{split}
\end{equation}
where   $\widetilde T(k,\xi)$  as in \eqref{Ttilde}  so that $\widetilde T(k,\frac{V_j}{12})$ coincides with $\widetilde T_j(k)$ defined in \eqref{T_tilde}.

%%%%%%%%%%%%%%%%%%%%%%
%

%
Since the jump matrices are exponentially close to the identity matrix uniformly, the contribution of the poles give the leading order asymptotic expansion, and the whole contours give exponentially small contribution. We can conclude that the solution of the Cauchy problem  \eqref{MKdV}, \eqref{ic0} in the domain $\xi>\frac{c_{\l}^2}{3} + \frac{c_{\r}^2}{6}$
has the following asympotics as $t\to\infty:$
\begin{itemize}
\item[(1)] For those $\xi, t$ such that  $|\Im g(\kappa_j,\xi)|>\delta$ for all $\kappa_j$ we have
$$q(x,t)=c_{\r}+\mathcal{O}(\e^{-C t}),$$
for some constant $C>0$.
\item [(2)] For those $\xi, t$ such that there exists $\kappa_j$,  $\Im \kappa_j>0$, with $|\Im g(\kappa_j, \xi)|<\delta,$ 
we have 
$$q(x,t)=
\begin{cases}
&q_{breath}(x,t; c_{\l}, \kappa_j, \widehat\nu_j)+\mathcal{O}(\e^{-C t})\quad \textrm{ if } \Re\kappa_j>0, \\
&q_{sol}(x,t; c_{\l}, \kappa_j, x_j)+\mathcal{O}(\e^{-C t})\quad \textrm{ if } \Re\kappa_j=0,
\end{cases}
$$
\noindent for some  constant  $C>0$. Here $\widehat\nu_j$  and $x_j$ are  defined in \eqref{chihat_j}.
\end{itemize}
We thus have finished the proof of Theorem \ref{thrm:asymp:rl}   part (a).
\color{black}

\section{Subleading  term of asymptotic expansion  in the left regions $\xi<-\frac{c_-^2}{2}$ and $\frac{-c_-^2}{2}<\xi<-\frac{c_-^2}{2}+c_+^2$}

\subsection{Second term of asymptotics in the utmost left constant region $\xi<-\frac{c_-^2}{2}$ or $x<-6c_-^2t$}
\label{sect_SecondTermUtmost}

\begin{figure}[ht!]
\centering
\begin{tikzpicture}
\draw[->,
%decoration={markings, mark=at position 0.5 with {\arrow{>}}}, 
decoration={markings, mark=at position 0.1 with {\arrow{>}}},
decoration={markings, mark=at position 0.465 with {\arrow{>}}}, postaction = {decorate},
very thick]
(-5,2) to (-4,2) [out = 0 ] to [in = 135] (-2,0)
to [out=-45, in=180] (0,-1.7);
\draw[very thick, 
decoration={markings, mark=at position 0.5 with {\arrow{>}}},
decoration={markings, mark=at position 0.9 with {\arrow{>}}},
postaction={decorate}
] (0,-1.7)
to 
[out=0, in=-135] (2,0)
[out=45, in=180] to (4,2) to[out=0] (5,2);

% lower part 
\draw[->,
%decoration={markings, mark=at position 0.5 with {\arrow{>}}}, 
decoration={markings, mark=at position 0.1 with {\arrow{>}}},
decoration={markings, mark=at position 0.465 with {\arrow{>}}}, postaction = {decorate},
very thick]
(-5,-2) to (-4,-2) [out = 0 ] to [in = -135] (-2,0)
to [out=45, in=180] (0,1.7);
\draw[very thick, 
decoration={markings, mark=at position 0.5 with {\arrow{>}}},
decoration={markings, mark=at position 0.9 with {\arrow{>}}},
postaction={decorate}
] (0,1.7)
to 
[out=0, in=135] (2,0)
[out=-45, in=180] to (4,-2) to[out=0] (5,-2);

%
%
%\draw[->,
%%decoration={markings, mark=at position 0.5 with {\arrow{>}}}, 
%decoration={markings, mark=at position 0.15 with {\arrow{>}}},
%decoration={markings, mark=at position 0.93 with {\arrow{>}}}, postaction = {decorate},
%very thick]
%(-5,-2) to (-4,-2) [out = 0 ] to [in = -135] (-2,0)
%to [out=45, in=180] (0,1.7)
%to 
%[out=0, in=135] (2,0)
%[out=-45, in=180] to (4,-2) to[out=0] (5,-2);

%vertical line
\draw[very thick, postaction = decorate, decoration={markings, mark= at position 0.25 with {\arrow{>}}},
decoration={markings, mark=at position 0.75 with {\arrow{>}}}]
(0,1.5) to (0,-1.5);

\scriptsize{
\node at (0.5, 0.5){$\begin{pmatrix}0&\i\\\i&0\end{pmatrix}$};

\node at (4,1.3){$\begin{pmatrix}1&0\\\frac{-r\e^{2\i t g}}{T^2}&1\end{pmatrix}$};
\node at (-4,1.3){$\begin{pmatrix}1&0\\\frac{-r\e^{2\i t g}}{T^2}&1\end{pmatrix}$};

\node at (4.2,-1.3){$\begin{pmatrix}1 & -\ol{r}T^2\e^{-2\i t g}\\0&1\end{pmatrix}$};
\node at (-4.2,-1.3){$\begin{pmatrix}1 & -\ol{r}T^2\e^{-2\i t g}\\0&1\end{pmatrix}$};
\draw[ dashed] (-5,0) to (5,0);
\node at (-2,0.4){$-k_0$};
\node at (2,0.4){$k_0$};
\node at (3,2){$L_1$};\node at (-3.3,2){$L_1$}; \node at (2.6,-1.6){$L_2$};
\node at (-3,-1.6){$L_2$};
\node at (1.2,1.3){$L_3$};\node at (1.3,-1.3){$L_4$};

\node at (0,2){$\begin{pmatrix}1 &\frac{-\ol{r}T^2\e^{-2\i t g}}{1+r\ol{r}}\\0&1\end{pmatrix}$};
\node at (0,-2.1){$\begin{pmatrix}1 & 0\\\frac{-r\e^{2\i t g}}{(1+r\ol{r})T^2}&1\end{pmatrix}$};
}
\end{tikzpicture}
\caption{The final RH problem in the utmost left constant region $\xi<-\frac{c_-^2}{2}.$}
\label{Fig_ParCyl_ut}
\end{figure}
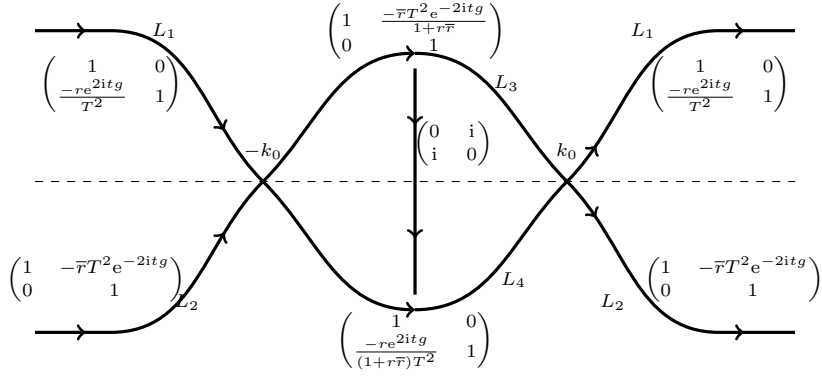

The final RH problem in the utmost left constant region $\xi<-\frac{c_-^2}{2}$ is depicted in  Figure \ref{Fig_ParCyl_ut}.
Note that the jump matrices are exponentially close to the identity matrix everywhere except for the stationary points $\pm k_0,$ where $k_0=\sqrt{-\xi-\frac{c_-^2}{2}}$ and $g'(k_0,\xi)=0.$
%
%The idea behind the parametrix analysis is to solve local RHPs where the jump matrix is not close to the identity matrix uniformly, and then construct a piecewise approximation to the original exact RHP. Unlike the other parametrixes like the Airy parametrix, the parametrix in parabolic cylinder functions cannot solve the local problem exactly, it only solves it approximately.
%
We will now construct a piecewise matrix-valued function, which satisfies approximately the jumps in the vicinities of the points $k=\pm k_0.$
%
%We will now construct approximate solution in the vicinities of the points $\pm k_0,$ called local parametrixes.
%

\medskip\noindent {\bf Function $\boldsymbol{g(k,\xi)}.$}
We start by inspecting the behavior of the phase function $g(k,\xi)$ in the neighborhoods of the points $\pm k_0.$ Note that $g$ is an odd function, $g(k)=-g(-k).$ We have
\begin{equation*}
g(k,\xi)=g(k_0,\xi)+\frac12g''(k_0,\xi)(k-k_0)^2+\mathcal{O}\((k-k_0)^3\),\quad k\to k_0,
\end{equation*}
where $g''(k_0)=\frac{-24(\xi+\frac{c_-^2}{2})}
{\sqrt{\frac{c_-^2}{2}-\xi}}>0,$ and 
\begin{equation*}
g(k,\xi)=-g(k_0,\xi)-\frac12g''(k_0,\xi)(k-k_0)^2+\mathcal{O}\((k-k_0)^3\),\quad k\to -k_0. \end{equation*}
This prompts us to introduce the local conformal changes of variables,
\begin{equation*}
z^2 =: g(k,\xi)-g(k_0,\xi),\;\mbox{as $ k\to k_0$},
\qquad
z_l^2 =: -g(k,\xi)-g(k_0,\xi),\;\mbox{as $ k\to -k_0$},
\end{equation*}
where $g(k_0,\xi)=-8\(-\xi+\frac{c_-^2}{2}\)^{\frac32}.$
We also introduce the rescaled variables 
\begin{equation*}
\zeta:= \sqrt{t\,}\,z,\quad \zeta_l:=\sqrt{t\,}\, z_l.
\end{equation*}
Then, using the expansion of the function $g(k)$ near $\pm k_0$ we have 
%\begin{equation*}
\begin{multline*}
z=\sqrt{\frac{g''(k_0,\xi)}{2}}(k-k_0)\(1+\mathcal{O}(k-k_0)\),\; k\to k_0,
\\
z_l = -\sqrt{\frac{g''(k_0,\xi)}{2}}(k+k_0)\(1+\mathcal{O}(k+k_0)\), k\to-k_0.
\end{multline*}
%\end{equation*}
Note that the {\bf l}eft variable $z_l$ is rotated $180$ degrees compared to the orientation of $k+k_0,$ while the right variable $z$ preserves the orientation of $k-k_0$.

\medskip\noindent {\bf Function $\boldsymbol{T(k).}$}
Next, we need to understand the behavior of the function $T(k)$  defined in \eqref{T_cl1} as $k$ approaches the points $\pm k_0.$ Note that the singular behavior of $T(k)$ is described 
 by the function $\(\frac{k-k_0}{k+k_0}\)^{-\i\nu},$ where $\nu=\frac{1}{2\pi}\ln(1+|r(k_0)|^2).$  We define the function $\chi(k)$ in such a way that 
\begin{equation*}
T(k)=:\(\frac{k-k_0}{k+k_0}\)^{-\i\nu(\xi)}\chi(k,\xi),
\quad \mbox{ where }
\nu(\xi)=\frac{1}{2\pi}\ln\(1+|r(k_0)|^2\),
\end{equation*}
and the function $\chi$ is meromorphic in $\mathbb{C}\setminus[-k_0,k_0]\cup[\i c_-,-\i c_-],$ and has a non-zero limit as $k\to\pm k_0$ (we thus have that the function $\chi$ is continuous at the points $\pm k_0,$ though it has discontinuity across the segment $[-k_0,k_0]$).

Note that the function $T(k)$ possesses the following symmetries:
\begin{equation*}\ol{T(-\ol k)} = T(k),\quad \ol{T(\ol k)} T(k)=1,
\quad
\mbox{ and thus     $|\chi(k_0)|=1$  and
$\chi(-k_0)=\frac{1}{\chi(k_0)}$.}
\end{equation*}
Singling out the main components of the function $T(k),$ we can write it as 
\begin{equation*}
T(k)=\zeta^{-\i\nu}\cdot\(\frac{k-k_0}{ \, (k+k_0) \sqrt{t\,}z}\)^{-\i\nu}\cdot\chi(k,\xi),\;\; k\to k_0,
\end{equation*}
\begin{equation*}
T(k)=\zeta_l^{\i\nu}\cdot\(\frac{(k-k_0)z_l\sqrt{t}}{k+k_0}\)^{-\i\nu}\cdot\chi(k,\xi),\;\;
k\to-k_0.
\end{equation*}

\subsubsection{Parabolic cylinder functions}
\label{sect_ParCyl}
The local solution  near the points $\pm k_0$  of the RH problem in Figure~\ref{Fig_ParCyl}  is obtained though the Parabolic Cylinder functions.
%Indeed if we consider the matrix function 
%\begin{equation*}
%P(k)=X(k)T(k)^{\sigma_3}e^{-itg(k,\xi)}
%\end{equation*}
%such matrix near $k_0$ has the jumps  illustrate in the Figure~\ref{Fig_P}  and similar jumps can be obtained at the point $-k_0$.
%\begin{figure}[ht!]
%\begin{tikzpicture}
%\draw[very thick, 
%decoration={markings, mark=at position 0.4 with {\arrow{>}}},
%decoration={markings, mark=at position 0.7 with {\arrow{>}}},
%postaction={decorate}
%] (0,-1.7)
%to 
%[out=0, in=-135] (2,0)
%[out=45, in=180] to (4,2);% to[out=0] (5,2);
%
%\draw[ very thick, ] (-1,0) to (2,0);
%\draw[very thick, 
%decoration={markings, mark=at position 0.3 with {\arrow{>}}},
%decoration={markings, mark=at position 0.7 with {\arrow{>}}},
%postaction={decorate}
%] (0,1.7)
%to 
%[out=0, in=135] (2,0)
%[out=-45, in=180] to (4,-2);
%% to[out=0] (5,-2);
%\scriptsize{
%\node at (0.5, 0.2){$(1+|r|^2)^{-\sigma_3}$};
%
%\node at (4,1.3){$\begin{pmatrix}1&0\\ -r&1\end{pmatrix}$};
%
%\node at (4.2,-1.3){$\begin{pmatrix}1 & -\ol{r}\\0&1\end{pmatrix}$};
%
%
%\node at (2,0.4){$k_0$};
%\node at (3,2){$L_1$}; \node at (2.6,-1.6){$L_2$};
%\node at (1.2,1.3){$L_3$};\node at (1.3,-1.3){$L_4$};
%
%\node at (0,2){$\begin{pmatrix}1 &\frac{-\ol{r}}{1+r\ol{r}}\\0&1\end{pmatrix}$};
%\node at (0,-2.1){$\begin{pmatrix}1 & 0\\\frac{-r}{(1+r\ol{r})}&1\end{pmatrix}$};
%}
%\end{tikzpicture}
%\caption{Jump of $P(k)$ near $k_0$}
%\label{Fig_P}
%\end{figure}
%
The goal of this section  is to construct a matrix function that reproduces in an approximate way the jumps of $X(k)$ in a neighbourhood of $\pm k_0$.
We follow \cite{Its81} to obtain the local parametrices.
The Parabolic Cylinder function $D_a(z)$, $a\in\mathbb{C},$ is  an entire function of $z,$  defined as the solution of the equation
\begin{equation*}
D_a''(z)+(a+\frac12-\frac{z^2}{4})D_a(z)=0,
\end{equation*}
with the asymptotic behaviour 
\begin{equation*}
D_a(z)=z^{a}\e^{-\frac14z^2}(1+\mathcal{O}(z^{-1}))\quad
\mbox{ as }z\to\infty \mbox{ inside the cone }
\arg z\in(-\frac{3\pi}{4},\frac{3\pi}{4}).
\end{equation*}
Furthermore, the parabolic cylinder function satisfies the  relations
\begin{equation}
\label{PC_rel}
\begin{split}
&
D_a(z)=D_a(-z)\e^{-\pi\i a}+\frac{\sqrt{2\pi}}{\Gamma(-a)}\e^{-\frac{\pi\i}{2}(a+1)}D_{-a-1}(\i z),
\\
&
D_a(z)=D_a(-z)\e^{\pi\i a}+\frac{\sqrt{2\pi}}{\Gamma(-a)}\e^{\frac{\pi\i}{2}(a+1)}D_{-a-1}(-\i z).
\end{split}
\end{equation}

Let $r_*$ and $ \rho_*$ be two nonzero  complex numbers such that $r_*\rho_*\neq -1$ and  denote  by $\nu=\frac{1}{2\pi}\ln\(1+r_*\rho_*\).$
Let us define the function
\begin{equation*}
\Psi(\zeta) = 
\begin{bmatrix}
2^{\i\nu}\e^{\frac{3\pi\nu}{4}}D_{-\i\nu}(2\e^{-\frac{3\pi\i}{4}}\zeta)
&
\beta_1 2^{-\i\nu}\e^{-\frac{\pi\nu}{4}}D_{\i\nu-1}(2\e^{-\frac{\pi\i}{4}}\zeta)
\\
\beta_2 2^{\i\nu}\e^{\frac{3\pi\nu}{4}}D_{-\i\nu-1}(2\e^{-\frac{3\pi\i}{4}}\zeta)
&
2^{-\i\nu}\e^{-\frac{\pi\nu}{4}}D_{\i\nu}(2\e^{-\frac{\pi\i}{4}}\zeta)
\end{bmatrix},
\quad \Im\zeta>0,
\end{equation*}
\begin{equation*}
\Psi(\zeta)=
\begin{bmatrix}
2^{\i\nu}\e^{-\frac{\pi\nu}{4}}D_{-\i\nu}(2\e^{\frac{\pi\i}{4}}\zeta)
&
-\beta_1 2^{-\i\nu}\e^{\frac{3\pi\nu}{4}}D_{\i\nu-1}(2\e^{\frac{3\pi\i}{4}}\zeta)
\\
-\beta_2 2^{\i\nu}\e^{-\frac{\pi\nu}{4}}D_{-\i\nu-1}(2\e^{\frac{\pi\i}{4}}\zeta)
&
2^{-\i\nu}\e^{\frac{3\pi\nu}{4}}D_{\i\nu}(2\e^{\frac{3\pi\i}{4}}\zeta)
\end{bmatrix},
\quad \Im\zeta<0,
\end{equation*}
where 
\begin{equation*}
\beta_1=\frac{-\i\sqrt{2\pi}\,2^{2\i\nu}\e^{\frac{\pi\nu}{2}}}{r_* \Gamma(\i\nu)}
=
\frac{\rho_* \Gamma(-\i\nu+1)2^{2\i\nu}}{\sqrt{2\pi}\, \e^{\frac{\pi\nu}{2}}},
\quad
\beta_2=
\frac{r_* \Gamma(1+\i\nu)}{\sqrt{2\pi}\, 2^{2\i\nu} \e^{\frac{\pi\nu}{2}}}
=
\frac{\i\sqrt{2\pi}\,\e^{\frac{\pi\nu}{2}}}{\rho_*2^{2\i\nu}\Gamma(-\i\nu)},
\end{equation*}
with $ \beta_1\beta_2=\nu.$
The function $\Psi(\zeta)$ is a piecewise analytic function   and using the relations \eqref{PC_rel}, it has the following jump over the real axis:
\begin{equation*}
\Psi_-(\zeta)=\Psi_+(\zeta)\begin{pmatrix}1 & -\rho_* \\ -r_* & 1+r_*\rho_*\end{pmatrix},\quad \zeta\in\mathbb{R},
\end{equation*}
where $\Psi_{\pm}(\zeta)$ are the boundary values of $\Psi(\zeta)$ as $\zeta$ approaches the oriented real line.

%Note that $\beta_1=\ol{\beta_2}$ if $\ol{r_*}=\rho_*,$ and 
% that $\beta_1=-\ol{\beta_2}$ if $\ol{r_*}=-\rho_*.$
%
%
%The function $\Psi$ has continuous boundary values on the real axis, but it does not have a uniform asymptotics at infinity. The next function $\Phi$ `improves' the behavior at infinity (it has a uniform asymptotics at infinity), at the cost of spoiling the behavior at the origin (it is not continuous at the origin anymore),
Next we introduce the function $\Phi(\zeta)$ defined as 
\begin{equation}
\label{Phi}
\Phi(\zeta)=
\begin{cases}
&\Psi(\zeta)\zeta^{\i\nu\sigma_3}\e^{\i\zeta^2\sigma_3},
\quad
\arg\zeta\in(\frac{\pi}{4},\frac{3\pi}{4})\cup(-\frac{3\pi}{4},-\frac{\pi}{4}),
\\
&
\Psi(\zeta)\zeta^{\i\nu\sigma_3}\e^{\i\zeta^2\sigma_3}
\begin{pmatrix}
1&0\\-r_*\zeta^{2\i\nu}\e^{2\i\zeta^2}
&1
\end{pmatrix},\quad\arg\zeta\in(0,\frac{\pi}{4}),
\\
&
\Psi(\zeta)\zeta^{\i\nu\sigma_3}\e^{\i\zeta^2\sigma_3}
\begin{pmatrix}
1& \rho_*\zeta^{-2\i\nu}\e^{-2\i\zeta^2} \\
0&1
\end{pmatrix},\quad\arg\zeta\in(-\frac{\pi}{4},0),
\\
&
\Psi(\zeta)\zeta^{\i\nu\sigma_3}\e^{\i\zeta^2\sigma_3}
\begin{pmatrix}
1& \frac{-\rho_*}{1+r_*\rho_*}\zeta^{-2\i\nu}\e^{-2\i\zeta^2} \\
0&1
\end{pmatrix},\quad\arg\zeta\in(\frac{3\pi}{4},\pi),
\\
&
\Psi(\zeta)\zeta^{\i\nu\sigma_3}\e^{\i\zeta^2\sigma_3}
\begin{pmatrix}
1 & 0 \\
 \frac{r_*}{1+r_*\rho_*}\zeta^{2\i\nu}\e^{2\i\zeta^2} & 1
\end{pmatrix},\quad\arg\zeta\in(-\pi, -\frac{3\pi}{4}).
\end{cases}
\end{equation}

%\begin{equation*}
%\Phi(\zeta)=
%\begin{cases}
%&\Psi(\zeta)\zeta^{\i\nu\sigma_3}\e^{\i\zeta^2\sigma_3},
%\quad
%\arg\zeta\in(\frac{\pi}{4},\frac{3\pi}{4})\cup(-\frac{3\pi}{4},-\frac{\pi}{4}),
%\\
%&\Psi(\zeta)\zeta^{\i\nu\sigma_3}\e^{\i\zeta^2\sigma_3}
%\begin{pmatrix}
%1&0\\-r_*\zeta^{2\i\nu}\e^{2\i\zeta^2}
%&1
%\end{pmatrix},\quad\arg\zeta\in(0,\frac{\pi}{4}),
%\\
%&
%\Psi(\zeta)\zeta^{\i\nu\sigma_3}\e^{\i\zeta^2\sigma_3}
%\begin{pmatrix}
%1& \rho_*\zeta^{-2\i\nu}\e^{-2\i\zeta^2} \\
%0&1
%\end{pmatrix},\quad\arg\zeta\in(-\frac{\pi}{4},0),
%\\
%&\Psi(\zeta)\zeta^{\i\nu\sigma_3}\e^{\i\zeta^2\sigma_3}
%\begin{pmatrix}
%1& \frac{-\rho_*}{1+r_*\rho_*}\zeta^{-2\i\nu}\e^{-2\i\zeta^2} \\
%0&1
%\end{pmatrix},\quad\arg\zeta\in(\frac{3\pi}{4},\pi),
%\\
%&\Psi(\zeta)\zeta^{\i\nu\sigma_3}\e^{\i\zeta^2\sigma_3}
%\begin{pmatrix}
%1 & 0 \\
% \frac{r_*}{1+r_*\rho_*}\zeta^{2\i\nu}\e^{2\i\zeta^2} & 1
%\end{pmatrix},\quad\arg\zeta\in(-\pi, -\frac{3\pi}{4}).
%\end{cases}
%\end{equation*}
The function $\Phi(\zeta)$ satisfies the RH problem  $\Phi_-(\zeta)=\Phi_+(\zeta)J_{\Phi}(\zeta),$
where

\begin{equation}
\label{J_Phi}\begin{split}
&J_{\Phi}(\zeta)
=\begin{pmatrix}
1&0\\-r_*\zeta^{2\i\nu}\e^{2\i\zeta^2}
&1
\end{pmatrix}, \;\;\zeta\in(0,\e^{\frac{\pi\i}{4}}\infty),\\
&J_{\Phi}(\zeta)=\begin{pmatrix}
1& -\rho_*\zeta^{-2\i\nu}\e^{-2\i\zeta^2} \\
0&1
\end{pmatrix},\;\; \zeta\in(0,\e^{-\frac{\pi\i}{4}}\infty),
\\
&J_{\Phi}(\zeta)=
\begin{pmatrix}
1& \frac{-\rho_*}{1+r_*\rho_*}\zeta^{-2\i\nu}\e^{-2\i\zeta^2} \\
0&1
\end{pmatrix},
\;\;\zeta\in(\e^{\frac{3\pi\i}{4}}\infty,0),
\\
&J_{\Phi}(\zeta)=
\begin{pmatrix}
1& 0 \\ \frac{-r_*}{1+r_*\rho_*}\zeta^{2\i\nu}\e^{2\i\zeta^2} & 1
\end{pmatrix},\;\;
\zeta\in(\e^{-\frac{3\pi\i}{4}}\infty,0).
\end{split}
\end{equation}
Besides, the function $\Phi(\zeta)$ has the following uniform asymptotics as $\zeta\to\infty,$
\begin{equation*}
\Phi(\zeta)=
\begin{bmatrix}
1 -\frac{\nu(1+\i\nu)}{8\zeta^2}(1 +\i\frac{(2+\i\nu)(3+\i\nu)}{16\zeta^2})
&
\frac{\e^{\frac{\pi\i}{4}}\beta_1}{2\zeta}
+\frac{\e^{-\frac{\pi\i}{4}}\beta_1(1-\i\nu)(2-\i\nu)}{16\zeta^3}
\\
\frac{\e^{\frac{3\pi\i}{4}}\beta_2}{2\zeta}
+\frac{\e^{-\frac{3\pi\i}{4}}\beta_2(1+\i\nu)(2+\i\nu)}{16\zeta^3}
&
1-\frac{\nu(1-\i\nu)}{8\zeta^2}(1-\i\frac{(2-\i\nu)(3-\i\nu)}{16\zeta^2})
\end{bmatrix}+\mathcal{O}(\frac{1}{\zeta^{5}}).
\end{equation*}
\begin{figure}[ht!]
\centering
\begin{tikzpicture}
\draw
[very thick,
decoration={markings, mark=at position 0.5 with {\arrow{>}}},
postaction = {decorate}]
(-1.5,1.5) to (0,0);
\draw
[very thick,
decoration={markings, mark=at position 0.5 with {\arrow{>}}},
postaction = {decorate}]
(-1.5, -1.5) to (0,0);
\draw
[very thick,
decoration={markings, mark=at position 0.5 with {\arrow{>}}},
postaction = {decorate}]
(0,0) to (1.5, 1.5);
\draw
[very thick,
decoration={markings, mark=at position 0.5 with {\arrow{>}}},
postaction = {decorate}]
(0,0) to (1.5, -1.5);
%
%\draw
%[very thick,
%decoration={markings, mark=at position 0.5 with {\arrow{>}}},
%postaction = {decorate}]
%(-2.5,0) to (0, 0);
\scriptsize{
\node at (-3.4,1.2){$\begin{pmatrix}
1& \frac{-\rho_*}{1+r_*\rho_*}\zeta^{-2\i\nu}\e^{-2\i\zeta^2}\\
0&1
\end{pmatrix}$};
\node at (-3.6,-1){$\begin{pmatrix}
1& 0 \\ \frac{-r_*}{1+r_*\rho_*}\zeta^{2\i\nu}\e^{2\i\zeta^2}& 1
\end{pmatrix}$};
\node at (2.8,1){$\begin{pmatrix}
1&0\\-r_*\zeta^{2\i\nu}\e^{2\i\zeta^2} 
&1
\end{pmatrix}$};
\node at (2.8,-1){$\begin{pmatrix}
1&  -\rho_*\zeta^{-2\i\nu}\e^{-2\i\zeta^2}  \\
0&1
\end{pmatrix}$};
\node at (0,-0.3){$0$};
\node at (0,0.7){$\Large{\Sigma}$};
%\node at (-1.5, 0.3){$(1+r_*\rho_*)^{-\sigma_3}$};
}
\end{tikzpicture}
\caption{The contour $\Sigma$ and jumps for the $\Phi(\zeta),$ $\Phi_{-}=\Phi_{+}J_{\Phi}.$}
\label{Fig_PC}
\end{figure}
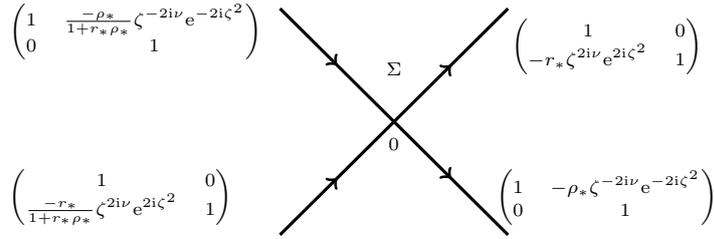

%The next piece which we will approximate is the function $T,$ which exhibits discontinuous behavior at the point $k=k_0.$
%Note that the `bad' behavior of $T$ at the point $k=k_0$ is played by the function $(k-k_0)^{-\i\nu},$ with $\nu=\frac{1}{2\pi}\ln(1+|r(k_0)|^2)>0,$ and hence the ratio of $T$ and the latter behaves well at the point $k=k_0,$
%\begin{equation*}
%T(k)=\(\frac{k-k_0}{k+k_0}\)^{-\i\nu}\chi(k),\quad \nu=\frac{1}{2\pi}\ln\(1+|r(k_0)|^2\)>0,
%\end{equation*}
%and $\chi$ is continuous at the point $k=k_0$ in the sense, that it has a non-zero limit. Note that the function $\chi$ has a discontinuity across $k\in[-k_0,k_0],$ but the limits as $k\to k_0$ from the lower and upper half-planes coincide, and hence there is no contradiction between continuity of $\chi$ at the point $k=k_0$ and discontinuity of the function $\chi$ across the segment $[-k_0,k_0].$
%
%

\subsubsection{Approximation, far away from breathers}
\label{sect_ApproximationParCylUtmost}
Now we are almost ready to define an approximation for the function $X(k)$. For this  purpose  we define  
 {\bf r}ight and {\bf l}eft parametrices
\begin{equation*}
\begin{split}
P_{r}(\zeta) = \Phi(\zeta)\Big|_{r_*=r(k_0),\rho_*=\ol{r(k_0)}}\ ,
\qquad
P_{l}(\zeta_l)={\scriptsize\begin{bmatrix}0&-1\\1&0\end{bmatrix}}
\Phi(\zeta_l)\Big|_{r_*=r(k_0),\rho_*=\ol{r(k_0)}}
{\scriptsize\begin{bmatrix}0&1\\-1&0\end{bmatrix}}.
\end{split}
\end{equation*}
where $\Phi$ has been defined in the previous Section~\ref{sect_ParCyl}.
Note that the parameters $r_*, \rho_*$ are the same for both $P_r$ and $P_l$, and that $P_l$ is additionally shuffled by  off-diagonal matrices.

Let  $\delta>0,$  such that the conformal changes of coordinates $z, z_l$ exist inside the disks  $${\mathcal U}_{\pm k_0}=\{k\in\C\, \,s.t.\;|k\pm k_0|\leq \delta\}.$$
We define the boundaries of the disks ${\mathcal U}_{\pm k_0}$ as $\partial {\mathcal U}_{\pm k_0}$ and we assume that they are oriented anticlockwise.
Furthermore we define 
\begin{align}
&P^{(k_0)}(k)=\mathcal{B}_r(k)P_{r}(\zeta(k))
\cdot\e^{-\i t g(k_0)\sigma_3}\chi(k_0)^{\sigma_3}
\(\frac{(k-k_0)}{(k+k_0)z\sqrt{t}}\)^{\i\nu\sigma_3},\quad k\in {\mathcal U}_{k_0}, \\
& P^{(-k_0)}(k)=\mathcal{B}_l(k)P_{l}(\zeta_l(k))
\cdot\e^{-\i t g(k_0)\sigma_3}\chi(k_0)^{\sigma_3}
\(\frac{(k-k_0)z_l\sqrt{t}}{k+k_0}\)^{\i\nu\sigma_3},\quad k\in {\mathcal U}_{-k_0},
\end{align}
where $\mathcal{B}_r(k)$ and $\mathcal{B}_l(k)$ are matrices that will be determine below.
We are now ready to define the approximate solution of $X(k)$  as 
\begin{equation*}
\begin{split}
X_{appr}(k)=\begin{cases}
P^{(k_0)}(k)\;
\mbox{   for $k\in  {\mathcal U}_{k_0},$}
\\
 P^{(-k_0)}(k)\;  \mbox{    for $k\in {\mathcal U}_{-k_0},$}
\\
M^{(mod)}(k)\; \mbox{for $k\in \C\backslash\{ {\mathcal U}_{-k_0}\cup  {\mathcal U}_{k_0}\}$},
\end{cases}
\end{split}
\end{equation*}
where $M^{(mod)}(k)=\begin{bmatrix}a_{\gamma}(k) & b_{\gamma}(k) \\b_{\gamma}(k) & a_{\gamma}(k)\end{bmatrix},$ $a_{\gamma}(k)=\frac12(\gamma(k)+\frac{1}{\gamma(k)}),$
$b_{\gamma}(k)=\frac12(\gamma(k)-\frac{1}{\gamma(k)}),$ $\gamma(k)=\sqrt[4]{\frac{k-\i c_-}{k+\i c_-}}.$

The matrices $\mathcal{B}_{r,l}(k)$ are obtained by requiring that 
\begin{equation*}
 P^{(\pm k_0)}(k)=M^{(mod)}(k)(I+o(1)) \quad  \mbox{as $t\to \infty$} \mbox{ and $k\in \partial{\mathcal U}_{\pm k_0}\backslash \Sigma^{\pm}$},  
 \end{equation*}
 where $\Sigma^{\pm }=\cup_{j=1}^4L_j\cap {\mathcal U}_{\pm k_0}$.
 This gives 
 \begin{equation}
 \label{Brl}
\begin{split}
\mathcal{B}_r(k)=M^{(mod)}(k)\e^{-\i t g(k_0,\xi)\sigma_3}\chi(k_0,\xi)^{\sigma_3}
\(\frac{k-k_0}{\sqrt{t}\,z(k+k_0)}\)^{-\i\nu\sigma_3}, \mbox{    for $k\in {\mathcal U}_{k_0},$}
\\
\mathcal{B}_l(k)=M^{(mod)}(k)\e^{\i t g(k_0,\xi)\sigma_3}\chi(k_0,\xi)^{-\sigma_3}
\(\frac{(k-k_0)z_l\sqrt{t}}{k+k_0}\)^{-\i\nu\sigma_3},  \mbox{    for $k\in {\mathcal U}_{-k_0}$}.
\end{split}
\end{equation}
We observe that the matrices $\mathcal{B}_{l,r}(k)$ are holomorphic inside the circles ${\mathcal U}_{\pm k_0}$.
%
%Note that the paramaters $r_*$ and $\rho_*$ are switched in the parametrixes at the  points $k_0$ and $-k_0,$ and the parametrix at the point $-k_0$ is multiplied by off-diagonal matrices.
 Next we consider the  error matrix $E(k)$ 
 \begin{equation*}
 E(k)=X(k)X_{appr}(k)^{-1}\,.
 \end{equation*}
Such  a matrix has jump $J_E(k)=(E_+(k))^{-1} E_-(k).$ It is easy to check that $J_E(k)$ is exponentially close to the identity as $t\to \infty$ for $
k\in \C\backslash  \overline{{\mathcal U}_{\pm k_0}}$.
We observe that since $P_{r,l}(\zeta)=I+\mathcal{O}(t^{-1/2})$ on the circles $ \partial{\mathcal U}_{\pm k_0}$
we also have  $J_{E}(k)=I+\mathcal{O}(t^{-1/2})$ as $t\to\infty$ and $k\in  \partial{\mathcal U}_{\pm k_0}$.
Finally, let us estimate the jump matrix $J_{E}$ on contours $\Sigma^\pm= \cup_{j=1}^4L_j\cap  \partial{\mathcal U}_{\pm k_0}$.
The terms which appear in the jump matrix $J_{E}(k)$   for $k\in \Sigma^+\cup\Sigma^-$ 
are of the form
$$\(r(k) - r(k_0)\)\zeta^{2\i\nu}\e^{2\i\zeta^2},
\quad 
\(\frac{r(k)}{1+r(k)\ol{r(\ol k)}} - \frac{r(k_0)}{1+|r(k_0)|^2}\)\zeta^{2\i\nu}\e^{2\i\zeta^2}$$
and they can be estimated as 
\begin{equation*}
\begin{split}
|\(r(k) - r(k_0)\)\zeta^{2\i\nu}\e^{2\i\zeta^2}|
&\leq C
\max\limits_{k\in\Sigma^\pm} \(|k-k_0|
\e^{-2|\zeta|^2}\)\\
& \leq C
\max\limits_{z\geq0}
\(|z|
\e^{-2t |z|^2}\)
=
\frac{C}{\sqrt{t}}
\max\limits_{|z|\geq0}
\(|z|\cdot
\e^{-2 |z|^2}\)
\end{split}
\end{equation*} 
with some generic constant $C>0$ (taking the maximum, we first allowed $|z|$ to vary over the whole positive numbers $[0,+\infty)$ instead of a finite interval, and second, we  make the transformation $|z|\leadsto\sqrt{t}|z|$),
%where $K_0>0$ 
%and $\tilde{k}\in \partial{\mathcal U}_{\pm k_0}$,  
and similarly
\begin{equation*}
\left|\(\frac{r(k)}{1+r(k)\ol{r(\ol k)}} - \frac{r(k_0)}{1+|r(k_0)|^2}\)\zeta^{2\i\nu}\e^{2\i\zeta^2}\right|=\mathcal{O}(\frac{1}{\sqrt{t}}).
\end{equation*}
% \todo{T.Check  if it is clear.
%\\
% A: I modified slightly the formula.}

We conclude that  jump matrix $J_{E}$ admits an estimate $I+\mathcal{O}(t^{-1/2})$ uniformly on the whole contour $\Sigma_{E}= \Sigma^\pm\cup\partial {\mathcal U}^\pm$.
\subsubsection{Second term of the asymptotics}
We thus can look for the  error matrix  $E(k)$ in the form 
\begin{equation}\label{Cauchy1}
E(k)=I+\mathcal{C}(F),
\end{equation}
where $\mathcal{C}$ is the Cauchy operator,
$$\mathcal{C}f(k)=\frac{1}{2\pi\i}\int_{\Sigma_{E}}\frac{f(s)\d s}{s-k},$$
where  $\Sigma_{E}= \Sigma^\pm\cup\partial{\mathcal U}_{\pm k_0}$.
Taking the limiting values in \eqref{Cauchy1}, and using the relation $\mathcal{C}_+-\mathcal{C}_-=I,$ we find $F(s)=E_{+}(s)(I-J_{E}(s)).$
Thus,
\begin{equation}\label{Cauchy2}
E(k)=I+\mathcal{C}[E_{+}(s)(I-J_{E}(s))],
\quad \mbox{or }
E_{}(k)=I+\frac{1}{2\pi\i}\int_{\Sigma_E}\frac{E_{+}(s)(I-J_{E}(s)) \d s}{s-k},
\end{equation}
where $E_{+}$ is the solution of the singular integral equation $E_{+}(k)=I+\mathcal{C}_+[E_{+}(s)(I-J_{E}(s))].$
Define
\begin{equation*}\begin{split}&q_{err}(x,t):=\lim\limits_{k\to\infty}2\i k(E(k)-I)_{12}=\lim\limits_{k\to\infty}2\i k(E(k)-I)_{21},
\\
&q_{appr}(x,t):=\lim\limits_{k\to\infty}2\i k(X_{appr}(k)-I)_{12}=\lim\limits_{k\to\infty}2\i k(X_{appr}(k)-I)_{21};
\end{split}\end{equation*}
then $q(x,t)=q_{err}(x,t)+q_{appr}(x,t),$ and $q_{appr}=c_-.$
From \eqref{Cauchy2} we find an expression for $q_{err},$
\begin{equation*}
q_{err}(x,t)=\frac{1}{\pi}\int_{\Sigma_E}\left[E_{+}(s)(J_{E}(s)-I)\right]_{21}\d s.
\end{equation*}
Subtracting and adding the identity matrix from $E_{+},$ and observing that since $J_{E}(k)=I+\mathcal{O}(t^{-1/2})$ uniformly on the contour $\Sigma_E$, then also 
$E(k)=I+\mathcal{O}(t^{-1/2}),$ we find that 
\begin{equation*}
q_{err}(x,t)=\frac{1}{\pi}\int_{\Sigma_E}(J_{E}(s)-I)_{21}\d s+\mathcal{O}(t^{-1}).
\end{equation*}
We start by estimating the integral over the parts of $L_j, j=1,2,3,4,$ inside the circles $|k\mp k_0|<\delta.$ They admit an estimate of the type
\begin{equation*}
\int_{0}^{\delta}y\e^{-2 ty^2}\d y<\frac{1}{t}\int_0^{\infty}(y\sqrt{t})\e^{-2ty^2}\d (y\sqrt{t})=\frac{C}{t},
\end{equation*}
where $C$ is a generic constant, and thus do not contribute to the $t^{-1/2}$ term.
The rest of the contour except for the circles around the points $\pm k_0$ gives an exponentially small contribution. Next we compute the integrals over the circles $|k\mp k_0|=\delta,$ which we orient in the counter-clock-wise direction. 
Due to the symmetries $\ol{E(-\ol k)}=E(k)={\scriptsize\begin{bmatrix}0&1\\-1&0\end{bmatrix}}E(-k){\scriptsize\begin{bmatrix}0&-1\\1&0\end{bmatrix}}$, the integrals over the circles $ \partial{\mathcal U}_{\pm k_0}$   have the following structure:
\begin{equation*}
\frac{1}{\pi}\int_{\partial{\mathcal U}_{k_0}}(J_{E}(s)-I)\d s = \begin{pmatrix}\ldots & \ol{A}\\ A & \ldots\end{pmatrix},
\quad
\frac{1}{\pi}\int_{\partial{\mathcal U}_{-k_0}}(J_{E}(s)-I)\d s = \begin{pmatrix}\ldots & A\\ \ol{A} & \ldots\end{pmatrix},
\end{equation*}
and hence $q_{err}=2\Re A+\mathcal{O}(t^{-1}).$
Note that for $k\in \partial{\mathcal U}_{k_0}$ the jump matrix equals
\begin{equation*}
\begin{split}
&J_E(k)-I=\mathcal{B}_r(k)\left[\begin{pmatrix}
0 & \e^{\frac{\pi\i}{4}}\beta_1
\\
\e^{\frac{3\pi\i}{4}}\beta_2 & 0
\end{pmatrix}\frac{1}{2\zeta}+\mathcal{O}(\zeta^{-2})\right]
(\mathcal{B}_r(k))^{-1}
\end{split}
\end{equation*}
where $\mathcal{B}_r(k)$ is defined in \eqref{Brl}.
%and $\zeta\asymp \sqrt{t}$ on the circle $C_r.$ 
Neglecting the $\mathcal{O}(\zeta^{-2})$ terms gives an error of the order $t^{-1},$ and the terms of the order $\zeta^{-1}$ have simple a pole at $k=k_0,$ and thus they can be computed by taking residues.
%\todo{T. Change notation, $F$ was already used, I introduced $\theta $
%\\
%A: $\theta$ was used as $kx+4k^3t.$ I change it to $\vartheta$}
Since 
\begin{equation*}
 |\e^{-\i t g(k_0)}\chi(k_0) \(\sqrt{t\,}\,2k_0\sqrt{\frac{g''(k_0)}{2}}\)^{\i\nu}|=1,
\end{equation*}
we denote 
\begin{equation*}
e^{\i\vartheta}=\e^{-\i t g(k_0)}\chi(k_0) \(\sqrt{t\,}\,2k_0\sqrt{\frac{g''(k_0)}{2}}\)^{\i\nu}.
\end{equation*}
Then
\begin{equation*}
\begin{split}
&A:=\frac{1}{\pi}\int_{\partial{\mathcal U}_{k_0}}(J_E(k)-I)_{21}\d k = 
\frac{1}{\pi}
\int_{\partial{\mathcal U}_{k_0}}
\(\frac{a_{\gamma}(k)^2}{e^{2\i\vartheta}}\e^{\frac{3\pi\i}{4}}\beta_2
-b_{\gamma}(k)^2e^{2\i \vartheta}\e^{\frac{\pi\i}{4}}\beta_1\)\frac{\d k}{2\zeta}
=
\\
&=2\i\, \mathrm{Res}_{k=k_0}
\(\frac{a_{\gamma}(k)^2}{e^{2\i \vartheta}}\e^{\frac{3\pi\i}{4}}\beta_2
-b_{\gamma}(k)^2 e^{2\i\vartheta}\e^{\frac{\pi\i}{4}}\beta_1\)\frac{\d k}{2\zeta}\\
&=
\frac{1}{\sqrt{\frac{t g''(k_0)}{2}}}
\(\frac{a_{\gamma}(k_0)^2}{ e^{2\i\vartheta}}\e^{\frac{-3\pi\i}{4}}\beta_2
-b_{\gamma}(k_0)^2 e^{2\i\vartheta}\e^{\frac{3\pi\i}{4}}\beta_1\).
\end{split}
\end{equation*}
We have $q_{err}(x,t)=A+\ol{A}+\mathcal{O}(t^{-1}),$ and thus we need to find the real part of the constant $A.$
Substituting the expressions for $a_{\gamma}(k)=\frac12(\gamma(k)+\frac{1}{\gamma(k)}),$
$b_{\gamma}(k)=\frac12(\gamma(k)-\frac{1}{\gamma(k)}),$ $\gamma(k)=\sqrt[4]{\frac{k-\i c_-}{k+\i c_-}},$
 and taking into account that $|\gamma(k_0)|=1$  and $\overline{\beta_1}=\beta_2$
we find
\begin{equation*}
\begin{split}
A=\frac{1}{\sqrt{\frac{t g''(k_0)}{2}}}&
\left\{
\frac12(\frac{2k_0}{\sqrt{k_0^2+c_-^2}} )
\underbrace{\left[
 e^{-2\i\vartheta}\e^{-\frac{3\pi\i}{4}}\beta_2
- e^{2\i\vartheta}\e^{\frac{3\pi\i}{4}}\beta_1
\right]}_{\in\i\mathbb{R}}+\right.\\
&\left.+
\frac12\(
 e^{-2\i\vartheta}\e^{-\frac{3\pi\i}{4}}\beta_2
+ e^{2\i\vartheta}\e^{\frac{3\pi\i}{4}}\beta_1\)
\right\}.
\end{split}
\end{equation*}
The underbraced expression is purely imaginary, and hence does not contribute to the real part of $A.$
We have
\begin{equation*}
2\Re A = \frac{2}{\sqrt{\frac{g''(k_0)t}{2}}}\Re\(e^{2\i\vartheta}\beta_1\e^{\frac{3\pi\i}{4}}\),
\end{equation*}
and substituting expressions for $\beta_1$ from Section \ref{sect_ParCyl},
we can write the latter in more detail,
\begin{equation*}
\begin{split}
2&\Re A=
 \frac{2}{\sqrt{\frac{g''(k_0)t}{2}}}
 \cdot
 \frac{\sqrt{2\pi\,}\,\e^{\frac{\pi\nu}{2}}}{|r(k_0)|\,|\Gamma(\i\nu)|}\cdot\\
 &\cdot
 \cos\left[
 \frac{\pi}{4}-\arg \left(r(k_0)\Gamma(\i\nu)\right)-2tg(k_0)+\arg\chi^2(k_0)
 +
 2\nu\ln\(4k_0\sqrt{\frac{tg''(k_0)}{2}}\)
 \right].
\end{split}
\end{equation*}
The second factor here can be computed using properties of the Gamma-function, $ \frac{\sqrt{2\pi\,}\,\e^{\frac{\pi\nu}{2}}}{|r(k_0)|\,|\Gamma(\i\nu)|} = \sqrt{\nu},$
\begin{equation*}\begin{split}
&q_{err}(x,t)=2\Re A+\mathcal{O}(t^{-1})
=\sqrt{\frac{8\nu}{t\cdot g''(k_0)}}\times\\
&\;\;\;\times
\cos\left[\frac{\pi}{4}
-2t g(k_0)+\nu\ln\(8 t k_0^2 g''(k_0)\)
-\arg \left(r(k_0)\Gamma(\i\nu)\right)+\arg\chi^2(k_0)
\right]
+\mathcal{O}(t^{-1}).
\end{split}\end{equation*}
Substituting finally the expressions for $g(k_0), g''(k_0),$ we get 
\begin{equation}\label{q_errUtmost}
\begin{split}
&q_{err}(x,t)=
\sqrt{\frac{\nu\sqrt{-\xi+\frac{c_-^2}{2}}}{3t(-\xi-\frac{c_-^2}{2})}}\times\\
&\;\;\times\cos
\left[
16t\(-\xi+\frac{c_-^2}{2}\)^{\frac32}
+\nu\ln\(\frac{192\, t\, (\xi+\frac{c_-^2}{2})^2}{\sqrt{-\xi+\frac{c_-^2}{2}}}\)
+
\phi(k_0)
\right]
+\mathcal{O}(t^{-1}),
\end{split}
\end{equation}
where the phase shift is given by the formula 
\begin{equation}\label{phaseUtmost}
\phi(k_0):=\frac{\pi}{4}-\arg r(k_0)-\arg\Gamma(\i\nu)+\arg\chi^2(k_0).
\end{equation}
Note that the formula \eqref{q_errUtmost} is very similar to the radiation part of the rarefaction wave \cite{M11}, but differs from it by the phase shift $\pi.$
Note also, that the error term is uniform in $\xi<\frac{-c_-^2}{2}-\varepsilon,$ for a  positive $\varepsilon>0.$ Looking at \eqref{q_errUtmost}, we see that it blows up on the line $\xi=-\frac{c_-^2}{2},$ which probably indicates the presence of a nontrivial transition region.

\subsection{Second term of the asymptotic expansion  in the middle left constant region $-\frac{c_-^2}{2}<\xi<-\frac{c_-^2}{2}+c_+^2$ or $-6c_-^2t<x<(-6c_-^2+12c_+^2)t$}
\label{sect_SecondTermMiddle}

%%%%
%%%%
%%%%
\begin{figure}[ht!]
\centering
\begin{tikzpicture}
%left upper part
\draw[
decoration={markings, mark=at position 0.365 with {\arrow{>}}},
decoration={markings, mark=at position 0.82 with {\arrow{<}}}, postaction = {decorate},
very thick]
(-5,0.5) to (-4,0.5) [out = 0 ] to [in = -135] (0,1)
to [out = 135, in =-90] (-1,2)
to [out=90, in=180] (0, 3);
\scriptsize{
\node at (-2.2, 2.3){$\begin{pmatrix}1 &\frac{-T^2\e^{-2\i t g}}{\widehat f} \\ 0&1\end{pmatrix}$};
}
%
%right upper part
\draw[
decoration={markings, mark=at position 0.365 with {\arrow{<}}},
decoration={markings, mark=at position 0.82 with {\arrow{<}}}, postaction = {decorate},
very thick]
(5,0.5) to (4,0.5) [out = 180 ] to [in = -45] (0,1)
to [out = 45, in = -90] (1,2)
to [out=90, in=0] (0, 3);
\scriptsize{
\node at (2.2, 2.3){$\begin{pmatrix}1 &\frac{T^2\e^{-2\i t g}}{\widehat f} \\ 0&1\end{pmatrix}$};
}
%
%left lower part
\draw[
decoration={markings, mark=at position 0.365 with {\arrow{>}}},
decoration={markings, mark=at position 0.82 with {\arrow{>}}}, postaction = {decorate},
very thick]
(-5,-0.5) to (-4,-0.5) [out = 0 ] to [in = 135] (0,-1)
to [out = -135, in =90] (-1,-2)
to [out=-90, in=-180] (0, -3);
\scriptsize{
\node at (-2, -2.3){$\begin{pmatrix}1 & 0 \\ \frac{\e^{2\i t g}}{\ol{\widehat f\,}\,T^2} &1\end{pmatrix}$};
}
%
%right lower part
\draw[
decoration={markings, mark=at position 0.365 with {\arrow{<}}},
decoration={markings, mark=at position 0.82 with {\arrow{>}}}, postaction = {decorate},
very thick]
(5,-0.5) to (4,-0.5) [out = -180 ] to [in = 45] (0,-1)
to [out = -45, in = 90] (1,-2)
to [out=-90, in=0] (0, -3);
\scriptsize{
\node at (2, -2.3){$\begin{pmatrix}1 & 0 \\ \frac{-\e^{2\i t g}}{\ol{\widehat f\,}\,T^2} &1\end{pmatrix}$};

\node at (-0.35,1){$\i d_0$}; \node at (-0.45,-1){$-\i d_0$};
\node at (-0.2,2.6){$\i c_-$}; \node at (-0.2,-2.6){$-\i c_-$};
\node at (-0.25,1.6){$\i c_+$}; \node at (-0.35,-1.6){$-\i c_+$};
}
\filldraw (0,1.6) circle (1pt);\filldraw (0,-1.6) circle (1pt);
\filldraw (0,2.5) circle (1pt);\filldraw (0,-2.5) circle (1pt);
%
% vertical part
\draw[very thick, postaction = decorate, decoration={markings, mark= at position 0.23 with {\arrow{>}}},
decoration={markings, mark=at position 0.8 with {\arrow{>}}}]
(0,2.5) to (0,-2.5);
\draw[dashed] (-5,0) to (5,0);
\scriptsize{
\node at (0.45, 2){$\begin{bmatrix}0&\i\\\i&0\end{bmatrix}$};
\node at (0.45, -2){$\begin{bmatrix}0&\i\\\i&0\end{bmatrix}$};

\node at (4,1.1){$\begin{pmatrix}1&0\\\frac{-r\e^{2\i t g}}{T^2}&1\end{pmatrix}$};
\node at (-4,1.1){$\begin{pmatrix}1&0\\\frac{-r\e^{2\i t g}}{T^2}&1\end{pmatrix}$};

\node at (4,-1.1){$\begin{pmatrix}1 & -\ol{r}T^2\e^{-2\i t g}\\0&1\end{pmatrix}$};
\node at (-4,-1.1){$\begin{pmatrix}1 & -\ol{r}T^2\e^{-2\i t g}\\0&1\end{pmatrix}$};

\node at (2.5, 0.7){$L_1$};\node at (-2.5, 0.7){$L_1$}; \node at (2.5,-0.7){$L_2$};
\node at (-2.5,-0.7){$L_2$};
\node at (-0.9,1.4){$L_5$};\node at (0.9,1.4){$L_7$};
\node at (-0.9,-1.3){$L_6$};\node at (0.9,-1.3){$L_8$};
}
\end{tikzpicture}
\caption{The final RH problem in the utmost left constant region $\xi<-\frac{c_-^2}{2}$.}
\label{Fig_ParCyl}
\end{figure}
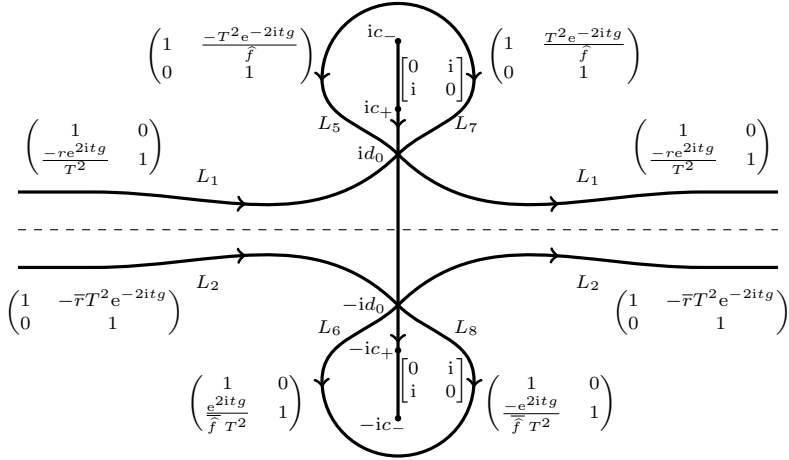
%%%%%
%%%%%
%%%%%
The jumps for the final RH problem~\ref{RHP_final_constant2}  for the matrix function $X(k)$   in the middle left constant region   considered in Section~\ref{middle_lcr} are depicted in  Figure \ref{Fig_ParCyl}.
Note that the jump matrices are exponentially close to the identity matrix everywhere except for interval $[\i c_-, -\i c_-],$  near  the stationary points $\pm \i d_0,$ where $d_0=\sqrt{\xi+\frac{c_-^2}{2}}$ where  $g'(\i d_0,\xi)=0$, where the function $g(k,\xi)$ has been defined in \eqref{g_breather}.
%
%The idea behind the parametrix analysis is to solve local RHPs where the jump matrix is not close to the identity matrix uniformly, and then construct a piecewise approximation to the original exact RHP. Unlike the other parametrixes like the Airy parametrix, the parametrix in parabolic cylinder functions cannot solve the local problem exactly, it only solves it approximately.
%
We will now construct a piecewise matrix-valued function, which satisfies approximately the jumps in the vicinities of the points $k=\pm \i d_0.$
The analysis is similar to that of Section \ref{sect_SecondTermUtmost}, but it is slightly more involved because of discontinuity across the imaginary axis. This however can be handled by multiplication of the matrix by an off-diagonal matrix in the left half-disk around the point $\i d_0.$ We provide analysis for the upper point  $\i d_0,$ since the constrution at the point $-\i d_0$ follows by symmetry.
%
%We will now construct approximate solution in the vicinities of the points $\pm k_0,$ called local parametrixes.
%

\medskip\noindent {\bf Function $\boldsymbol{g(k,\xi).}$}
We start by inspecting the behavior of the phase function $g(k,\xi)$ in the neighborhoods of the points $\pm \i d_0.$ We recall  that $$g(k,\xi)=2\(2k^2-c_{-}^2+6\xi\right)\sqrt{k^2+c_{-}^2}\,$$ has 
 a discontinuity across $[\i c_-,-\i c_-]$,   because we chose $\sqrt{k^2+c_{-}^2}\,$ analytic in $\C\backslash[\i c_-,-\i c_-]$. Furthermore, we have the symmetry $g(k)=-\ol{g(-\ol k)}.$ 
 For $\delta >0$ let $${\mathcal U}_{\pm\i d_0}=\{k\in\C\,, s.t.\;|k\mp\i d_0|<\delta\}.$$ We also define ${\mathcal U}_{\i d_0}^\pm$  the semicircles on the left and right part of the complex plane with respect to the imaginary axis.
 The expansion of  $g(k,\xi)$ near  the point $k=\i d_0$ takes the form
\begin{equation*}
g(k,\xi)=\pm\(g_+(\i d_0,\xi)+\frac12g''_+(\i d_0,\xi)(k-\i d_0)^2+\mathcal{O}\(
(k-\i d_0)^3\)\)
,\; k\to \i d_0,\;k\in{\mathcal U}_{\i d_0}^\pm,\end{equation*}
where
\begin{equation*}
g_+(\i d_0) = -8\(\frac{c_-^2}{2}-\xi\)^{\frac32}<0,
\qquad g_+''(\i d_0)=\frac{-24(\xi+\frac{c_-^2}{2})}
{\sqrt{\frac{c_-^2}{2}-\xi}}<0.
\end{equation*}
 We  introduce the local conformal changes of variables,
\begin{equation*}
z:=\begin{cases}
\sqrt{g(k)- g_{+}(\i d_0)}\, , \quad k\in {\mathcal U}_{\i d_0}^+, \\
\sqrt{-g(k)- g_{+}(\i d_0)}\, , \quad \quad k\in {\mathcal U}_{\i d_0}^-\,.
\end{cases}
\end{equation*}
We have that
\begin{equation*}
z=\sqrt{\frac{|g''(k_0)|}{2}}\(\frac{k-\i d_0}{-\i}\)\(1+\mathcal{O}(k-\i d_0)\),\ k\in {\mathcal U}_{\i d_0} .
\end{equation*}
We also introduce $$\zeta:= \sqrt{t\,}\,z.$$
Note that the local variable $z$ is rotated by $90$ degrees with respect  to the orientation of $k-\i d_0$.
% so that the interval $(\i d_0,\i (d_0-r))$ is transformed into $z\in(0, \tilde r)$ for some $\tilde r;$ this is done to make application of the construction of Section \ref{sect_ParCyl} more straightforward.
Summing up, we have
\begin{equation*}\begin{split}
&\pm 2\i t g(k,\xi)=2\i t g_+(\i d_0,\xi)+2 \i t z^2 = 2\i t g(k_0,\xi) + 2\i\zeta^2,\ k\in   {\mathcal U}_{\i d_0}^\pm \, .
\end{split}\end{equation*}
%Note also that $g(k_0,\xi)=-8\(-\xi+\frac{c_-^2}{2}\)^{\frac32}.$

\medskip\noindent {\bf Function $\boldsymbol{T(k).}$}
Next, we need to understand the behavior of the function $T(k)$ as $k$ approaches the point $\i d_0$.  We recall that the function $T$ satisfies the jump relations 
$T_+(k)T_-(k)=\frac{1}{|a(k)|^2}, k\in(\i c_-,\i d_0),$ and $T_+(k)T_-(k)=1, k\in(\i d_0, -\i d_0),$ and 
$T_+(k)T_-(k)=\frac{1}{|a|^2},$ for $k\in(-\i d_0, -\i c_-)$. To understand the behavior of the function $T(k)$ at the point $k=\i d_0$ it is convenient to introduce the function 
\begin{equation*}
F(k)=\begin{cases}
&T(k),  \mbox{  for $\Re k>0$}, \\
&\frac{1}{T(k)}, \mbox{  for  $\Re k<0.$}
\end{cases}
\end{equation*}
 Then the function $F(k)$ satisfies the jump $\frac{F_+(k)}{F_-(k)}=\frac{1}{|a(k)|^2}$ on the interval $(\i c_-,\i d_0),$ and it is continuous across the interval $(\i d_0,-\i d_0).$
After straighforward computation we conclude that
\begin{equation*}
T(k)=\begin{cases}
&\(\frac{k-\i d_0}{-(k+\i d_0)}\)^{-\i \nu}\chi(k),\,\mbox{  for $\Re k>0$,}\\
&\(\frac{k-\i d_0}{-(k+\i d_0)}\)^{\i \nu}\frac{1}{\chi(k)},\mbox{  for $\Re k<0$,}
\end{cases}
\end{equation*}
and the latter formula is the definition of the  function $\chi(k),$ which is discontinuous across {\color{black} $[\i c_-,\i d_0]\cup [-\i d_0,-\i c_-]$ and continuous across $[\i d_0,-\i d_0],$} and has a non-zero limit as $k\to\pm \i d_0.$ Here
\begin{equation*}
\nu=\frac{-1}{2\pi}\ln{|a(\i d_0)|^2}<0.
\end{equation*}
By Lemma \ref{lem_abr},  we have that $|a(\i d_0)|^2-|b(\i d_0)|^2=1$ and hence indeed $\nu<0.$
Note that $|\chi(\i d_0)|=1$ in view of the symmetry relation 
$\ol{T(-\ol k)} = T(k).$
Singling out the main components of the function $T(k),$ we can write it as 
\begin{equation*}
\begin{split}
&T(k)=\zeta^{-\i\nu}\cdot\(\frac{k-\i d_0}{-(k+\i d_0) \sqrt{t\,} \, z}\)^{-\i\nu}\cdot\chi(k,\xi),\;\; \mbox{for $k\in {\mathcal U}_{\i d_0}^+,$}\\
\\
&T(k)=\zeta^{\i\nu}\cdot\(\frac{k-\i d_0}{-(k+\i d_0) \sqrt{t\,}\, z}\)^{\i\nu}\cdot\frac{1}{\chi(k,\xi)},\;\;
\mbox{for   $ k\in{\mathcal U}_{\i d_0}^-.$ }
\end{split}
\end{equation*}

\begin{figure}[ht!]
\begin{tikzpicture}
\draw
[very thick,
decoration={markings, mark=at position 0.5 with {\arrow{>}}},
postaction = {decorate}]
(-1.5,1.5) to (0,0);
\draw
[very thick,
decoration={markings, mark=at position 0.5 with {\arrow{>}}},
postaction = {decorate}]
(-1.5, -1.5) to (0,0);
\draw
[very thick,
decoration={markings, mark=at position 0.5 with {\arrow{>}}},
postaction = {decorate}]
(0,0) to (1.5, 1.5);
\draw
[very thick,
decoration={markings, mark=at position 0.5 with {\arrow{>}}},
postaction = {decorate}]
(0,0) to (1.5, -1.5);
\scriptsize{
\node at (-3.6,1){$\begin{pmatrix}
1& -a_+(\i d_0)\ol{b_+(-\i d_0)}\zeta^{-2\i\nu}\e^{-2\i\zeta^2} \\
0&1
\end{pmatrix}$};
\node at (-3.5,-1){$\begin{pmatrix}
1& 0 \\ -\ol{a_+(-\i d_0)} b_+(\i d_0) \zeta^{2\i\nu}\e^{2\i\zeta^2} & 1
\end{pmatrix}$};
\node at (3,1){$\begin{pmatrix}
1&0\\-r_+(\i d_0) \zeta^{2\i\nu}\e^{2\i\zeta^2}
&1
\end{pmatrix}$};
\node at (3,-1){$\begin{pmatrix}
1& \ol{r_+(\i d_0)}\zeta^{-2\i\nu}\e^{-2\i\zeta^2} \\
0&1
\end{pmatrix}$};
\node at (0,-0.3){$0$};
\node at (0.9, 1.3){$L_{1r}$};
\node at (0.9, -1.3){$L_{1l}$};
\node at (-0.9, 1.3){$L_{7}$};
\node at (-0.9, -1.3){$L_{5}$};
}
\end{tikzpicture}
\caption{Choice of the parameters   for the jump matrix $J_{\Phi}$ in \eqref{J_Phi} of  the parabolic cylinder functions $\Phi(\zeta)$  defined in \eqref{Phi}.}
\label{Fig_PCmiddle}
\end{figure}
In Figure \ref{Fig_PCmiddle} we plot the image  of the contours $L_1, L_5$ and $L_7$ passing through  the point $k=\i d_0$ via the map $k\to \zeta(k)$.  We denoted by $L_{1r}$ and $L_{1l}$ the images of the parts of the contour $L_1$ with $\Re k>0$ and $\Re k<0,$ respectively.

\subsubsection{Approximation, far away from breathers}
Now we are almost ready to define an approximation for the function $X(k)$  that solves the RH problem~\ref{RHP_final_constant2} in the middle  left constant region.
 For this purpose we use the local parametrix $\Phi(\zeta)$ in terms of parabolic cylinder functions obtained  in section \ref{sect_ParCyl}.
{\color{black} We define the  upper parametrix as
\begin{equation}
\label{Pu}
\begin{split}
P_{u}(k)=B(k)  \Phi(\zeta(k))\Big|_{\{r_*=r_+(\i d_0),\ \rho_*=-\ol{r_+(\i d_0)}\}}\cdot
\left(\frac{\e^{\i t g_+(\i d_0,\xi)}}{\chi(\i d_0)}\right)^{\sigma_3}
\(\frac{-(k-\i d_0)}{(k+\i d_0) z \sqrt{t}\, }\)^{\i\nu\sigma_3},
\end{split}
\end{equation}

Note that the parameter $r_*$ is the same as  in  Section \ref{sect_ApproximationParCylUtmost}, but the parameter $\rho_*$ has the opposite sign compared to the Section \ref{sect_ApproximationParCylUtmost}.  The matrix $B(k)$  in the above expression will be determined below.

Now we are ready to define an approximation of the function $X(k).$  We define
\begin{equation*}
\begin{split}
X_{appr}(k)=\begin{cases}
P_{u}(k),\quad \mbox{for  $ k\in {\mathcal U}^+_{\i d_0}$ },\\
P_{u}(k)
\begin{bmatrix}0&\i\\\i&0
\end{bmatrix},
\quad \mbox{for  $ k\in {\mathcal U}^-_{\i d_0}$, }\\
{\scriptsize \begin{bmatrix}0&1\\-1&0\end{bmatrix}}\ol{P_{u}(\ol k)}
{\scriptsize \begin{bmatrix}0&-1\\1&0\end{bmatrix}},
 \mbox{  for  $ k\in {\mathcal U}^+_{-\i d_0},$ }\\
-\i{\scriptsize \begin{bmatrix}0&1\\-1&0\end{bmatrix}}\ol{P_{u}(\ol k)}
\sigma_3,
 \mbox{  for  $ k\in {\mathcal U}^-_{-\i d_0},$ }\\
M^{(mod)}(k),\quad |k\mp \i d_0|>\delta,\;
\end{cases}
\end{split}
\end{equation*}
}
%
%Note that the paramaters $r_*$ and $\rho_*$ are switched in the parametrixes at the  points $k_0$ and $-k_0,$ and the parametrix at the point $-k_0$ is multiplied by off-diagonal matrices.
where $M^{(mod)}(k)=\begin{bmatrix}a_{\gamma}(k) & b_{\gamma}(k) \\b_{\gamma}(k) & a_{\gamma}(k)\end{bmatrix},$ $a_{\gamma}(k)=\frac12(\gamma(k)+\frac{1}{\gamma(k)}),$
$b_{\gamma}(k)=\frac12(\gamma(k)-\frac{1}{\gamma(k)}),$ $\gamma(k)=\sqrt[4]{\frac{k-\i c_-}{k+\i c_-}}.$
The matrix $B(k)$ in \eqref{Pu} is obtained by requesting that   the error matrix
 \begin{equation*}
 E(k)=X(k)X_{appr}(k)^{-1}
 \end{equation*}
has its jumps 
as close to the identity matrix  as possible. In particular for $k\in \partial{\mathcal U}^+_{\i d_0}$  we  have $J_E(k)=E^{-1}_+(k)E_-(k)=(P_u(k))^{-1}M^{(mod)}(k)$ where we assume that the circle 
${\mathcal U}_{\i d_0}$ is oriented anticlockwise.
Requesting that $J_E=I+o(1)$ we obtain 
{\color{black}
\begin{equation*}
B(k)=\begin{cases}
&M^{(mod)}(k)\left(\frac{\chi(\i d_0,\xi)}{\e^{\i t g_+(\i d_0,\xi)}}\right)^{\sigma_3}
\(\frac{-(k-\i d_0)}{(k+\i d_0) z \sqrt{t\,}\,}\)^{-\i\nu\sigma_3},\quad k\in {\mathcal U}^+_{\i d_0}, \\
&M^{(mod)}(k)\begin{bmatrix}0&\i\\\i&0
\end{bmatrix}^{-1}
\left(\frac{\chi(\i d_0,\xi)}{\e^{\i t g_+(\i d_0,\xi)}}\right)^{\sigma_3}
\(\frac{-(k-\i d_0)}{(k+\i d_0) z \sqrt{t\,}\,}\)^{-\i\nu\sigma_3},\;\; k\in {\mathcal U}^-_{\i d_0}.
\end{cases}
\end{equation*}
We observe that with the above definition the matrix $B(k)$ is analytic in $ {\mathcal U}_{\i d_0}$.
}

In a similar  way as  in Section \ref{sect_SecondTermUtmost}, we can verify that the jump matrix $J_{E}(k)=E^{-1}_+(k)E_-(k)$ 
is of the order $J_E(k)=I+\mathcal{O}(t^{-1/2})$ uniformly for $k\in \partial  {\mathcal U}_{\pm\i d_0}$ and on the contours 
$(L_1\cup L_5\cup L_7)\cap {\mathcal U}_{\i d_0}$ and $(L_2\cup L_6\cup L_8)\cap {\mathcal U}_{-\i d_0}$,
 while on the remaining contours the jump matrix $J_E$ is exponentially close to the identity matrix.
 Only the circles $ \partial {\mathcal U}_{\pm\i d_0}$ contribute to the $t^{-1/2}$-term of the asymptotic expansion  of the MKdV solution $q(x,t)$, while the remaining terms give a  higher order  contribution.
 We thus have 
\begin{equation*}q(x,t)=q_{appr}(x,t)+ q_{err}(x,t),\end{equation*}
where
\begin{equation*}
q_{err}(x,t)=\frac{1}{\pi}\int_{ {\mathcal U}_{\i d_0}}(J_E(k)-I)_{21}\d k + \frac{1}{\pi}\int_{ {\mathcal U}_{-\i d_0}}(J_E(k)-I)_{21}\d k
+\mathcal{O}(t^{-1}).
\end{equation*}
The symmetries $\ol{E(-\ol k)}=E(k)={\scriptsize\begin{bmatrix}0&-1\\1&0\end{bmatrix}}
E(-k)=\ol{E(\ol k)}
{\scriptsize\begin{bmatrix}0&1\\-1&0\end{bmatrix}}$
imply that the integrals over the oriented circles $\partial {\mathcal U}_{\pm \i d_0}$ have the following structure as $t\to\infty$:
\begin{equation*}
\frac{1}{\pi}\int_{\partial {\mathcal U}_{ \i d_0}}(J_E(k)-I)\d k
=\begin{pmatrix}
\ldots & B\\ A &\ldots
\end{pmatrix},
\qquad
\frac{1}{\pi}\int_{\partial {\mathcal U}_{- \i d_0}}(J_E(k)-I)\d k
=\begin{pmatrix}
\ldots & A\\ B &\ldots
\end{pmatrix},
\end{equation*}
and both $A,B$ are real.
Hence, $q_{err}=A+B+\mathcal{O}(t^{-1}),$ and it suffices to make   the computation only on the cirlcle $\partial  {\mathcal U}_{\i d_0}$.
%We have,
%\begin{equation*}
%\begin{split}
%&\mbox{Res}_{k=\i d_0}[J_E(k)-I]=M^{(mod)}_+(\i d_0)
%\left(\frac{\chi(\i d_0)}{\e^{\i t g_+(\i d_0)}}\right)^{\sigma_3}
%\mbox{Res}_{k=\i d_0}\left[\(\frac{k-\i d_0}{-\sqrt{t}\, z(k+\i d_0)}\)^{-\i\nu\sigma_3}\cdot\right.\\
%\\
%&\left.
%\left[
%\frac{\begin{pmatrix}
%0 & \e^{\frac{\pi\i}{4}}\beta_1\\
%\e^{\frac{3\pi\i}{4}}\beta_2 & 0
%\end{pmatrix}}{2\zeta}+\mathcal{O}(\zeta^{-2})\right]
%\(\frac{-(k-\i d_0)}{\sqrt{t}\, z(k+\i d_0)}\)^{\i\nu\sigma_3}
%\left(\frac{\e^{\i t g_+(\i d_0)}}{\chi(\i d_0)}\right)^{\sigma_3}
%M^{(mod)}_+(\i d_0)^{-1}\\,.
%\end{split}
%\end{equation*}
For simplifying the notation we define
\begin{equation*}
\begin{split}
f &:= \e^{-\i t g_+(\i d_0)}\chi(\i d_0) \lim\limits_{k\to\i d_0}\(\frac{k-\i d_0}{-(k+\i d_0)z\sqrt{t}}\)^{-\i\nu}\\
&=\e^{-\i t g_+(\i d_0)}\chi(\i d_0) \(2 d_0\sqrt{\frac{|g''_+(\i d_0)| t}{2}}\)^{\i\nu},
\end{split}
\end{equation*}
and observe that $|f|=1$.   Then the integral takes the form
\begin{equation*}
\begin{split}
&\int\limits_{ \partial {\mathcal U}_{\i d_0}}(J_E(k)-I)_{12}\d k =\e^{\frac{\pi\i}{4}}\sqrt{\frac{2\pi^2}{|g''_+(\i d_0)| t}}\left(a_{\gamma,+}(\i d_0)^2f^2\beta_1-\i
\frac{b_{\gamma,+}(\i d_0)^2}{f^2}\beta_2\right)+\mathcal{O}(\frac{1}{t}),\\
&\int\limits_{ \partial {\mathcal U}_{\i d_0}}(J_E(k)-I)_{21}\d k =\e^{\frac{3\pi\i}{4}}\sqrt{\frac{2\pi^2}{|g''_+(\i d_0)| t}}\left(\frac{a_{\gamma,+}(\i d_0)^2}{f^2}\beta_2+\i
b_{\gamma,+}(\i d_0)^2f^2\beta_1\right)
+\mathcal{O}(\frac{1}{t})\,.
\end{split}
\end{equation*}
We can use the relations $\beta_2=-\ol{\beta_1}$ and $\ol{a_{\gamma,+}(\i d_0) }=\i b_{\gamma,+}(\i d_0),$
$\ol{b_{\gamma,+}(\i d_0) }=\i a_{\gamma,+}(\i d_0)$
to verify that indeed $A,B\in\mathbb{R}.$
We have 
$q_{err}(x,t)=A+B+\mathcal{O}(t^{-1}),$ and using the relation $a_{\gamma,+}(\i d_0)^2-b_{\gamma,\i}(\i d_0)^2=1,$ we obtain, after some simplifications,
\begin{equation*}
\begin{split}
q_{err}(x,t)&=\sqrt{\frac{2}{|g''_+(\i d_0)|\, t}}\(\e^{\frac{\pi\i}{4}}\,\beta_1 f^2
+\frac{\e^{\frac{-\pi\i}{4}}\ol{\beta_1}}{f^2}\)+
{\mathcal O}(t^{-1})\\
&=
\sqrt{\frac{2}{|g''_+(\i d_0)|\, t}}\cdot 2\Re\({\e^{\frac{\pi\i}{4}}\,\beta_1}{f^2}\)+{\mathcal O}(t^{-1}).
\end{split}
\end{equation*}
Substituting the expression for $f,$ $g_+(\i d_0)$ and $g''_+(\i d_0)$  
%\begin{equation*}q_{err}(x,t)
%=
%\sqrt{\frac{8}{|g''_+(\i d_0)|\, t}}
%\Re
%\(
%\frac{\e^{\frac{-\pi\i}{4}}\sqrt{2\pi}2^{2\i\nu}\e^{\frac{\pi\nu}{2}}}{r_+(\i d_0)\,\Gamma(\i\nu)}
%\cdot
%\e^{-2\i t g_+(\i d_0)}\chi^2(\i d_0)
%\(2d_0\sqrt{\frac{|g''_+(\i d_0)|\,t}{2}}\)^{2\i\nu}
%\)+\mathcal{O}(t^{-1}).
%\end{equation*}
%To simplify the amplitude part of the latter, we use
and the  relation
\begin{equation*}
\frac{\sqrt{2\pi}\, \e^{\frac{\pi\nu}{2}}}{|r_+(\i d_0)|\,|\Gamma(\i\nu)|}=\sqrt{|\nu|},
\end{equation*}
%which follows from the reflection formula for the Gamma function and property $|a_+(\i d_0)|^2-|b_+(\i d_0)|^2=~1.$
%Substituting further the expressions for $g_+(\i d_0)$ and $g''_+(\i d_0),$ 
in the expression for  $q_{err}(x,t)$,  we   finally obtain 
\begin{equation}\label{q_errMiddle}
\begin{split}
q_{err}(x,t)
=\sqrt{\frac{|\nu|\sqrt{\frac{c_-^2}{2}-\xi}}
{3(\xi+\frac{c_-^2}{2})t}}
&\cos\(16t\(\frac{c_-^2}{2}-\xi\)^{\frac32}+\nu\ln\(\frac{192t(\xi+\frac{c_-^2}{2})^2}
{\sqrt{\frac{c_-^2}{2}-\xi}}\)
+\phi(k_0)\)+\\
&\hfill\quad\qquad\qquad+\mathcal{O}(t^{-1})\,,
\end{split}
\end{equation}
where the phase shift is given by the formula
\begin{equation}\label{phaseMiddle}
\phi(k_0)=-\frac{\pi}{4}-\arg r_+(\i d_0)-\arg\Gamma(\i\nu)+\arg\chi^2(\i d_0).
\end{equation}
Note that the formula \eqref{q_errMiddle} is very similar to the formula \eqref{q_errUtmost}, but the term $\frac{\pi}{4}$ enters with the opposite sign in the phase shift, and hence the phase shift in formula \eqref{phaseMiddle} differs by $\pi/2$ from the one in \eqref{phaseUtmost}.

\color{black}

%{
%\color{black}
%\begin{framed}
% What about the situation when several poles lie on the line $\Im g_{\r}(k,\xi)=0?$
%Will it be sum of two breathers, or more difficult situation?
%The maximum height of the breather is $4\kappa_2?$
%The maximum height of the soliton is $2\kappa_2?$
%Breather is twice higher than soliton?
%\end{framed}
%}

%%%%%
%%%%%
%%%%%%%%%%%%%%%%%%%%%%%%%%

%\begin{appendices}
%\appendix
\begin{appendices}

%begin of new input
%\input{ParCylAppend} already merged with {SecondTermParCyl}
%end of new input

\section{ Travelling wave solution of MKdV  and Whitham modulation equations\label{Appendix_travelling}}
 We look for the solution  of the  MKdV equation
\begin{equation*}u_t+6u^2u_x+u_{xxx}=0\,
\end{equation*}
 in the form \begin{equation*}u(x,t)=\eta(\theta),\quad \textrm{ where } \theta =kx-\omega t,\quad {\mathcal V}:=\frac{\omega}{k}\,.\end{equation*}
Substituting this ansatz into MKdV equation, we get after two integrations
\begin{equation}\label{v_1
}k^2\eta_{\theta}^2=-\eta^4+v\eta^2+B\eta+A,\end{equation}
where $A,B$ are constants of integration.
We assume that the above polynomials has four real roots   $e_1<e_2<e_3<e_4,$   so that
\begin{equation}
\begin{split}
\label{v}
k^2&\eta_{\theta}^2=-(\eta-e_1)(\eta-e_2)(\eta-e_3)(\eta-e_4),
\qquad 
\sum\limits_{j=1}^4e_j=0,
\qquad
{\mathcal V}=-\sum_{i<j}e_ie_j.
\end{split}
\end{equation}
Finite real valued periodic motion can take place   when $\eta$ varies between $[e_1,e_2]$  or $[e_3,e_4]$. We develop the calculations   for the latter case,
the first  case can be derived by the symmetry $\eta\to -\eta$ and $e_j\to -e_{5-j}$.
We obtain the integral 
$$k\int\limits^{\eta}_{e_3}\frac{\d \widetilde\eta}{\sqrt{(e_1-\widetilde\eta)(\widetilde\eta-e_2)(\widetilde\eta-e_3)(\widetilde\eta-e_4)}}=\theta.$$
After integration we arrive to the expression
\begin{equation}\label{eqn_eta_2}u_{per}(x,t)=\eta(\theta)=e_1+\frac{e_3-e_1}{1-\(\frac{e_4-e_3}{e_4-e_1}\)\mathrm{cn}^2\(\frac{\sqrt{(e_3-e_1)(e_4-e_2)}\ }{2}\ (x-vt)|m\)},\end{equation}
where the function   $\mathrm{cn}(z|m)$ is the Jacobi elliptic function of modulus 
\begin{equation*}
m^2=\dfrac{(e_3-e_2)(e_4-e_1)}{(e_4-e_2)(e_3-e_1)}.
\end{equation*}
Since $e_1+e_2+e_3+e_4=0$ we make a change of variables $\{e_1,e_2,e_3,e_4\}\to \{\beta_1,\beta_2,\beta_3\}$
defined as  $$e_1=-\beta_1-\beta_2-\beta_3,\quad e_2=\beta_1+\beta_2-\beta_3,\quad e_3=\beta_1+\beta_3-\beta_2,\quad e_4=\beta_2+\beta_3-\beta_1,$$
with $\beta_3>\beta_2>\beta_1$.
The inverse transformation is 
\begin{equation*}
\beta_1=\frac{e_2+e_3}{2}=-\frac{e_1+e_4}{2},\;\;\beta_2=\frac{e_2+e_4}{2}=-\frac{e_1+e_3}{2},\;\;\beta_3=\frac{e_3+e_4}{2}=-\frac{e_1+e_2}{2}.
\end{equation*}

 This choice of variables relates the elliptic solution to the   band spectrum of the  Zakharov-Shabat  linear operator \eqref{x-eq}.
In this way we obtain  the expression 
\begin{equation}
\label{periodicA}
u_{per}(x,t)=-\beta_1-\beta_2-\beta_3+\frac{2(\beta_2+\beta_3)(\beta_1+\beta_3)}{\beta_2+\beta_3-(\beta_2-\beta_1)\mathrm{cn}^2\(\sqrt{\beta_3^2-\beta_1^2} (x-{\mathcal V}t)+x_0|m\)},\end{equation}
%\begin{equation}\label{q_trav_intro}q_{periodic}(x,t)=e_4+\frac{(e_1-e_4)(e_2-e_4)}{e_1-e_4-(e_1-e_2)\mathrm{cn}^2\(\frac{\sqrt{(e_1-e_3)(e_2-e_4)}}{2}\ (x-vt)+x_0|m\)},\end{equation}
where $\beta_3>\beta_2>\beta_1$,  the speed ${\mathcal V}=2(\beta_1^2+\beta_2^2+\beta_3^2)$ and $x_0$ is an arbitrary phase. The elliptic modulus $m$  transforms to  
% $ m=\frac{(e_1-e_2)(e_3-e_4)}{(e_1-e_3)(e_2-e_4)}$ with $e_1>e_2>e_3>e_4$, $\sum_{j=1}^4 e_j=0$
\begin{equation}
\label{m_a}
m^2=\dfrac{\beta_2^2-\beta_1^2}{\beta_3^2-\beta_1^2}.
\end{equation}
When the periodic motion takes place between $e_1$ and $e_2$ the quantity $\beta_j\to\tilde{\beta}_j$ with $\tilde{\beta}_1=-\beta_1$, $\tilde{\beta}_2=\beta_2$ and $\tilde{\beta}_3=\beta_3$.\\
 The limiting case of the travelling wave  solutions are 
 \begin{itemize}
 \item[(a)] $e_1\to e_2$  which means  $\beta_1\to-\beta_2$  or $m\to 0$;
 \item[(b)] $e_3\to e_4$ which means $\beta_2\to\beta_1$  or $m\to 0$;
\item[(c)]  $e_2\to e_3$ which implies $\beta_2\to\beta_3$ or $m\to 1$.
 \end{itemize}
In the case (a) and (b) we have that $\mathrm{cn}(z|m)\to \cos z$   as $m\to0$. Parametrizing $\beta_2=\beta+\delta$ and $\beta_1=\beta-\delta$  the periodic solution \eqref{periodicA}  in case (b) 
 has an expansion of the form
\begin{equation*}
u_{per}(x,t)=\beta_3-2\delta+4\delta\cos^2\(\sqrt{\beta_3^2-\beta^2} (x-(4\beta^3+2\beta_3t)+x_0\)+O(\delta^2).
\end{equation*}
which is the small amplitude limit.
In  case (a)  and in  the limit $\beta_1\to -\beta_2=-\beta$ we obtain
\begin{equation*}
u_{per}(x,t)\simeq-\beta_3+\frac{2(\beta_3^2-\beta^2)}{\beta+\beta_3-2\beta\cos^2\(\sqrt{\beta_3^2-\beta^2} (x-(4\beta^2+2\beta_3)t)+x_0|m\)},
\end{equation*}
which is a  trigonometric solution.
The soliton limit is  case (c) when $m\to 1 $ or $\beta_3=\beta_2=\beta$, and $\mathrm{cn}(z|m)\to \mbox{sech}(z)$
so that 
\begin{equation*}
\begin{split}
u_{per}(x,t)&\simeq \beta_1+\dfrac{2\beta(\beta+\beta_1)}{\beta\,\mbox{cosh}(2\sqrt{\beta^2-\beta_1^2} (x-(4\beta^2+2\beta_1)t)+x_0)+\beta_1},
\end{split}
\end{equation*}
which coincides with the one-soliton solution  \eqref{soliton_intro}  on the constant background  by identifying $\beta$ with the spectral parameter $\kappa_0$ and $\beta_1$ with the constant background $c$.
\section{Whitham modulation equations\label{WhithamApp}}
The  Whitham modulation equations are the   slow modulations of the wave parameters $\beta_j=\beta_j(X,T)$ of the travelling  wave solution \eqref{periodicA}, where $X=\varepsilon x $ and $T=\varepsilon t$, where $\varepsilon $ is a small parameter.
These equations were first derived by Whitham \cite{W} for the KdV case using the method of averaging over conservation laws.
The same equations  can be  obtained requiring that the travelling wave solution with slow variation of the  wave 
parameters  satisfies the KdV equation up to an error of order $O(\varepsilon)$. 
After  Whitham's work the  modulations equations have been obtained for a large set of partial differential equations. In particular for the MKdV equation   they have been obtained  by Driscoll and O'Neil \cite{Driscol} using the original Whitham averaging procedure.
The Whitham modulation equations for MDKV wave parameters $\beta_3>\beta_2>\beta_1$ take the form
\begin{equation*}
\dfrac{\partial}{\partial T}\beta_j+W_j(\beta_1,\beta_2,\beta_3)\dfrac{\partial}{\partial X}\beta_j=0,\quad j=1,2,3,
\end{equation*}
where the speeds $W_j=W_j(\beta_1,\beta_2,\beta_3)$ are 
\begin{align}
\label{speed}
&W_j(\beta_1,\beta_2,\beta_3)=2(\beta_1^2+\beta_2^2+\beta_3^2)+4\dfrac{\prod_{k\neq j}(\beta_j^2-\beta_k^2)}{\beta_j^2+\alpha},\\
\label{alphan}
&\alpha=-\beta^2_3+(\beta_3^2-\beta_1^2)\dfrac{E(m)}{K(m)},\quad m^2=\dfrac{\beta_2^2
-\beta_1^2}{\beta_3^2-\beta_1^2},
\end{align}
where $E(m)=\int_0^{\pi/2}\sqrt{1-m^2\sin \psi^2}d\psi$ is the complete elliptic integral of the second kind.
At the boundary  one has 
\begin{itemize}
\item $m\to 0$ or $\beta_2= \pm\beta_1$,   then  $W_1(\beta_1,\beta_1,\beta_3)=W_2(\beta_1,\beta_1,\beta_3)=-6\beta_3^2+12\beta_1^2$ and $W_3(\beta_1,\beta_1,\beta_3)=6\beta_3^3$;
\item $m\to 1$ or $\beta_2=\pm \beta_3$, then  $W_2(\beta_1,\beta_3,\beta_3)=W_3(\beta_1,\beta_3,\beta_3)=4\beta_3^2+2\beta_1^2$ and $W_1(\beta_1,\beta_3,\beta_3)=6\beta_1^2$.
\end{itemize}
The solution of  Whitham modulation equations   is  obtained as a boundary value problem, namely for an initial monotone decreasing initial 
data $q_0(X)$, the solution satisfies the boundary conditions
\begin{itemize}
\item when $\beta_1=\beta_2$, then $\beta_3(X,T)=q(X,T) $   that solves the equation $q_T+6q^2q_X=0$ with initial data $q_0(X)$;
\item  when $\beta_2=\beta_3$ then $\beta_1(X,T)=q(X,T) $   that solves the equation $q_T+6q^2q_X=0$ with initial data $q_0(X)$.
\end{itemize}
We observe that the Whitham modulation equations for MKdV can be obtained from the Whitham modulation equations for KdV  with speeds $V_j(r_1,r_2,r_3)$ as 
\begin{equation*}
W_j(\beta_1,\beta_2,\beta_3)=V_j(\beta_1^2,\beta_2^2,\beta_3^2),\quad j=1,2,3.
\end{equation*}

It was shown in \cite{Levermore}  that for $r_3>r_2>r_1$  one has 
\begin{itemize}
\item strict  hyperbolicity:   $V_3(r_1,r_2,r_3)>V_2(r_1,r_2,r_3)>V_1(r_1,r_2,r_3),$
\item genuine nonlinearity:  $\dfrac{\partial }{\partial r_j}V_j(r_1,r_2,r_3)>0$.
\end{itemize}
It follows that the Whitham equations for MKdV are strictly hyperbolic and genuinely nonlinear only  when all  $\beta_3>\beta_2>\beta_1>0$.
Indeed one has 
\begin{equation*}
\dfrac{\partial}{\partial \beta_j}W_j(\beta_1,\beta_2,\beta_3)=2\beta_j\dfrac{\partial}{\partial r_j}V_j(r_1,r_2,r_3)|_{r_k=\beta_k^2},
\end{equation*}
so that when one of the $\beta_j=0$  the equations are not genuinely nonlinear.
 
In particular, if we consider the self-similar solution $\beta_j(X,T)=\beta_j(z)$ with $z=X/T=x/t$ we have
that the  Whitham equations reduce to the form
\begin{equation}
\label{Whithamself}
(W_j-z)\dfrac{\partial }{\partial z}\beta_j=0,\quad j=1,2,3.
\end{equation}
The  initial data that corresponds to such  a self-similar solution is the  step initial data 
\begin{equation}
\label{B4}
q_0(X)=\begin{cases}c_{+}\qquad {\rm as}\quad
X>0\\c_{-}\qquad {\rm as }\quad X<0. \end{cases}
\end{equation}
We observe that such initial data is invariant under the scaling $X\to \varepsilon x$ so that 
we can also consider the initial data $q_0(x)$ in the fast variable $x$.
The   self-similar  Whitham equations \eqref{Whithamself} can be uniquely solved  in the form 
\begin{equation*}
\beta_3=c_-, \;\;\beta_1=c_+,\;\; \dfrac{X}{T}=\dfrac{x}{t}=W_2(c_+,\beta_2,c_-),\quad c_-\leq\beta_2\leq c_+\,,
\end{equation*}
as long as we can invert $\beta_2$ as a function of $x/t$ within the region 
\begin{equation*}
=-6c_+^2+12c_-^2=W_2(c_+,c_+,c_-)< \frac{x}{t} <W_3(c_+,c_-,c_-)=4c_-^2+2c_+^2. 
\end{equation*}
This happens only when $W_2(c_+,\beta_2,c_-)$ is a strictly monotone function of $\beta_2$, that is when $\beta_2>0$, namely 
 when $c_->c_+>0$.
Thus we have obtained the solution appearing in the long-time asymptotic regime for  the Cauchy problem for MKdV considered in this manuscript.

When $c_->-c_+>0$  the solution  of the MKdV  will still  develop a  dispersive shock wave, but the Whitham equations are not in general strictly hyperbolic and genuinely nonlinear. In this case   it can be shown that  the solution of  \eqref{Whithamself}   is obtained as follows. We define $z^*=W_1(c_+,-c_+,c_-)=W_2(c_+,-c_+,c_-)=-6c_-^2+12c_+^2$.
Then for $z^*\leq \frac{x}{t}\leq W_3(c_+,c_-,c_-)$  we have that the solution of \eqref{Whithamself} with the initial data \eqref{B4}  is obtained as 
\begin{equation}
\label{solw1}
\beta_3=c_-, \;\;\beta_1=c_+,\;\; \dfrac{X}{T}=\dfrac{x}{t}=W_2(c_+,\beta_2,c_-),
\end{equation}
which is solvable for $\beta_2=\beta_2(x/t)$ in view of the genuinely nonlinearity.
For $ -6c_-^2=W_1(0,0,c_-)  \leq \frac{x}{t}\leq z^*=-6c_-^2+12c_+^2$  we have 
\begin{equation}
\label{solw2}
\beta_3=c_-, \;\; \dfrac{X}{T}=\dfrac{x}{t}=W_1(\beta_1,\beta_2,c_-),\;\; \dfrac{X}{T}=\dfrac{x}{t}=W_2(\beta_1,\beta_2,c_-).
\end{equation}
A  plot of the solution is shown in Figure~\ref{fig_whitham}.
 \begin{figure}[htb]
 \begin{center}
  \includegraphics[scale=0.3]{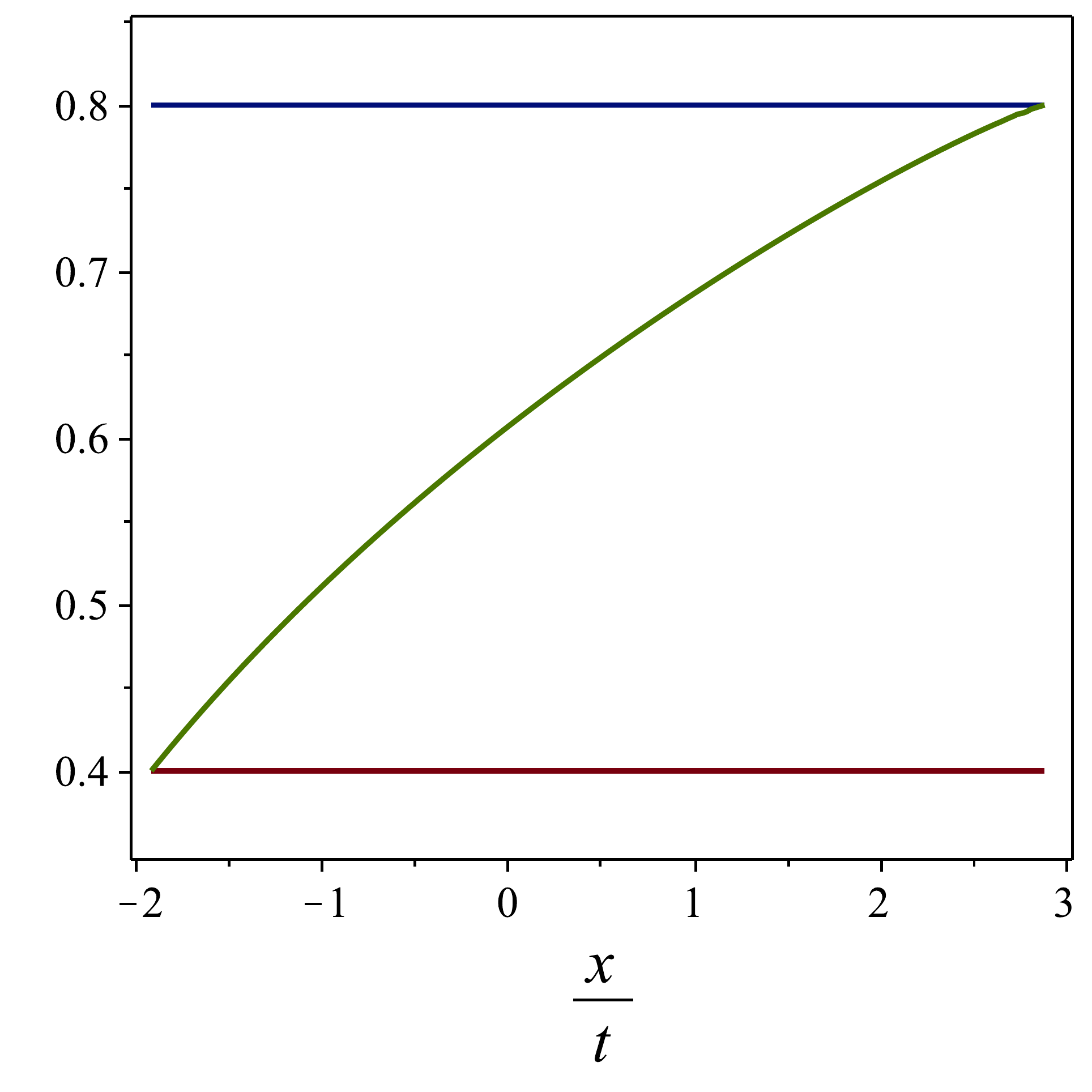}
  \includegraphics[scale=0.31]{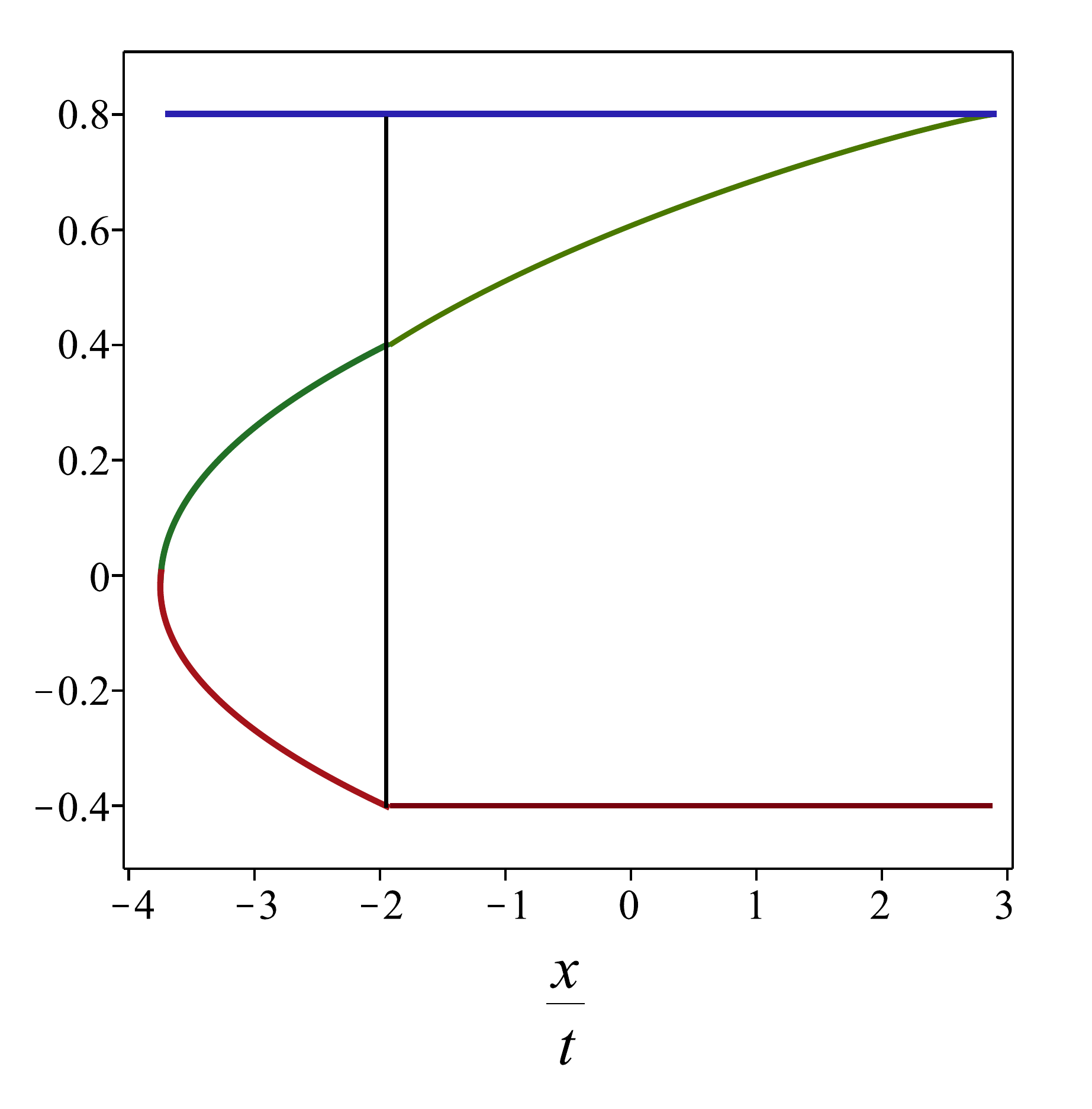}
\end{center}
\caption{The   self-similar  solution of the  Whitham equations \eqref{Whithamself} with  initial condition \eqref{B4}. The quantity  $\beta_1$  is  in red, $\beta_2$ in green and $\beta_3=c_-$ in blue and    $c_-=0.8$ and $c_+=0.4$ (left  figure) and $c_+=-0.4$ (right figure). In the right figure, the vertical black line distinguished between the  solution \eqref{solw1}
(right) and the solution \eqref{solw2} (left).  }
 \label{fig_whitham}
 \end{figure}

The proof of the solvability of the above system of equations   for $\beta_1=\beta_1(x/t)$ and $\beta_2=\beta_2(x/t)$ 
requires further analysis  that has been considered in  a similar setting
 in \cite{GTP}   and also in \cite{M15} for the  modulation equations \cite{AB} for the   Camassa-Holm equation. For nonstrictly hyperbolic equations see also \cite{C1,C2,C3}.
A discussion of the   various cases arising in the long-time asymptotics for the  MKdV  solution 
can be found in \cite{EHS}  A review  of dispersive shock waves for KdV can be found in \cite{Grava_LN}.

\section{ Proofs of properties of Jost solutions}\label{sect_appendC}
Here we prove Lemma \ref{lem_prop_Jost}. For the inspiration, we used the work \cite{Jakovleva}, but our functional spaces are a bit different. Namely, in order to be able to tackle not only continuous, but also discontinuous initial functions $q_0(x),$ like $q_0(x)=c_{\pm},$ for $ \pm x>0,$ we work with functions of bounded variations, $BV_{loc}(\mathbb{R}),$ rather than merely continuous functions. Such functions admit representation
\begin{equation*}
q_0(x)=_{ _{\hskip-3mm a.e.}}q_{ac}(x)+\sum\limits_j \alpha_j \chi_{(-\infty, x_j)}(x),\end{equation*}
 where the set $\left\{x_j\right\}$ is at most countable,  $\sum\limits_j |\alpha_j| =:\alpha<\infty,$
and $\chi_{A}(x)$ is the indicator function of a set $A$. Furthermore, $q_{ac}$ is an absolutely continuous function, i.e.
\begin{equation*}q_{ac}(x)=q_{ac}(x_0)+\int\limits^x_{x_0}q'_{ac}(\widetilde x)\d \widetilde x,\qquad q_{ac}'(.)\in L^1_{loc}(\mathbb{R}).\end{equation*} It follows that $q_0(x)$ is locally bounded.
Furthermore, the associated distributional derivative is the signed measure $\d q_0(x),$
\begin{equation*}
\d q_0(x)=q'_{ac}(x)\d x-\sum\limits_j \alpha_j \delta(x-x_j),
\end{equation*}
where $\delta$ is the Dirac $\delta-$distribution (a.k.a. Dirac $\delta$-function).
For the readers who would like to restrict themselves only to smooth (absolutely continuous) $q_0(x),$ we recommend to think $\alpha_j\equiv 0,$ in which case $\d q_0(x) = q_0'(x) \d x$ is the usual derivative.
We have the following proposition

\begin{proposition}\label{prop1}
\begin{enumerate} \item[(a)] Let $q_0(x)\in BV_{loc}(\mathbb{R}),$ and 
\begin{equation}\label{q0_cond1}\begin{split}&\int\limits_{-\infty}^{0}|x| |q_0(x)-c_-|\d x+ \int\limits^{+\infty}_{0}|x| |q_0(x)-c_+|\d x < \infty,
\end{split}\end{equation}
and $ \mathrm{ess}\sup\limits_{\hskip-5mm x\in\mathbb{R}}|q_0(x)|<\infty. $\\
Then there are exist solutions $\Phi^{\pm}(x;k)$ of $x-$equation \eqref{x-eq}, such that
\begin{equation}
\label{integral_representation_Jost_solution}
\Phi^{\pm}(x;k) = E^{\pm}(x;k) + \int\limits_{\pm\infty} ^ x L^{\pm}(x,y) E^{\pm}(y;k) \d y,
\qquad \color{black}{k\in\mathbb{R}\cup(\i c_{\pm},-\i c_{\pm}),} \end{equation}
where 
$$L^{\pm}(x,y)=\begin{pmatrix}
            L_{1}^{\pm}(x,y) & L_{2}^{\pm}(x,y)\\ -L_{2}^{\pm}(x,y) & L_{1}^{\pm}(x,y)
           \end{pmatrix}
$$
and 
$\mp\int_{\pm\infty} ^ x |L_j^{\pm}(x,y)| \d y <\infty,$ where $j=1,2.$ 
Moreover, $L^{\pm}(x,y)$ admit the following estimates in terms of the auxiliary quantities
\begin{equation*}
\begin{split}
\hat\sigma^{\pm}(x) = &\mp\int\limits_{\pm\infty} ^ x |q_0(\widetilde x)-c_{\pm}|\d\widetilde x\geq 0,
\quad
\hat\sigma_{1}^{\pm}(x) = \mp\int\limits_{\pm\infty} ^ x \hat\sigma^{\pm}(\widetilde x) \d\widetilde x\geq 0,\\
&M^{\pm}(x) = 
\mathrm{ess}\sup\limits_{\hskip-5mm\widetilde u\in(\pm\infty, x)}|q_0(\widetilde u)|,
\end{split}
\end{equation*}
which exist in view of \eqref{q0_cond1}.
Then  
\begin{equation}\label{estimates_Jost _kernel}
|L^{\pm}(x,y)| \leq C[M^{\pm}(x), \hat\sigma_{1}^{\pm}(x) ] \cdot \hat\sigma^{\pm}(\frac{x+y}{2}),
\end{equation}
where $C[M^{\pm}(x), \hat\sigma_{1}^{\pm}(x) ]$ is a constant, which depend only on $M^{\pm}(x),$  $\hat\sigma_{1}^{\pm}(x)$, and whose explicit form is not important for us.
\item[(b)] Let in addition to assumptions (a)  suppose that 
\begin{equation}\label{q0_cond2}\mp\int\limits_{\pm\infty}^0|\d q_0(x)|<\infty,
\qquad \mp\int\limits_{\pm\infty}^0 \mathrm{ess}\sup\limits_{\hskip-5mm \widetilde x\in {(\pm\infty, x)}}|q_0(x)|\,\d x<\infty  \end{equation}
Then (we supress $\pm$ everywhere for the ease of reading)
\begin{equation*}\begin{split}
\Phi(x;k) &= E(x;k) + \frac{1}{\i k}(q_0(x)-c)\sigma_1 E(x;k)-
\frac{1}{\i k}\int\limits_{\infty}^x \sigma_1 E(y;k)\d_yq_0(\frac{x+y}2)-\\
&
-\frac{\widehat L(x,x)\sigma_3E(x;k)}{\i k}-
\frac{1}{\i k}\int\limits_{\infty}^x \(q_0(\frac{x+y}2)-c\)\sigma_1 Q_c E(y;k)\d y+\\
&+\frac{1}{\i k}\int\limits_{\infty}^x\(\widehat L_y(x,y)\sigma_3 + \widehat L\sigma_3Q_c\)E(y;k)\d y,
\end{split}\end{equation*}
where $\sigma_3=\begin{pmatrix}1&0\\0&-1\end{pmatrix},$  $\sigma_{1}=\begin{pmatrix}0&1\\1&0\end{pmatrix},$  $Q_c=\begin{pmatrix}0&c\\-c&0\end{pmatrix}$,
and \begin{equation*}\left| \int_{\infty}^x |\widehat L_y(x,y)|\d y\right|<\infty,
\quad 
\left|\int_{\infty}^x |\widehat L(x,y)|\d y\right|<\infty.\end{equation*}
Moreover, (below $j=1$ or $j=2$)
\begin{equation*}
|\partial_y L_j(x,y)|\leq C[M(x), \hat\sigma_{1}(x), \hat\sigma_{2}(x)]\cdot \(\hat\sigma(\frac{x+y}{2}) + \mathrm{ess}\sup\limits_{\hskip-5mm\widetilde x\in(\infty, x)}|q_0(\widetilde x)| \),
\end{equation*}
where $C[M(x), \hat\sigma_{1}(x), \hat\sigma_2(x) ]$ is a constant, which depend only on $M^{\pm}(x),$  $\hat\sigma_{1}(x)$, $\hat\sigma_2(x)$ and whose explicit form is not important for us, and $\hat\sigma_2(x)$ stands for $\hat\sigma_2^\pm(x)$  where 
\begin{equation*}\hat\sigma_{2}^{\pm}(x) := \mp\int\limits_{\pm\infty}^x|\d q_0(x)|\geq 0.\end{equation*}
\item[(c)] Let in addition to assumptions (a) and  (b) suppose that 
\begin{equation}\label{q0_cond3}\pm\int\limits_{\mp\infty}^0 \e^{C|x|} |q_0(x)-c| \d x<\infty,\end{equation}
with some $C>0.$ Then\begin{equation*}\pm\int\limits_{\mp\infty}^0\e^{ C|x|} |L(x,y)|\d y<\infty, \quad \pm\int\limits_{\mp\infty}^0\e^{ C|x|} |L_y(x,y)|\d y<\infty.\end{equation*}
\end{enumerate}
\end{proposition}

\begin{remark}
The second of conditions (a) follows from the first of conditions (b).
\end{remark}

\begin{remark}
Conditions (b) of the Lemma are satisfied if, for example,
\begin{equation*}\mp\int\limits_{\pm\infty}^0|x| |\d q_0(x)|<\infty.\end{equation*}
\end{remark}

\begin{remark}
Both conditions (a), (b) of the Lemma are satisfied if, for example,
\begin{equation*}\mp\int\limits_{\pm\infty}^0|x|^2 |\d q_0(x)|<\infty \mbox{ and } q_0(x) = c_-+\int_{-\infty}^x \d q_0(y) = 
c_+-\int^{+\infty}_x \d q_0(y).
%\quad {i.e.} \  \int_{-\infty}^{+\infty} \d q_0(x) =c_+-c_-.
\end{equation*}
\end{remark}

\begin{proof} We start the proof of Proposition~\ref{prop1}.
With a bit of abuse of notations, we will suppress the superscript and subscript  $\pm.$
Assume first that $q_0$ is an absolutely continuous function. Then one can check directly that
\begin{equation*}\Phi(x,k)=E(x,k)+\int\limits_{\infty}^{x}L(x,y)E(y,k)\d y\end{equation*}
is a solution to \eqref{x-eq} provided that
the kernel
$$L(x,y)=\begin{pmatrix}
            L_{1}(x,y) & L_{2}(x,y)\\ -L_{2}(x,y) & L_{1}(x,y)
           \end{pmatrix}
$$
satisfies the following system of equations:

\begin{equation}\label{L1L2sys}\begin{cases}
 L_{1,x}(x,y)+L_{1,y}(x,y)=-(q_0(x)+c)L_{2}(x,y),\quad y\in(x,\infty),\\
  L_{2, x}(x,y)-L_{2, y}(x,y)=(q_0(x)-c)L_{1}(x,y),\quad y\in(x,\infty),
\\
L_{2}(x,x)=\frac{(q_0(x)-c)}{2},\quad \lim\limits_{y\to\infty}L_{2}(x,y)=0.
\end{cases}\end{equation}
In its turn, solutions of \eqref{L1L2sys} can be found as solutions of appropriate integral equations. Namely, denote 
$$u=\frac{x+y}{2},\quad v=\frac{x-y}{2},\qquad H_{1}(u,v)\equiv L_{1}(x,y),\quad H_{2}(u,v)\equiv L_{2}(x,y),$$
then we come to the integral equations
\begin{equation}\label{int_eq_1}
\begin{split}
&H_{1}(u,v)=-\int\limits_{\infty}^u(q_0(\widetilde u+v)+c)H_2(\widetilde u,v)\d \widetilde u,\\
&H_{2}(u,v)=\frac{(q_0(u)-c)}{2}+\int\limits_{0}^v(q_0(u+\widetilde v)-c)H_1(u, \widetilde v)\d \widetilde v,
\end{split}
\end{equation}
or 
\begin{equation}\label{int_eq_2}
\begin{split}
&
H_1(u,v)=\frac{-1}2\int\limits_{\infty}^u(q_0(\widetilde u+v)+c)(q_0(\widetilde u)-c)\d\widetilde u-\\
&\qquad\qquad-\int\limits_{\infty}^u(q_0(\widetilde u+v)+c)\int\limits_{0}^v(q_0(\widetilde u+\widetilde v)-c)
H_1(\widetilde u,\widetilde v)\d\widetilde v\d \widetilde u,
\\
&
H_2(u,v)=\frac{(q_0(u)-c)}{2}-\int\limits_{0}^v(q_0(u+\widetilde v)-c)\int\limits_{\infty}^u(q_0(\widetilde u+\widetilde v)+c)
H_2(\widetilde u,\widetilde v)\d\widetilde u\d\widetilde v.\end{split}\end{equation}
We see that $H_2(x,y)$ has the same level of regularity as $q_0$ does. Hence, system \eqref{L1L2sys} is appropriate for absolutely continuous functions $q_0,$ but it should have another meaning if we would switch to discontinuous functions $q_0$ of bounded variations.
One might think that now we will need to operate with functions of bounded variations of two variables. However, in our case functions of bounded variations of one variable suffices. Indeed, subtracting function $q_0(u+v),$ we obtain a regular function.
To this end, denote 
\begin{equation*} \begin{split}&\widehat L_1(x,y) = L_1(x,y),\quad \widehat L_2(x,y)=L_2(x,y)-\frac12\(q(\frac{x+y}{2})-c\),
\\
&
\widehat H_1(u,v)=H_1(u,v),\quad \widehat H_2(u,v) = H_2(u,v) -\frac12(q(u)-c).
\end{split}\end{equation*}
Furthermore, $\widehat H_1, \widehat H_2$ satisfy the following system of  integral equations:
\begin{equation}\label{int_eq_1a}
\begin{split}&\widehat H_{1}(u,v)=-\frac12 \int\limits_{\infty}^u(q_0(\widetilde u)-c)(q_0(\widetilde u+v)+c)\d \widetilde u-\int\limits_{\infty}^u(q_0(\widetilde u+v)+c)\widehat H_2(\widetilde u,v)\d \widetilde u,\\
&  
\widehat H_{2}(u,v)=\int\limits_{0}^v(q_0(u+\widetilde v)-c)\widehat H_1(u, \widetilde v)\d \widetilde v,\end{split}\end{equation}
or 
\begin{equation}\label{int_eq_2a}
\begin{split}
\widehat H_1(u,v)=&-\frac12\int\limits_{\infty}^u(q_0(\widetilde u+v)+c)(q_0(\widetilde u)-c)\d\widetilde u-\\
&-\int\limits_{\infty}^u(q_0(\widetilde u+v)+c)\int\limits_{0}^v(q_0(\widetilde u+\widetilde v)-c)
\widehat H_1(\widetilde u,\widetilde v)\d\widetilde v\d \widetilde u,
\\
\widehat H_2(u,v)=&\frac{-1}2\int\limits_{0}^v(q_0(u+\widetilde v)-c)\int\limits_{\infty}^u(q_0(\widetilde u+\widetilde v)+c)
(q_0(\widetilde u)-c)\d\widetilde u\d\widetilde v
\\
&-\int\limits_{0}^v(q_0(u+\widetilde v)-c)\int\limits_{\infty}^u(q_0(\widetilde u+\widetilde v)+c)
\widehat H_2(\widetilde u,\widetilde v)\d\widetilde u\d\widetilde v.\end{split}\end{equation}
We see that by subtracting from $L_2$ the `irregular part', which is $q_0,$  we gain in regularity of $\widehat L_2.$ Now $\widehat L_1, \widehat L_2$ have one more level of regularity than $q_0,$ and hence we are able to integrate the integral representation for $\Phi$ by parts, which gives us $k^{-1}$ term in the large $k$ asymptotic expansion.

More precisely, apply the successive approximation method to the first of equations \eqref{int_eq_2a},
i.e. represent $\widehat H_1 = \sum\limits_{j=0}^{\infty} \widehat H_1^{(j)},$
where $H_1^{(0)}$ is the inhomogeneous part of the r.h.s. of the equation, and $H_1^{(j)}$ is the homogeneous part of the r.h.s. of the equation, applied to $H_1^{(j-1)}.$ Then we have by induction that
\begin{equation*}|H_1^{j}|\leq\frac{\(M(u+v)\)^{j+1} \cdot \hat\sigma(u) \cdot (\hat\sigma_1(u+v)-\hat\sigma_1(u))^j}{2\cdot j!},\end{equation*}
where
\begin{equation*}\hat\sigma(u)=\left|\int_{\infty}^u|q(\widetilde u)-c|\d\widetilde u\right|,
\quad 
\hat\sigma_1(u)=\left|\int_{\infty}^u \hat\sigma(\widetilde u) \d\widetilde u\right| = \left|\int_{\infty}^u(u-\widetilde u)|q_0(\widetilde u)|\d\widetilde u\right|,
\end{equation*}
and $M(z) = \mathrm{ess}\sup\limits_{\hskip-5mm\widetilde u\in (\infty,z)}|q_0(\widetilde u)+c|.$
Indeed, the induction step goes as follows:
\begin{equation*}\begin{split}
&|\widehat H_1^{(j+1)}(u,v)| \leq \left|\int\limits_{\infty}^u(q_0(\widetilde u+v)+c)\int\limits_{0}^v(q_0(\widetilde u+\widetilde v)-c)
\widehat H_1^{(j)}(\widetilde u,\widetilde v)\d\widetilde v\d \widetilde u\right|
\\
&
\leq
\left|\int\limits_{\infty}^u|q_0(\widetilde u+v)+c|\int\limits_{0}^v|q_0(\widetilde u+\widetilde v)-c|
\frac{M(\widetilde u+\widetilde v)^{j+1}}{2\cdot j!}\hat\sigma(\widetilde u)\(\hat\sigma_1(\widetilde u+\widetilde v)-\hat\sigma_1(\widetilde u)\)^{j}\d\widetilde v\d \widetilde u\right|.
\end{split}
\end{equation*}
In the latter we estimate \\
$\hat\sigma(\widetilde u)\leq \hat\sigma(u),\quad \hat\sigma_1(\widetilde u+\widetilde v)-\hat\sigma_1(\widetilde u)\leq\hat\sigma_1(\widetilde u+v)-\hat\sigma_1(\widetilde u), \quad M(\widetilde u+\widetilde v)\leq M(u+v),\quad |q_0(\widetilde u+v)+c|\leq^{^{^{\hskip-3mm a.e.}}} M(u+v),$
and thus obtain
\begin{equation*}
|H_1^{(j+1)}(u,v)|
\leq
\frac{M(\widetilde u+\widetilde v)^{j+2}}{2\cdot j!}\cdot \hat\sigma(\widetilde u)
\left|\int\limits_{\infty}^u \(\hat\sigma_1(\widetilde u+v)-\hat\sigma_1(\widetilde u)\)^{j} \int\limits_{0}^v|q_0(\widetilde u+\widetilde v)-c|
\d\widetilde v\d \widetilde u\right|.
\end{equation*}
Here we use the definition of $\hat\sigma$ as the integral of $|q-c|,$ and further make use of $|\hat\sigma_1'(u)|=\hat\sigma(u),$ and integrating by parts, obtain the induction step.

Hence, 
\begin{equation*}|\widehat H_1(u,v)|\leq K_1(u,v) := \frac12M(u+v)\hat\sigma(u)\exp\left[M(u+v)\cdot(\hat\sigma_1(u+v)-\hat\sigma_1(u))\right],\end{equation*}
and similarly,
\begin{equation*}
\begin{split}
|\widehat H_2(u,v)|\leq K_2(u,v)&:= \frac12M(u+v)\,\hat\sigma(u)\,(\hat\sigma(u+v)-\hat\sigma(u))\times\\
&\times \exp\left[M(u+v)\,(\hat\sigma_1(u+v)-\hat\sigma_1(u))\right],
\end{split}
\end{equation*}
and henceforth,
\begin{equation*}|\widehat L_1(x,y)|\leq \widetilde C \hat\sigma(\frac{x+y}{2}),\quad |\widehat L_2(x,y)|\leq \widetilde C \hat\sigma(\frac{x+y}{2}),\end{equation*}
where $\widetilde C$ is a generic constant, which does not depend on $y.$ Since $\hat\sigma\in L_1,$ we henceforth proved the statement (a) of the Lemma.

Passing to the part (b), it is sufficient to prove that $\left|\int_{\infty}^ x |L_y(x,y)|\d y\right|<\infty$, and for this it is sufficient to estimate $\partial_u\widehat H_1,$ $\partial_v\widehat H_1,$ $\partial_u\widehat H_2,$ $\partial_v\widehat H_2.$
We have 
\begin{equation*}\partial_u \widehat H_1(u,v) = -(q_0(u+v)+c)\left\{\widehat H_2(u,v) + \frac12(q_0(u)-c)\right\},
\end{equation*}
$$\partial_v \widehat H_2(u,v) = (q_0(u+v)-c) \widehat H_1(u,v),$$
and hence it is enough to estimate $\partial_v\widehat H_1,$ $\partial_u\widehat H_2.$
We have 
\begin{equation}\label{int_eq_1a_der}
\begin{split}&\partial_v\widehat H_{1}(u,v)=\int\limits_{\infty}^u \left[ \frac{q_0(\widetilde u)-c}{-2}-\widehat H_2(\widetilde u, v)\right] \d_{\widetilde u} q_0(\widetilde u+v)
-\int\limits_{\infty}^u(q_0(\widetilde u+v)+c) \partial_v \widehat H_2(\widetilde u,v)\d \widetilde u, \\
&  
\partial_u \widehat H_{2}(u,v)=\int\limits_{0}^v(q_0(u+\widetilde v)-c)\partial_u \widehat H_1(u, \widetilde v)\d \widetilde v
+
\int\limits_{0}^v \widehat H_1(u, \widetilde v) \d_{\widetilde v} q_0(u+\widetilde v) ,\end{split}\end{equation}
and to estimate, we split $\d q_0(z) = q_{ac}'(z)\d z+ \sum_j\alpha_j \delta(z-x_j).$
We thus have 
\begin{equation*}
\begin{split}
|\partial_v \widehat H_1(u,v)| \leq & \hat\sigma_2(u+v) \(\frac12 \mathrm{ess}\sup\limits_{\hskip-5mm\widetilde u\in (\infty, u)} |q_0(\widetilde u)-c| +
K_2(u,v)\) +\\
&+ M(u+v) \, \hat\sigma(u+v) \, K_1(u,v),
\\ 
|\partial_u \widehat H_2(u, v)| \leq & M(u+v)\, (\hat\sigma(u+v)-\hat\sigma(u))
\(\frac12 \mathrm{ess}\sup\limits_{\hskip-5mm\widetilde u\in (\infty, u)} |q_0(\widetilde u)-c| + K_2(u,v)\)+ 
 \\ & + \(\hat\sigma_2(u+v)-\hat\sigma_2(u)\) K_1(u, v),
\end{split}
\end{equation*}
Since \begin{equation*}\left|\int_{\infty}^0 \mathrm{ess}\sup\limits_{\hskip-5mm\widetilde u\in (\infty, u)} |q_0(\widetilde u)-c| \d u\right|  + 
\left|\int_{\infty}^{0}\hat\sigma(u) \d u\right| <\infty,\end{equation*}
we obtain the statement of $(b).$
The statement (c) follows immediately from the statement (b)
and from the estimate (let us stick to $-\infty$ for certainty)
\begin{equation*}
\begin{split}
\int\limits_{-\infty}^{x} \hat{\sigma}(y)\e^{C|y|}\d y& =  
\int\limits_{-\infty}^{x} \e^{C|y|}\int\limits_{-\infty}^y|q_0(z)-c|\d z\d y = \\
&=
\int\limits_{-\infty}^{x}\d z  |q_0(z)-c| \int\limits^{x}_z \e^{C|y|} \d y = 
\frac{1}{C}\int\limits_{-\infty}^{x}  |q_0(z)-c|  \(\e^{C|x|} - \e^{C|z|}\)\d z.
\end{split}
\end{equation*}

\end{proof}

\begin{corollary}\label{cor_analyticity_Jost}\begin{enumerate}
            \item[(a)]
Provided that conditions \eqref{q0_cond1} are satisfied, the first column $\Phi^{\l}_1$ is analytic in \\ $\Im\sqrt{k^2+c_{\l}^2}>0$, i.e. in $\left\{k: \Im k>0\right\}\setminus[\i c_-,0],$ the second column $\Phi^{\l}_2$ is analytic in $\Im\sqrt{k^2+c_{\l}^2}<0,$ and the first column of the right Jost solution $\Phi^{\r}_1$
is analytic in $\Im \sqrt{k^2+c_{\r}^2}<0,$ and the second column $\Phi^{\r}_2$ is analytic in $\Im \sqrt{k^2+c_{\r}^2}>0.$

\item[(b)] Furthermore, define the transition matrix
$T(k):= \Phi^+(x;k)^{-1}\Phi^-(x;k),$ and denote its elements by
\begin{equation*}T(k)=\begin{pmatrix}a(k) & -\ol{b(\ol k)} \\ b(k) & \ol{a(\ol k)}\end{pmatrix},\quad \mbox{and furthermore }\quad r(k) = \frac{b(k)}{a(k)}.\end{equation*}
Then $a(.)$ is analytic in $k\in\left\{k:\ \Im k>0\right\}\setminus [\i c_-,0],$ continuous up to the boundary with the exception of the points $k\in\left\{\i c_-, \i c_+, 0\right\},$
and has uniform w.r.t. $\arg k \in[0,\pi]$ asymptotics $a(k)=1+\mathcal{O}(k^{-1})$ as $k\to\infty.$
Function $r(k)$ is defined and continuous for $k\in\mathbb{R}\setminus\left\{ 0 \right\},$ and has the asymptotics
$b(k)=\mathcal{O}(k^{-1})$ as $k\to\pm\infty, $ $k\in\mathbb{R}.$

\item[(c)] Suppose that in addition to \eqref{q0_cond1} the following conditions are satisfied: $c_-\geq c_+\geq 0,$ and 
 \begin{equation}
 \label{exp_decreasing}\int\limits_{-\infty}^0|q_0(x)-c_{\l}|\e^{2|x|\sqrt{(c_{\l}+\delta)^2-c_{\l}^2}}\d
x+\int\limits_{0}^{+\infty}|q_0(x)-c_{\r}|\e^{2x\sqrt{(c_{\l}+\delta)^2-c_{\r}^2}}\d
x<\infty,
\end{equation}
for some $\delta>0$. \\
Then the Jost solutions $\Phi^{\l},$ $\Phi^{\r}$ are analytic in a $\delta-$neighborhood of the contour $k\in\Sigma=\mathbb{R}\cup[\i c_{\l},-\i c_{\l}].$
\item[(d)] Furthermore, $a(k)$ and $b(k)$ have analytic extension to $k\in U_{\delta}(\Sigma)\setminus [\i c_-, -\i c_-],$ and have the asymptotics
$b(k) = \mathcal{O}(k^{-1})$, as $k\to\infty, $ $k\in U_{\delta}(\Sigma).$
\end{enumerate}
\end{corollary}
\begin{proof}
The  part (a)  of the corollary follows from the fact that under \eqref{q0_cond1} the kernels in \eqref{integral_representation_Jost_solution} are summable:
$$ \int\limits_{-\infty}^x|L^\l(x,y)|\d y<\infty,\quad \quad \int\limits^{+\infty}_x|L^\r(x,y)|\d y<\infty.$$
 Statement (c) of the corollary follows directly from the estimates \eqref{estimates_Jost _kernel} and the fact that under \eqref{exp_decreasing} 
$$\int\limits_{-\infty}^0 \e^{2|y| \sqrt{(c_{\l}+\delta)^2-c_{\l}^2}}\hat{\sigma}^\l(y)\d y<\infty,
\qquad\int\limits^{+\infty}_0 \e^{2y \sqrt{(c_{\l}+\delta)^2-c_{\r}^2}}\hat{\sigma}^\r(y)\d y<\infty.$$
The first of the above estimates gives us that $\Phi^{\l}$ is analytic in $0<|\Im\sqrt{k^2+c_{\l}^2}|<\sqrt{(c_{\l}+\delta)^2-c_{\l}^2},$ which includes  a   $\delta-$neighborhood of $\Sigma,$
and the second estimate gives us analyticity of $\Phi^{\r}$ in 
$|\Im \sqrt{k^2+c_{\r}^2}|<\sqrt{(c_{\l}+\delta)^2-c_{\r}^2},$ which also includes  a $\delta-$neighborhood of $\Sigma.$

To prove part   (d) of the corollary  let us  
introduce 
%\begin{equation*}
%\Phi^-(x;k):= \begin{pmatrix}\varphi^-(x;k) & -\ol{\psi^-(x;\ol k)}  \\ \psi^-(x;k) & \ol{\phi^-(x;\ol k)} \end{pmatrix},
%\qquad
%\Phi^+(x;k) := \begin{pmatrix}  -\ol{\psi^+(x;\ol k)} & \varphi^+(x;k)  \\  \ol{\phi^+(x;\ol k)} & \psi^+(x;k) \end{pmatrix},
%\end{equation*}
 $\chi^{\pm}(k)=\sqrt{k^2+c^2_{\pm}}$, $\gamma(k;c)=\sqrt[4]{\dfrac{k-\i c}{k+\i c}}$,   $a^{\pm}_\gamma(k)=\dfrac{1}{2}\(\gamma(k,c_\pm)+\frac{1}{\gamma(k,c_\pm)}\right)$,$b^{\pm}_\gamma(k)=\dfrac{1}{2}\(\gamma(k,c_\pm)-\frac{1}{\gamma(k,c_\pm)}\right)$,
 and
 $$
 f_1(x;k):=\int\limits_{\infty}^{x}(L_1(x,y)a_\gamma(k)+L_2(x,y)b_\gamma(k))\e^{\i(x-y)\chi(k)}\d y
 $$
 and 
 $$
f_2(x;k):= \int\limits_{\infty}^{x}(-L_2(x,y)a_\gamma(k)+L_1(x,y)b_\gamma(k))\e^{\i(x-y)\chi(k)}\d y
$$
 where we drop everywhere the superscript $\pm$ in the definition of $f_1(x;k)$ and $f_2(x;k)$
 and
 $$
 f_3^+(x;k):= \int\limits_{+\infty}^{x}(-L_2^+(x,y)b_\gamma^+(k)+L_1^+(x,y)a_\gamma^+(k))\e^{\i(y-x)\chi^+(k)}\d y,
 $$
 $$
 f_4^+(x;k):=\int\limits_{+\infty}^{x}(L_1^+(x,y)b_\gamma^+(k)+L_2^+(x,y)a_\gamma^+(k))\e^{\i(y-x)\chi^+(k)}\d y.
 $$
Then, substituting the expressions 
$$\Phi^-(x;k)=:\begin{pmatrix}\varphi^-(x;k) & -\overline{\psi^-(x;\overline{k})} \\
\psi^-(x;k) & \overline{\varphi^-(x;\overline{k})} \end{pmatrix}
\
\mbox{ and } \
\Phi^+(x;k)=:\begin{pmatrix} -\overline{\psi^+(x;\overline{k})} & \varphi^+(x;k) \\
\overline{\varphi^+(x;\overline{k})} &  \psi^+(x;k)\end{pmatrix},$$
 into the definition of the transition matrix $T(k),$
we find the following expressions for $a, b$
(which are $x$-independent, as they should be):
\begin{equation}\label{b}
\begin{split}
b(k) =& \e^{-\i x(\chi^+(k)+\chi^-(k)) }\left[\(a^-_\gamma(k) +f_1^-(x;k) \)\(b_\gamma^+(k) + f^+_2(x;k)\)-\right.\\
&\left.-\(b_\gamma^-(k) +f_2^-(x;k)\)(\(a_\gamma^+(k) + f_1^+(x;k)\)\right]\end{split}\end{equation}
and
\begin{equation*}
\begin{split}
a(k) =& \e^{\i (x(\chi^+(k)-\chi^-(k)) }\left[(a_\gamma^-(k) +f_1^-(x;k)) (a_\gamma^+(k) +f_3^+(x;k))
-\right.\\
&\left.-(b_\gamma^-(k) + f_2^-(x;k))(b_\gamma^+(k) +f_4^+(x;k) )
\right].
\end{split}\end{equation*}
Integrating here by parts, and using estimates
$|L^{\pm}_j(x,y)|\leq \widetilde C(x) \hat{\sigma}^\pm((x+y)/2),$ $|\partial_y L^{\pm}_j(x,y)|\leq \widetilde C(x) \hat{\sigma}^{\pm}((x+y)/2),$
we see that all the integrals are convergent in the domain $k\in U_{\delta}(\Sigma),$ and, furthermore, all the exponential terms with  dependence on $k$  are estimated by $\e^{|x| |\Im \chi(k)|}\leq \e^{\widetilde C |x|},$ and hence are bounded as $k\to\infty, |\Im k|\leq\delta.$
\end{proof}

\begin{proof}[Proof of Lemma \ref{lem_prop_Jost}]
%The 1st property follows by the  \href{"https://en.wikipedia.org/wiki/Liouville%27s_formula"} {Liouville's formula}
% from the fact that the left multiplier in the r.h.s. of \eqref{x-eq} is trace-less.
 The second property follows from the  integral representation in  Proposition \ref{prop1}.
  The large $k$ asymptotics in property 4  follows from integral representations in Proposition \ref{prop1} (b), and the possibility of integrating them by parts, as in Proposition \ref{prop1} (b).
  \end{proof}

\section{Zeros of $a(k)$ in $\mathrm{Im}\sqrt{k^2+c_{\l}^2}\geq0.$}\label{sect_nongen}
In this section we %prove Lemma~\ref{lem_zeros} and we  
derive a RH problem in the case of higher order poles, which gives higher order solitons or breathers.

\begin{lemma}\label{lem_zeros_a_multiple}
Let $k_0,$ $\mathrm{Im}k_0>0,$ $k_0\notin[\i c_{\l},\i c_{\r}]$ be a zero of $a(.)$ of the order $n\geq1,$ i.e.
$$a(k_0)=\cdots a^{(n-1)}(k_0)=0,\quad a^{(n)}(k_0)\neq0.$$
Then there exist coefficients $$\mu_0,\cdots,\mu_{n-1}$$
independent on $x,t,$ such that the following equalities hold:\\
$$f_{\l}^{(m)}=\sum\limits_{k=0}^{m}C_m^k\mu_{m-k}f_{\r}^{(k)},\quad m=0,\cdots, n-1,\quad C_{m}^{k}:=\frac{m!}{k!(m-k)!}$$
where $f_{\pm}=\begin{pmatrix}\varphi_{\pm},\psi_{\pm}\end{pmatrix}^T,$
i.e.
$\begin{cases}
 f_{\l}=\mu_{0}f_{\r},\\
 \dot{f}_{\l}=\mu_{0}\dot{f}_{\r}+\mu_1f_{\r},\\
 \ddot{f}_{\l}=\mu_{0}\ddot{f}_{\r}+2\mu_1\dot{f}_{\r}+\mu_2 f_{\r},\\
 \dddot{f}_{\l}=\mu_{0}\dddot{f}_{\r}+3\mu_1\ddot{f}_{\r}+3\mu_2 \dot{f}_{\r}+\mu_3 f_{\r}.\\\cdots
\end{cases}
$
\end{lemma}
\begin{proof}
 Denote $$f=\Phi_{\l,1}, \quad g=\Phi_{\r,2}$$
the first and the second columns of $\Phi_{\l,\r},$ respectfully.
$f,g$ are analytic in the upper half-plane.
The spectral coefficient $a(k)$ is given by the formula
$$a(k)=\det[f(x,t;k),g(x,t;k)].$$

The base of induction is the fact that if $a(k_0)=0,$ then $$\exists\mu_0 \textrm{ such that }\quad f=\mu_0 g.$$
%{\color{black}\begin{framed}
%             I forgot how to prove that $\mu_0$ does not depend on $x,t?$
%            \end{framed}
%}

Further, suppose that the statement of the lemma is satisfied for $m=1,\cdots, m-1.$
By the Lebesgue formula,
$$a^{(m)}(k_0)=\sum\limits_{k=0}^{m}\det[f^{(m-k)},g^{(k)}].$$
In the above formula we split off the term containing $f^{m},$ and for all the smaller order derivatives of $f$ we substitute their expression in terms of the derivatives of $g.$ We get
\begin{equation}\label{a_der_m_prom}a^{(m)}(k_0)=\det[f^{(m)},g]+\sum\limits_{k=1}^{m}\sum\limits_{j=0}^{m-k}C_{m}^kC_{m-k}^j\mu_{m-k-j}\det[g^{(j)},g^{(k)}].\end{equation}
Now, in the above expression each term $\det[g^{(j)},g^{(k)}]$ with $k\geq1,$ $k\neq j,$ will appear twice: once as
$$C_{m}^kC_{m-k}^j\mu_{m-k-j}\det[g^{(j)},g^{(k)}],$$
and another time as 
$$C_{m}^jC_{m-j}^k\mu_{m-j-k}\det[g^{(k)},g^{(j)}].$$
Since $$C_{m}^jC_{m-j}^k=C_{m}^kC_{m-k}^j=\frac{m!}{j!k!(m-k-j)!},$$
these two terms will cancel each other. The same is for $j=k.$ Hence, the only terms which remain in \eqref{a_der_m_prom} are those corresponding to $j=0,$ i.e.
\begin{equation}\label{a_der_m_prom_2}a^{(m)}(k_0)=\det[f^{(m)},g]-\sum\limits_{k=1}^{m}C_{m}^k\mu_{m-k}\det[g^{(k)},g]=
\det[f^{(m)}-\sum\limits_{k=1}^{m}C_{m}^k\mu_{m-k}g^{(k)},g],\end{equation}
and hence $a^{(m)}(k_0)=0$ implies the existence of a coefficient $\mu_{m}$ such that 
$$f^{(m)}-\sum\limits_{k=1}^{m}C_{m}^k\mu_{m-k}g^{(k)}=\mu_{m}g,\quad \textrm{ or }\quad f^{(m)}=\sum\limits_{k=0}^{m}C_{m}^k\mu_{m-k}g^{(k)}.$$
This finishes the proof.
%{\color{black}\begin{framed}
%             The only thing is why $\mu_j$'s do not depend on $x,t?$
%            \end{framed}
%}
\end{proof}

\begin{lemma}\label{lem_reg_multi}
 Let $k_0,$ $\mathrm{Im}k_0>0,$ $k_0\notin[\i c_{\l},\i c_{\r}]$ be a zero of $a(.)$ of the $n^{th}$ order, $n\geq1,$ i.e.
$$a(k_0)=\ldots=a^{(n-1)}(k_0)=0,\quad a^{(n)}(k_0)\neq0;$$
let the coefficients $\mu_0,\ldots,\mu_{n-1}$  (independent of $x,t$!) be as in Lemma \ref{lem_zeros_a_multiple}, i.e.
\begin{equation}\label{f_l_in_f_r}f_{\l}^{(m)}(k_0)=\sum\limits_{j=0}^{m}C_{m}^{j}\mu_jf_{\r}^{(m-j)}(k_0),\quad m=0,\ldots,n-1, \qquad\textrm{ where }C_{m}^{j}:=\frac{m!}{j!(m-j)!} \end{equation}
are the binomial coefficients,
and let $T_1(x,t),\ldots,T_{n}(x,t)$ be the coefficients in the Taylor expansion of the function 
$\frac{\e^{2\i\theta(x,t;k)-2\i\theta(x,t;k_0)}}{a(k)}$ at the point $k=k_0,$ i.e.
\begin{equation}\label{a_decomposition}\frac{\e^{2\i\theta(x,t;k)}}{a(k)}=\sum\limits_{q=1}^{n}\frac{T_{q}(x,t)\ \e^{2\i\theta(x,t; k_0)}}{(k-k_0)^q}+\mathcal{O}(1),\quad k\to k_0.\end{equation}
Then 
\begin{equation}\label{regular_expression}
\frac{1}{a(k)}f_{\l}(x,t;k)\e^{\i\theta(x,t;k)}-\left[\sum\limits_{j=1}^{n}\frac{A_j(x,t)\ \e^{2\i\theta(x,t;k_0)}}{(k-k_0)^j}\right]f_{\r}(x,t;k)\e^{-\i\theta(x,t;k)}
\end{equation}
is regular in a vicinity of the point $k=k_0$ provided that 
\begin{equation}\label{A_j}A_{j}(x,t)=\sum\limits_{m=j}^{n}T_{m}(x,t)\frac{\mu_{m-j}}{(m-j)!}=\sum\limits_{m=0}^{n-j}T_{m+j}(x,t)\frac{\mu_{m}}{m!},\quad j=1,\ldots,n.\end{equation}
\end{lemma}
\begin{remark}
 It is remarkable that the coefficients $A_j$ are determined only by $\mu_j,$ $j=0,\ldots,n-1,$ and not by $f_{\r}$ and its derivatives at $k=k_0.$
\end{remark}
\begin{remark}
 Definition of $A_j, j=1,\ldots,n,$ requires up to the $2n-1^{st}$ derivative of $a(.)$ at the point $k=k_0,$ i.e. $a^{(n)}(k_0),\ldots,a^{(2n-1)}(k_0).$
\end{remark}

\begin{proof}
 Regularity of the expression \eqref{regular_expression} is equivalent to the regularity of 
\begin{equation}\label{regular_expression_2}
\frac{1}{a(k)}f_{\l}(x,t;k)\e^{2\i\theta(x,t;k)-2\i\theta(x,t;k_0)}-\left[\sum\limits_{j=1}^{n}\frac{A_j(x,t)}{(k-k_0)^j}\right]f_{\r}(x,t;k).
\end{equation}
Let's first treat the left summand in \eqref{regular_expression_2}. We have $$f_{\l}=\sum\limits_{m=0}^{n-1}\frac{1}{m!}f_{\l}^{(m)}(k_0)(k-k_0)^m+\mathcal{O}\((k-k_0)^n\),$$
and substituting in the above formula the expressions \eqref{f_l_in_f_r} of $f_{\l}^{(m)}$ in terms of $f_{\r}^{(j)},$ we obtain
$$f_{\l}=\sum\limits_{m=0}^{n-1}\sum\limits_{j=0}^{m}\frac{\mu_{m-j} f_{\r}^{(j)}(k_0)}{j!(m-j)!}(k-k_0)^m+\mathcal{O}\((k-k_0)^n\).$$
Multiplying the above formula by \eqref{a_decomposition}, we obtain
\begin{equation}\label{first_summand}
\frac{\e^{2\i\theta(x,t;k)}}{a(k)\e^{2\i\theta(x,t;k_0)}}f_{\l}(x,t;k)=\sum\limits_{p=1}^{n}\sum\limits_{m=0}^{n-p}T_{m+p}(x,t)\left[\sum\limits_{j=0}^{m}\frac{\mu_{m-j}f_{\r}^{(j)}(k_0)}{j!(m-j)!}\right](k-k_0)^{-p}+\mathcal{O}(1).\end{equation}
The second summand in \eqref{regular_expression_2} has the following decomposition in the Taylor series:
\begin{equation}\label{second_summand}\left[\sum\limits_{j=1}^{n}\frac{A_j(x,t)}{(k-k_0)^j}\right]f_{\r}(x,t;k) = \sum\limits_{p=1}^{n}\sum\limits_{j=p}^{n}\frac{A_j(x,t)f_{\r}^{(j-p)}(k_0)}{(j-p)!}(k-k_0)^{-p}+\mathcal{O}(1).\end{equation}
Equating terms $(k-k_0)^{-p},$ $p=1,\ldots,n,$ in \eqref{first_summand} and \eqref{second_summand}, we come to the system of equations for $A_j:$
\begin{equation}\label{equations_for_A_j}\sum\limits_{j=p}^{n}\frac{A_jf_{\r}^{(j-p)}(k_0)}{(j-p)!}=\sum\limits_{m=0}^{n-p}T_{m+p}\sum\limits_{j=0}^{m}\frac{\mu_{m-j}f_{\r}^{(j)}(k_0)}{j!(m-j)!},\quad p=1,\ldots,n.\end{equation}
The system \eqref{equations_for_A_j} allows   us to determine $A_j, j=1,\ldots,n,$ one by one, starting from $A_n.$
However, it is not clear at this point that $A_j$ do not depend on $f_{\r}$ and its derivatives. In order to solve the system \eqref{equations_for_A_j}, let us rearrange the terms in \eqref{equations_for_A_j}:
\begin{equation}\label{equations_for_A_j_2}
\sum\limits_{j=0}^{n-p}\frac{A_{j+p}f_{\r}^{(j)}(k_0)}{j!}=\sum\limits_{j=0}^{n-p}\left[\sum\limits_{m=j}^{n-p}T_{m+p}\frac{\mu_{m-j}}{(m-j)!}\right]\frac{f_{\r}^{(j)}(k_0)}{j!},\quad p=1,\ldots,n.
\end{equation}
The above expression \eqref{equations_for_A_j_2} is satisfied provided that 
\begin{equation}\label{A_j+p}A_{j+p}=\sum\limits_{m=j}^{n-p}T_{m+p}\frac{\mu_{m-j}}{(m-j)!},\quad p=1,\ldots,n,\quad j=0,\ldots,n-p,\end{equation}
however, \eqref{A_j+p} is not yet the definition of $A_{j+p},$ since the r.-h.-s. might depend on different choices for $j,p$ with $j+p.$ However, reorganizing terms in \eqref{A_j+p}, we obtain
$$A_j(x,t)=\sum\limits_{m=j}^{n}T_{m}(x,t)\frac{\mu_{m-j}}{(m-j)!},\quad j=1,\ldots,n.$$ 
\end{proof}

\subsection{Handling pole conditions}
The analysis of  RH problems with  simple poles  was first introduced  in \cite{GT09}  for KdV  (see also \cite{M15} for the CH equation).

%
%%%%%%%%%%%%%%%% BEGIN COMMENT DOUBLE TRIPLE MULTIPLE
% If we want to change the triangularity of the jump matrix around multipole $\kappa_j,$
% $$J=\begin{pmatrix}
%      1&0\\\sum\limits_{l=1}^{n_j}\frac{A_l}{(k-\kappa_j)^l} & 1
%     \end{pmatrix},
% $$
% we define
% $$M^{(1)}=\begin{cases}
%                      \begin{pmatrix}
%                       1 & \frac{-1}{\sum\limits_{l=1}^{n_j}\frac{A_l}{(k-\kappa_j)^l}} \\ \sum\limits_{l=1}^{n_j}\frac{A_l}{(k-\kappa_j)^l} & 0
%                      \end{pmatrix}
% \begin{pmatrix}
%  \frac{(k-\kappa_j)^{n_j}}{(k-\ol{\kappa_j})^{n_j}} & 0 \\ 0 & \frac{(k-\ol{\kappa_j})^{n_j}}{(k-\kappa_j)^{n_j}}
% \end{pmatrix},\quad \textrm{inside } C_j,
% \\
% \begin{pmatrix}\frac{(k-\kappa_j)^{n_j}}{(k-\ol{\kappa_j})^{n_j}} & 0 \\ 0 & \frac{(k-\ol{\kappa_j})^{n_j}}{(k-\kappa_j)^{n_j}}
% \end{pmatrix},\quad \textrm{outside} C_j.
% \end{cases}
% $$
% First of all, $M^{(1)}$ is regular at the point $\kappa_j.$
% Second, the jump for $M^{(1)}$ has other triangularity:
% $$J^{(1)}=\begin{pmatrix}
%            1& \frac{1}{\sum\limits_{l=1}^{n_j}\frac{A_l}{(k-\kappa_j)^l}} \cdot \(\frac{k-\ol\kappa_j}{k-\kappa_j}\)^{2n_j} \\ 0 & 1
%           \end{pmatrix}, \quad k\in C_j.
% $$
%%%%%%%%%%%%%%%% END COMMENT DOUBLE TRIPLE MULTIPLE
%

Here we explain how to proceed, if our spectral problem has poles of multiple order.
First of all, the pole condition 3 of the RH problem \ref{RH_problem_1} will change to the following conditions:
for $\kappa_j, \Re \kappa_j>0, \Im \kappa_j>0,$   let  $M_1, M_2$ be  the first and the second columns of $M.$ Then the following pole conditions hold:
\begin{equation*}
\begin{split}
&
M_1(k) - \left[\sum\limits_{l=1}^{n_j}\frac{\nu_j^{(l)}\ e^{2\i \theta(x,t;k)}}{(k-\kappa_j)^l}\right]M_2(k) = \mathcal{O}(1) , k\to \kappa_j,\\
&
M_1(k) - \left[\sum\limits_{l=1}^{n_j}\frac{\overline{\nu_j^{(l)}}\ e^{2\i \theta(x,t;k)}}{(-1)^l(k+\ol{\kappa_j})^l}\right]M_2(k) = \mathcal{O}(1) , k\to -\ol{\kappa_j},
\\
&
M_2(k) + \left[\sum\limits_{l=1}^{n_j}\frac{\ol{\nu_j^{(l)}}\ e^{-2\i \theta(x,t;k)}}{(k-\ol{\kappa_j})^l}\right]M_1(k) = \mathcal{O}(1) , k\to \ol{\kappa_j},\\
&
M_2(k) + \left[\sum\limits_{l=1}^{n_j}\frac{{\nu_j^{(l)}}\ e^{-2\i \theta(x,t;k)}}{(-1)^{l}(k+{\kappa_j})^l}\right]M_1(k) = \mathcal{O}(1) , k\to -{\kappa_j},
\\
&
M_2(k)=\mathcal{O}(1), k\to \kappa_j, \mbox { and } k \to -\ol{\kappa_j},\qquad  M_1(k)=\mathcal{O}(1), k\to \ol{\kappa_j} \mbox{ and } k\to -\kappa_j. \ \ 
\end{split}
\end{equation*}

\subsubsection{Applications: Solitons and breathers of multiple order}
They are generated by the following meromorphic RH problem.

\begin{RHP}\label{RHP_for_multi_solitons_breathers}
\begin{enumerate}Find a meromorphic $2\times2$ matrix-valued function $M(x,t;k)$ such that
\item $M$ is meromorphic in $\mathbb{C}$ with poles at $\kappa,$ $\ol{\kappa}$, $-\kappa,$ $-\ol{\kappa},$
for some  $\Re\kappa\geq0,$ $\Im\kappa>0;$
\item pole conditions, upper half-plane:
$$M_{1}(k)-\left[\sum\limits_{j=1}^{n}\frac{A_j\ \e^{2\i \theta(\kappa)}}{(k-\kappa)^j}\right]M_2(k)=\mathcal{O}(1),\quad k\to \kappa,$$
$$M_{1}(k)-\left[\sum\limits_{j=1}^{n}\frac{(-1)^j\ \ol{A_j}\ \e^{2\i \theta(-\ol\kappa)} }{\(k+\ol{\kappa}\)^j}\right]M_2(k)=\mathcal{O}(1),\quad k\to -\ol{\kappa},$$
lower half-plane:
$$M_{2}(k)-\left[\sum\limits_{j=1}^{n}\frac{-\ol{A_j}\ \e^{-2\i \theta(\ol\kappa)} }{\(k-\ol{\kappa}\)^j}\right]M_1(k)=\mathcal{O}(1),\quad k\to \ol{\kappa},$$
$$M_{2}(k)-\left[\sum\limits_{j=1}^{n}\frac{(-1)^{j-1}A_j \ \e^{2\i \theta(\kappa)}}{(k+\kappa)^j}\right]M_1(k)=\mathcal{O}(1),\quad k\to -\kappa;$$

\item asymptotics: $M(x,t;k)\to I$ as $k\to\infty.$
\end{enumerate}
\end{RHP}
Here $A_j = A_j(x,t)$ are as in \eqref{a_decomposition}, \eqref{A_j}, i.e.
\begin{itemize}
 \item simple pole $n=1,\ a(\kappa)=0,\ \dot{a}(\kappa)\neq0.$ In this case we have 
$$T_1=\frac{\e^{2\i\theta(\kappa)}}{\dot{a}(\kappa)},\quad A_1=\mu_0T_1\,.$$
Here $\theta(k)=4k^3t+xk,$ and the dependence on $x,t$ comes only from $\theta(k).$

\item double-pole $n=2, \ a(\kappa)=0,\ \dot{a}(\kappa)=0,\ \ddot{a}(\kappa)\neq0.$ In this case we have 
$$T_2=\frac{2\e^{2\i\theta(\kappa)}}{\ddot{a}(\kappa)},\quad T_1=\frac{2\i\e^{2\i\theta(\kappa)}\left[6\ddot{a}(\kappa)\(x+12\kappa^2t\)+\i\dddot{a}(\kappa)\right]}{3\(\ddot{a}(\kappa)\)^2},$$
and $ A_2=T_2\mu_0,\quad A_1=\mu_0T_1+\mu_1T_2,$
\item triple-pole $n=3, \ a(\kappa)=0,\ \dot{a}(\kappa)=0,\ \ddot{a}(\kappa)=0, \dddot{a}(\kappa)\neq0.$ In this case we have 
$$T_3=\frac{6\e^{2\i\theta(\kappa)}}{\dddot{a}(\kappa)},\quad T_2=\frac{24\ \i\ \dddot{a}(\kappa)\(x+12\kappa^2t\)-3a^{(4)}(\kappa)}{2\(\dddot{a}(\kappa)\)^2}\e^{2\i\theta(\kappa)},$$
\begin{multline*}
T_1\e^{-2\i\theta(\kappa)}=\frac{-12\left[(x+12\kappa^2t)^2-12\i\kappa t\right]}{\dddot{a}(\kappa)}-\\
-\frac{3\i a^{(4)}(\kappa)\(x+12\kappa^2t\)}{\({\dddot{a}(\kappa)}\)^2}
+\frac{15\( a^{(4)}(\kappa) \)^2 -12a^{(3)}(\kappa)a^{(5)}(\kappa)}{40\({\dddot{a}(\kappa)}\)^3},
\end{multline*}
$$ A_3=T_3\mu_0,\quad A_2=\mu_0T_2+\mu_1T_3,\quad A_1=\mu_0T_1+\mu_1T_2+\frac{\mu_2}{2}T_3.$$
\end{itemize}

Taking into account symmetry \eqref{Symmetries}, the solution can be found in the form
$$M(k)=\begin{pmatrix}
        1+\sum\limits_{j=1}^n\frac{\alpha_j(x,t)}{(k-\kappa)^j} + \sum\limits_{j=1}^n\frac{(-1)^j\ol{\alpha_j(x,t)}}{(k+\ol\kappa)^j} & 
\sum\limits_{j=1}^n\frac{-\ol{\beta_j(x,t)}}{(k-\ol\kappa)^j} + \sum\limits_{j=1}^n\frac{(-1)^{j-1}\ \beta_j(x,t)}{(k+\kappa)^j}
 \\
\sum\limits_{j=1}^n\frac{\beta_j(x,t)}{(k-\kappa)^j} + \sum\limits_{j=1}^n\frac{(-1)^j\ol{\beta_j(x,t)}}{(k+\ol\kappa)^j}
 & 1+ \sum\limits_{j=1}^n\frac{\ol{\alpha_j(x,t)}}{(k-\ol\kappa)^j} + \sum\limits_{j=1}^n\frac{(-1)^j\ \alpha_j(x,t)}{(k+\kappa)^j}
       \end{pmatrix},
$$
(If $\kappa=-\ol{\kappa}\neq0,$ then we make a straightforward reduction of the above expression).
Then we find the solution of the MKdV by the formula
\begin{equation*}q(x,t)=2\i\(\beta_1(x,t)-\ol{\beta_1(x,t)}\)=-4\Im\beta_1(x,t).\end{equation*}

If $\Im\kappa=0,\ \Re{\kappa>0}$, then we have a multiple soliton, if $\Im \kappa>0,\ \Re\kappa>0,$ then we have a multiple breather.
As was pointed out in \cite{WO82}, double-pole soliton can be obtained as a limit of a breather with $\Re\kappa\to0.$

\end{appendices}
\vskip 1cm
\noindent{\bf Acknowledgment.}
This manuscript was partially developed while T.G. was Chair Morlet at the  Centre International de Rencontres Mathematiques, (CIRM), Luminy France.
T.G. thanks CIRM for  the warm hospitality.  T.G. wish to thank   GNFM (Gruppo Nazionale di Fisica matematica) and INDAM.

\end{document}